%% file: RSA-20260423-for-arxiv.tex
\PassOptionsToPackage{unicode}{hyperref}
\PassOptionsToPackage{hyphens}{url}
\PassOptionsToPackage{dvipsnames,svgnames,x11names}{xcolor}
\documentclass[12pt]{article}

\usepackage{easyReview}
\usepackage{appendix}

\usepackage{amsmath,amssymb,bbm, amsthm, commath}
\numberwithin{equation}{section}
\theoremstyle{plain}
\newtheorem{thm}{Theorem}[section]

\newtheorem{lemma}[thm]{Lemma}
\theoremstyle{definition}

\newtheorem{remark}{Remark}[section]

\usepackage{enumitem}
\usepackage[flushleft]{threeparttable}
\usepackage{multirow,float,setspace,stfloats}
\usepackage{standalone}
\usepackage{tikz} 
\usetikzlibrary{fit,positioning,shadows.blur}
\usepackage{iftex}
\ifPDFTeX
  \usepackage[T1]{fontenc}
  \usepackage[utf8]{inputenc}
  \usepackage{textcomp} 
\else 
  \usepackage{unicode-math}
  \defaultfontfeatures{Scale=MatchLowercase}
  \defaultfontfeatures[\rmfamily]{Ligatures=TeX,Scale=1}
\fi
\usepackage{lmodern}
\ifPDFTeX\else  
\fi
\IfFileExists{upquote.sty}{\usepackage{upquote}}{}
\IfFileExists{microtype.sty}{
  \usepackage[]{microtype}
  \UseMicrotypeSet[protrusion]{basicmath} 
}{}
\makeatletter
\@ifundefined{KOMAClassName}{
  \IfFileExists{parskip.sty}{%
    \usepackage{parskip}
  }{
    \setlength{\parindent}{0pt}
    \setlength{\parskip}{6pt plus 2pt minus 1pt}}
}{
  \KOMAoptions{parskip=half}}
\makeatother
\usepackage{xcolor}
\makeatletter
\ifx\paragraph\undefined\else
  \let\oldparagraph\paragraph
  \renewcommand{\paragraph}{
    \@ifstar
      \xxxParagraphStar
      \xxxParagraphNoStar
  }
  \newcommand{\xxxParagraphStar}[1]{\oldparagraph*{#1}\mbox{}}
  \newcommand{\xxxParagraphNoStar}[1]{\oldparagraph{#1}\mbox{}}
\fi
\ifx\subparagraph\undefined\else
  \let\oldsubparagraph\subparagraph
  \renewcommand{\subparagraph}{
    \@ifstar
      \xxxSubParagraphStar
      \xxxSubParagraphNoStar
  }
  \newcommand{\xxxSubParagraphStar}[1]{\oldsubparagraph*{#1}\mbox{}}
  \newcommand{\xxxSubParagraphNoStar}[1]{\oldsubparagraph{#1}\mbox{}}
\fi
\makeatother

\usepackage{longtable,booktabs,array}
\usepackage{calc} 
\usepackage{etoolbox}
\makeatletter
\patchcmd\longtable{\par}{\if@noskipsec\mbox{}\fi\par}{}{}
\makeatother
\IfFileExists{footnotehyper.sty}{\usepackage{footnotehyper}}{\usepackage{footnote}}
\makesavenoteenv{longtable}
\usepackage{graphicx}
\makeatletter
\def\maxwidth{\ifdim\Gin@nat@width>\linewidth\linewidth\else\Gin@nat@width\fi}
\def\maxheight{\ifdim\Gin@nat@height>\textheight\textheight\else\Gin@nat@height\fi}
\makeatother
\setkeys{Gin}{width=\maxwidth,height=\maxheight,keepaspectratio}
\makeatletter
\def\fps@figure{htbp}
\makeatother

\makeatletter
\@ifpackageloaded{caption}{}{\usepackage{caption}}
\AtBeginDocument{%
\ifdefined\contentsname
  \renewcommand*\contentsname{Table of contents}
\else
  \newcommand\contentsname{Table of contents}
\fi
\ifdefined\listfigurename
  \renewcommand*\listfigurename{List of Figures}
\else
  \newcommand\listfigurename{List of Figures}
\fi
\ifdefined\listtablename
  \renewcommand*\listtablename{List of Tables}
\else
  \newcommand\listtablename{List of Tables}
\fi
\ifdefined\figurename
  \renewcommand*\figurename{Figure}
\else
  \newcommand\figurename{Figure}
\fi
\ifdefined\tablename
  \renewcommand*\tablename{Table}
\else
  \newcommand\tablename{Table}
\fi
}
\@ifpackageloaded{float}{}{\usepackage{float}}
\floatstyle{ruled}
\@ifundefined{c@chapter}{\newfloat{codelisting}{h}{lop}}{\newfloat{codelisting}{h}{lop}[chapter]}
\floatname{codelisting}{Listing}

\makeatother
\makeatletter
\@ifpackageloaded{caption}{}{\usepackage{caption}}
\@ifpackageloaded{subcaption}{}{\usepackage{subcaption}}
\makeatother

\ifLuaTeX
  \usepackage{selnolig}  
\fi
\usepackage[]{natbib}
\usepackage{bookmark}

\IfFileExists{xurl.sty}{\usepackage{xurl}}{} 
\hypersetup{
  pdftitle={Title},
  pdfauthor={Author 1; Author 2},
  pdfkeywords={3 to 6 keywords, that do not appear in the title},
  colorlinks=true,
  linkcolor={blue},
  filecolor={Maroon},
  citecolor={Blue},
  urlcolor={Blue},
  pdfcreator={LaTeX via pandoc}}

\newcommand{\anon}{1}

\begin{document}

\def\spacingset#1{\renewcommand{\baselinestretch}%
{#1}\small\normalsize} \spacingset{1}


\if1\anon
{
  \title{\bf Random Subset Averaging\thanks{We thank Yuhong Yang for useful comments. We also benefit a lot from the discussion with Jiandong Wang. Wenhao Cui gratefully acknowledges financial support from the National Natural Science Foundation of China (NSFC, 72473006, 72103014). Jie Hu's research is supported by the Postdoctoral Fellowship Program and China Postdoctoral Science Foundation under Grant Number BX20240182. }}
  \author{Wenhao Cui\\
    School of Economics and Management, Beihang University \\
    Jie Hu \thanks{Correspondence should be addressed to Jie Hu: $<$\href{mailto:}{hujie\_86@163.com}$>$. \\
    }\\
    Yau Mathematical Sciences Center, Tsinghua University \\
	Haitao Zheng\\
    School of Economics and Management, Beihang University \\ }
  \maketitle
} \fi

\if0\anon
{
  \bigskip
  \bigskip
  \bigskip
  \begin{center}
    {\LARGE\bf Random Subset Averaging}
\end{center}
  \medskip
} \fi

\bigskip
\begin{abstract}
We propose a new ensemble prediction method, Random Subset Averaging (RSA), tailored for settings with many correlated covariates, including extreme regimes in which the number of predictors far exceeds the sample size and the covariates exhibit strong dependence. RSA constructs candidate models via a binomial random subset strategy and aggregates their predictions through a two-round weighting scheme, yielding a hierarchical aggregation structure that separates model-fit and subset-construction uncertainty. All tuning parameters are selected via cross-validation, requiring no prior knowledge of covariate relevance. We establish the asymptotic optimality of RSA under mild rate conditions, allowing for data-dependent first-round weights. Under orthogonal designs, RSA incurs no asymptotic approximation loss relative to flat Mallows averaging while substantially relaxing the associated rate conditions, and achieves a lower finite-sample risk bound than both nested Mallows averaging and random subset regression. Simulation studies demonstrate that RSA consistently delivers accurate and stable predictive performance across a wide range of sample sizes, dimensional settings, sparsity levels and correlation structures, outperforming conventional model selection and ensemble learning methods. An empirical application to financial return forecasting further illustrates its practical utility. 
\end{abstract}

\noindent%
\vfill

\newpage
\spacingset{1.8} 

\section{Introduction}

Forecasting with many correlated covariates arises in a wide range of applications, including macroeconomic forecasting, asset pricing, genomics, and climate studies. In such settings, the number of predictors often grows with the sample size and may exceed it, while substantial dependence among predictors is common and may intensify with dimensionality. High dimensionality amplifies estimation error, rendering it infeasible to include all predictors in a single model, whereas strong dependence induces near-singularity in the design matrix, complicating both variable selection and parameter estimation. 

Importantly, correlation among predictors is not inherently detrimental to forecasting. Redundant covariates may contain overlapping information and can serve as substitutes for one another, potentially stabilizing predictions when properly exploited. The central challenge is therefore not to eliminate dependence, but to utilize it without inducing instability. 

A natural approach is model selection, implemented via information criteria such as AIC or BIC \citep{akaike1974new,schwarz1978estimating}, or regularization methods such as Lasso, SCAD, and MCP \citep{tibshirani1996regression,fan2001variable,zhang2010nearly}. These methods aim to identify a single predictive model by balancing fit and complexity. However, since the selected model arises from a data-dependent selection procedure, such approaches inherently ignore model uncertainty \citep{yuan2005combining,nan2014variable}. In the presence of correlated predictors or weak signals, many models often exhibit similar in-sample performance, so that small perturbations in the data or tuning parameters can lead to substantially different selected models and unstable forecasts. 

Ensemble methods mitigate this instability by aggregating predictors across multiple models. Examples include feature bagging \citep{ho1998random}, random forests \citep{breiman2001random}, and random subspace regression \citep{elliott2013complete,boot2019forecasting}. While these methods can reduce variance, their performance depends critically on how the candidate models are constructed. Fixed subset sizes or ad hoc randomization may introduce bias when model complexity is unknown, and equal-weight aggregation ignores heterogeneity in predictive performance, leading to inefficiency \citep{liang2011optimal,zhang2016dominance}.

Model averaging provides a more principled alternative by assigning optimal convex weights to candidate models under explicit risk criteria \citep{hansen2007least,wan2010least}. Although theoretically appealing and effective in dense-signal environments \citep{peng2022improvability}, existing model averaging methods typically rely on structured candidate sets, such as nested models, and often require prior knowledge of variable ordering. In high-dimensional settings with complex dependence, such structure is rarely available, and data-driven construction of candidate models, via marginal screening \citep{ando2014model} or variable selection paths \citep{zhang2019parsimonious}, remains inherently unstable. 

A key observation underlying this paper is that forecast instability arises from two distinct sources of uncertainty: (i) model-fit uncertainty within a given subset of predictors, and (ii) uncertainty induced by the construction of the subset itself. Existing methods implicitly conflate these two sources, treating model uncertainty as a single-layer problem. 

This paper proposes Random Subset Averaging (RSA), a hierarchical aggregation framework that explicitly separates these two sources of uncertainties. RSA constructs candidate models via binomial random subsets, allowing model sizes to vary and avoiding reliance on prespecified ordering or screening. Predictions are then aggregated through a two-round convex weighting scheme: the first round addresses model-fit uncertainty within subsets, while the second accounts for randomness in subset construction. This decomposition yields a multi-layer aggregation that stabilizes prediction while retaining the informational benefits of correlated predictors.

From a theoretical perspective, we show that RSA achieves asymptotic optimality under mild conditions, even when the first-round weights are data-dependent. This result extends classical model averaging theory from single-layer aggregation to a broader class of hierarchical aggregation procedures. Under orthogonal designs, RSA incurs no asymptotic approximation loss relative to flat Mallows averaging, while substantially relaxing the rate conditions required for optimal weighting. It also attains strictly lower finite-sample risk than nested model averaging and random subspace regression. 

From a statistical perspective, the hierarchical structure of RSA transforms a high-dimensional aggregation problem into a sequence of lower-dimensional subproblems, thereby enlarging the admissible class of candidate models without sacrificing optimality. This is particularly important in high-dimensional settings, where stringent rate conditions on the number of candidate models restrict the effective model class and limit the ability to account for model instability. The hierarchical structure of RSA relaxes these constraints, enabling substantially richer candidate constructions and improving stability while preserving optimality. Simulation and empirical results further demonstrate that RSA performs comparably to existing methods under weak correlation, while delivering substantial gains in predictive accuracy and stability as dependence strengthens.

The remainder of the paper is organized as follows. Section 2 introduces the RSA estimator and develops its theoretical properties. Section 3 reports simulation evidence, Section 4 presents the empirical application, and Section 5 concludes.

\section{Random Subset Averaging and Its Properties}\label{sec2}

Section \ref{sec2.1} introduces the RSA estimator as a hierarchical aggregation scheme for ensemble prediction. Section \ref{sec2.2} establishes its asymptotic optimality under mild conditions, extending classical model averaging theory to multi-layer aggregation. Section \ref{sec2.3} provides a risk comparison in an orthogonal setting to highlight its statistical advantages relative to existing methods.

\subsection{The Random Subset Averaging Estimator}\label{sec2.1}

We consider a homoscedastic linear regression model throughout this study:
\begin{align}
	y_i = \mu_i + e_i,\text{ where } i=1,2,\ldots,N, \text{ and } 
	\mu_i = \sum_{j = 1}^{K}\beta_j x_{ij} = x_i^\top \beta,\label{eq2.1}
\end{align}
where $x_i = (x_{i1},\ldots,x_{iK})^\top $ and $\beta = (\beta_1,\ldots,\beta_K)^\top $ are both $K$-dimensional column vectors. The number of regressors $K$ are allowed to grow with the sample size $N$ and can even exceed it. The response variable $y_i$ is real-valued, and we assume $E(e_i \mid x_i) = 0$ and $E(e_i^2 \mid x_i) = \sigma^2$. The model in Eq. \eqref{eq2.1} encompasses both classical and high-dimensional regimes and can be written in matrix form as $Y = X\beta + e$, where $Y$ is an $N \times 1$ response vector, $X$ is an $N \times K$ design matrix, and the error term $e$ satisfies $E(e \mid X) = 0$ and $E(ee^\top \mid X) = \sigma^2 I_N$.

RSA defines a hierarchical aggregation scheme that combines random subset construction with a two-round convex weighting procedure. Its key feature is a two-layer architecture (Figure \ref{fig:RSA}) that decomposes model uncertainty into within-subset and between-subset components. In the first layer, predictions are constructed from models based on randomly selected subsets of covariates. These predictions are aggregated within each subgroup via convex weighting, thereby forming the second-layer elements. A second round of convex weighting is then applied to these elements to produce the final estimator. This two-round construction stabilizes prediction while preserving flexibility in model construction.

\begin{figure}[htbp]
	\centering
	\includestandalone[width=\textwidth]{plot-of-two-layer-RSA-v2}
	\caption{Graphical illustration of RSA.}
	\label{fig:RSA}
\end{figure}

\underline{Design of the first layer.}
To fix ideas, let $R = diag(r_1, \dots, r_K)$ be a random selection matrix, where each $r_j$ is an independent Bernoulli random variable with selection probability $p_j \in [0,1]$. The resulting candidate model is given by $Y = XR\beta_{R}+ U,$ where $XR$ represents a randomly selected subset of the original $K$ covariates. The selection probabilities $p_j$ can, in principle, vary across covariates to reflect prior beliefs or empirical relevance. However, determining optimal covariate-specific probabilities is nontrivial and left for future work. For simplicity, we adopt a common selection probability $p$ in both simulations and empirical analyses, yielding an expected subset size of $Kp$, which can be substantially smaller than $N$ and thus serves as a form of dimension reduction in high-dimensional settings.

Relying on a single selection matrix $R$ may lead to unstable and potentially misspecified predictions, particularly when important covariates are omitted. To mitigate this issue, we generate independent selection matrices $R_m^{(\ell)} $ for $m = 1, \dots, M$ and $\ell = 1, \dots, L$, thereby forming $L$ groups of $M$ candidate models. The corresponding design matrices $XR_m^{(\ell)}$ are non-nested and heterogeneous in dimension, which helps reduce the risk of systematic misspecification. For each design matrix $XR_{m}^{(\ell)}$, the first-layer prediction is defined as the best linear prediction of $Y$, given by
\begin{equation} \label{eq:muhatj}
	\hat{\mu}_{m}^{(\ell)} = XR_{m}^{(\ell)}\hat{\beta}_{R_{m}^{(\ell)}} = XR_{m}^{(\ell)}(R_{m}^{(\ell)}X^\top  XR_{m}^{(\ell)})^{-}R_{m}^{(\ell)}X^\top  Y,
\end{equation}
where $A^{-}$ denotes the Moore–Penrose generalized inverse of matrix $A$.

\underline{Design of the second layer.} Within each group $\ell$, the predictions are aggregated via a convex combination:
\begin{equation}\label{eq:muhat}
	\hat{\mu}^{(\ell)} = \sum_{m=1}^{M} \hat{w}_{m}^{(\ell)}\hat{\mu}_{m}^{(\ell)},
\end{equation}
where $\hat{w}_{m}^{(\ell)}$ for $m = 1, \dots, M$ lies in the simplex. The weights may be specified deterministically, for example as $\hat{w}_{m}^{(\ell)} = \frac{1}{M}$, corresponding to naive averaging in ensemble methods, or in data-driven manner to improve performance. A common data-driven choice is the Mallows weights, defined as the solution to the following optimization problem:
\begin{equation} \label{eq:whatj}
	(\hat{w}_1^{(\ell)}, \dots, \hat{w}_M^{(\ell)}) = \arg \min_{(w_1, \dots, w_M)\in H_M} \left\Vert Y - \sum_{m=1}^{M}w_{m}\hat{\mu}_{m}^{(\ell)}\right\Vert^2 + 2\sigma^2\sum_{m=1}^{M} w_{m} k_m^{(\ell)},
\end{equation}
with $H_M \equiv\left\{w\in[0,1]^M:\sum_{m = 1}^{M}w_{m} = 1\right\}$. Here, $k_m^{(\ell)}$ denotes the number of selected covariates in model $m$ of group $\ell$, which equals the trace of the projection matrix $P_{XR_m^{(\ell)}}$, with $P_A = A(A^\top A)^{-}A^\top $ for any matrix $A$.

\underline{Design of the output layer.} To further reduce variability induced by random subset construction, RSA performs a second-round aggregation across groups. Given the $L$ group-level predictions $\hat{\mu}^{(\ell)}$ from the second layer, the final RSA estimator is formed as a convex combination:
\begin{equation} \label{eq:RSAprediction}
	\hat{\mu}_{RSA} = \sum_{\ell=1}^{L} \hat{\mathbbm{w}}_{\ell}\hat{\mu}^{(\ell)},
\end{equation} 
where the weights are obtained by minimizing a Mallows criterion:
\begin{equation} \label{eq:whatell}
	(\hat{\mathbbm{w}}_{1}, \dots, \hat{\mathbbm{w}}_{L}) = \arg \min_{(\mathbbm{w}_1, \dots, \mathbbm{w}_L) \in H_L} \left\Vert Y - \sum_{\ell=1}^{L}\mathbbm{w}_{\ell}\hat{\mu}^{(\ell)} \right\Vert^2 + 2\sigma^2\sum_{\ell=1}^{L} \mathbbm{w}_{\ell} k^{(\ell)},
\end{equation}
with the feasible set defined as $H_L \equiv\left\{\mathbbm{w}\in[0,1]^L:\sum_{\ell = 1}^{L}\mathbbm{w}_{\ell} = 1\right\}$. Here, $k^{(\ell)} = \sum_{m = 1}^{M}\hat{w}_m^{(\ell)}k_m^{(\ell)}$ denotes the effective model dimension associated with $\hat{\mu}^{(\ell)}$. 

This two-layer construction is not merely algorithmic, but reflects a decomposition of model uncertainty into within-subset and between-subset components. The second-round aggregation reduces variability and improves stability, particularly in settings where individual submodels are weak or unstable.

\begin{remark}[Implicit Decorrelation via Random Subset Sampling]
	Random subset sampling substantially reduces the probability that highly correlated covariates are selected jointly. Consider an AR(1) dependence structure with correlation $\rho^{|i-j|}$. For a given threshold $t \in (0,1)$ and selection probability $p$, the probability that the maximal pairwise correlation among the selected covariates exceeds $t$ is of order $K p^2 \frac{\log t}{\log \rho}$. By choosing $p$ sufficiently small so that this quantity remains negligible, the selected covariates are approximately weakly dependent with high probability. In this sense, random subset sampling acts as an implicit decorrelation mechanism: even if the original predictors exhibit strong dependence, the induced submodels tend to operate in a regime closer to orthogonality. Moreover, stronger dependence (larger $\rho$) increases clustering among predictors, so maintaining a low-correlation structure requires a smaller selection probability $p$. The data-driven choice of $p$ via cross-validation in our simulations is consistent with this implication, suggesting that RSA adaptively exploits dependence to stabilize prediction. 
\end{remark}

\begin{remark}
	RSA is related to several ensemble methods that introduce randomness to improve predictive stability, including dropout, stacking, and random forests. Like these methods, RSA incorporates randomness via covariate subsampling and aggregates across the resulting submodels. Unlike heuristic aggregation schemes, however, RSA employs a structured two-round convex weighting procedure, yielding a tractable and theoretically transparent framework. In particular, RSA promotes diversity through random covariate subspaces while maintaining a unified model structure and data-driven aggregation. Overall, RSA can be viewed as a structured ensemble method that integrates randomness with principled aggregation, while retaining flexibility and theoretical interpretability through its two-layer architecture.
\end{remark}

\subsection{Asymptotic Optimality of Random Subset Averaging} \label{sec2.2}

In this section, we establish the asymptotic optimality of the RSA estimator defined in Eq. (\ref{eq:RSAprediction}). A distinctive feature of RSA is its hierarchical aggregation structure, in which the final estimator depends on both within-group and between group weighting schemes. In particular, RSA permits data‑dependent weighting in the first‑round aggregation, so that the resulting estimator is not necessarily linear. While this additional flexibility can improve performance, it also introduce new challenges for establishing theoretical guarantees.

To analyze this two-layer structure, we formulate the problem in terms of a second-round weighting parameter $\mathbbm{w}$. Let the squared $L_2$ loss be defined as $\mathcal{L}_N(\mathbbm{w}) = (\hat{\mu}_{RSA}(\mathbbm{w})-\mu)^\top (\hat{\mu}_{RSA}(\mathbbm{w})-\mu)$, with associated risk, $\mathcal{R}_N(\mathbbm{w}) = E\left[\mathcal{L}_N(\mathbbm{w})|X,R\right]$, where the expectation is taken conditional on the full set of covariates $X$ and all random selection matrices $R = \{R_m^{(\ell)}: m = 1,\dots, M, \ell = 1, \dots, L\}$. We assume that the first-round weights $\hat{w}^{(\ell)} = (\hat{w}_1^{(\ell)}, \dots, \hat{w}_M^{(\ell)})^\top $ converge in probability to deterministic (data-independent) limits $w^{(\ell)} = (w_1^{(\ell)}, \dots, w_M^{(\ell)})^\top$. Under this assumption, the RSA estimator with second-round weights $ \mathbbm{w} $ can be written as: $$\hat{\mu}_{RSA}(\mathbbm{w}) = \sum_{\ell=1}^{L} \mathbbm{w}_{\ell} \hat{\mu}^{(\ell)} = \sum_{\ell=1}^{L} \mathbbm{w}_{\ell} \sum_{m=1}^{M} \hat{w}_{m}^{(\ell)} P_{XR_{m}^{(\ell)}} Y,$$
We define its asymptotic counterpart $\tilde{\mu}_{RSA}(\mathbbm{w})$ by replacing the random weights with their limits: $$\tilde{\mu}_{RSA}(\mathbbm{w}) = \sum_{\ell=1}^{L} \mathbbm{w}_{\ell} \sum_{m=1}^{M} w_{m}^{(\ell)} P_{XR_{m}^{(\ell)}} Y.$$ The corresponding loss and risk functions are:
\begin{align*}
	\tilde{\mathcal{L}}_N(\mathbbm{w}) = (\tilde{\mu}_{RSA}(\mathbbm{w})-\mu)^\top (\tilde{\mu}_{RSA}(\mathbbm{w})-\mu), \quad \tilde{\mathcal{R}}_N(\mathbbm{w}) = E\left[\tilde{\mathcal{L}}_N(\mathbbm{w})|X,R\right].
\end{align*}

The following theorem shows that, despite the hierarchical structure and the presence of data-dependent intermediate weights, RSA achieves the same asymptotic risk as the oracle convex combination over the induced model class. In this sense, RSA extends classical model averaging optimality from single-layer convex aggregation to to a broader class of multi-layer aggregation procedures.
\begin{thm}\label{theorem2.0}
	Let $\xi_{N}=\inf _{\mathbbm{w} \in H_L} \tilde{\mathcal{R}}_{N}(\mathbbm{w})$ and assume the following conditions hold:
	\begin{enumerate}
		\renewcommand{\labelenumi}{(\roman{enumi})}
		\item $E(e_i^4|x_i)< \infty$.
		\item $E[\Vert \hat{w}^{(\ell)} - w^{(\ell)} \Vert^4|X,R] = O(r_{N,M}^4)$ for all $\ell = 1,\ldots,L$.
		\item $\xi_{N}^{-1} r_{N,M}^2 LMN \to 0$, $\xi_{N}^{-1} r_{N,M} L\sqrt{MN} \to 0$, and $\xi_{N}^{-1}L^2 \to 0$.
	\end{enumerate}
	Then, we have 
	\begin{align*}
		\frac{\mathcal{L}_{N}(\hat{\mathbbm{w}})}{\inf _{\mathbbm{w} \in H_L} \mathcal{L}_{N}(\mathbbm{w})} \stackrel{p}{\rightarrow} 1.
	\end{align*} 
\end{thm}
The proof is given in Appendix.

\begin{remark}[Verifiability of Condition (ii)]\label{remark2.2}
	Condition (i) is a standard moment requirement for establishing asymptotic optimality. Condition (ii) requires the first-round weights to converge in probability to a (possibly pseudo-true) deterministic limit. For the Mallows weights defined in Eq. \eqref{eq:whatj}, which arise from a constrained least squares problem, \citet{liew1976inequality} shows that $r_{N,M} = \sqrt{M/N}$ when the limiting weight lies in the interior of $H_M$; faster rate may be achieved when the constraint binds. In the special case of deterministic weighting, such as equal weights, the convergence rate becomes zero. These results indicate that condition (ii) is mild and is satisfies by a broad class of weighting schemes commonly used in practice. 
\end{remark}

\begin{remark}[Asymptotic Optimality at Each Stage] \label{remark2.3}
	Using Mallows weights in the first-round aggregation provides a natural route to constructing effective group-level predictors. As discussed in Remark \ref{remark2.2}, the rate $r_{N,M} = \sqrt{M/N}$ represents the least favorable case. Under this rate, condition (iii) reduces to $\xi_{N}^{-1}  LM^2 \to 0$ and $\xi_{N}^{-1}L^2 \to 0$. In particular, the condition $\xi_{N}^{-1}  LM^2 \to 0$ implies that, within each group, the first-round aggregation closely approximates the best convex combination of the corresponding candidate models. This follows from the fact that $\xi_{N} \le \xi^{(\ell)}_{N}$, where $\xi^{(\ell)}_{N}$ denotes the minimal achievable risk over convex combinations of candidate models within group $\ell$. As a result, the effective complexity requirement at the group level is no stronger than that of the overall problem. Together with condition (i), this ensures that the first-round aggregation achieves the same asymptotic performance as the optimal convex combination within each group (see \citet{zhang2021new}). The second-round aggregation then performs optimal weighting over the induced class of group-level predictors. Consequently, each aggregation step is asymptotically optimal in its respective stage.
\end{remark}

\begin{remark}
	RSA can be naturally extended to a multi-round aggregation framework, in which additional layers are introduced to further refine intermediate predictors. Under suitable conditions on the convergence of intermediate weights, the resulting estimator continues to achieve asymptotic optimality over the induced model class, without incurring additional loss from the multi-layer construction. From a computational perspective, RSA is well suited to parallel implementation, as both subset construction and within-group aggregation can be carried out independently across groups. This feature helps mitigate the computational cost associated with multi-round extensions.
\end{remark}

The performance of the proposed RSA estimator depends on the tuning parameters $(p,M,L)$. The parameter $p$ controls model complexity by determining the expected subset size, approximately $Kp$, for each candidate model. The parameters $M$ and $L$ specify the number of candidate models and the number of groups, respectively. We select $(p,M,L)$ via 5-fold CV over a predefined grid, choosing the combination that minimizes the prediction error. Simulation results indicate that $p$ plays the primary role, as it directly governs the complexity of individual candidate models. In contrast, $M$ and $L$ have a relatively limited impact on performance, provided they are chosen to be sufficiently large (e.g., $L = 30$ in our simulations). Supporting evidence is provided in Section \ref{sec3.1} and the supplementary material.

\subsection{Orthogonal design: Shrinkage Structure and Risk Comparison}\label{sec2.3}

Theorem \ref{theorem2.0} establishes the optimality of the second-round procedure, while Remark \ref{remark2.3} characterizes the behavior of the first-round aggregation under Mallows weighting. While these results provide optimality guarantees at each stage, they do not directly imply that the two-round procedure is jointly optimal. A natural question is therefore whether the hierarchical aggregation in RSA can achieve the same oracle performance as flat Mallows averaging when applied to the same pool of $LM$ candidate submodels, or whether any loss is incurred due to the hierarchical structure. To address this question, we consider an orthogonal design, as is standard in the model averaging literature (e.g., \citet{peng2022improvability}, \citet{peng2024optimality}), under which covariates are mutually orthogonal. This setting eliminates cross-coordinate interactions and allows for a transparent characterization of the shrinkage structure induced by hierarchical weighting. We first compare the risk of RSA with that of flat Mallows averaging and then extend the comparison to other commonly used procedures to further illustrate the advantages of RSA.

Consider the linear model
\begin{equation*}
	y_i = \sum_{j = 1}^{K}\beta_j x_{ij} + e_i= x_i^\top \beta+ e_i, \quad i=1,2,\ldots,N,
\end{equation*}
where $X^\top X = NI_K$, $E(e_i|x_i) = 0$ and $E(e_i^2|x_i) = \sigma^2$. Assume that $|\beta_j|$ is non-increasing in $j$ and $\sum_{j=1}^{K}\beta_j^2 < \infty$.

We first characterize the minimal squared $L_2$ risk of RSA and flat Mallows as $M, L \to \infty$. Two regimes are considered: (i) uniform selection probability $p_j = p$, and (ii) covariate-specific probabilities $p_j$. Throughout, $A \simeq B$ denotes that $A = B + o(B)$ as $M, L \to \infty$. Recall $\xi_{N}$ is the minimal squared risk of RSA. 

\begin{thm}\label{theorem2.3}
	Under orthogonality,
	\begin{enumerate}
		\item If $p_j = p$ for all $j$, 
		\begin{align*}
			\xi_{N} \simeq \frac{K\sigma^2\sum_{j=1}^{K}N \beta_j^2}{\sum_{j=1}^{K}N\beta_j^2+ K\sigma^2}, \quad \xi_{N}^{flat} \simeq \frac{K\sigma^2\sum_{j=1}^{K}N \beta_j^2}{\sum_{j=1}^{K}N\beta_j^2+ K\sigma^2}
		\end{align*}
		\item If $p_j$ varies across coordinates,
		\begin{align*}
			\xi_{N} \simeq \sum_{j=1}^{K}\frac{N\beta_j^2\sigma^2}{N\beta_j^2+\sigma^2}, \quad \xi_{N}^{flat} \simeq \sum_{j=1}^{K}\frac{N\beta_j^2\sigma^2}{N\beta_j^2+\sigma^2}.
		\end{align*}
	\end{enumerate} 	
\end{thm}
The proof is given in Appendix.

Theorem \ref{theorem2.3} shows that, under orthogonality and sufficiently rich candidate construction, RSA incurs no asymptotic approximation loss relative to flat Mallows averaging. Furthermore, when the selection probabilities $p_j$ are appropriately chosen, both procedures achieve the minimal attainable risk under orthogonality, given by $\sum_{j=1}^{K}\frac{N\beta_j^2\sigma^2}{N\beta_j^2+\sigma^2}$. This risk corresponds to the optimal coordinatewise shrinkage rule obtained by minimizing  
\begin{align*}
	E \Vert \hat{\mu}(\theta) - \sum_{j=1}^{K} \beta_j x_{ij} \Vert^2 = N \sum_{j=1}^{K} \left[ (1-\theta_j)^2\beta_j^2 + \theta_j^2\frac{\sigma^2}{N} \right],
\end{align*}
over $\theta \in [0,1]^K$, where $\hat{\mu}(\theta) = \sum_{j=1}^{K} \theta_j x_{ij} \hat{\beta}_j$ and $\hat{\beta}_j$ is the marginal coefficient. The combination of this result with Theorem \ref{theorem2.0} and Remark \ref{remark2.3} shows that RSA attains the minimal achievable risk within the class of convex shrinkage estimators under orthogonality, thereby establishing joint optimality in this class.

Despite achieving the same asymptotic risk as flat Mallows averaging, RSA offers improved scalability. While flat Mallows averaging requires that $(\xi_{N}^{flat})^{-1} (LM)^2 \to 0$ (see \cite{zhang2021new}), the hierarchical structure of RSA reduces the effective dimensionality of the weight-learning problem, requiring only $(\xi_{N})^{-1} (LM^2 + L^2) \to 0$. Since $\xi_{N}^{flat} \le \xi_{N}$, RSA can accommodate a larger pool of candidate models without incurring additional asymptotic cost. In effect, RSA transforms a single ($LM$)-dimensional simplex optimization into a sequence of lower-dimensional subproblems, thereby relaxing the admissible growth rate of the candidate models relative to the sample size.

We next compare the bounds in Theorem \ref{theorem2.3} with those of nested model averaging (MA), random projection regression (RPR), and random subset regression (RSR). All these methods admit the reduced-form representation:
\begin{equation*}
	y_i = x_i^\top R\beta_R + u_i,
\end{equation*}
where $R$ is a data-driven dimension reduction matrix mapping the covariates from dimension $K$ to $P$. Different methods correspond to different constructions of $R$. In nested MA, $R$ sequentially select the first $1, 2, \dots, K$ covariates, followed by convex aggregation. In RPR, the entries of $R$ are i.i.d. $N(0,1/\sqrt{P})$. In RSR \citep{elliott2013complete,boot2019forecasting}, $R$ selects $P$ covariates uniformly at random and averages the resulting predictions equally. Under orthogonality, this unified representation permits closed-form expressions for the minimal risks.

\begin{lemma}\label{lemma2.4}
	Let $\xi_{N}^{MA}$, $\xi_{N}^{RPR}$, $\xi_{N}^{RSR}$ denote the minimal squared $L_2$ risks of MA, RPR, and RSR. Then
	\begin{align*}
		\xi_{N}^{MA} = \sigma^2 + \sum_{j=2}^{K}\frac{N\beta_j^2\sigma^2}{N\beta_j^2+\sigma^2}, \quad \xi_{N}^{RPR} = \xi_{N}^{RSR} = \frac{K\sigma^2\sum_{j=1}^{K}N \beta_j^2}{\sum_{j=1}^{K}N\beta_j^2+ K\sigma^2},
	\end{align*}
	where the minimum of $\xi_{N}^{RPR}$ and $\xi_{N}^{RSR}$ is achieved at $P = \frac{K\sum_{j = 1}^{K}N\beta_j^2}{K\sigma^2+\sum_{j = 1}^{K}N\beta_j^2}$.
\end{lemma}
The proof is given in Appendix.

If $p_j = p$, Theorem \ref{theorem2.3} and Lemma \ref{lemma2.4} imply 
$$\xi_{N} \simeq \xi_{N}^{flat} \simeq \xi_{N}^{RPR} = \xi_{N}^{RSR}.$$ 
Thus under homogeneous sampling, RSA matches the asymptotic risk of random projection and random subset methods. If $p_j$ varies across coordinates, 
$$ \xi_{N} \simeq \xi_{N}^{flat} \simeq \sum_{j=1}^{K}\frac{N\beta_j^2\sigma^2}{N\beta_j^2+\sigma^2} \leq \xi_{N}^{RPR} = \xi_{N}^{RSR} ,$$
where the inequality follows from Jensen’s inequality. Hence, coordinate-adaptive subsampling improves upon global random subspace methods. Moreover, $$\xi_{N} < \xi_{N}^{MA},$$
and the gap depends primarily on the signal strength of the leading coordinates. For fixed $\sigma^2$, smaller $|\beta_1|$ enlarges the advantage of RSA over nested MA, reflecting the benefit of adaptive shrinkage in weak-signal regimes. In summary, under orthogonality, RSA matches RPR and RSR under homogeneous sampling, strictly improves upon them under heterogeneous sampling, and dominates nested MA except in cases where the leading coordinate is overwhelmingly strong. 


Nested MA is known to achieve full asymptotic optimality, defined as attaining the asymptotic risk over all nested models \citep{peng2024optimality}. In view of Theorem \ref{theorem2.3} and Lemma \ref{lemma2.4}, RSA weakly dominates nested MA in terms of asymptotic risk under orthogonality and therefore inherits this property under the same conditions. Notably, the attainability of nested MA relies on additional structural assumptions, such as a prescribed ordering of covariates, whereas RSA does not depend on such assumptions or a nested construction, thereby offering greater flexibility. Supplementary simulations based on the design of \citet{peng2024optimality} further indicate that RSA yields finite-sample gains under weak correlation and exhibits improved robustness when covariates are strongly correlated.
Moreover, Theorem \ref{theorem2.3} implies that the oracle risk characterization of RSA continues to hold in high-dimensional regimes with $K > N$ under orthogonality, as the variable sampling probabilities can adapt to signal strength. In contrast, nested MA requires truncating the number of candidate models to be no larger than the sample size to ensure estimability, which restricts its applicability in high-dimensional settings.

\begin{remark}[RSA as Adaptive Shrinkage]
	The risk improvement of RSA can be understood from a shrinkage perspective. Under orthogonality, the OLS estimator satisfies $\hat{\beta} = \frac{1}{N} X^\top Y$, while the ridge estimator takes the form $\hat{\beta}_{\lambda}^{ridge} = \frac{N}{N+\lambda}\hat{\beta} $, where $\lambda$ is the regularization parameter. For a random selection matrix $R$, the corresponding estimator can be written as $ R\hat{\beta}_R = \left[R(R^\top X^\top XR)^{-}R^\top \right]X^\top Y $. Under homogeneous sampling, taking expectation w.r.t. $R$ yields $ E_R[R\hat{\beta}_R] = p\hat{\beta} $, which corresponds to a global ridge-type shrinkage. When the selection probabilities $p_j$ vary across coordinates, the induced shrinkage becomes coordinate-specific. In this case, $ E_R[R\hat{\beta}_R] = diag(\frac{N\beta_1^2}{\sigma^2 + N\beta_1^2}, \ldots, \frac{N\beta_K^2}{\sigma^2 + N\beta_K^2})\hat{\beta} $. Thus, RSA induces coordinatewise shrinkage with factor $ \frac{N\beta_j^2}{\sigma^2 + N\beta_j^2} $, which coincides with the oracle shrinkage intensity under orthogonality. In this sense, RSA can be interpreted as an adaptive ridge-type estimator, with shrinkage intensities governed by the signal-to-noise ratios. This perspective highlights that heterogeneous sampling enables RSA to approximate the oracle coordinatewise shrinkage, thereby improving risk efficiency.
\end{remark}

\section{Simulation Study}\label{sec3}

We conduct Monte Carlo simulations to evaluate the finite-sample performance of the proposed RSA method relative to several widely used benchmarks, including random subset regression (RSR; \citealp{elliott2013complete, boot2019forecasting}), random forests (RF; \citealp{breiman2001random}), Lasso \citep{tibshirani1996regression}, SCAD \citep{fan2001variable}, MCP \citep{zhang2010nearly}, parsimonious model averaging (PMA; \citealp{zhang2019parsimonious}), model averaging via cross-validated weights (MCV; \citealp{ando2014model}), and Mallows model averaging (MMA; \citealp{hansen2007least}). In addition, we implement Mallows averaging over candidate models generated by random subsets, which corresponds to a one-round version of RSA and is denoted by RSA.1. We denote RSA with cross-validated tuning parameters by RSA.o and RSA with fixed parameter settings by RSA.f. 

RSR and RF construct ensemble predictors by aggregating forecasts from base learners trained on randomly selected subsets of covariates. Lasso, SCAD, and MCP are standard regularization-based methods for high-dimensional regression. In contrast, MMA, PMA and MCV combine forecasts across multiple candidate models, with PMA and MCV specifically designed for high-dimensional settings. Among the latter, both methods achieve asymptotic optimality under model misspecification: PMA constructs candidate models along the solution path of the adaptive lasso, while MCV generates candidates based on marginal correlations and relaxes the sum-to-one constraint on averaging weights. Although, in principle, RSA cannot achieve lower risk than flat Mallows averaging over all $LM$ candidate models, the latter requires substantially stronger rate conditions to ensure asymptotic optimality and practical feasibility. To facilitate a meaningful and implementable comparison, we therefore include RSA.1, corresponding to the degenerate case $L = 1$ with $M$ candidate models. This specification preserves the random-subset generation mechanism while avoiding the stringent dimensionality requirements associated with flat Mallows averaging over $LM$ models.

\subsection{Simulation Setup} \label{sec3.1}

We consider the following linear data-generating process (DGP):
\begin{equation}\label{simultionmodel}
	y_i = x_i^\top\beta + e_i, i = 1, \ldots, N, 
\end{equation}
where $x_i \in \mathbb{R}^K$. The sample size is set to $N \in \{200, 400, 800\}$, and the number of covariates is defined as $K = \delta N$, with $\delta \in \{0.1, 0.5, 1, 1.5\}$. The number of nonzero coefficients in $\beta$, denoted $K^{*}$, is set to $K$ when $\delta = 0.1$ and to $0.3K$ otherwise. This design spans both low- and high-dimensional regimes with varying degrees of sparsity. 

The nonzero coefficients in $\beta$ follow two decay patterns: (i) polynomial decay, $\{j^{-0.51}:j = 1,\dots, K^{*}\}$; and (ii) exponential decay,  $\{\exp(-j^{0.25}): j = 1, \dots, K^{*}\}$. These coefficients are randomly positioned within $\beta$, with all remaining coefficients set to zero. The faster exponential decay yields smaller signal magnitudes and fewer effectively relevant covariates, thereby inducing a sparser setting relative to the polynomial decay case. Notably, the random placement of nonzero coefficients violates the implicit ordering assumption underlying nested candidate models in MMA. Moreover, when $\delta \ge 1$, the number of covariates exceeds the sample size $(K \geq N)$, rendering conventional MMA infeasible. To address this, we construct nested MMA candidate models using only the first $N - 2$ covariates under a natural ordering. 

The covariate vector $x_i$ is generated from a multivariate normal distribution with mean zero and covariance matrix $\Sigma$, where $\Sigma_{ij} = \rho^{|i-j|}$ for $i,j = 1, \dots, K$. We consider $\rho \in \{0.1, 0.9\}$ to represent low and high-correlation settings. Additional results based on random correlation matrices, capturing more complex dependence structures encountered in practice, are reported in the Supplementary Material.

We evaluate one-step-ahead forecasting performance using the Mean Squared Forecast Error (MSFE): $ MSFE = \frac{1}{N_{test}} \sum_{i=1}^{N_{test}} (\hat{y}_{i} - x_i^\top\beta)^2 $ and assess in-sample fit via the Mean Squared Error (MSE): $ MSE = \frac{1}{N_{train}} \sum_{i=1}^{N_{train}} (\hat{y}_{i} - x_i^\top\beta)^2 $. The training sample size is set to be twice that of the testing sample.

Tuning parameters are selected via cross-validation when applicable, and default choices are used otherwise. Each configuration is evaluated over 100 Monte Carlo replications. In addition, we employ the Model Confidence Set (MCS) test \citep{hansen2011model} to formally compare predictive performance across methods.

Figure \ref{fig:CVdemonstration} illustrates a representative cross-validation landscape for RSA with $N = 800$ under the polynomial decay setting. Darker blue blocks indicate lower CV errors. The optimal selection probability decreases as the number of covariates increases, whereas the number of candidate models has a relatively limited effect, under both low ($\rho=0.1$) and high ($\rho = 0.9$) correlation regimes. Moreover, the optimal selection probability is generally smaller under high correlation, reflecting the adverse effect of multicollinearity. This pattern is consistent with the theoretical insight that stronger dependence increases redundancy among predictors, leading to more aggressive shrinkage through smaller sampling probabilities. Similar patterns for other configurations are documented in the supplementary material. 




\begin{figure}[!h]
	\centering
	
	\begin{minipage}{0.95\textwidth}
		\centering
		\includegraphics[width=0.24\linewidth]{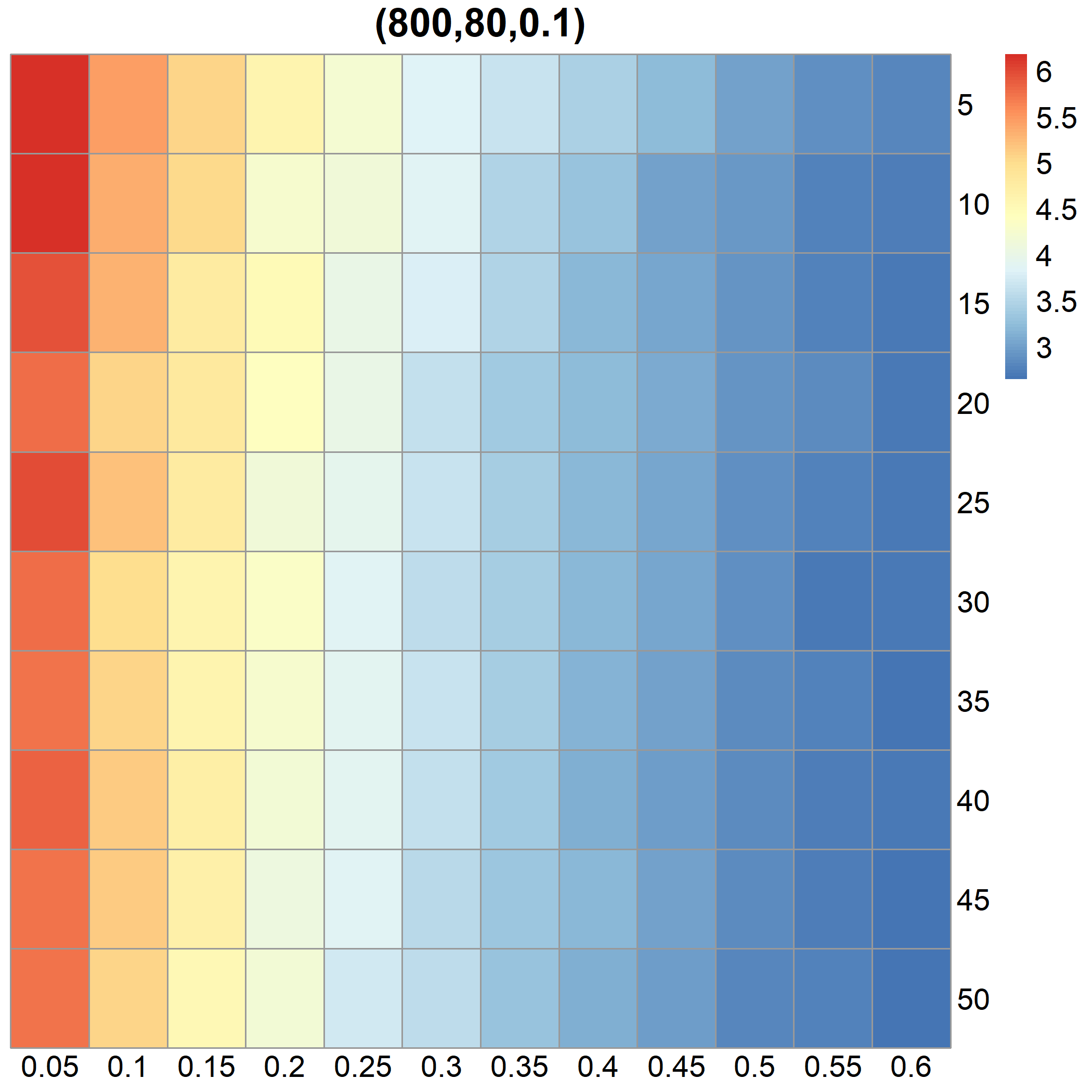}
		\includegraphics[width=0.24\linewidth]{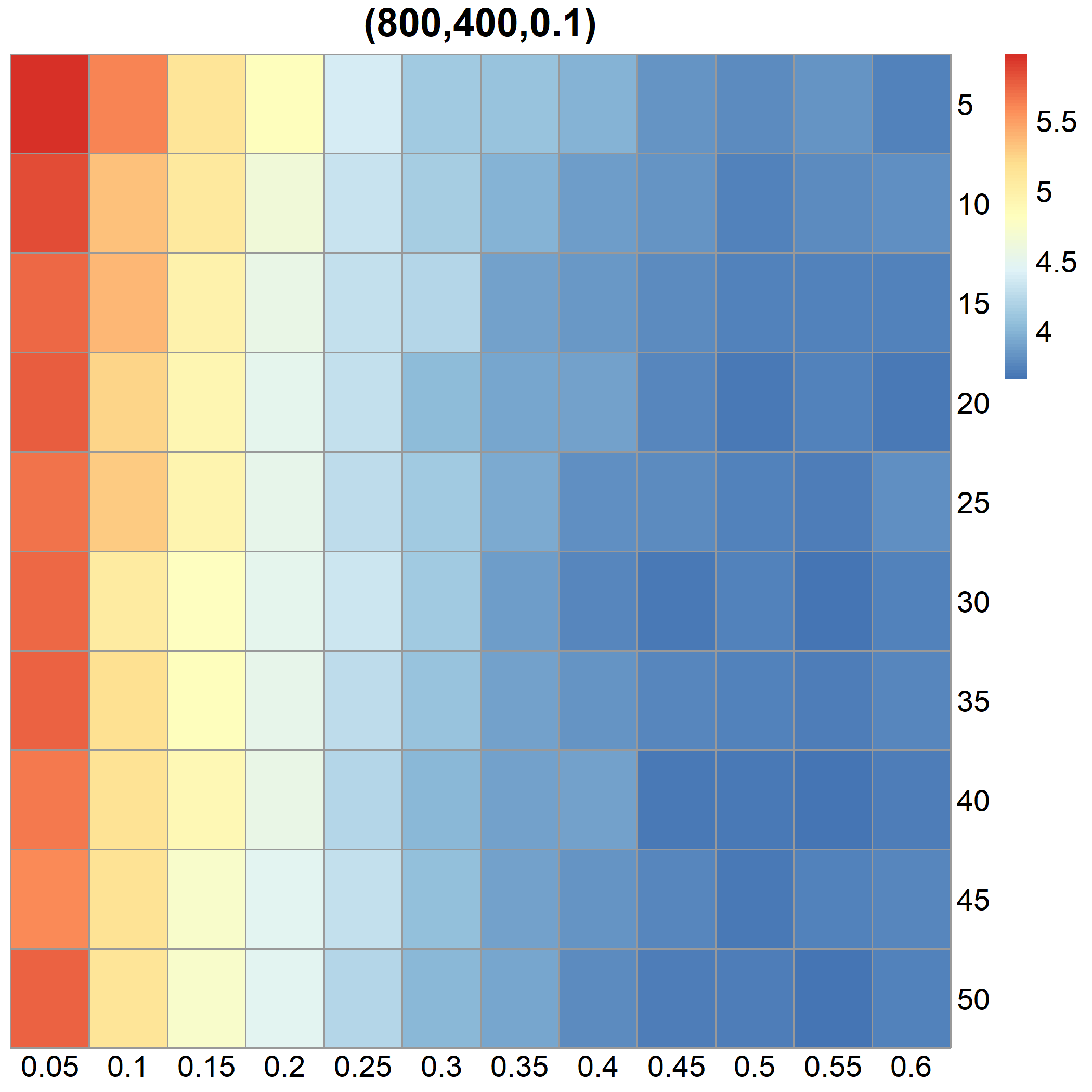}
		\includegraphics[width=0.24\linewidth]{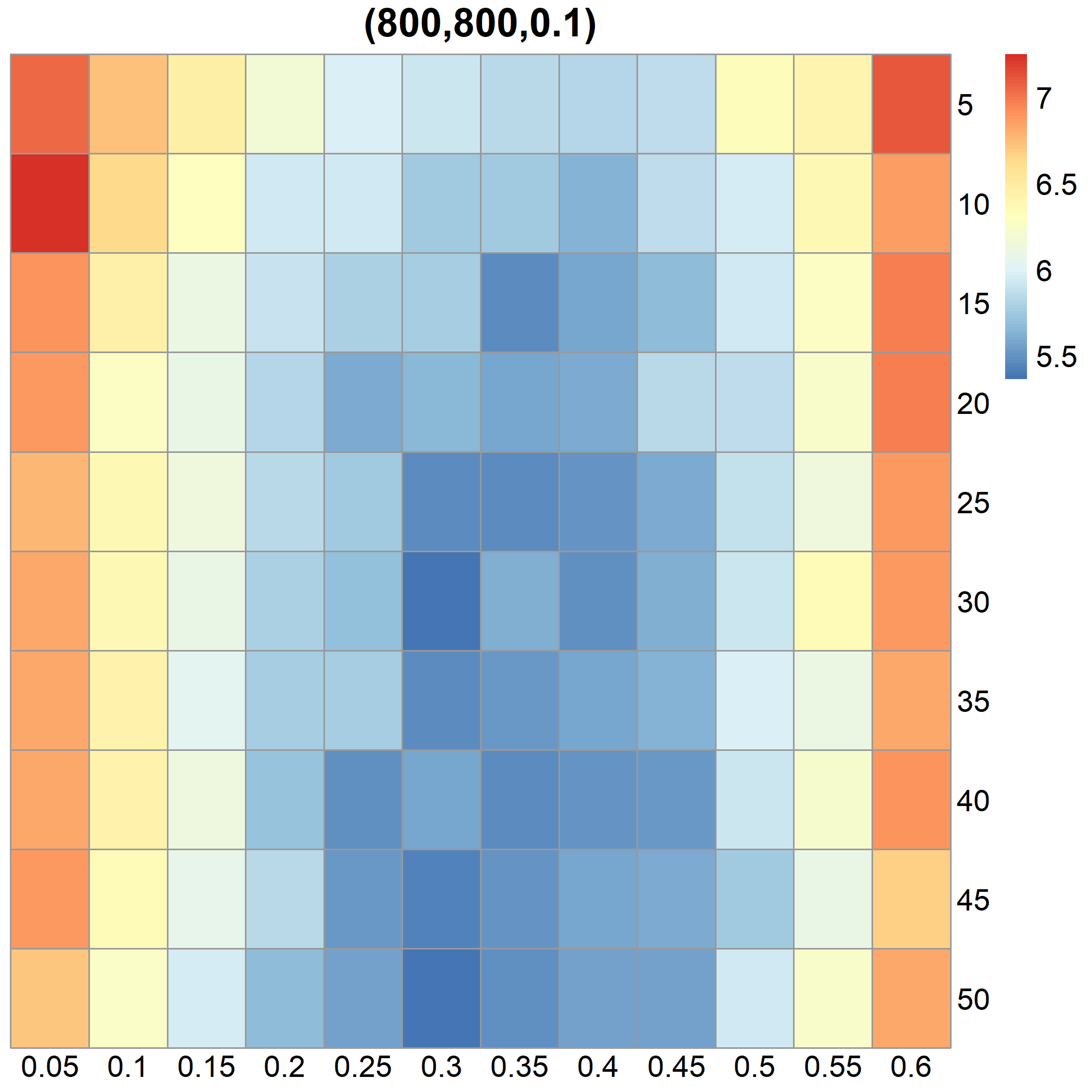}
		\includegraphics[width=0.24\linewidth]{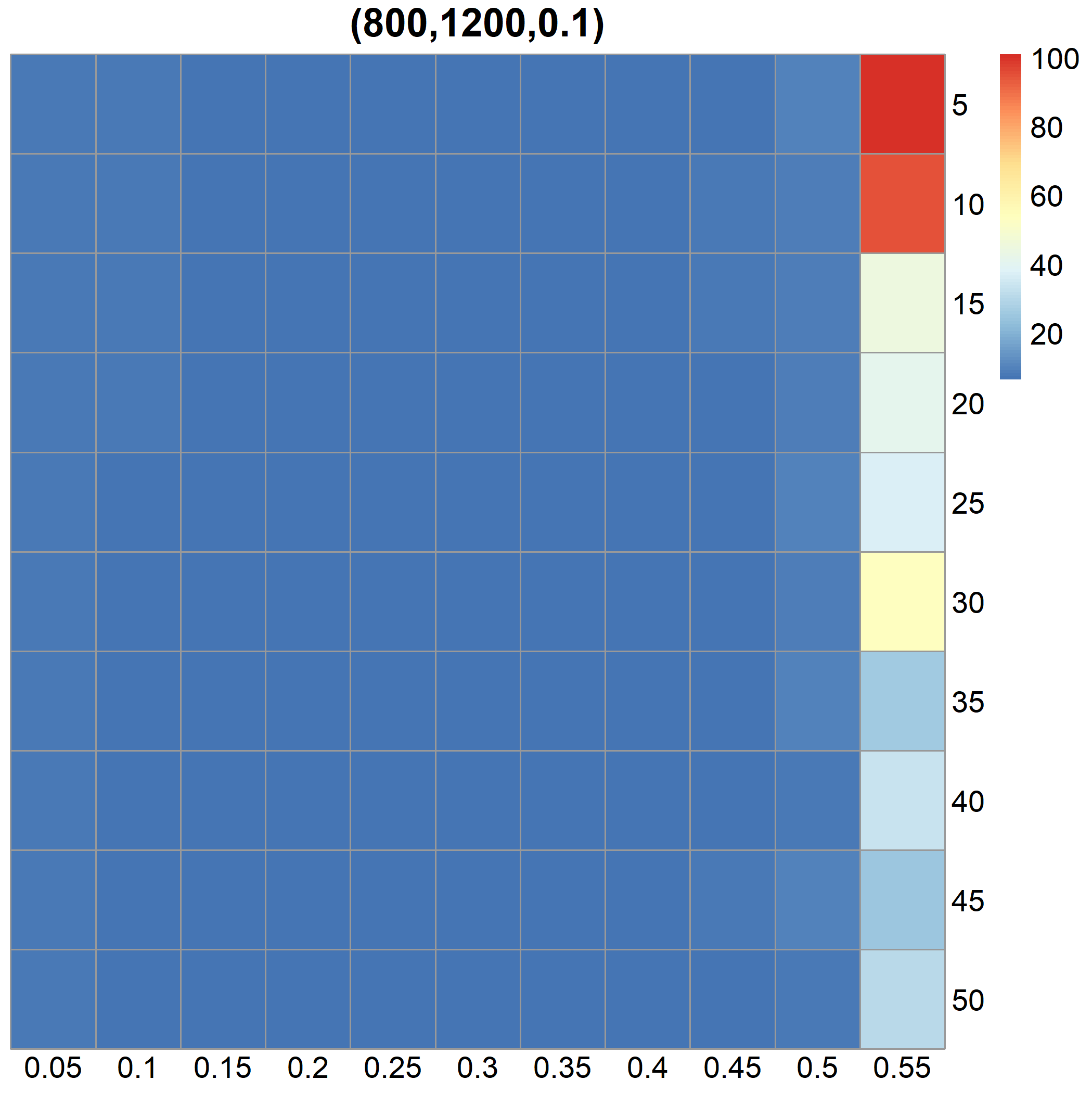}
		
		\vspace{3pt}
		\small (a) $\rho = 0.1$
	\end{minipage}
	
	\vspace{6pt}
	
	\begin{minipage}{0.95\textwidth}
		\centering
		\includegraphics[width=0.24\linewidth]{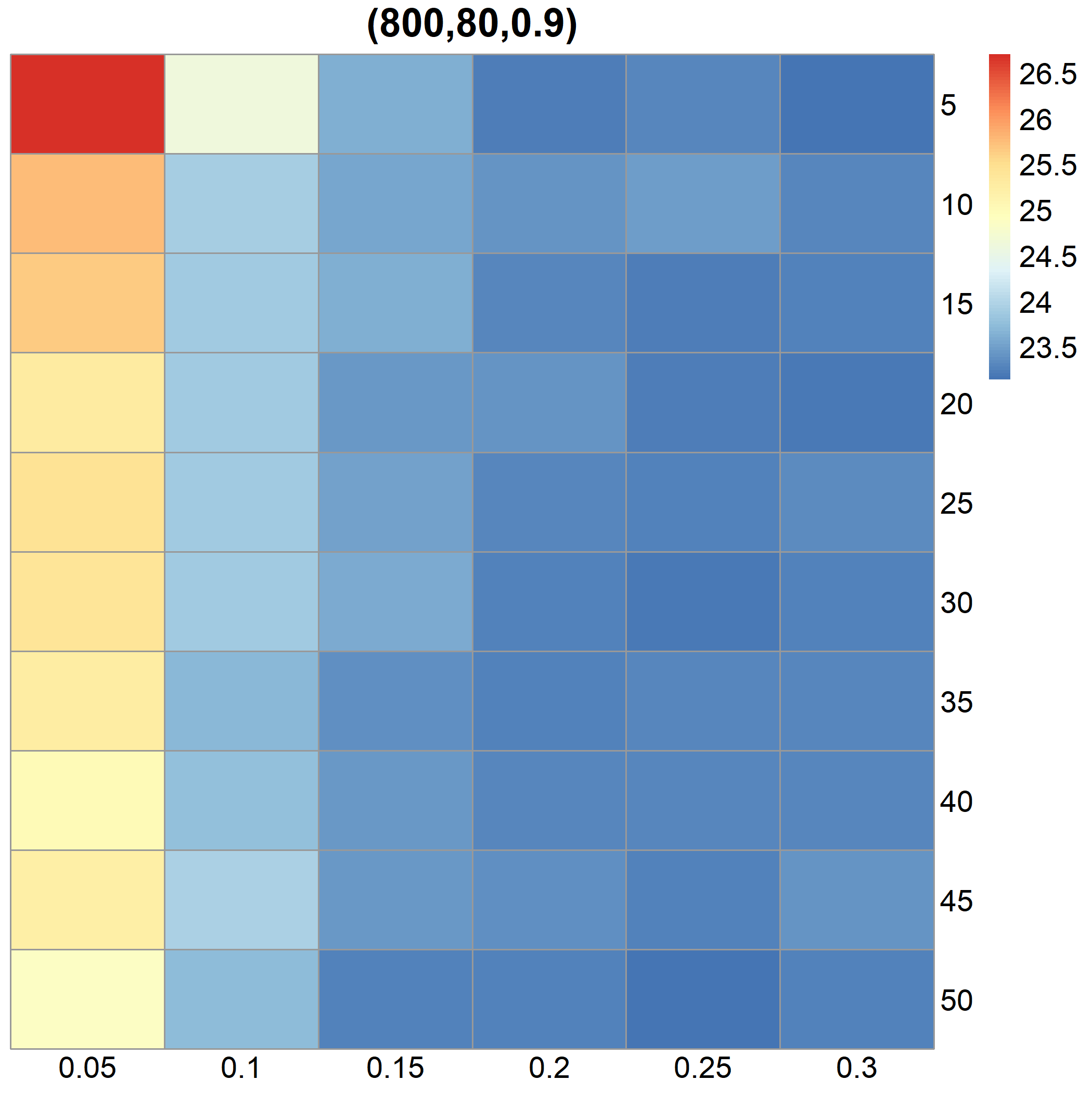}
		\includegraphics[width=0.24\linewidth]{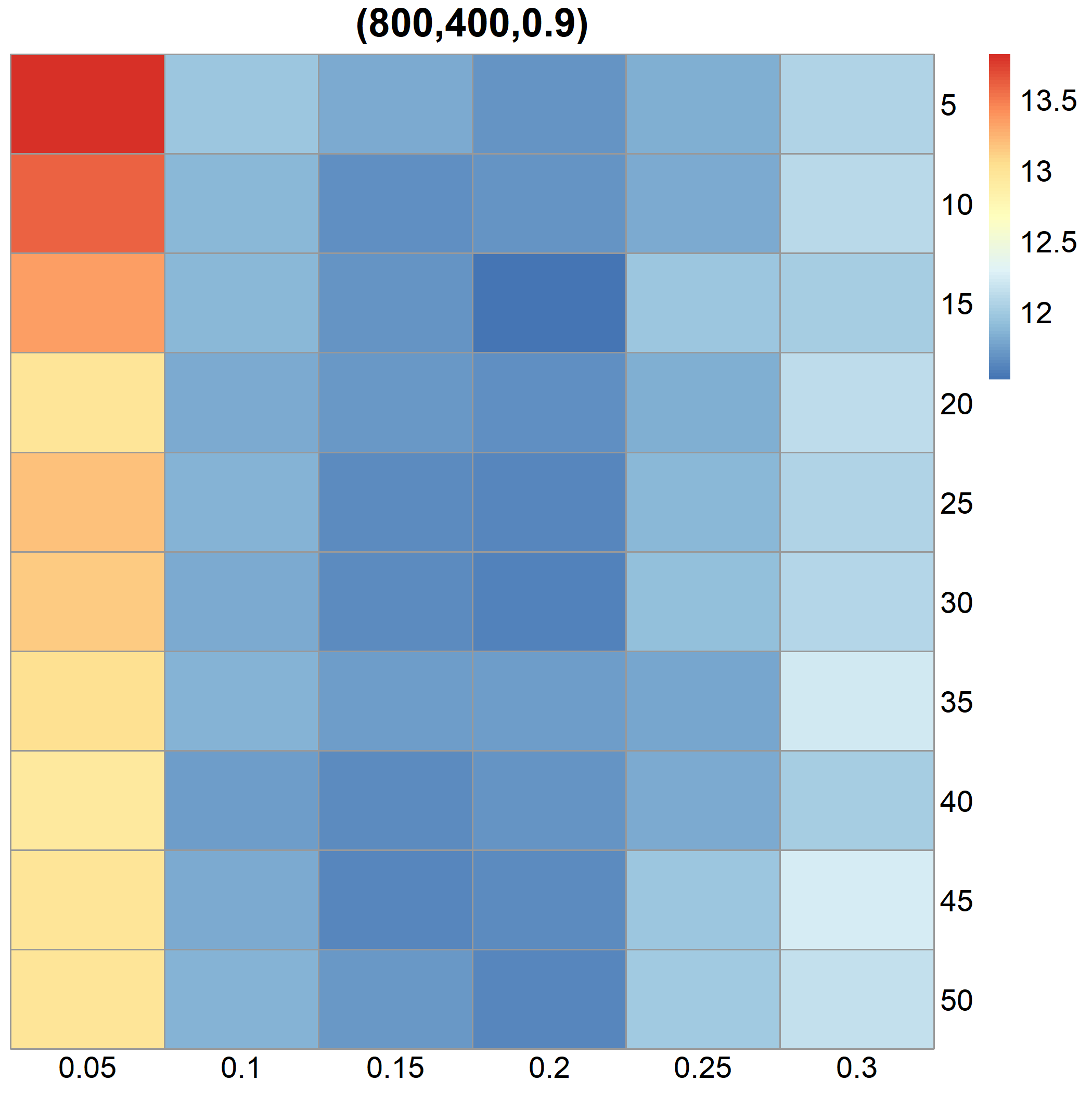}
		\includegraphics[width=0.24\linewidth]{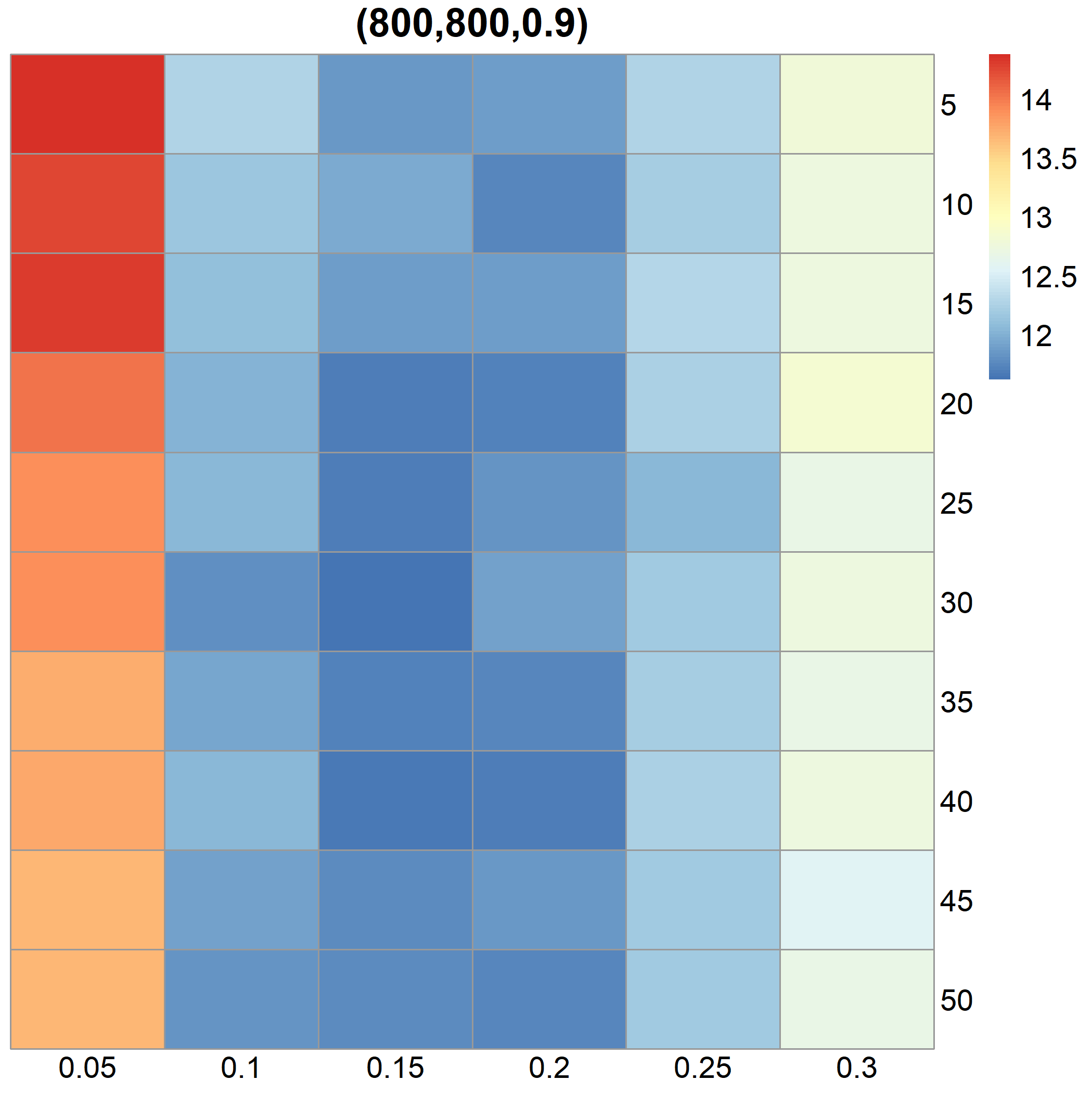}
		\includegraphics[width=0.24\linewidth]{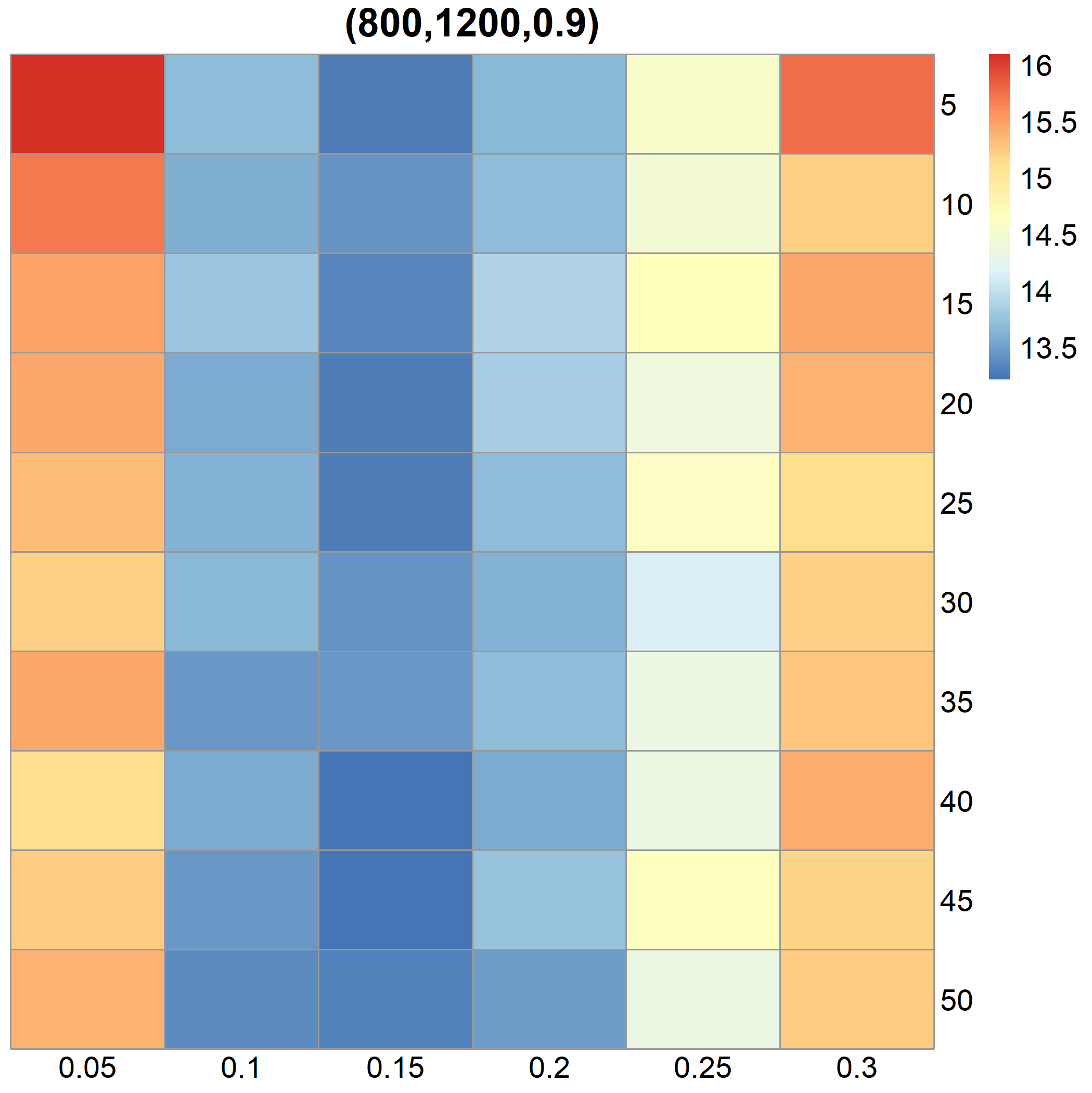}
		
		\vspace{3pt}
		\small (b) $\rho = 0.9$
	\end{minipage}
	
	\caption{Cross-validation results under polynomial decay with $N=800$. Values in parentheses denote $(N, K, \rho)$. The horizontal axis represents the selection probability $p$, and the vertical axis indicates the number of candidate models $M$. Darker regions correspond to $(p, M)$ combinations yielding lower CV errors.}
	\label{fig:CVdemonstration}
\end{figure}

\subsection{MSFE Comparison under Low and High Correlation}\label{sec3.2}
Tables \ref{tab:polyrho0.1} and \ref{tab:exprho0.1} report MSFE results under low correlation ($\rho = 0.1$) for polynomially and exponentially decaying signals, respectively, while Tables \ref{tab:polyrho0.9} and \ref{tab:exprho0.9} present the corresponding results under high correlation ($\rho = 0.9$), allowing for a direct comparison of performance across correlation regimes. 

\begin{table}[t]
	\centering
	\caption{MSFE comparison for $\rho = 0.1$ under polynomially decaying coefficients.}
	\scalebox{0.80}{
		\begin{threeparttable}
			\begin{tabular}{ccccccccccccc}
				\toprule
				$N$     & $K$     & RSA.o & RSA.f & RSA.1 & RSR   & RF    & Lasso & SCAD  & MCP   & PMA   & MCV   & MMA \\
				\midrule
				\multirow[t]{8}[2]{*}{200} & \multirow[t]{2}[1]{*}{20} & 0.27  & 1.47  & 0.33  & 3.12  & 1.71  & \textbf{0.21} & 0.21  & 0.21  & 0.70  & 0.74  & \textbf{0.21} \\
				&       & (0.09) & (0.29) & (0.11) & (0.46) & (0.3) & \textbf{(0.07)} & (0.07) & (0.07) & (0.22) & (0.22) & \textbf{(0.07)} \\
				& \multirow[t]{2}[0]{*}{100} & \textbf{0.85} & 1.84  & 0.96  & 3.26  & 2.60  & 1.07  & 0.99  & 1.07  & 1.10  & 1.44  & 1.26 \\
				&       & \textbf{(0.18)} & (0.32) & (0.2) & (0.47) & (0.42) & (0.33) & (0.28) & (0.36) & (0.26) & (0.32) & (0.31) \\
				& \multirow[t]{2}[0]{*}{200} & \textbf{1.84} & 2.66  & 2.19  & 3.92  & 3.58  & 2.57  & 2.22  & 2.41  & 2.07  & 2.43  & 1555.74 \\
				&       & \textbf{(0.33)} & (0.43) & (0.4) & (0.53) & (0.49) & (0.72) & (0.49) & (0.56) & (0.44) & (0.51) & (13606.73) \\
				& \multirow[t]{2}[1]{*}{300} & \textbf{2.59} & 3.09  & 3.15  & 4.16  & 4.03  & 3.51  & 3.21  & 3.21  & 2.87  & 3.05  & 518.32 \\
				&       & \textbf{(0.49)} & (0.44) & (0.55) & (0.52) & (0.53) & (0.84) & (0.75) & (0.66) & (0.58) & (0.55) & (1741.73) \\
				\midrule
				\multirow[t]{8}[2]{*}{400} & \multirow[t]{2}[1]{*}{40} & 0.36  & 2.13  & 0.49  & 3.81  & 2.41  & \textbf{0.22} & \textbf{0.22} & \textbf{0.22} & 0.88  & 0.87  & \textbf{0.22} \\
				&       & (0.08) & (0.25) & (0.12) & (0.37) & (0.27) & \textbf{(0.06)} & \textbf{(0.06)} & \textbf{(0.06)} & (0.18) & (0.18) & \textbf{(0.06)} \\
				& \multirow[t]{2}[0]{*}{200} & \textbf{0.97} & 2.39  & 1.12  & 3.88  & 3.20  & 1.24  & 1.15  & 1.20  & 1.23  & 1.62  & 1.39 \\
				&       & \textbf{(0.15)} & (0.27) & (0.17) & (0.41) & (0.35) & (0.23) & (0.27) & (0.28) & (0.25) & (0.28) & (0.25) \\
				& \multirow[t]{2}[0]{*}{400} & 2.33  & 3.12  & 2.87  & 4.41  & 4.06  & 2.67  & 2.35  & 2.50  & \textbf{2.19} & 2.61  & 420.07 \\
				&       & (0.3) & (0.39) & (0.42) & (0.49) & (0.46) & (0.45) & (0.38) & (0.36) & \textbf{(0.33)} & (0.38) & (1125.75) \\
				& \multirow[t]{2}[1]{*}{600} & \textbf{2.92} & 3.57  & 3.59  & 4.72  & 4.57  & 3.71  & 3.33  & 3.44  & \textbf{2.97} & 3.31  & 1147.42 \\
				&       & \textbf{(0.34)} & (0.38) & (0.41) & (0.49) & (0.5) & (0.58) & (0.43) & (0.48) & \textbf{(0.37)} & (0.41) & (6375.5) \\
				\midrule
				\multirow[t]{8}[2]{*}{800} & \multirow[t]{2}[1]{*}{80} & 0.42  & 2.68  & 0.52  & 4.35  & 3.04  & 0.27  & 0.27  & 0.27  & 1.06  & 1.00  & \textbf{0.26} \\
				&       & (0.06) & (0.22) & (0.07) & (0.29) & (0.23) & (0.05) & (0.05) & (0.05) & (0.11) & (0.13) & \textbf{(0.05)} \\
				& \multirow[t]{2}[0]{*}{400} & \textbf{1.09} & 2.88  & 1.26  & 4.32  & 3.68  & 1.36  & 1.25  & 1.26  & 1.37  & 1.70  & 1.57 \\
				&       & \textbf{(0.12)} & (0.25) & (0.15) & (0.32) & (0.3) & (0.21) & (0.18) & (0.18) & (0.17) & (0.2) & (0.19) \\
				& \multirow[t]{2}[0]{*}{800} & 2.40  & 3.57  & 2.84  & 4.86  & 4.54  & 2.84  & 2.52  & 2.60  & \textbf{2.27} & 2.72  & 3280.55 \\
				&       & (0.22) & (0.31) & (0.27) & (0.35) & (0.35) & (0.4) & (0.26) & (0.24) & \textbf{(0.19)} & (0.25) & (18600.18) \\
				& \multirow[t]{2}[1]{*}{1200} & 3.15  & 4.08  & 3.74  & 5.27  & 5.05  & 3.80  & 3.38  & 3.46  & \textbf{3.04} & 3.47  & 1310.68 \\
				&       & (0.27) & (0.31) & (0.35) & (0.37) & (0.34) & (0.42) & (0.34) & (0.37) & \textbf{(0.29)} & (0.33) & (4041.43) \\
				\bottomrule
			\end{tabular}%
			\vspace{1ex}
			{\raggedright Note: RSA.o represents the RSA method with CV-determined parameters and RSA.f refers to the RSA method with fixed parameters, specifically $M=L=30$ and $p=0.1$. RSA.1 denotes single-round Mallows averaging with the number of candidate models determined by CV. Values in bold indicate the top performers within the 95\% MCS test while values in parentheses represent the standard deviation of the reported MSFEs.\par}
		\end{threeparttable}
	}
	\label{tab:polyrho0.1}%
\end{table}%

\begin{table}[t]
	\centering
	\caption{MSFE comparison for $\rho = 0.1$ under exponentially decaying coefficient.}
	\scalebox{0.80}{
		\begin{threeparttable}
			\begin{tabular}{ccccccccccccc}
				\toprule
				$N$     & $K$     & RSA.o & RSA.f & RSA.1 & RSR   & RF    & Lasso & SCAD  & MCP   & PMA   & MCV   & MMA \\
				\midrule
				\multirow[t]{8}[2]{*}{200} & \multirow[t]{2}[1]{*}{20} & 0.08  & 0.40  & 0.11  & 0.73  & 0.44  & \textbf{0.05} & \textbf{0.05} & \textbf{0.05} & 0.17  & 0.18  & \textbf{0.05} \\
				&       & (0.02) & (0.07) & (0.03) & (0.1) & (0.07) & \textbf{(0.02)} & \textbf{(0.02)} & \textbf{(0.02)} & (0.05) & (0.06) & \textbf{(0.02)} \\
				& \multirow[t]{2}[0]{*}{100} & \textbf{0.21} & 0.49  & 0.25  & 0.75  & 0.65  & 0.27  & 0.25  & 0.25  & 0.26  & 0.36  & 0.30 \\
				&       & \textbf{(0.04)} & (0.09) & (0.05) & (0.11) & (0.1) & (0.08) & (0.07) & (0.08) & (0.06) & (0.08) & (0.07) \\
				& \multirow[t]{2}[0]{*}{200} & \textbf{0.53} & 0.71  & 0.62  & 0.92  & 0.91  & 0.67  & 0.63  & 0.65  & 0.57  & 0.63  & 77.31 \\
				&       & \textbf{(0.09)} & (0.11) & (0.1) & (0.13) & (0.13) & (0.2) & (0.15) & (0.14) & (0.12) & (0.13) & (320.28) \\
				& \multirow[t]{2}[1]{*}{300} & \textbf{0.68} & 0.82  & 0.83  & 1.00  & 1.03  & 0.90  & 0.84  & 0.90  & 0.78  & 0.81  & 295.11 \\
				&       & \textbf{(0.12)} & (0.14) & (0.15) & (0.16) & (0.17) & (0.2) & (0.15) & (0.17) & (0.16) & (0.14) & (1755.67) \\
				\midrule
				\multirow[t]{8}[2]{*}{400} & \multirow[t]{2}[1]{*}{40} & 0.09  & 0.56  & 0.11  & 0.88  & 0.61  & 0.06  & 0.06  & 0.06  & 0.22  & 0.21  & \textbf{0.06} \\
				&       & (0.02) & (0.06) & (0.02) & (0.09) & (0.07) & (0.02) & (0.02) & (0.02) & (0.04) & (0.05) & \textbf{(0.02)} \\
				& \multirow[t]{2}[0]{*}{200} & \textbf{0.25} & 0.64  & 0.28  & 0.91  & 0.82  & 0.31  & 0.27  & 0.28  & 0.31  & 0.40  & 0.35 \\
				&       & \textbf{(0.04)} & (0.08) & (0.05) & (0.1) & (0.1) & (0.07) & (0.05) & (0.06) & (0.06) & (0.08) & (0.05) \\
				& \multirow[t]{2}[0]{*}{400} & \textbf{0.50} & 0.79  & 0.60  & 1.01  & 1.00  & 0.62  & 0.57  & 0.60  & 0.52  & 0.62  & 98.70 \\
				&       & \textbf{(0.07)} & (0.09) & (0.08) & (0.1) & (0.1) & (0.1) & (0.1) & (0.11) & (0.08) & (0.1) & (265.1) \\
				& \multirow[t]{2}[1]{*}{600} & 0.72  & 0.88  & 0.84  & 1.08  & 1.11  & 0.81  & 0.75  & 0.76  & \textbf{0.68} & 0.76  & 256.81 \\
				&       & (0.1) & (0.11) & (0.11) & (0.12) & (0.13) & (0.14) & (0.12) & (0.11) & \textbf{(0.1)} & (0.1) & (926.26) \\
				\midrule
				\multirow[t]{8}[2]{*}{800} & \multirow[t]{2}[1]{*}{80} & 0.11  & 0.71  & 0.13  & 1.03  & 0.79  & 0.06  & 0.06  & 0.06  & 0.25  & 0.23  & \textbf{0.06} \\
				&       & (0.02) & (0.06) & (0.02) & (0.07) & (0.06) & (0.01) & (0.01) & (0.01) & (0.03) & (0.04) & \textbf{(0.01)} \\
				& \multirow[t]{2}[0]{*}{400} & \textbf{0.27} & 0.76  & 0.33  & 1.02  & 0.95  & 0.32  & 0.28  & 0.28  & 0.30  & 0.39  & 0.38 \\
				&       & \textbf{(0.03)} & (0.07) & (0.04) & (0.08) & (0.08) & (0.05) & (0.04) & (0.04) & (0.04) & (0.04) & (0.05) \\
				& \multirow[t]{2}[0]{*}{800} & 0.55  & 0.86  & 0.63  & 1.09  & 1.07  & 0.58  & 0.50  & 0.51  & \textbf{0.45} & 0.56  & 846.08 \\
				&       & (0.05) & (0.07) & (0.06) & (0.08) & (0.08) & (0.09) & (0.06) & (0.06) & \textbf{(0.06)} & (0.06) & (2657.81) \\
				& \multirow[t]{2}[1]{*}{1200} & 0.72  & 0.90  & 0.86  & 1.11  & 1.11  & 0.71  & 0.63  & 0.61  & \textbf{0.53} & 0.64  & 685.52 \\
				&       & (0.06) & (0.06) & (0.08) & (0.08) & (0.07) & (0.08) & (0.06) & (0.06) & \textbf{(0.05)} & (0.07) & (1658.35) \\
				\bottomrule
			\end{tabular}%
			\vspace{1ex}
			{\raggedright Note: RSA.o represents the RSA method with CV-determined parameters and RSA.f refers to the RSA method with fixed parameters, specifically $M=L=30$ and $p=0.1$. RSA.1 denotes single-round Mallows averaging with the number of candidate models determined by CV. Values in bold indicate the top performers within the 95\% MCS test while values in parentheses represent the standard deviation of the reported MSFEs.\par}
		\end{threeparttable}
	}
	\label{tab:exprho0.1}%
\end{table}%

Under low correlation, MMA performs competitively only when signals are dense ($K = 0.1N$). In this case, the nested candidate model set contains the true model, and the convex combination in MMA yields predictive gains relative to single-model selection approaches, although variable selection methods achieve comparable performance when they correctly recover the true model. As dimensionality increases and signals become sparse due to the presence of irrelevant variables, variable selection methods dominate owing to their consistency properties under weak correlation. With larger sample sizes, PMA eventually outperforms variable selection by more accurately estimating optimal aggregation weights over candidate models. 

In contrast, RSA delivers the best performance in small-sample, high-dimensional settings and consistently outperforms RSR across all scenarios, highlighting the effectiveness of its binomial sampling strategy and two-round weighting scheme in balancing model complexity and predictive accuracy. The same qualitative pattern is observed under the weak-signal setting in Table \ref{tab:exprho0.1}. Additional simulations in high-dimensional regimes with many relevant variables further confirm the superiority of RSA under low correlation over all competing methods; see Supplementary Material \ref{app:manyrelevantvariables}. Notably, RSA ranks among the top two methods under polynomial decay, whereas its relative advantage diminishes under exponential decay as $N$ and $K$ grow. This pattern suggests that RSA remains competitive except in regimes where signal magnitudes decay sufficiently rapidly that the effective signal-to-noise ratio becomes negligible.

\begin{table}[t]
	\centering
	\caption{MSFE comparison for $\rho = 0.9$ under polynomially decaying coefficients.}
	\scalebox{0.80}{
		\begin{threeparttable}
			\begin{tabular}{ccccccccccccc}
				\toprule
				$N$     & $K$     & RSA.o & RSA.f & RSA.1 & RSR   & RF    & Lasso & SCAD  & MCP   & PMA   & MCV   & MMA \\
				\midrule
				\multirow[t]{8}[2]{*}{200} & \multirow[t]{2}[1]{*}{20} & 1.09  & \textbf{0.89} & 2.08  & 2.44  & 2.00  & \textbf{0.88} & 1.38  & 1.31  & 2.16  & 1.40  & 1.56 \\
				&       & (0.45) & \textbf{(0.38)} & (0.87) & (0.88) & (0.61) & \textbf{(0.4)} & (0.53) & (0.48) & (0.72) & (0.51) & (0.53) \\
				& \multirow[t]{2}[0]{*}{100} & \textbf{1.09} & 1.17  & 1.20  & 1.83  & 3.75  & 1.58  & 2.13  & 2.16  & 2.55  & 1.82  & 5.08 \\
				&       & \textbf{(0.31)} & (0.33) & (0.35) & (0.5) & (0.8) & (0.74) & (0.69) & (0.63) & (0.64) & (0.67) & (1.24) \\
				& \multirow[t]{2}[0]{*}{200} & \textbf{1.79} & 1.88  & 1.96  & 2.25  & 6.94  & 3.02  & 3.65  & 3.80  & 4.50  & 2.84  & 979.01 \\
				&       & \textbf{(0.51)} & (0.54) & (0.53) & (0.62) & (1.32) & (1.34) & (1.13) & (0.86) & (1.09) & (1)   & (2382.3) \\
				& \multirow[t]{2}[1]{*}{300} & \textbf{2.58} & \textbf{2.59} & 3.09  & 2.80  & 9.62  & 4.45  & 4.97  & 5.51  & 5.99  & 4.05  & 4631.86 \\
				&       & \textbf{(0.59)} & \textbf{(0.58)} & (0.68) & (0.67) & (1.56) & (1.62) & (1.4) & (1.4) & (1.31) & (1.17) & (20243.33) \\
				\midrule
				\multirow[t]{8}[2]{*}{400} & \multirow[t]{2}[1]{*}{40} & \textbf{0.81} & 1.02  & 0.94  & 2.89  & 3.00  & 1.11  & 1.72  & 1.69  & 2.87  & 3.09  & 2.14 \\
				&       & \textbf{(0.29)} & (0.35) & (0.3) & (0.72) & (0.68) & (0.55) & (0.55) & (0.49) & (0.64) & (0.86) & (0.61) \\
				& \multirow[t]{2}[0]{*}{200} & \textbf{1.10} & 1.26  & 1.59  & 2.05  & 5.76  & 1.59  & 2.18  & 2.21  & 2.74  & 2.44  & 6.10 \\
				&       & \textbf{(0.22)} & (0.25) & (0.33) & (0.37) & (0.81) & (0.49) & (0.51) & (0.45) & (0.53) & (0.73) & (1.04) \\
				& \multirow[t]{2}[0]{*}{400} & \textbf{1.75} & 1.97  & 2.09  & 2.51  & 9.67  & 2.92  & 3.70  & 3.86  & 4.74  & 3.16  & 2805.92 \\
				&       & \textbf{(0.31)} & (0.35) & (0.37) & (0.38) & (1.01) & (1.01) & (0.73) & (0.79) & (0.85) & (0.67) & (10015.95) \\
				& \multirow[t]{2}[1]{*}{600} & \textbf{2.44} & 2.60  & 2.76  & 2.89  & 11.93 & 4.58  & 5.53  & 5.63  & 6.38  & 4.30  & 26095.99 \\
				&       & \textbf{(0.38)} & (0.39) & (0.43) & (0.43) & (1.34) & (1.46) & (1.2) & (1.05) & (0.94) & (0.83) & (225420.01) \\
				\midrule
				\multirow[t]{8}[2]{*}{800} & \multirow[t]{2}[1]{*}{80} & \textbf{0.78} & 1.26  & 1.16  & 3.39  & 5.54  & 1.17  & 1.92  & 1.90  & 3.39  & 4.76  & 2.50 \\
				&       & \textbf{(0.23)} & (0.36) & (0.3) & (0.6) & (0.75) & (0.37) & (0.43) & (0.4) & (0.61) & (1.38) & (0.47) \\
				& \multirow[t]{2}[0]{*}{400} & \textbf{1.10} & 1.51  & 1.36  & 2.30  & 8.45  & 1.70  & 2.21  & 2.31  & 3.09  & 2.86  & 6.62 \\
				&       & \textbf{(0.17)} & (0.25) & (0.2) & (0.34) & (0.85) & (0.48) & (0.37) & (0.4) & (0.45) & (0.59) & (0.88) \\
				& \multirow[t]{2}[0]{*}{800} & \textbf{1.88} & 2.28  & 2.18  & 2.86  & 12.72 & 3.12  & 3.92  & 4.11  & 5.13  & 3.65  & 8016.38 \\
				&       & \textbf{(0.28)} & (0.34) & (0.34) & (0.39) & (1.06) & (0.62) & (0.62) & (0.63) & (0.62) & (0.61) & (25866.27) \\
				& \multirow[t]{2}[1]{*}{1200} & \textbf{2.50} & 2.83  & 2.82  & 3.20  & 14.95 & 4.51  & 5.21  & 5.68  & 6.76  & 4.56  & 5602.16 \\
				&       & \textbf{(0.3)} & (0.35) & (0.36) & (0.4) & (1.12) & (0.75) & (0.65) & (0.67) & (0.7) & (0.66) & (9122.91) \\
				\bottomrule
			\end{tabular}%
			\vspace{1ex}
			{\raggedright Note: RSA.o represents the RSA method with CV-determined parameters and RSA.f refers to the RSA method with fixed parameters, specifically $M=L=30$ and $p=0.1$. RSA.1 denotes single-round Mallows averaging with the number of candidate models determined by CV. Values in bold indicate the top performers within the 95\% MCS test while values in parentheses represent the standard deviation of the reported MSFEs.\par}
		\end{threeparttable}
	}
	\label{tab:polyrho0.9}%
\end{table}%

\begin{table}[t]
	\centering
	\caption{MSFE comparison for $\rho = 0.9$ under exponentially decaying coefficient.}
	\scalebox{0.80}{
		\begin{threeparttable}
			\begin{tabular}{ccccccccccccc}
				\toprule
				$N$     & $K$     & RSA.o & RSA.f & RSA.1 & RSR   & RF    & Lasso & SCAD  & MCP   & PMA   & MCV   & MMA \\
				\midrule
				\multirow[t]{8}[2]{*}{200} & \multirow[t]{2}[1]{*}{20} & \textbf{0.18} & 0.20  & 0.19  & 0.57  & 0.47  & 0.21  & 0.32  & 0.34  & 0.52  & 0.32  & 0.37 \\
				&       & \textbf{(0.07)} & (0.08) & (0.08) & (0.19) & (0.16) & (0.1) & (0.13) & (0.11) & (0.14) & (0.11) & (0.11) \\
				& \multirow[t]{2}[0]{*}{100} & \textbf{0.28} & \textbf{0.28} & 0.28  & 0.45  & 0.97  & 0.41  & 0.55  & 0.56  & 0.65  & 0.49  & 1.35 \\
				&       & \textbf{(0.09)} & \textbf{(0.08)} & (0.08) & (0.12) & (0.19) & (0.23) & (0.17) & (0.19) & (0.17) & (0.2) & (0.38) \\
				& \multirow[t]{2}[0]{*}{200} & \textbf{0.47} & \textbf{0.48} & 0.51  & 0.57  & 1.84  & 0.83  & 1.04  & 0.98  & 1.19  & 0.76  & 501.63 \\
				&       & \textbf{(0.11)} & \textbf{(0.12)} & (0.11) & (0.15) & (0.34) & (0.39) & (0.3) & (0.22) & (0.28) & (0.2) & (2435.62) \\
				& \multirow[t]{2}[1]{*}{300} & \textbf{0.61} & 0.62  & 0.68  & 0.65  & 2.36  & 1.11  & 1.31  & 1.39  & 1.52  & 0.97  & 466.23 \\
				&       & \textbf{(0.15)} & (0.15) & (0.17) & (0.15) & (0.37) & (0.4) & (0.44) & (0.33) & (0.3) & (0.29) & (1269) \\
				\midrule
				\multirow[t]{8}[2]{*}{400} & \multirow[t]{2}[1]{*}{40} & \textbf{0.22} & 0.27  & 0.23  & 0.74  & 0.80  & 0.28  & 0.44  & 0.45  & 0.79  & 0.79  & 0.57 \\
				&       & \textbf{(0.08)} & (0.1) & (0.08) & (0.17) & (0.16) & (0.12) & (0.13) & (0.13) & (0.19) & (0.23) & (0.15) \\
				& \multirow[t]{2}[0]{*}{200} & \textbf{0.28} & 0.34  & 0.35  & 0.52  & 1.52  & 0.42  & 0.58  & 0.59  & 0.76  & 0.62  & 1.57 \\
				&       & \textbf{(0.06)} & (0.07) & (0.06) & (0.09) & (0.18) & (0.15) & (0.12) & (0.13) & (0.12) & (0.17) & (0.26) \\
				& \multirow[t]{2}[0]{*}{400} & \textbf{0.44} & 0.49  & 0.49  & 0.59  & 2.36  & 0.76  & 0.97  & 1.01  & 1.21  & 0.75  & 459.74 \\
				&       & \textbf{(0.07)} & (0.08) & (0.08) & (0.09) & (0.26) & (0.26) & (0.2) & (0.22) & (0.19) & (0.15) & (1206.14) \\
				& \multirow[t]{2}[1]{*}{600} & \textbf{0.56} & 0.60  & 0.63  & 0.67  & 2.75  & 1.04  & 1.26  & 1.30  & 1.45  & 0.98  & 4355.83 \\
				&       & \textbf{(0.09)} & (0.1) & (0.09) & (0.1) & (0.27) & (0.33) & (0.27) & (0.2) & (0.23) & (0.19) & (16858.64) \\
				\midrule
				\multirow[t]{8}[2]{*}{800} & \multirow[t]{2}[1]{*}{80} & \textbf{0.22} & 0.32  & 0.32  & 0.83  & 1.39  & 0.33  & 0.52  & 0.51  & 0.89  & 1.19  & 0.68 \\
				&       & \textbf{(0.06)} & (0.09) & (0.08) & (0.15) & (0.15) & (0.1) & (0.11) & (0.1) & (0.15) & (0.37) & (0.11) \\
				& \multirow[t]{2}[0]{*}{400} & \textbf{0.27} & 0.37  & 0.36  & 0.54  & 2.08  & 0.41  & 0.55  & 0.57  & 0.78  & 0.71  & 1.65 \\
				&       & \textbf{(0.04)} & (0.05) & (0.04) & (0.07) & (0.19) & (0.12) & (0.08) & (0.09) & (0.1) & (0.15) & (0.22) \\
				& \multirow[t]{2}[0]{*}{800} & \textbf{0.42} & 0.48  & 0.62  & 0.59  & 2.68  & 0.65  & 0.84  & 0.86  & 1.05  & 0.78  & 3373.20 \\
				&       & \textbf{(0.05)} & (0.06) & (0.08) & (0.07) & (0.2) & (0.13) & (0.15) & (0.12) & (0.14) & (0.1) & (25710.25) \\
				& \multirow[t]{2}[1]{*}{1200} & \textbf{0.50} & 0.56  & 0.58  & 0.65  & 2.81  & 0.83  & 0.96  & 1.04  & 1.20  & 0.89  & 2261.93 \\
				&       & \textbf{(0.06)} & (0.07) & (0.06) & (0.07) & (0.2) & (0.17) & (0.14) & (0.13) & (0.14) & (0.12) & (9162.98) \\
				\bottomrule
			\end{tabular}%
			\vspace{1ex}
			{\raggedright Note: RSA.o represents the RSA method with CV-determined parameters and RSA.f refers to the RSA method with fixed parameters, specifically $M=L=30$ and $p=0.1$. RSA.1 denotes single-round Mallows averaging with the number of candidate models determined by CV. Values in bold indicate the top performers within the 95\% MCS test while values in parentheses represent the standard deviation of the reported MSFEs.\par}
		\end{threeparttable}
	}
	\label{tab:exprho0.9}%
\end{table}%

Under high correlation, RSA with CV tuning (RSA.o) generally attains the lowest out-of-sample predictive error, except when signals are dense ($K = 0.1N$) with polynomial decay at $N = 200$. This illustrates the strong adaptivity of RSA in forecasting with highly correlated predictors. In contrast, methods that do not rely on random subset construction exhibit marked performance deterioration as correlation increases, as seen by comparing Tables \ref{tab:polyrho0.1} with \ref{tab:polyrho0.9} and Tables \ref{tab:exprho0.1} with \ref{tab:exprho0.9}. This behavior aligns with theory: strong collinearity inflates the inverse Gram matrix, weakens identification of relevant predictors, and degrades the performance of OLS-based and variable selection methods that rely on fitting the full predictor set. These effects are further amplified in high-dimensional model averaging, where substitutability among highly correlated predictors leads to unstable candidate models and volatile forecasts. 

By contrast, the improved performance of random subset-based methods (RSA.o, RSA.f, RSA.1 and RSR) reflects the benefit of controlling model complexity while exploiting dependence. Among them, RSA variants tend to outperform RSR due to their allowance for heterogeneous subset sizes and optimal convex aggregation. Moreover, the advantage of RSA.o over RSA.f and RSA.1 underscores the value of data-driven tuning and adaptive regularization of weak signals. When the sample size is very small (e.g., $ N=100 $), strong collinearity can still affect RSA.o, in which case RSR may perform better due to its lower estimation burden arising from fixed subset sizes. Nevertheless, RSA.o typically maintains an advantage over RSR. Additional evidence is provided in Supplementary Material \ref{app:manyrelevantvariables}. 
Furthermore, RSA consistently ranks first in highly correlated designs, whereas its relative advantage attenuates when predictors are weakly correlated. This divergence suggests that RSA effectively exploits cross-predictor dependence to stabilize estimation in low signal-to-noise regimes.

Overall, the simulation results indicate that RSA performs comparably to existing benchmark methods when predictors are weakly correlated, while achieving larger out-of-sample risk reductions as predictor dependence strengthens. These gains reflect RSA’s ability to stabilize prediction by separating model-fit uncertainty from the randomness induced by subset selection. Consistent patterns across Tables \ref{tab:polyrho0.1}--\ref{tab:exprho0.9}, including random covariance designs (Supplementary Material \ref{app:randomcov}) and small-$n$ large-$K$ settings, support this conclusion. The close alignment between in-sample MSE (Supplementary Material \ref{app:mse}) and out-of-sample errors further indicates stable generalization of RSA.

\section{Empirical Illustration on Asset Return Forecasting}\label{sec4}

Predicting asset returns remains a central challenge in financial economics, with important implications for portfolio construction, risk management, and market efficiency. Early studies, such as \citet{fama1993common}, emphasized the role of fundamental factors, including size (SMB) and value (HML). Since then, a large number of predictors have been proposed, often referred to as the ``Factor Zoo'' \citep{cochrane2011presidential}. While these predictors can exhibit strong in-sample performance, their out-of-sample reliability is often limited \citep{mclean2016does,chen2021open}.

This proliferation of predictors highlights a fundamental challenge in high-dimensional forecasting: although many covariates may contain useful information, naively incorporating them can lead to instability due to multicollinearity and overfitting. Existing approaches typically address this issue through the use of theoretically motivated factors \citep{fama2015five}, machine-learning-based dimension reduction \citep{gu2020empirical} or factor selection \citep{feng2020taming,wan2024mining,hwang2022bayesian}. In contrast, RSA adopts a different strategy: rather than selecting a subset of predictors, it stabilizes prediction by combining randomized subset construction with adaptive aggregation, thereby exploiting information across the full predictor set while controlling model complexity. 

We evaluate RSA using the high-frequency factor dataset constructed by \citet{pelger2019understanding}, which is well suited for high-dimensional forecasting. The dataset comprises 332 U.S. stocks and was originally used to extract latent factors via high-frequency principal component analysis. Its high dimensionality and strong cross-sectional dependence reflect the empirical challenges that RSA is designed to address. To mitigate microstructure frictions, we aggregate intraday returns into daily observations, yielding 3,270 time series spanning January 2004 to December 2016. These series are used to forecast daily S\&P 500 returns. To account for structural breaks associated with the 2008–2009 financial crisis, we exclude the crisis period and split the sample into pre-crisis (2004–2007) and post-crisis (2010–2016) subperiods.

As shown in Figure \ref{fig:correlation}, predictors exhibit strong dependence in both subperiods. Such dependence poses challenges for conventional selection-based methods, as collinearity weakens identification and amplifies estimation variance. To facilitate comparability across methods, we orthogonalize predictors prior to estimation. We adopt a rolling-window forecasting scheme with a 252-day estimation window and forecast horizons up to 22 days ahead, which corresponds approximately to using one year of past data to forecast one month ahead.

\begin{figure}[!h]
	\centering
	\begin{minipage}[b]{0.48\textwidth}
		\centering
		\includegraphics[width=0.85\linewidth, trim={0 0cm 0 0cm}]{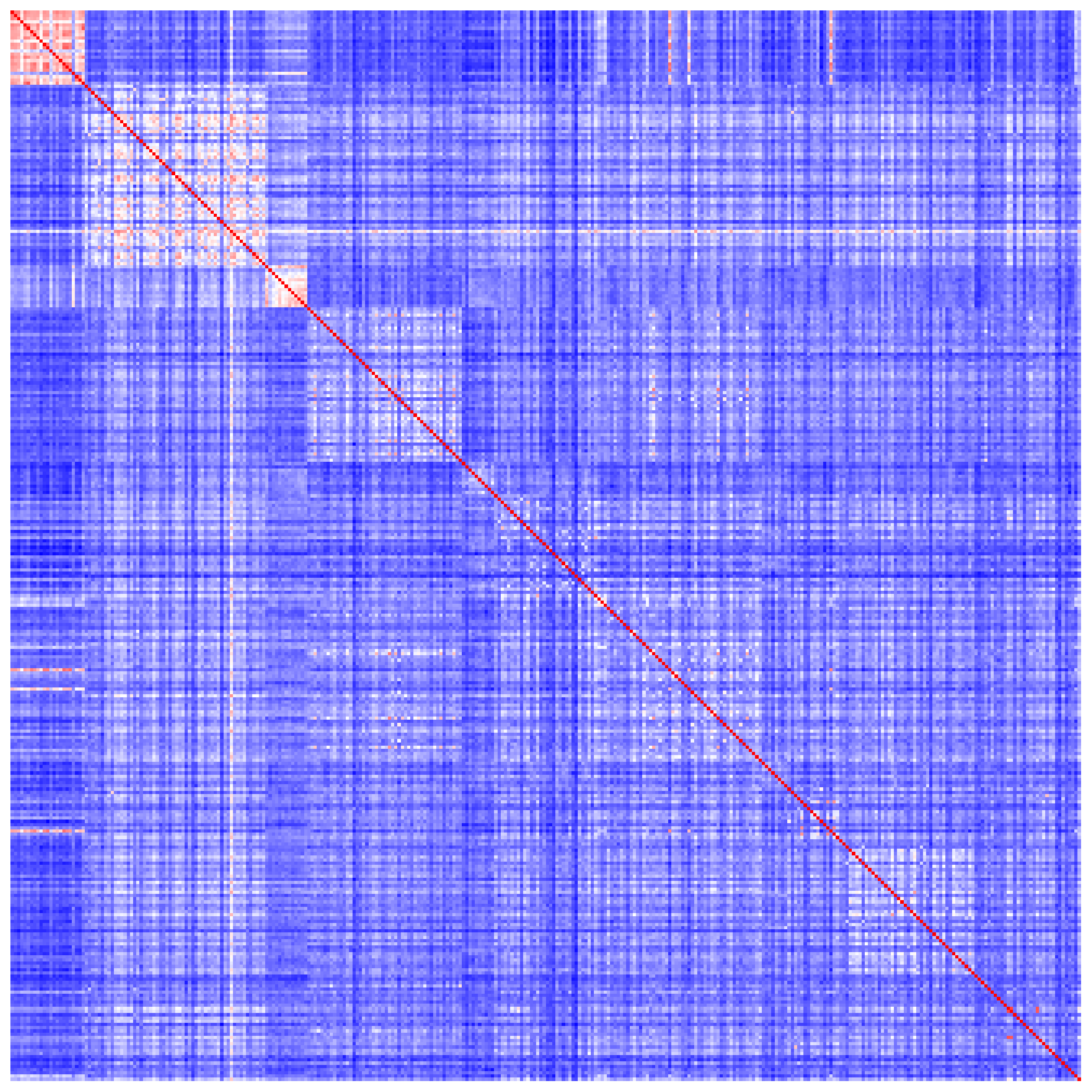}
		
		\small (a) Pre-crisis
	\end{minipage}
	\hfill
	\begin{minipage}[b]{0.48\textwidth}
		\centering
		\includegraphics[width=0.85\linewidth, trim={0 0cm 0 0cm}]{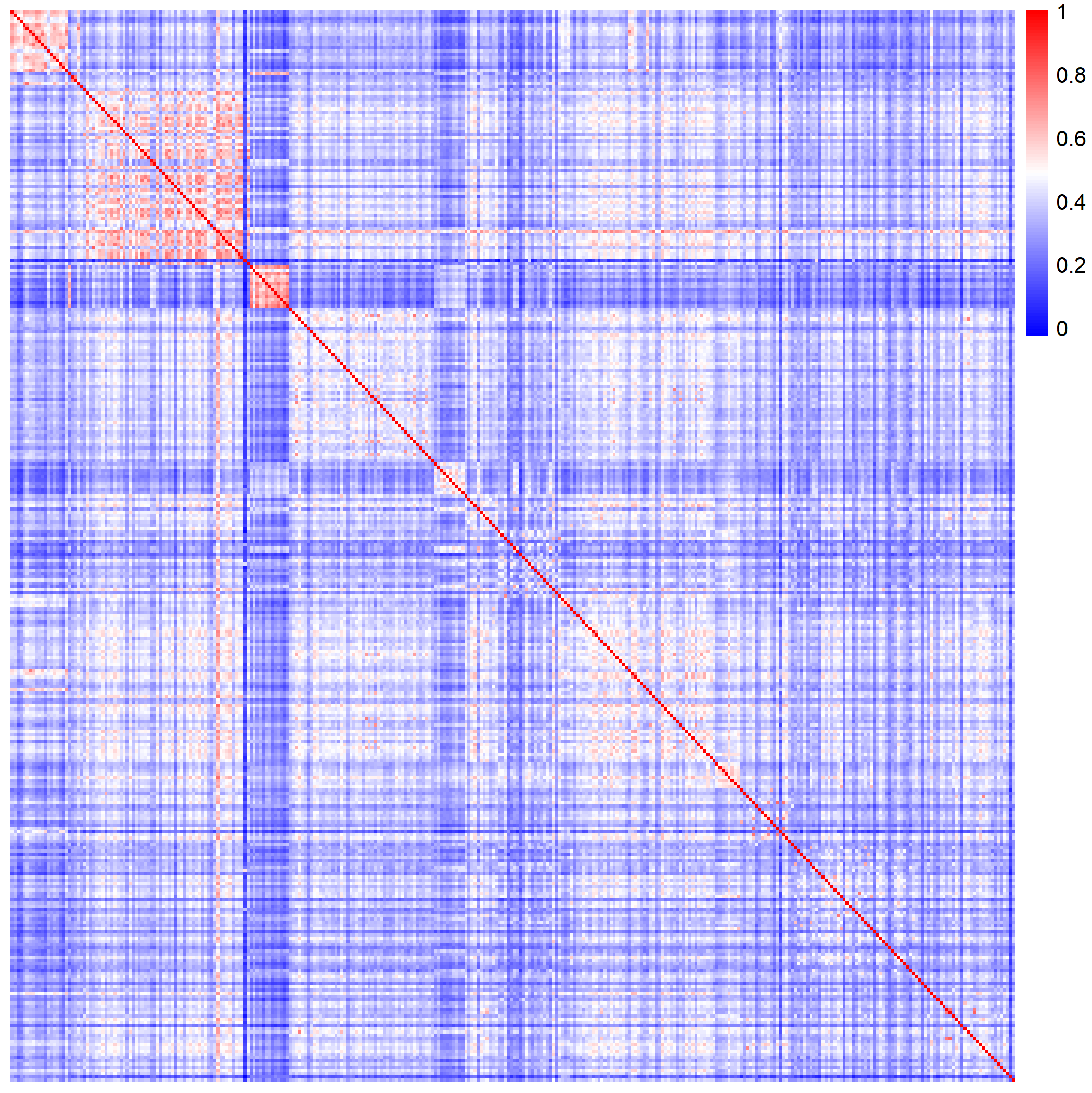}
		
		\small (b) Post-crisis
	\end{minipage}
	\caption{Correlation structure of the original factors across periods.}
	\label{fig:correlation}
\end{figure}

We compare RSA with several benchmark methods, including Lasso, SCAD, and MCP as selection-based approaches, RSR and RF as ensemble methods, and PMA and MCV as representative model averaging procedures for high-dimensional settings. Predictive performance is evaluated using the daily MSFE. Since RSA with CV-tuned parameters tends to outperform its fixed-parameter counterpart, we focus on the CV-selected specification and its one-round version RSA.1.


\begin{figure}[!h]
	\centering
	\begin{minipage}[b]{0.48\textwidth}
		\centering
		\includegraphics[width=\linewidth, trim={0 0cm 0 0cm}]{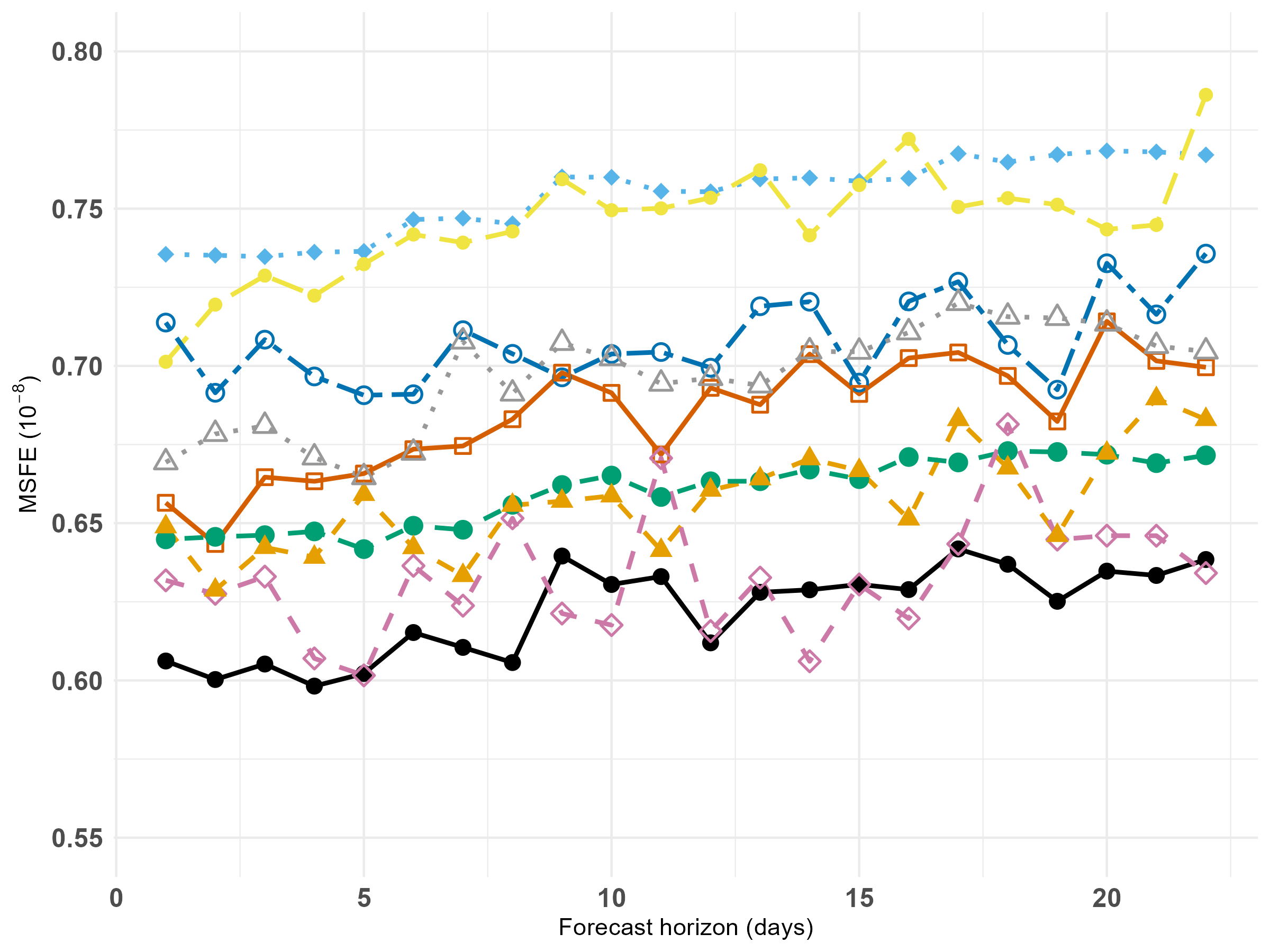}
		
		\small (a) Pre-crisis
	\end{minipage}
	\hfill
	\begin{minipage}[b]{0.48\textwidth}
		\centering
		\includegraphics[width=\linewidth, trim={0 0cm 0 0cm}]{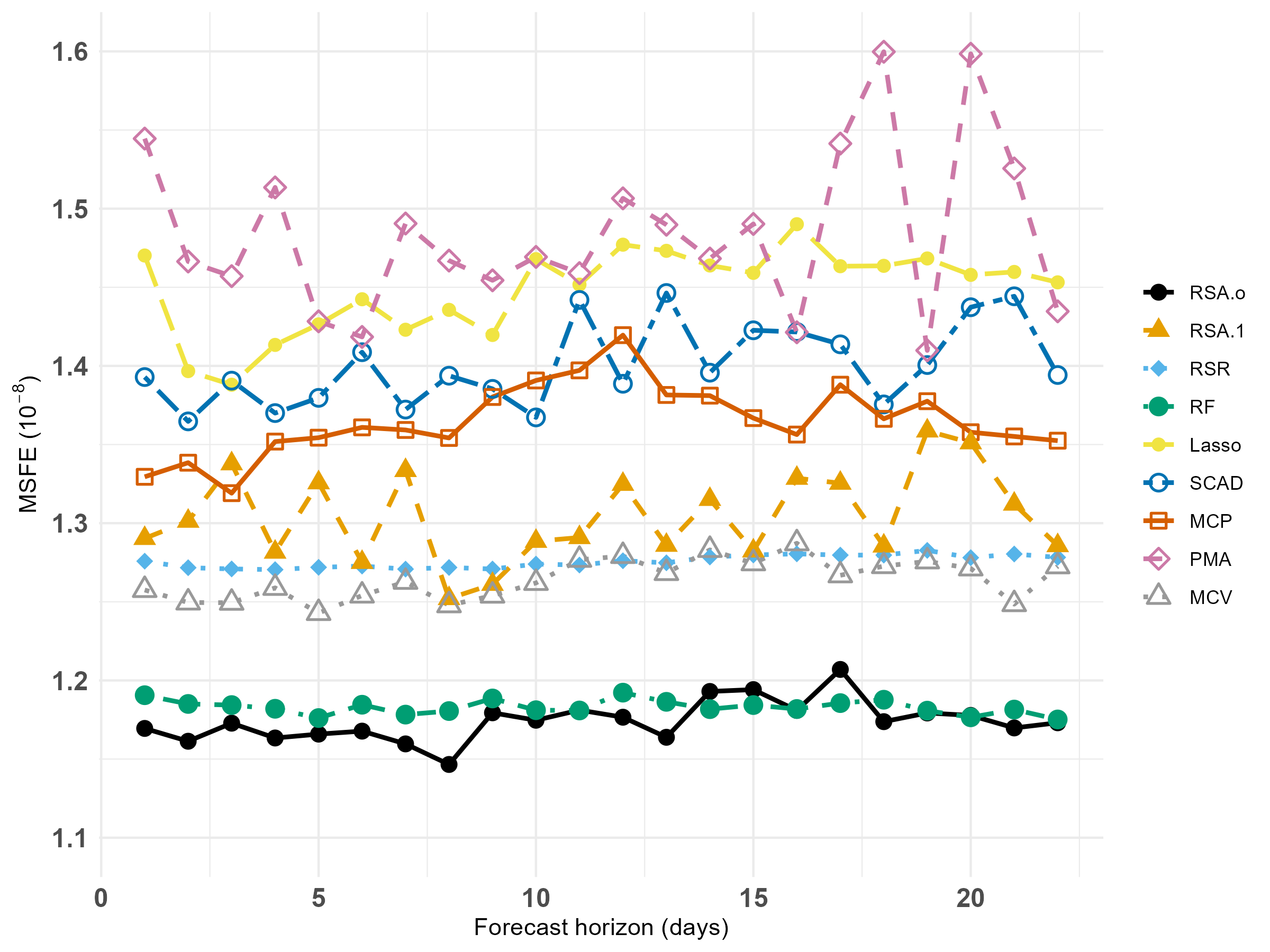}
		
		\small (b) Post-crisis
	\end{minipage}
	\caption{Mean MSFE for each period.}
	\label{fig:empiricalmsfemean}
\end{figure}

Figure \ref{fig:empiricalmsfemean} reports the MSFE across forecast horizons in both subperiods, with standard deviations within each horizon reported in the Supplementary Material. A common pattern emerges: RSA achieves a favorable bias-variance trade-off, combining low prediction error with stable performance across horizons.  
In the pre-crisis period, RSA performs competitively across all horizons, with PMA occasionally achieving similar MSFE levels but exhibiting substantially higher variability, both across horizons and over time. In the post-crisis period, increased macroeconomic uncertainty leads to a general deterioration in predictive accuracy across all methods. Nevertheless, RSA remains among the most stable and accurate procedures, while RF achieves comparable average performance at the cost of greater variability. In contrast, selection-based methods (Lasso, SCAD, MCP) and alternative model averaging methods (RSA.1, MCV, PMA) exhibit systematically higher prediction error and greater instability, reflecting their sensitivity to model selection uncertainty. Furthermore, the inferior predictive performance of RSR relative to RSA underscores the advantage of combining subset models of varying sizes with an optimally designed weighting scheme.

These empirical findings are consistent with the theoretical insights developed earlier. By combining random subset construction with adaptive aggregation, RSA effectively separates model uncertainty from subset uncertainty, thereby mitigating the adverse effects of multicollinearity.
Random subset construction reduces variance by avoiding reliance on a single, potentially unstable model, while adaptive weighting shrinks noisy components and emphasizes stronger signals. This two-layer structure allows RSA to stabilize prediction without discarding information, which is particularly important in environments with strongly correlated predictors.

Formal comparisons based on the MCS tests further support these findings. RSA is selected as a superior model with frequencies of 68.2\% in the pre-crisis period and 77.3\% in the post-crisis period across the 22 forecast horizons. This suggests that the performance gains of RSA are systematic rather than driven by isolated horizons.

\section{Conclusions}\label{sec5}

This paper proposes Random Subset Averaging (RSA), a new ensemble method for high-dimensional forecasting with complex covariate dependence. RSA constructs candidate models via a binomial random subset strategy and aggregates them through a two-round convex weighting scheme, striking a balance between model complexity and predictive accuracy, while enhancing stability. We establish its asymptotic optimality under general conditions with data-dependent first-round weights, and derive finite-sample risk bounds under orthogonal designs. Theoretical analysis shows that RSA incurs no asymptotic approximation loss relative to flat Mallows averaging and can outperform nested model averaging and random subspace methods when selection probabilities are appropriately tuned. Simulation studies and an empirical application further provide consistent evidence of RSA's strong and stable predictive performance relative to variable selection, model averaging, and machine-learning-based ensembles.

While RSA shares conceptual similarities with ensemble methods such as dropout, stacking, and random forests, it is distinguished by its structured two-layer architecture and convex aggregation scheme. These design features enhance robustness and computational tractability while enabling rigorous theoretical analysis. As a result, RSA offers a principled and computationally feasible approach to high-dimensional prediction, particularly in settings characterized by model uncertainty and strong covariate multicollinearity. Looking ahead, RSA can be extended to broader modeling frameworks, including classification problems, generalized linear models, and multi-round or adaptive ensemble constructions. This perspective highlights the role of structured randomization and adaptive shrinkage in stabilizing prediction under dependence.



%
%

\bibliography{ref-1.bib}

\newpage
\appendix
\appendixpage

\setcounter{table}{0}
\setcounter{figure}{0}
\renewcommand{\thetable}{\Alph{section}.\arabic{table}}
\renewcommand{\thefigure}{\Alph{section}.\arabic{figure}}

\newpage
\section*{Supplementary Material: Additionally numerical results}

In the supplementary material, we provide additional results from the simulation study and empirical analysis that supplement the main findings of the paper. 

Section \ref{app:implement} describes the implementation details of the competing methods used in the simulation experiments.

Section \ref{app:comparewithpeng} conducts simulation studies following the design of \citet{peng2024optimality}, extended to allow for correlated covariates. These results are intended to illustrate the finite-sample improvements for RSA established in Section \ref{sec2.3}, relative to MMA and RSR, two benchmark methods with well-understood theoretical properties. 

Section \ref{app:manyrelevantvariables} presents additional simulation studies in small-sample, high-dimensional settings, providing further evidence in support of the superior performance of RSA.o documented in Section \ref{sec3.2}.

Section \ref{app:randomcov} reports MSFE comparisons under a random covariance structure, complementing the results in Section \ref{sec3.2}. This setting captures scenarios in which the dependence structure among covariates is complex and not directly observable. RSA continues to perform comparable to its performance under low-correlation designs, suggesting that its effectiveness extends to more intricate dependence structures.

Section \ref{app:cv} presents heatmaps summarizing the CV results for the selection probability and the number of candidate models. Across all scenarios, the selection probability $p$ plays a more influential role in determining predictive performance, whereas the number of candidate models has a comparatively modest effect.

Section \ref{app:mse} reports in-sample MSE comparisons corresponding to the simulation studies in this paper. The close alignment between in-sample and out-of-sample errors provides additional empirical support for RSA's favorable generalization properties.

Finally, Section \ref{app:empirical} presents additional results from the empirical analysis, including asset return volatility in the pre- and post-crisis periods, heatmaps for tuning parameter selection, the standard deviation of MSFE across forecast horizons, and MSE comparisons across methods. Taking together, these findings further support the practical relevance of RSA.

\newpage

\subsection{Implementation detail} \label{app:implement}

RSA.o is implemented with CV-tuning. RSA.f is implemented under a fixed parameter specification with $p = 0.1 $ and $ M = L = 30$. RSA.1 is implemented using a one-round convex aggregation of candidate models generated via random subsets, employing the same CV-tuned probability and number of candidates as RSA.o. The RSR method is implemented following \cite{elliott2013complete}. The RF is implemented using the R package \texttt{randomForestSRC} \citep{ishwaran2023fast}.

The Lasso is implemented via the \texttt{glmnet} package, while SCAD and MCP are implemented using \texttt{ncvreg}. For PMA, we construct a nested set of candidate models based on the solution path of the Adaptive Lasso \citep{zhang2019parsimonious}. The MCV method is implemented following \citet{ando2014model}. 

For MMA, we construct a nested set of candidate models based on the natural ordering of covariates. When this approach is not directly applicable in high-dimensional settings, we instead use the first $N-2$ covariates to form the nested models. Notably, when the nonzero coefficients are ordered accordingly, this construction coincides with the canonical nested model set described in \cite{hansen2007least}. However, when the nonzero coefficients are randomly located, the resulting nested model set may fail to exhibit desirable theoretical properties. 

\subsection{Finite sample improvement of RSA under \citet{peng2024optimality}}\label{app:comparewithpeng}

We follow the simulation setup of \citet{peng2024optimality}, while allowing for correlated covariates. Sample sizes are set to $N \in \{100, 300, 1000\}$ with $K = \lfloor 2N/3 \rfloor$ covariates, where $\lfloor \cdot \rfloor$ denotes the floor function. Covariates are drawn from a multivariate normal distribution with mean zero and covariance $\Sigma_{ij} = \rho^{|i-j|}$ for $\rho \in \{0.1, 0.9\}$. The case $\rho = 0.1$ approximates independent covariates as in \citet{peng2024optimality}, while $\rho = 0.9$ introduces strong dependence, beyond the scope of their theory. The coefficients $\beta_j$ follow either (i) polynomial decay $\beta_j = j^{-\alpha_1}$ with $\alpha_1 = 0.51$, or (ii) exponential decay $\beta_j = \exp(-j^{\alpha_2})$ with $\alpha_2 = 0.25$. Errors $e_i$ are drawn from $N(0, \sigma_e^2)$, calibrated to yield a signal-to-noise ratio of $0.7$, defined as $\frac{Var(x_i^\top\beta)}{Var(x_i^\top\beta) + \sigma_{e}^2}$.

\begin{table}[htb]
	\centering
	\caption{In-sample MSE results.}
	\scalebox{0.80}{
		\begin{threeparttable}
			\begin{tabular}{ccccccccccccccc}
				\toprule
				DGP   & $\rho$   & $N$     & $K$     & RSA.o & RSA.f & RSA.1 & RSR   & RF    & Lasso & SCAD  & MCP   & PMA   & MCV   & MMA \\
				\midrule
				\multirow[t]{6}[2]{*}{poly} & \multirow[t]{3}[1]{*}{0.1} & 100   & 66    & 1.04  & 2.13  & 1.13  & 3.61  & 1.60  & 1.44  & 1.63  & 1.91  & 1.93  & 2.17  & \textbf{1.01} \\
				&       & 300   & 200   & 1.19  & 3.01  & 1.34  & 4.60  & 1.91  & 1.70  & 1.84  & 2.08  & 2.18  & 2.36  & \textbf{1.10} \\
				&       & 1000  & 666   & 1.33  & 4.00  & 1.47  & 5.50  & 2.28  & 1.98  & 2.06  & 2.30  & 2.48  & 2.62  & \textbf{1.17} \\
				& \multirow[t]{3}[1]{*}{0.9} & 100   & 66    & \textbf{3.41} & 3.59  & 3.78  & 4.65  & 8.40  & 4.93  & 6.01  & 6.75  & 7.98  & 5.02  & 10.02 \\
				&       & 300   & 200   & \textbf{3.79} & 3.84  & 4.38  & 6.58  & 13.22 & 5.72  & 7.51  & 7.68  & 10.77 & 7.39  & 13.00 \\
				&       & 1000  & 666   & \textbf{3.65} & 4.09  & 6.40  & 7.58  & 19.19 & 6.53  & 8.69  & 8.96  & 13.18 & 10.42 & 15.11 \\
				\midrule
				\multirow[t]{6}[2]{*}{exp} & \multirow[t]{3}[1]{*}{0.1} & 100   & 66    & 0.29  & 0.56  & 0.33  & 0.88  & 0.40  & 0.36  & 0.43  & 0.49  & 0.49  & 0.56  & \textbf{0.26} \\
				&       & 300   & 200   & 0.27  & 0.73  & 0.29  & 1.05  & 0.45  & 0.39  & 0.41  & 0.46  & 0.49  & 0.53  & \textbf{0.23} \\
				&       & 1000  & 666   & 0.30  & 0.82  & 0.34  & 1.11  & 0.47  & 0.37  & 0.34  & 0.36  & 0.38  & 0.42  & \textbf{0.17} \\
				& \multirow[t]{3}[1]{*}{0.9} & 100   & 66    & 0.94  & \textbf{0.92} & 0.98  & 1.07  & 2.19  & 1.34  & 1.69  & 1.75  & 2.13  & 1.20  & 2.66 \\
				&       & 300   & 200   & \textbf{0.80} & 0.85  & 0.99  & 1.45  & 3.08  & 1.40  & 1.83  & 1.83  & 2.39  & 1.66  & 2.85 \\
				&       & 1000  & 666   & \textbf{0.68} & 0.74  & 0.81  & 1.56  & 3.61  & 1.31  & 1.57  & 1.65  & 2.01  & 1.93  & 2.14 \\
				\bottomrule
			\end{tabular}%
			\vspace{1ex}
			{\raggedright Note: ``poly'' refers to polynomially decaying coefficients, while ``exp'' denotes exponentially decaying coefficients. RSA.o represents the RSA method with CV-determined parameters, and RSA.f refers to the RSA method with fixed parameters, specifically $M=L=30$ and $p=0.1$. Values in bold indicate the smallest MSE. \par}
		\end{threeparttable}
	}
	\label{tab:mse12}%
\end{table}%

\begin{table}[htb]
	\centering
	\caption{Out-of-sample MSFE results.}
	\scalebox{0.80}{
		\begin{threeparttable}
			\begin{tabular}{ccccccccccccccc}
				\toprule
				DGP   & $\rho$   & $N$     & $K$     & RSA.o & RSA.f & RSA.1 & RSR   & RF    & Lasso & SCAD  & MCP   & PMA   & MCV   & MMA \\
				\midrule
				\multirow[t]{12}[2]{*}{poly} & \multirow[t]{6}[1]{*}{0.1} & \multirow[t]{2}[1]{*}{100} & \multirow[t]{2}[1]{*}{66} & \textbf{1.91} & 3.03  & 2.14  & 4.23  & 3.88  & 2.93  & 3.31  & 3.80  & 3.08  & 3.33  & \textbf{1.87} \\
				&       &       &       & \textbf{(0.54)} & (0.71) & (0.71) & (0.81) & (0.81) & (0.9) & (1.12) & (1.12) & (0.81) & (0.86) & \textbf{(0.67)} \\
				&       & \multirow[t]{2}[0]{*}{300} & \multirow[t]{2}[0]{*}{200} & \textbf{1.95} & 3.72  & 2.28  & 5.24  & 4.84  & 3.18  & 3.29  & 3.85  & 3.23  & 3.43  & \textbf{1.89} \\
				&       &       &       & \textbf{(0.33)} & (0.47) & (0.38) & (0.58) & (0.6) & (0.57) & (0.52) & (0.65) & (0.46) & (0.47) & \textbf{(0.36)} \\
				&       & \multirow[t]{2}[0]{*}{1000} & \multirow[t]{2}[0]{*}{666} & 2.13  & 4.67  & 2.41  & 6.23  & 5.92  & 3.49  & 3.44  & 3.84  & 3.39  & 3.67  & \textbf{1.91} \\
				&       &       &       & (0.17) & (0.33) & (0.19) & (0.42) & (0.38) & (0.34) & (0.3) & (0.33) & (0.26) & (0.24) & \textbf{(0.21)} \\
				& \multirow[t]{6}[1]{*}{0.9} & \multirow[t]{2}[0]{*}{100} & \multirow[t]{2}[0]{*}{66} & \textbf{3.94} & 4.15  & 4.62  & 5.53  & 11.42 & 6.13  & 8.05  & 9.11  & 10.34 & 5.82  & 18.99 \\
				&       &       &       & \textbf{(1.66)} & (1.67) & (1.87) & (2.24) & (3.37) & (4.51) & (4.07) & (4.51) & (2.94) & (2.66) & (6.38) \\
				&       & \multirow[t]{2}[0]{*}{300} & \multirow[t]{2}[0]{*}{200} & \textbf{4.16} & 4.26  & 4.92  & 7.44  & 19.14 & 6.78  & 9.63  & 9.65  & 13.17 & 8.87  & 23.45 \\
				&       &       &       & \textbf{(0.96)} & (1.03) & (1.08) & (1.49) & (3.29) & (2.99) & (2.8) & (2.46) & (1.96) & (2.93) & (5.3) \\
				&       & \multirow[t]{2}[1]{*}{1000} & \multirow[t]{2}[1]{*}{666} & \textbf{3.97} & 4.57  & 7.47  & 8.68  & 33.79 & 7.48  & 10.69 & 10.86 & 15.75 & 12.23 & 25.82 \\
				&       &       &       & \textbf{(0.53)} & (0.71) & (1.82) & (1.04) & (2.37) & (1.62) & (1.7) & (1.4) & (1.58) & (1.72) & (2.99) \\
				\midrule
				\multirow[t]{12}[2]{*}{exp} & \multirow[t]{6}[1]{*}{0.1} & \multirow[t]{2}[1]{*}{100} & \multirow[t]{2}[1]{*}{66} & \textbf{0.52} & 0.82  & 0.58  & 1.04  & 1.01  & 0.76  & 0.96  & 1.09  & 0.87  & 0.91  & \textbf{0.52} \\
				&       &       &       & \textbf{(0.16)} & (0.17) & (0.14) & (0.23) & (0.22) & (0.25) & (0.35) & (0.36) & (0.25) & (0.22) & \textbf{(0.18)} \\
				&       & \multirow[t]{2}[0]{*}{300} & \multirow[t]{2}[0]{*}{200} & 0.44  & 0.92  & 0.49  & 1.21  & 1.18  & 0.72  & 0.74  & 0.86  & 0.74  & 0.78  & \textbf{0.38} \\
				&       &       &       & (0.07) & (0.13) & (0.07) & (0.15) & (0.15) & (0.16) & (0.14) & (0.15) & (0.12) & (0.12) & \textbf{(0.08)} \\
				&       & \multirow[t]{2}[0]{*}{1000} & \multirow[t]{2}[0]{*}{666} & 0.46  & 0.96  & 0.54  & 1.27  & 1.24  & 0.57  & 0.52  & 0.54  & 0.50  & 0.56  & \textbf{0.23} \\
				&       &       &       & (0.04) & (0.07) & (0.05) & (0.08) & (0.08) & (0.07) & (0.05) & (0.05) & (0.04) & (0.04) & \textbf{(0.03)} \\
				& \multirow[t]{6}[1]{*}{0.9} & \multirow[t]{2}[0]{*}{100} & \multirow[t]{2}[0]{*}{66} & \textbf{1.11} & \textbf{1.06} & 1.13  & 1.19  & 2.80  & 1.66  & 2.25  & 2.35  & 2.72  & 1.43  & 4.96 \\
				&       &       &       & \textbf{(0.46)} & \textbf{(0.44)} & (0.52) & (0.49) & (0.95) & (1.17) & (0.96) & (1.2) & (1)   & (0.93) & (1.69) \\
				&       & \multirow[t]{2}[0]{*}{300} & \multirow[t]{2}[0]{*}{200} & \textbf{0.84} & 0.90  & 1.13  & 1.61  & 4.09  & 1.62  & 2.29  & 2.28  & 2.87  & 1.92  & 4.81 \\
				&       &       &       & \textbf{(0.2)} & (0.21) & (0.27) & (0.33) & (0.58) & (0.67) & (0.66) & (0.56) & (0.53) & (0.7) & (1.08) \\
				&       & \multirow[t]{2}[1]{*}{1000} & \multirow[t]{2}[1]{*}{666} & \textbf{0.72} & 0.80  & 0.88  & 1.75  & 5.04  & 1.46  & 1.82  & 1.92  & 2.27  & 2.39  & 3.00 \\
				&       &       &       & \textbf{(0.11)} & (0.13) & (0.14) & (0.24) & (0.41) & (0.36) & (0.29) & (0.24) & (0.25) & (0.35) & (0.36) \\
				\bottomrule
			\end{tabular}%
			\vspace{1ex}
			{\raggedright Note: ``poly'' refers to polynomially decaying coefficients, while ``exp'' denotes exponentially decaying coefficients. RSA.o represents the RSA method with CV-determined parameters and RSA.f refers to the RSA method with fixed parameters, specifically $M=L=30$ and $p=0.1$. Values in bold indicate the top performers within the 95\% MCS test while values in parentheses represent the standard deviation of the reported MSFEs.\par}
		\end{threeparttable}
	}
	\label{tab:msfe12}%
\end{table}%

Table \ref{tab:mse12} reports the in-sample MSE for all methods, with the smallest values highlighted in bold. Under low correlation ($\rho = 0.1$), MMA consistently attains the lowest MSE across all designs, in line with the findings of \cite{peng2022improvability}. Notably, RSA.o, whose tuning parameters are selected via CV method, remains highly competitive, achieving MSE values close to those of MMA. In contrast, under strong correlation ($\rho = 0.9$), RSA.o uniformly outperforms all competing methods.

Table \ref{tab:msfe12} presents the corresponding out-of-sample MSFE results, with models included in the 95\% MCS highlighted in bold. The out-of-sample patterns mirrors the in-sample findings: MMA achieves the lowest MSFE under weak correlation, with RSA.o delivering comparable performance. Under strong correlation, however, RSA.o emerges as the sole best-performing method, while all other alternatives exhibit substantially larger forecast errors. Moreover, RSA.o attains the smallest MSFE standard deviations across all settings, indicating superior stability in predictive performance. 

Overall, these simulations demonstrate that RSA delivers robust predictive performance, both in- and out-of-sample. While MMA outperforms variable selection under polynomial decay and matches its performance under exponential decay when covariates are weakly correlated \citep{peng2022improvability}, RSA remains stable across both decay patterns and clearly dominates in the presence of strong predictor dependence. Moreover, despite their structural similarities, RSA consistently outperforms RSR and RF, underscoring the advantage of its two-layer aggregation structure.

\subsection{MSFE Comparison under Many Relevant Covariates} \label{app:manyrelevantvariables}

Sections \ref{sec3.2} demonstrate RSA’s strong performance in high-dimensional settings with correlated covariates. Here, we examine a practical scenario with many potentially relevant covariates but limited sample size. The relatively small sample size makes it infeasible to include all relevant variables in a single model and undermines the selection consistency of variable selection methods. In contrast, RSA and, to some extent, RSR should be more robust by aggregating predictions from multiple smaller models. We also expect RSA to outperform RSR due to its binomial random subset strategy and two-round convex weighting scheme.

\begin{table}[!b]
	\centering
	\caption{MSFE comparison under many relevant covariates: polynomial decay.}
	\scalebox{0.80}{
		\begin{threeparttable}
			\begin{tabular}{cccccccccccccc}
				\toprule
				$\rho$   & $N$     & $K$     & RSA.o & RSA.f & RSA.1 & RSR   & RF    & Lasso & SCAD  & MCP   & PMA   & MCV   & MMA \\
				\midrule
				\multirow[t]{12}[2]{*}{0.1} & \multirow[t]{6}[1]{*}{100} & \multirow[t]{2}[1]{*}{100} & \textbf{2.74} & 3.56  & 3.77  & 4.48  & 4.47  & 4.56  & 4.42  & 5.01  & 4.30  & 4.12  & 316.15 \\
				&       &       & \textbf{(0.66)} & (0.8) & (1.04) & (0.86) & (0.94) & (1.96) & (1.41) & (1.71) & (1.64) & (1.01) & (957.46) \\
				&       & \multirow[t]{2}[0]{*}{125} & \textbf{3.29} & 3.92  & 3.90  & 4.85  & 4.97  & 4.95  & 4.92  & 5.40  & 4.66  & 4.40  & 277.50 \\
				&       &       & \textbf{(0.81)} & (0.81) & (0.92) & (0.92) & (0.99) & (1.31) & (1.41) & (1.51) & (1.35) & (0.92) & (908.41) \\
				&       & \multirow[t]{2}[0]{*}{150} & \textbf{3.78} & 4.26  & 4.51  & 5.10  & 5.27  & 5.43  & 5.46  & 5.77  & 5.18  & 4.87  & 941.12 \\
				&       &       & \textbf{(0.86)} & (0.87) & (1.01) & (1.05) & (1.05) & (1.51) & (1.49) & (1.51) & (3.21) & (1.04) & (3218.87) \\
				& \multirow[t]{6}[1]{*}{300} & \multirow[t]{2}[0]{*}{300} & \textbf{2.91} & 4.31  & 3.72  & 5.48  & 5.45  & 4.70  & 4.51  & 4.97  & 4.24  & 4.45  & 1358.36 \\
				&       &       & \textbf{(0.45)} & (0.52) & (0.58) & (0.66) & (0.61) & (0.88) & (0.75) & (0.82) & (0.53) & (0.57) & (4978.29) \\
				&       & \multirow[t]{2}[0]{*}{375} & \textbf{3.32} & 4.48  & 3.88  & 5.56  & 5.70  & 5.25  & 4.92  & 5.34  & 4.57  & 4.76  & 987.35 \\
				&       &       & \textbf{(0.47)} & (0.53) & (0.58) & (0.62) & (0.68) & (1.17) & (0.81) & (0.88) & (0.58) & (0.62) & (3384.89) \\
				&       & \multirow[t]{2}[1]{*}{450} & \textbf{4.34} & 4.63  & 5.07  & 5.66  & 5.85  & 5.78  & 5.36  & 5.64  & 4.77  & 5.04  & 1556.13 \\
				&       &       & \textbf{(0.57)} & (0.56) & (0.74) & (0.7) & (0.75) & (1.12) & (0.82) & (0.81) & (0.69) & (0.72) & (7902.62) \\
				\midrule
				\multirow[t]{12}[2]{*}{0.9} & \multirow[t]{6}[1]{*}{100} & \multirow[t]{2}[1]{*}{100} & 5.55  & 5.51  & 5.85  & \textbf{4.84} & 15.96 & 9.01  & 12.19 & 13.10 & 14.10 & 7.33  & 1773.48 \\
				&       &       & (1.95) & (2.09) & (1.83) & \textbf{(1.89)} & (4.06) & (4.77) & (4.98) & (5.93) & (3.75) & (3.53) & (5384.73) \\
				&       & \multirow[t]{2}[0]{*}{125} & 6.16  & 6.35  & 7.25  & \textbf{5.14} & 20.30 & 10.67 & 13.10 & 14.43 & 16.94 & 8.27  & 2379.49 \\
				&       &       & (2.25) & (2.59) & (2.23) & \textbf{(1.93)} & (4.6) & (7.24) & (4.92) & (5.34) & (5.38) & (3.21) & (6632.76) \\
				&       & \multirow[t]{2}[0]{*}{150} & 6.79  & 6.79  & 7.86  & \textbf{5.70} & 22.58 & 13.04 & 16.27 & 16.51 & 18.22 & 9.49  & 5891.74 \\
				&       &       & (2.46) & (2.4) & (2.58) & \textbf{(2.11)} & (4.77) & (8.04) & (5.62) & (5.9) & (5.02) & (3.51) & (31263.39) \\
				& \multirow[t]{6}[1]{*}{300} & \multirow[t]{2}[0]{*}{300} & \textbf{4.63} & 5.09  & 5.15  & 5.37  & 27.11 & 8.91  & 12.30 & 12.24 & 16.53 & 6.79  & 6260.09 \\
				&       &       & \textbf{(0.96)} & (1.11) & (1.16) & (1.17) & (4.17) & (3.06) & (2.55) & (2.27) & (3.05) & (1.83) & (14356.44) \\
				&       & \multirow[t]{2}[0]{*}{375} & \textbf{5.75} & \textbf{5.74} & 6.56  & \textbf{5.75} & 31.17 & 10.60 & 14.54 & 14.60 & 18.58 & 7.98  & 7595.74 \\
				&       &       & \textbf{(1.23)} & \textbf{(1.24)} & (1.36) & \textbf{(1.3)} & (3.66) & (3.41) & (3.3) & (2.75) & (3.41) & (2.32) & (24330.51) \\
				&       & \multirow[t]{2}[1]{*}{450} & \textbf{6.05} & 6.53  & 6.75  & \textbf{6.15} & 33.21 & 13.14 & 17.24 & 17.70 & 20.67 & 9.52  & 9206.80 \\
				&       &       & \textbf{(1.12)} & (1.4) & (1.22) & \textbf{(1.33)} & (4.39) & (4.82) & (3.94) & (3.86) & (3.54) & (2.49) & (23827.29) \\
				\bottomrule
			\end{tabular}%
			\vspace{1ex}
			{\raggedright Note: RSA.o represents the RSA method with CV-determined parameters and RSA.f refers to the RSA method with fixed parameters, specifically $M=L=30$ and $p=0.1$. Values in bold indicate the top performers within the 95\% MCS test while values in parentheses represent the standard deviation of the reported MSFEs.\par}
		\end{threeparttable}
	}
	\label{tab:msfe56}%
\end{table}%

In this section, we set $N \in \{100, 300\}$ and define the number of covariates as $K=\delta N$, where $\delta \in \{1, 1.25, 1.5\}$. The values in coefficient vector $\beta$ follow either a polynomial decay $\beta_j = j^{-0.51}$ or an exponential decay $\beta_j = \exp(-j^{0.25})$ for $j = 1, \dots, K$. Although all entries in $\beta$ are nonzero, the exponential decay results in a sparser effective structure. The MSFE results for these two decay patterns are summarized in Tables \ref{tab:msfe56} and \ref{tab:msfe910}.

Under low correlation $(\rho = 0.1)$, the RSA method consistently outperforms alternative approaches across all sample sizes. In contrast, under high correlation $(\rho = 0.9)$, RSA demonstrate superior performance in large-sample cases, whereas RSR yields the best results when the sample size is small.

As previously discussed, RSR restricts each candidate model to a fixed number of covariates, whereas RSA allows the number of covariates to vary following a binomial distribution. Consequently, RSA has a positive probability of generating candidate models with a large number of covariates. When the sample size is small, these larger models with strong covariate correlations exhibit high estimation variability, resulting in reduced prediction accuracy. 
However, as the sample size increases, this variability diminishes, leading to improved predictive performance of RSA.  

\begin{table}[!h]
	\centering
	\caption{MSFE comparison under many relevant covariates: exponential decay.}
	\scalebox{0.80}{
		\begin{threeparttable}
			\begin{tabular}{cccccccccccccc}
				\toprule
				$\rho$   & $N$     & $K$     & RSA.o & RSA.f & RSA.1 & RSR   & RF    & Lasso & SCAD  & MCP   & PMA   & MCV   & MMA \\
				\midrule
				\multirow[t]{12}[2]{*}{0.1} & \multirow[t]{6}[1]{*}{100} & \multirow[t]{2}[1]{*}{100} & \textbf{0.67} & 0.90  & 0.75  & 1.07  & 1.09  & 1.15  & 1.14  & 1.30  & 1.12  & 1.02  & 80.18 \\
				&       &       & \textbf{(0.18)} & (0.21) & (0.19) & (0.22) & (0.24) & (0.37) & (0.35) & (0.4) & (1.08) & (0.26) & (237.48) \\
				&       & \multirow[t]{2}[0]{*}{125} & \textbf{0.83} & 0.96  & 0.90  & 1.10  & 1.16  & 1.37  & 1.34  & 1.41  & 1.22  & 1.15  & 75.61 \\
				&       &       & \textbf{(0.23)} & (0.17) & (0.23) & (0.19) & (0.2) & (0.59) & (0.36) & (0.34) & (0.57) & (0.24) & (235.53) \\
				&       & \multirow[t]{2}[0]{*}{150} & \textbf{0.91} & 1.03  & 1.03  & 1.15  & 1.23  & 1.43  & 1.37  & 1.49  & 1.31  & 1.18  & 70.27 \\
				&       &       & \textbf{(0.21)} & (0.22) & (0.21) & (0.23) & (0.25) & (0.44) & (0.36) & (0.41) & (0.37) & (0.27) & (204.69) \\
				& \multirow[t]{6}[1]{*}{300} & \multirow[t]{2}[0]{*}{300} & \textbf{0.64} & 0.95  & 0.72  & 1.16  & 1.20  & 0.94  & 0.90  & 0.99  & 0.84  & 0.92  & 141.87 \\
				&       &       & \textbf{(0.09)} & (0.11) & (0.1) & (0.13) & (0.14) & (0.22) & (0.15) & (0.15) & (0.14) & (0.13) & (422.05) \\
				&       & \multirow[t]{2}[0]{*}{375} & \textbf{0.71} & 0.97  & 0.85  & 1.16  & 1.22  & 1.05  & 0.96  & 1.02  & 0.90  & 0.95  & 332.14 \\
				&       &       & \textbf{(0.1)} & (0.11) & (0.12) & (0.13) & (0.13) & (0.21) & (0.14) & (0.15) & (0.12) & (0.13) & (1546.52) \\
				&       & \multirow[t]{2}[1]{*}{450} & \textbf{0.90} & 0.98  & 1.00  & 1.15  & 1.22  & 1.08  & 1.02  & 1.10  & 0.94  & 0.98  & 3451.48 \\
				&       &       & \textbf{(0.13)} & (0.14) & (0.14) & (0.14) & (0.15) & (0.2) & (0.17) & (0.18) & (0.15) & (0.15) & (33382.76) \\
				\midrule
				\multirow[t]{12}[2]{*}{0.9} & \multirow[t]{6}[1]{*}{100} & \multirow[t]{2}[1]{*}{100} & 1.45  & 1.41  & 2.10  & \textbf{1.31} & 4.30  & 2.42  & 2.96  & 3.11  & 3.78  & 1.95  & 841.95 \\
				&       &       & (0.57) & (0.55) & (0.77) & \textbf{(0.61)} & (1.15) & (1.61) & (1.11) & (1.1) & (1.18) & (1.03) & (4077.29) \\
				&       & \multirow[t]{2}[0]{*}{125} & 1.98  & 1.50  & 2.72  & \textbf{1.21} & 5.08  & 2.66  & 3.55  & 3.69  & 4.05  & 2.04  & 330.94 \\
				&       &       & (0.71) & (0.53) & (0.9) & \textbf{(0.49)} & (1.11) & (1.56) & (1.63) & (1.65) & (1.23) & (0.91) & (641.36) \\
				&       & \multirow[t]{2}[0]{*}{150} & 1.84  & 1.66  & 2.20  & \textbf{1.32} & 5.38  & 3.14  & 3.96  & 3.98  & 4.42  & 2.47  & 483.43 \\
				&       &       & (0.62) & (0.63) & (0.85) & \textbf{(0.59)} & (1.47) & (1.75) & (1.63) & (1.53) & (1.4) & (1.33) & (1114.05) \\
				& \multirow[t]{6}[1]{*}{300} & \multirow[t]{2}[0]{*}{300} & \textbf{0.88} & 0.98  & 1.02  & 1.05  & 5.36  & 1.69  & 2.34  & 2.36  & 3.13  & 1.39  & 2514.16 \\
				&       &       & \textbf{(0.17)} & (0.21) & (0.2) & (0.23) & (0.64) & (0.49) & (0.47) & (0.46) & (0.52) & (0.4) & (12983.34) \\
				&       & \multirow[t]{2}[0]{*}{375} & \textbf{0.96} & 1.07  & 1.23  & 1.08  & 5.61  & 1.95  & 2.62  & 2.74  & 3.35  & 1.52  & 1815.31 \\
				&       &       & \textbf{(0.21)} & (0.22) & (0.25) & (0.21) & (0.7) & (0.75) & (0.64) & (0.56) & (0.56) & (0.41) & (6117.41) \\
				&       & \multirow[t]{2}[1]{*}{450} & \textbf{1.04} & 1.10  & 1.47  & 1.08  & 5.64  & 2.11  & 2.69  & 2.89  & 3.38  & 1.66  & 2920.46 \\
				&       &       & \textbf{(0.21)} & (0.25) & (0.3) & (0.24) & (0.76) & (0.64) & (0.58) & (0.55) & (0.6) & (0.47) & (11731.58) \\
				\bottomrule
			\end{tabular}%
			\vspace{1ex}
			{\raggedright Note: RSA.o represents the RSA method with CV-determined parameters and RSA.f refers to the RSA method with fixed parameters, specifically $M=L=30$ and $p=0.1$. Values in bold indicate the top performers within the 95\% MCS test while values in parentheses represent the standard deviation of the reported MSFEs.\par}
		\end{threeparttable}
	}
	\label{tab:msfe910}%
\end{table}%

\subsection{MSFE comparison under random covariance matrix} \label{app:randomcov}

In empirical applications, covariates correlations are often complex and unpredictable, motivating our evaluation of model performance under randomly generated covariance structures. Tables \ref{tab:polyrhoRA} to \ref{tab:msfemorecovRA} report the corresponding MSFE results for Sections \ref{sec3.2} and \ref{app:manyrelevantvariables}. Table \ref{tab:polyrhoRA} shows the MSFE results under polynomially decaying coefficients. In this setting, RSA.o consistently delivering superior out-of-sample predictive accuracy, especially when both the sample size $N$ and the number of covariates $K$ are large. 

\begin{table}[htbp]
	\centering
	\caption{MSFE comparison for random covariance matrix (polynomial decay).}
	\scalebox{0.80}{
		\begin{threeparttable}
			\begin{tabular}{ccccccccccccc}
				\toprule
				$N$     & $K$     & RSA.o & RSA.f & RSA.1 & RSR   & RF    & Lasso & SCAD  & MCP   & PMA   & MCV   & MMA \\
				\midrule
				\multirow[t]{8}[2]{*}{200} & \multirow[t]{2}[1]{*}{20} & 0.18  & 0.28  & 0.23  & 1.43  & 0.57  & 0.15  & 0.17  & 0.18  & 0.29  & 0.27  & \textbf{0.15} \\
				&       & (0.06) & (0.07) & (0.07) & (0.24) & (0.13) & (0.05) & (0.06) & (0.06) & (0.09) & (0.07) & \textbf{(0.05)} \\
				& \multirow[t]{2}[0]{*}{100} & \textbf{0.34} & \textbf{0.33} & 0.43  & 0.48  & 0.64  & 0.44  & 0.40  & 0.42  & 0.42  & 0.42  & 0.37 \\
				&       & \textbf{(0.06)} & \textbf{(0.07)} & (0.08) & (0.07) & (0.1) & (0.2) & (0.15) & (0.19) & (0.12) & (0.07) & (0.1) \\
				& \multirow[t]{2}[0]{*}{200} & 0.62  & 0.64  & 0.75  & \textbf{0.60} & 1.10  & 1.15  & 0.99  & 1.06  & 0.96  & 0.73  & 783.16 \\
				&       & (0.15) & (0.16) & (0.15) & \textbf{(0.12)} & (0.22) & (0.67) & (0.22) & (0.25) & (0.22) & (0.17) & (2029.71) \\
				& \multirow[t]{2}[1]{*}{300} & 0.63  & \textbf{0.60} & 0.71  & 0.62  & 0.83  & 0.90  & 0.84  & 0.87  & 0.79  & 0.71  & 213.14 \\
				&       & (0.1) & \textbf{(0.09)} & (0.11) & (0.1) & (0.13) & (0.27) & (0.17) & (0.18) & (0.14) & (0.13) & (532.68) \\
				\midrule
				\multirow[t]{8}[2]{*}{400} & \multirow[t]{2}[1]{*}{40} & 0.11  & 0.22  & 0.17  & 0.53  & 0.45  & 0.05  & 0.05  & 0.05  & 0.13  & 0.28  & \textbf{0.05} \\
				&       & (0.02) & (0.03) & (0.03) & (0.05) & (0.05) & (0.01) & (0.01) & (0.01) & (0.02) & (0.09) & \textbf{(0.01)} \\
				& \multirow[t]{2}[0]{*}{200} & \textbf{0.30} & 0.36  & 0.35  & 0.55  & 0.71  & 0.43  & 0.37  & 0.37  & 0.36  & 0.45  & 0.42 \\
				&       & \textbf{(0.05)} & (0.05) & (0.06) & (0.06) & (0.08) & (0.14) & (0.08) & (0.08) & (0.08) & (0.06) & (0.08) \\
				& \multirow[t]{2}[0]{*}{400} & \textbf{0.56} & \textbf{0.57} & 0.62  & 0.60  & 0.94  & 0.95  & 0.86  & 0.87  & 0.77  & 0.68  & 19434.24 \\
				&       & \textbf{(0.07)} & \textbf{(0.07)} & (0.09) & (0.06) & (0.12) & (0.34) & (0.15) & (0.16) & (0.13) & (0.08) & (185436.61) \\
				& \multirow[t]{2}[1]{*}{600} & 0.95  & 0.81  & 1.10  & \textbf{0.68} & 1.16  & 1.42  & 1.25  & 1.42  & 1.05  & 0.86  & 5772.27 \\
				&       & (0.15) & (0.13) & (0.2) & \textbf{(0.11)} & (0.18) & (0.74) & (0.3) & (0.31) & (0.2) & (0.18) & (33637.05) \\
				\midrule
				\multirow[t]{8}[2]{*}{800} & \multirow[t]{2}[1]{*}{80} & 0.19  & 0.36  & 0.26  & 0.76  & 0.73  & 0.15  & 0.16  & 0.16  & 0.26  & 0.29  & \textbf{0.14} \\
				&       & (0.02) & (0.04) & (0.04) & (0.06) & (0.06) & (0.03) & (0.03) & (0.04) & (0.03) & (0.04) & \textbf{(0.02)} \\
				& \multirow[t]{2}[0]{*}{400} & \textbf{0.34} & 0.43  & 0.41  & 0.59  & 0.77  & 0.52  & 0.43  & 0.45  & 0.43  & 0.46  & 0.57 \\
				&       & \textbf{(0.04)} & (0.05) & (0.05) & (0.05) & (0.06) & (0.11) & (0.06) & (0.07) & (0.08) & (0.05) & (0.07) \\
				& \multirow[t]{2}[0]{*}{800} & \textbf{0.48} & 0.48  & 0.54  & 0.61  & 0.78  & 0.66  & 0.56  & 0.56  & 0.51  & 0.60  & 871.22 \\
				&       & \textbf{(0.05)} & (0.05) & (0.05) & (0.05) & (0.06) & (0.13) & (0.07) & (0.06) & (0.07) & (0.05) & (2035.43) \\
				& \multirow[t]{2}[1]{*}{1200} & \textbf{0.60} & 0.60  & 0.65  & 0.61  & 0.83  & 1.00  & 0.85  & 0.91  & 0.81  & 0.68  & 3464.36 \\
				&       & \textbf{(0.05)} & (0.06) & (0.06) & (0.05) & (0.07) & (0.23) & (0.09) & (0.11) & (0.08) & (0.06) & (12970.35) \\
				\bottomrule
			\end{tabular}%
			\label{tab:addlabel}%
			\vspace{1ex}
			{\raggedright Note: RSA.o represents the RSA method with CV-determined parameters and RSA.f refers to the RSA method with fixed parameters, specifically $M=L=30$ and $p=0.1$. Values in bold indicate the top performers within the 95\% MCS test while values in parentheses represent the standard deviation of the reported MSFEs.\par}
		\end{threeparttable}
	}
	\label{tab:polyrhoRA}%
\end{table}%

Table \ref{tab:expRAcov} presents the MSFE results under an exponentially decaying coefficient structure. The findings closely mirror those in Table \ref{tab:polyrhoRA}, with RSA exhibiting strong predictive performance in most settings.

\begin{table}[htbp]
	\centering
	\caption{MSFE comparison for random covariance matrix (exponential decay).}
	\scalebox{0.80}{
		\begin{threeparttable}
			\begin{tabular}{ccccccccccccc}
				\toprule
				$N$     & $K$     & RSA.o & RSA.f & RSA.1 & RSR   & RF    & Lasso & SCAD  & MCP   & PMA   & MCV   & MMA \\
				\midrule
				\multirow[t]{8}[2]{*}{200} & \multirow[t]{2}[1]{*}{20} & \textbf{0.19} & 0.25  & 0.23  & 1.26  & 0.60  & 0.21  & 0.23  & 0.23  & 0.34  & 0.31  & \textbf{0.19} \\
				&       & \textbf{(0.06)} & (0.07) & (0.08) & (0.23) & (0.12) & (0.06) & (0.08) & (0.08) & (0.11) & (0.09) & \textbf{(0.07)} \\
				& \multirow[t]{2}[0]{*}{100} & \textbf{0.29} & 0.32  & 0.33  & 0.53  & 0.57  & 0.45  & 0.43  & 0.44  & 0.37  & 0.43  & 0.52 \\
				&       & \textbf{(0.07)} & (0.06) & (0.07) & (0.07) & (0.1) & (0.17) & (0.15) & (0.16) & (0.08) & (0.09) & (0.12) \\
				& \multirow[t]{2}[0]{*}{200} & 0.36  & \textbf{0.34} & 0.43  & 0.37  & 0.55  & 0.58  & 0.54  & 0.57  & 0.53  & 0.43  & 107.61 \\
				&       & (0.06) & \textbf{(0.05)} & (0.07) & (0.05) & (0.08) & (0.2) & (0.12) & (0.11) & (0.1) & (0.07) & (253.29) \\
				& \multirow[t]{2}[1]{*}{300} & \textbf{0.42} & 0.43  & 0.50  & 0.60  & 0.77  & 0.59  & 0.52  & 0.54  & 0.50  & 0.56  & 608.35 \\
				&       & \textbf{(0.08)} & (0.09) & (0.1) & (0.09) & (0.13) & (0.18) & (0.11) & (0.13) & (0.11) & (0.1) & (3945.3) \\
				\midrule
				\multirow[t]{8}[2]{*}{400} & \multirow[t]{2}[1]{*}{40} & 0.20  & 0.38  & 0.25  & 1.10  & 0.87  & 0.19  & 0.21  & 0.21  & 0.36  & 0.37  & \textbf{0.17} \\
				&       & (0.04) & (0.07) & (0.06) & (0.12) & (0.09) & (0.05) & (0.06) & (0.06) & (0.07) & (0.07) & \textbf{(0.05)} \\
				& \multirow[t]{2}[0]{*}{200} & 0.45  & \textbf{0.42} & 0.50  & 0.65  & 0.92  & 0.66  & 0.55  & 0.57  & 0.55  & 0.55  & 0.79 \\
				&       & (0.08) & \textbf{(0.08)} & (0.09) & (0.08) & (0.13) & (0.26) & (0.13) & (0.13) & (0.12) & (0.08) & (0.15) \\
				& \multirow[t]{2}[0]{*}{400} & \textbf{0.72} & \textbf{0.71} & 0.82  & 0.77  & 1.23  & 1.24  & 1.09  & 1.18  & 1.10  & 0.90  & 820.55 \\
				&       & \textbf{(0.11)} & \textbf{(0.11)} & (0.13) & (0.09) & (0.17) & (0.51) & (0.22) & (0.25) & (0.21) & (0.12) & (2538.56) \\
				& \multirow[t]{2}[1]{*}{600} & 0.46  & \textbf{0.44} & 0.54  & 0.53  & 0.69  & 0.66  & 0.56  & 0.58  & 0.53  & 0.59  & 564.69 \\
				&       & (0.05) & \textbf{(0.05)} & (0.07) & (0.05) & (0.07) & (0.19) & (0.09) & (0.1) & (0.08) & (0.06) & (2659.09) \\
				\midrule
				\multirow[t]{8}[2]{*}{800} & \multirow[t]{2}[1]{*}{80} & 0.16  & 0.28  & 0.19  & 0.46  & 0.54  & 0.15  & 0.16  & 0.16  & 0.28  & 0.33  & \textbf{0.13} \\
				&       & (0.02) & (0.03) & (0.02) & (0.04) & (0.04) & (0.02) & (0.03) & (0.03) & (0.03) & (0.04) & \textbf{(0.02)} \\
				& \multirow[t]{2}[0]{*}{400} & 0.36  & 0.40  & 0.43  & 0.55  & 0.70  & 0.40  & \textbf{0.36} & 0.36  & \textbf{0.35} & 0.46  & 0.46 \\
				&       & (0.03) & (0.03) & (0.04) & (0.04) & (0.05) & (0.08) & \textbf{(0.05)} & (0.05) & \textbf{(0.05)} & (0.05) & (0.05) \\
				& \multirow[t]{2}[0]{*}{800} & \textbf{0.58} & \textbf{0.58} & 0.66  & 0.68  & 0.95  & 1.07  & 0.81  & 0.82  & 0.76  & 0.72  & 1476.47 \\
				&       & \textbf{(0.06)} & \textbf{(0.06)} & (0.06) & (0.05) & (0.09) & (0.33) & (0.12) & (0.12) & (0.13) & (0.06) & (4792.74) \\
				& \multirow[t]{2}[1]{*}{1200} & 1.13  & \textbf{0.75} & 1.31  & 0.79  & 1.10  & 1.16  & 0.91  & 0.93  & 0.86  & 0.84  & 1784.91 \\
				&       & (0.1) & \textbf{(0.08)} & (0.14) & (0.06) & (0.09) & (0.44) & (0.13) & (0.11) & (0.1) & (0.09) & (5095.25) \\
				\bottomrule
			\end{tabular}%
			\vspace{1ex}
			{\raggedright Note: RSA.o represents the RSA method with CV-determined parameters and RSA.f refers to the RSA method with fixed parameters, specifically $M=L=30$ and $p=0.1$. Values in bold indicate the top performers within the 95\% MCS test while values in parentheses represent the standard deviation of the reported MSFEs.\par}
		\end{threeparttable}
	}
	\label{tab:expRAcov}%
\end{table}%

Table \ref{tab:msfemorecovRA} presents the MSFE result in settings with many relevant covariates. RSA continue to deliver the highest out-of-sample prediction accuracy regardless of signal strength, while RSA with fixed parameters is selected most frequently by the MCS test at 95\% significance level. 

\begin{table}[htbp]
	\centering
	\caption{MSFE comparison under many relevant covariates.}
	\scalebox{0.80}{
		\begin{threeparttable}
			\begin{tabular}{cccccccccccccc}
				\toprule
				DGP &   $N$     & $K$     & RSA.o & RSA.f & RSA.1 & RSR   & RF    & Lasso & SCAD  & MCP   & PMA   & MCV   & MMA \\
				\midrule
				\multirow[t]{12}[2]{*}{poly} & \multirow[t]{6}[1]{*}{100} & \multirow[t]{2}[1]{*}{100} & 0.53  & 0.53  & 0.58  & \textbf{0.51} & 0.79  & 0.92  & 0.82  & 0.85  & 0.80  & 0.62  & 209.55 \\
				&       &       & (0.14) & (0.15) & (0.15) & \textbf{(0.13)} & (0.22) & (0.39) & (0.27) & (0.31) & (0.2) & (0.17) & (999.23) \\
				&       & \multirow[t]{2}[0]{*}{125} & \textbf{0.72} & 0.74  & 0.80  & 0.83  & 1.18  & 1.36  & 1.18  & 1.20  & 1.10  & 1.02  & 126.11 \\
				&       &       & \textbf{(0.18)} & (0.18) & (0.22) & (0.17) & (0.26) & (0.69) & (0.38) & (0.35) & (0.28) & (0.23) & (457.22) \\
				&       & \multirow[t]{2}[0]{*}{150} & 0.74  & 0.72  & 0.85  & \textbf{0.68} & 1.00  & 1.19  & 1.05  & 1.07  & 1.00  & 0.89  & 182.80 \\
				&       &       & (0.2) & (0.19) & (0.22) & \textbf{(0.15)} & (0.22) & (0.64) & (0.41) & (0.33) & (0.22) & (0.24) & (720.33) \\
				& \multirow[t]{6}[1]{*}{300} & \multirow[t]{2}[0]{*}{300} & 0.73  & \textbf{0.70} & 0.87  & \textbf{0.69} & 1.14  & 1.34  & 1.13  & 1.23  & 1.07  & 0.81  & 252055.87 \\
				&       &       & (0.11) & \textbf{(0.12)} & (0.14) & \textbf{(0.1)} & (0.18) & (0.65) & (0.23) & (0.25) & (0.17) & (0.13) & (2502293.78) \\
				&       & \multirow[t]{2}[0]{*}{375} & 0.67  & \textbf{0.65} & 0.75  & 0.69  & 0.99  & 1.02  & 0.91  & 0.95  & 0.91  & 0.77  & 877.74 \\
				&       &       & (0.1) & \textbf{(0.1)} & (0.11) & (0.09) & (0.14) & (0.32) & (0.15) & (0.17) & (0.15) & (0.11) & (2773.95) \\
				&       & \multirow[t]{2}[1]{*}{450} & \textbf{0.62} & \textbf{0.62} & 0.73  & 0.69  & 0.93  & 1.04  & 0.88  & 0.91  & 0.83  & 0.80  & 201.49 \\
				&       &       & \textbf{(0.09)} & \textbf{(0.09)} & (0.11) & (0.08) & (0.12) & (0.32) & (0.16) & (0.15) & (0.14) & (0.12) & (388.3) \\
				\bottomrule
				\multirow[t]{12}[2]{*}{exp} & \multirow[t]{6}[1]{*}{100} & \multirow[t]{2}[1]{*}{100} & 0.18  & 0.18  & 0.20  & \textbf{0.16} & 0.29  & 0.31  & 0.27  & 0.26  & 0.23  & 0.21  & 206.10 \\
				&       &       & (0.04) & (0.05) & (0.05) & \textbf{(0.04)} & (0.08) & (0.13) & (0.1) & (0.08) & (0.06) & (0.06) & (1552.5) \\
				&       & \multirow[t]{2}[0]{*}{125} & 0.16  & 0.16  & 0.19  & \textbf{0.15} & 0.23  & 0.25  & 0.25  & 0.25  & 0.21  & 0.19  & 43.31 \\
				&       &       & (0.04) & (0.04) & (0.05) & \textbf{(0.04)} & (0.05) & (0.09) & (0.08) & (0.07) & (0.06) & (0.05) & (286.79) \\
				&       & \multirow[t]{2}[0]{*}{150} & 0.24  & 0.22  & 0.26  & \textbf{0.20} & 0.37  & 0.40  & 0.35  & 0.35  & 0.32  & 0.28  & 81.71 \\
				&       &       & (0.08) & (0.06) & (0.07) & \textbf{(0.05)} & (0.1) & (0.2) & (0.14) & (0.15) & (0.11) & (0.09) & (436.55) \\
				& \multirow[t]{6}[1]{*}{300} & \multirow[t]{2}[0]{*}{300} & \textbf{0.10} & 0.11  & 0.11  & 0.14  & 0.18  & 0.15  & 0.13  & 0.14  & 0.12  & 0.15  & 34.22 \\
				&       &       & \textbf{(0.01)} & (0.02) & (0.02) & (0.02) & (0.02) & (0.04) & (0.02) & (0.02) & (0.02) & (0.02) & (98.95) \\
				&       & \multirow[t]{2}[0]{*}{375} & 0.11  & \textbf{0.11} & 0.14  & 0.13  & 0.16  & 0.14  & 0.12  & 0.13  & 0.11  & 0.12  & 33.82 \\
				&       &       & (0.01) & \textbf{(0.01)} & (0.02) & (0.02) & (0.02) & (0.03) & (0.02) & (0.02) & (0.02) & (0.02) & (90.97) \\
				&       & \multirow[t]{2}[1]{*}{450} & 0.17  & \textbf{0.16} & 0.18  & 0.17  & 0.24  & 0.28  & 0.23  & 0.24  & 0.21  & 0.20  & 275.30 \\
				&       &       & (0.02) & \textbf{(0.02)} & (0.03) & (0.02) & (0.03) & (0.1) & (0.04) & (0.05) & (0.03) & (0.03) & (1214.28) \\
				\bottomrule
			\end{tabular}%
			\vspace{1ex}
			{\raggedright Note: ``poly'' refers to polynomially decaying coefficients, while ``exp'' denotes exponentially decaying coefficient. RSA.o represents the RSA method with CV-determined parameters and RSA.f refers to the RSA method with fixed parameters, specifically $M=L=30$ and $p=0.1$. Values in bold indicate the top performers within the 95\% MCS test while values in parentheses represent the standard deviation of the reported MSFEs.\par}
		\end{threeparttable}
	}
	\label{tab:msfemorecovRA}%
\end{table}%

\subsection{CV selection for the selection probability and the number of candidate models} \label{app:cv}

\subsubsection{CV results for Section \ref{sec3.2}}
Figures~\ref{fig:case0.1}, \ref{fig:case0.9}, and \ref{fig:caseepolyRA} report the cross-validation results under three different scenarios: when the covariates are weakly correlated, highly correlated, and randomly correlated, respectively, with polynomially decaying regression coefficients. The results indicate that the selection probability plays a significant role in the performance of the RSA estimator.

\begin{figure*}[!h]
	\centering
	
	\begin{minipage}{0.95\textwidth}
		\centering
		\includegraphics[width=0.24\linewidth]{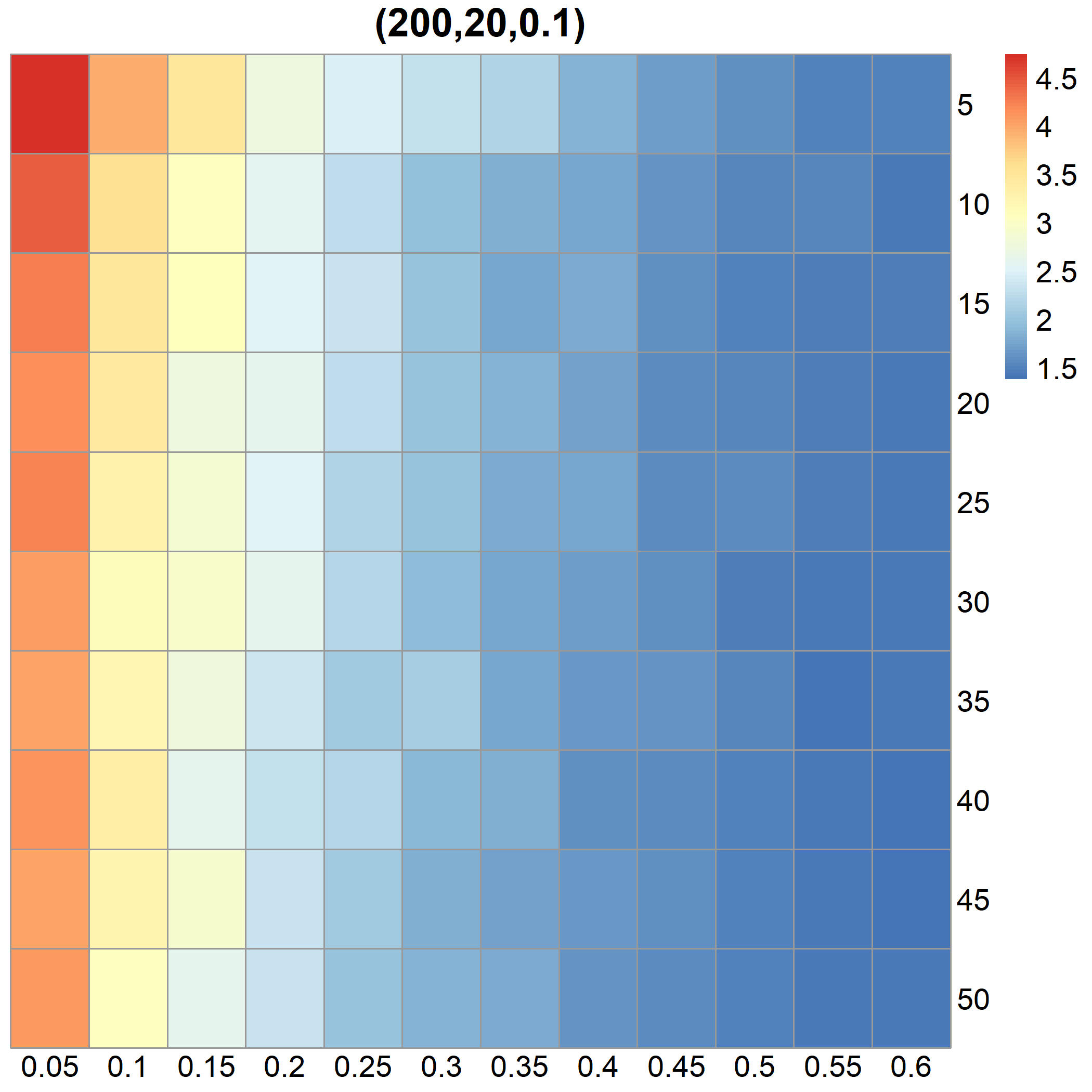}
		\includegraphics[width=0.24\linewidth]{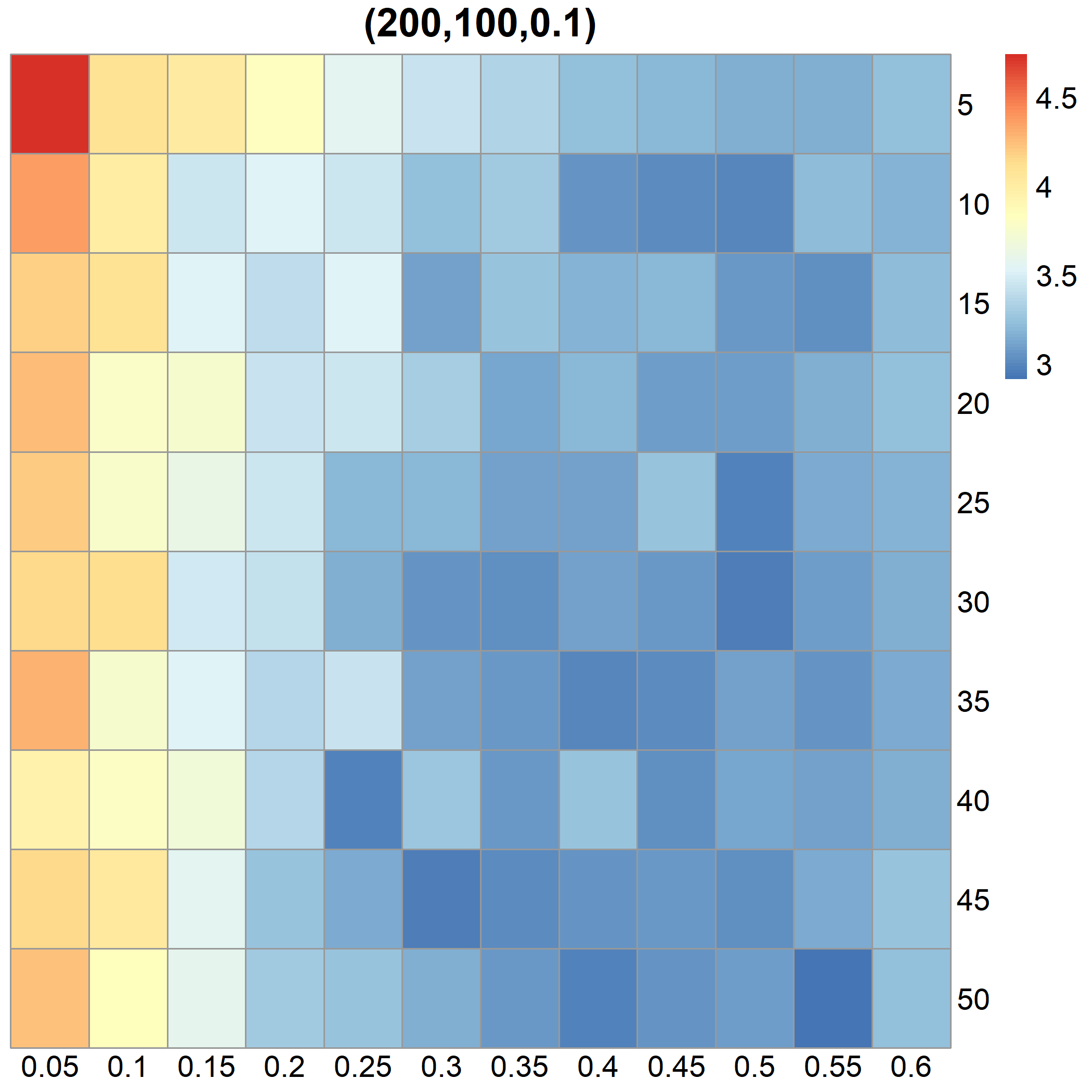}
		\includegraphics[width=0.24\linewidth]{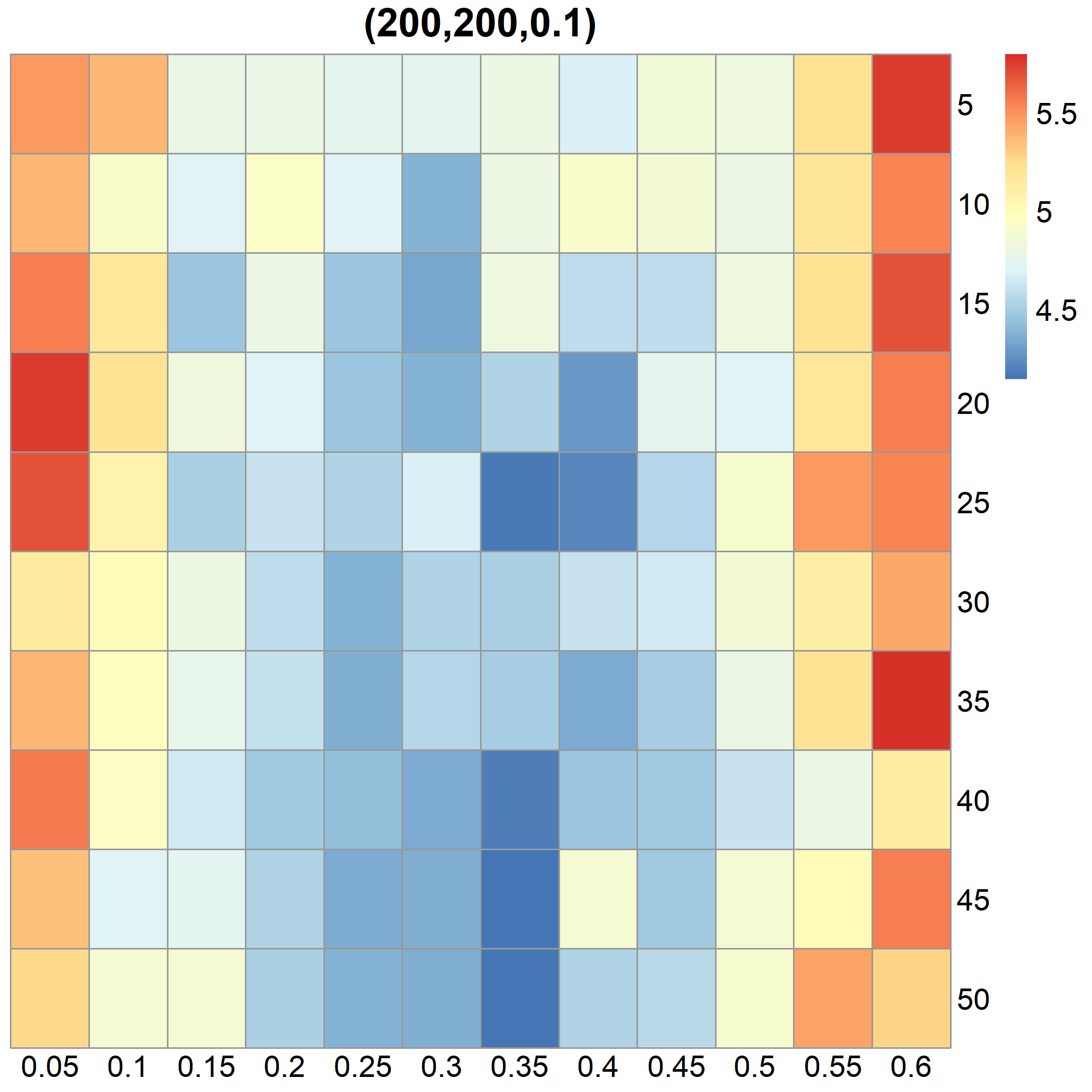}
		\includegraphics[width=0.24\linewidth]{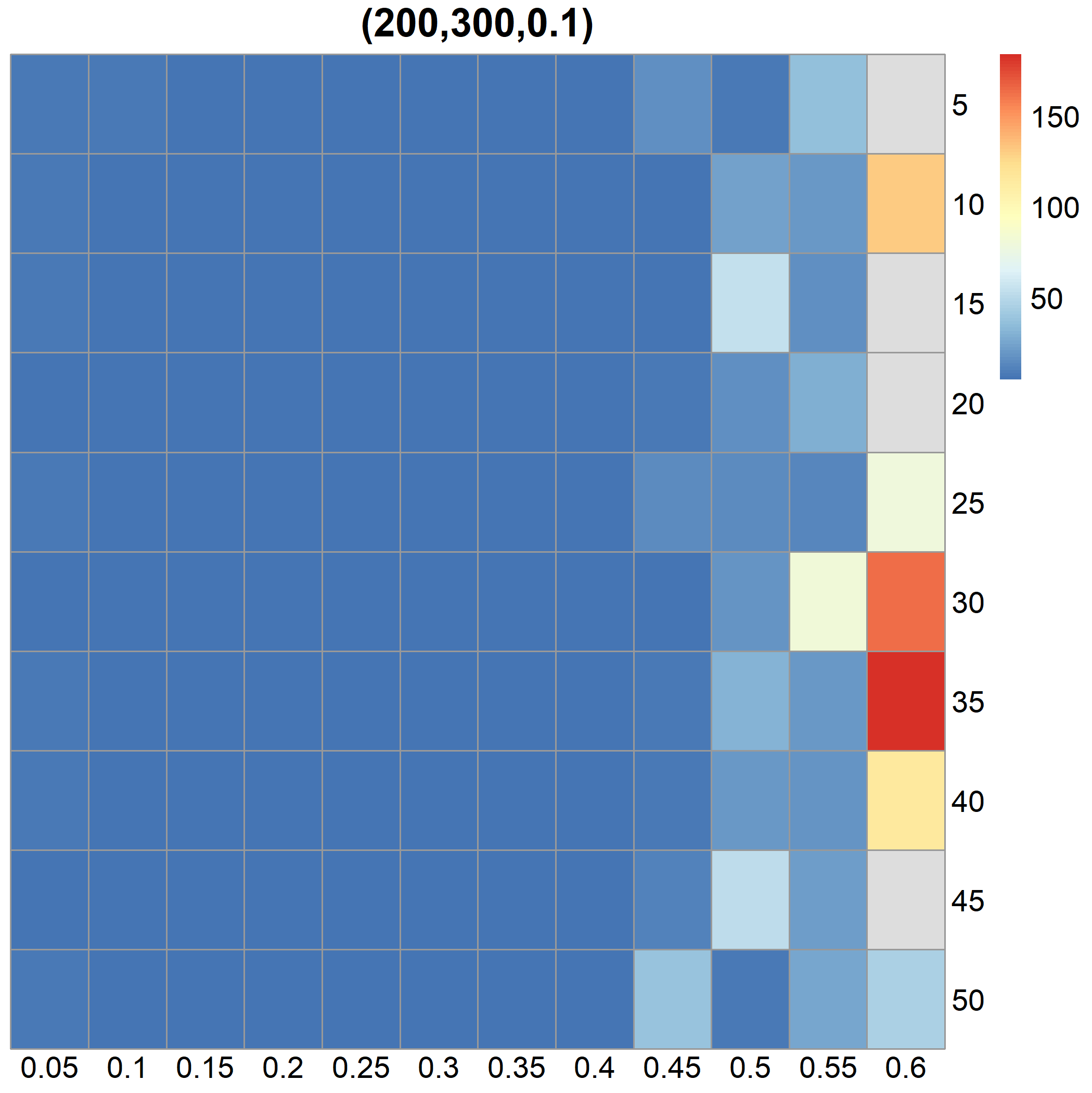}
		
		\vspace{3pt}
		\small (a) $N = 200$
	\end{minipage}
	
	\vspace{6pt}
	
	\begin{minipage}{0.95\textwidth}
		\centering
		\includegraphics[width=0.24\linewidth]{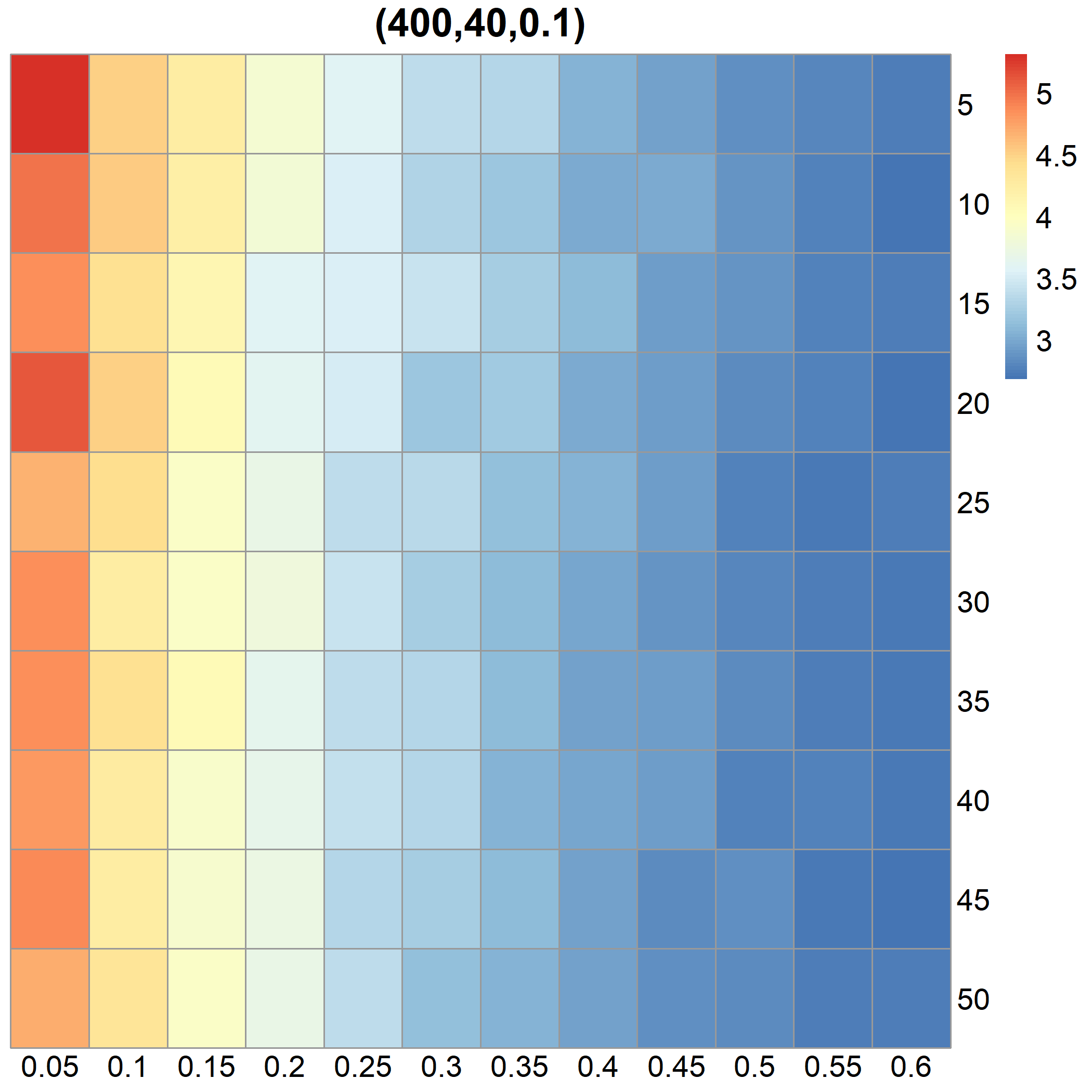}
		\includegraphics[width=0.24\linewidth]{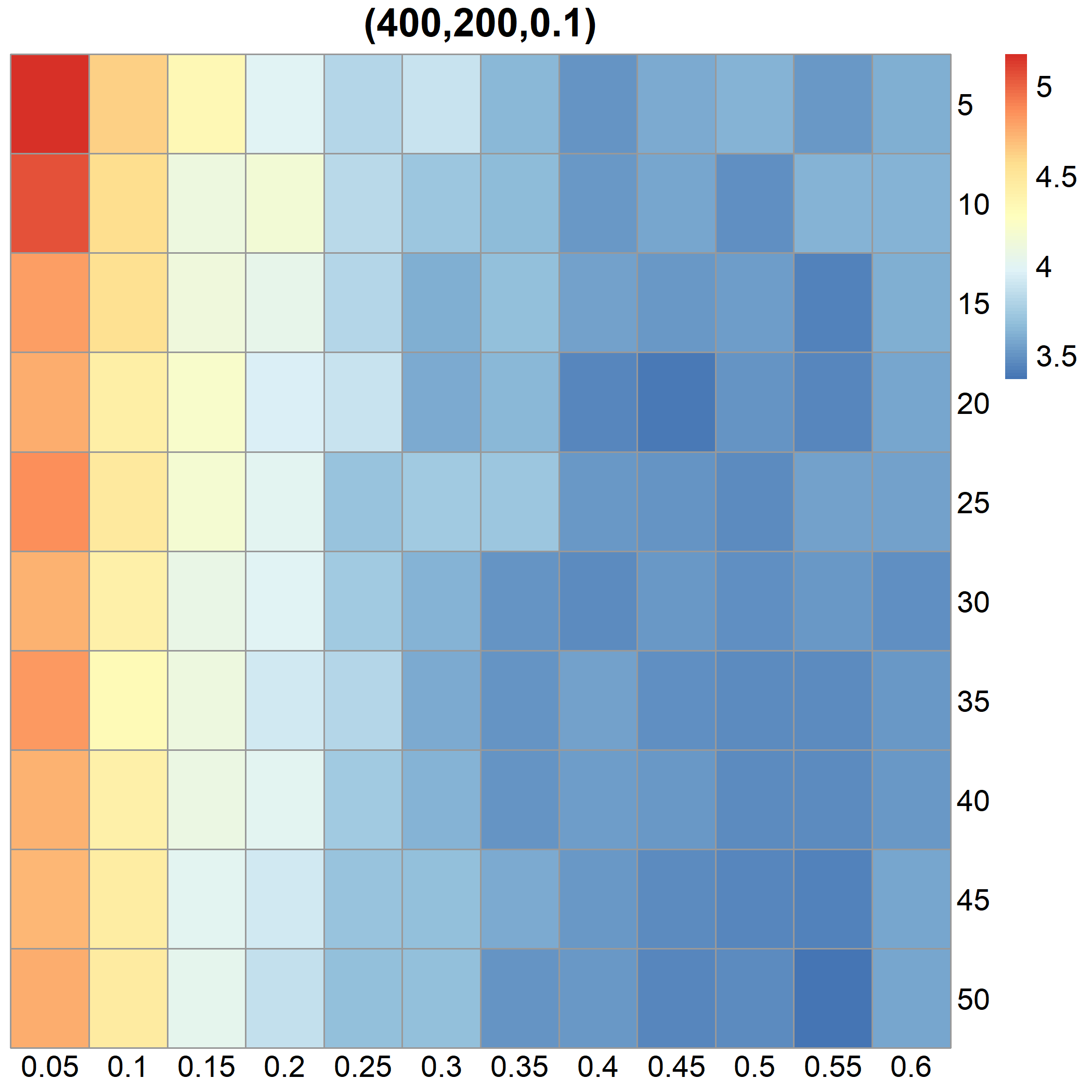}
		\includegraphics[width=0.24\linewidth]{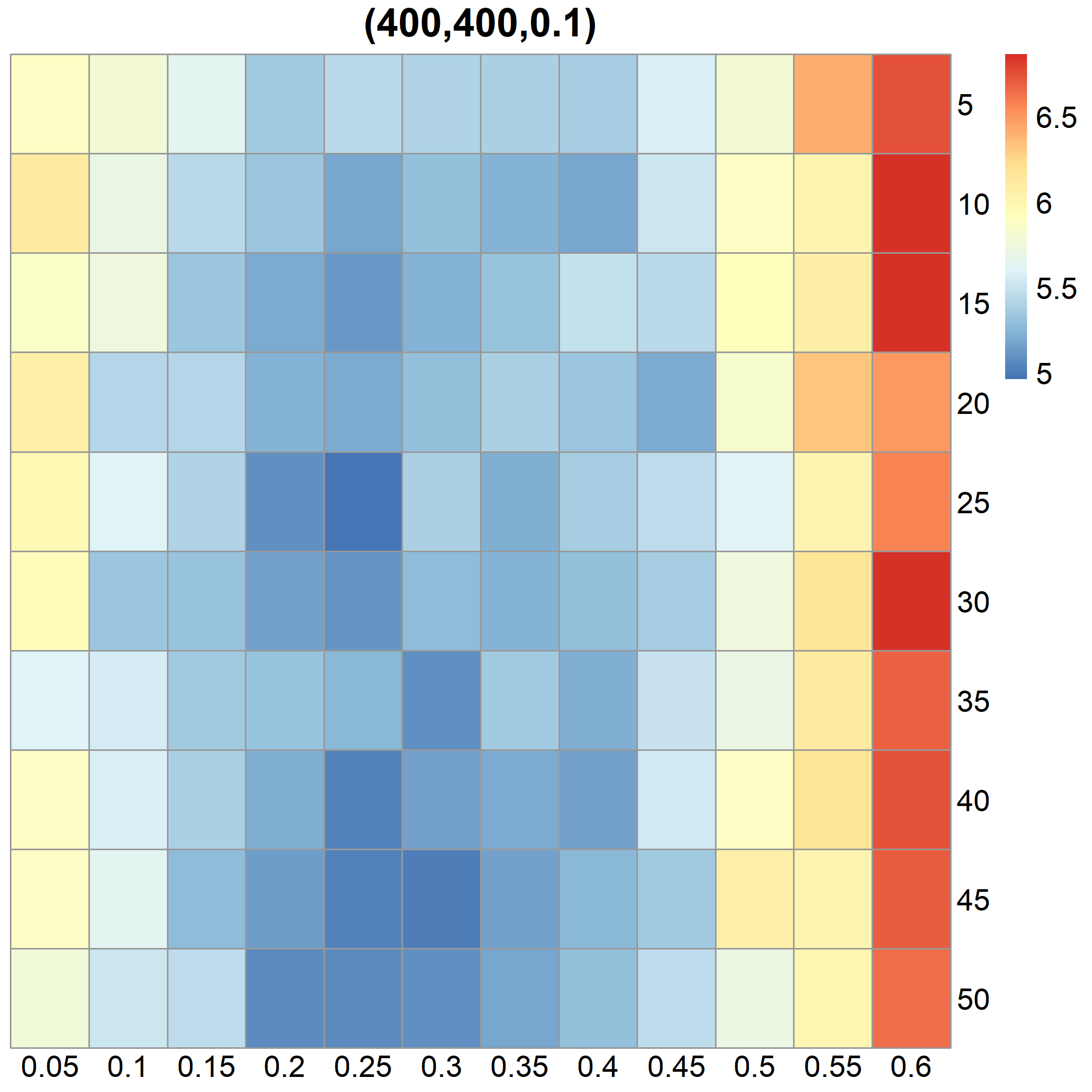}
		\includegraphics[width=0.24\linewidth]{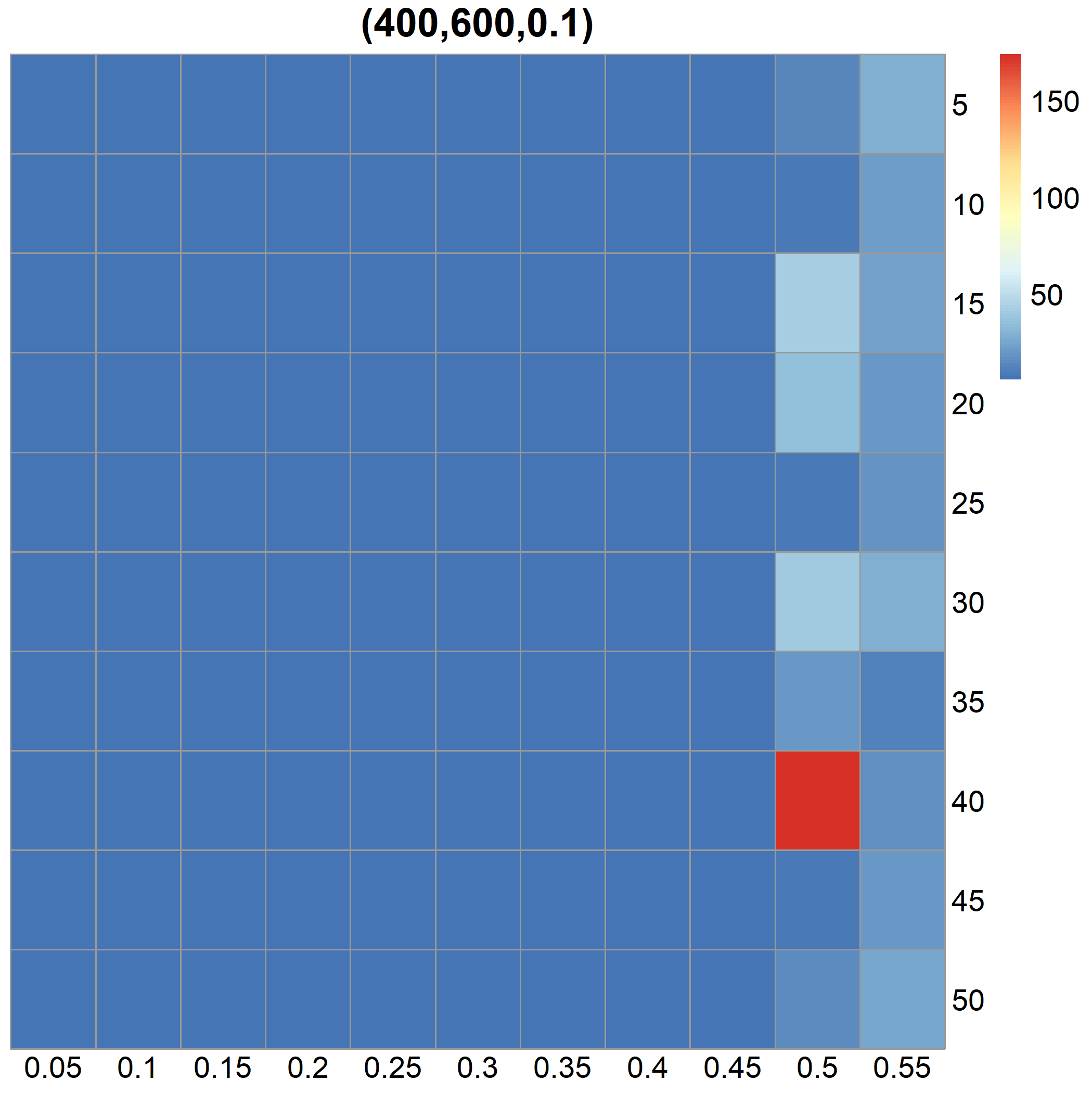}
		
		\vspace{3pt}
		\small (b) $N = 400$
	\end{minipage}
	
	\vspace{6pt}
	
	\begin{minipage}{0.95\textwidth}
		\centering
		\includegraphics[width=0.24\linewidth]{figs/3-FCOV-800-80-0.1.png}
		\includegraphics[width=0.24\linewidth]{figs/3-FCOV-800-400-0.1.png}
		\includegraphics[width=0.24\linewidth]{figs/3-FCOV-800-800-0.1.png}
		\includegraphics[width=0.24\linewidth]{figs/3-FCOV-800-1200-0.1.png}
		
		\vspace{3pt}
		\small (c) $N = 800$
	\end{minipage}
	
	\caption{Cross-validation results for Section \ref{sec3.2} under polynomial decay with $\rho = 0.1$. Values in parentheses denote $(N, K, \rho)$. The horizontal axis represents the selection probability $p$, and the vertical axis indicates the number of candidate models $M$. Darker regions correspond to $(p, M)$ combinations yielding lower cross-validation errors.}
	\label{fig:case0.1}
\end{figure*}

\newpage
\begin{figure*}[!h]
	\centering
	
	\begin{minipage}{0.95\textwidth}
		\centering
		\includegraphics[width=0.24\linewidth]{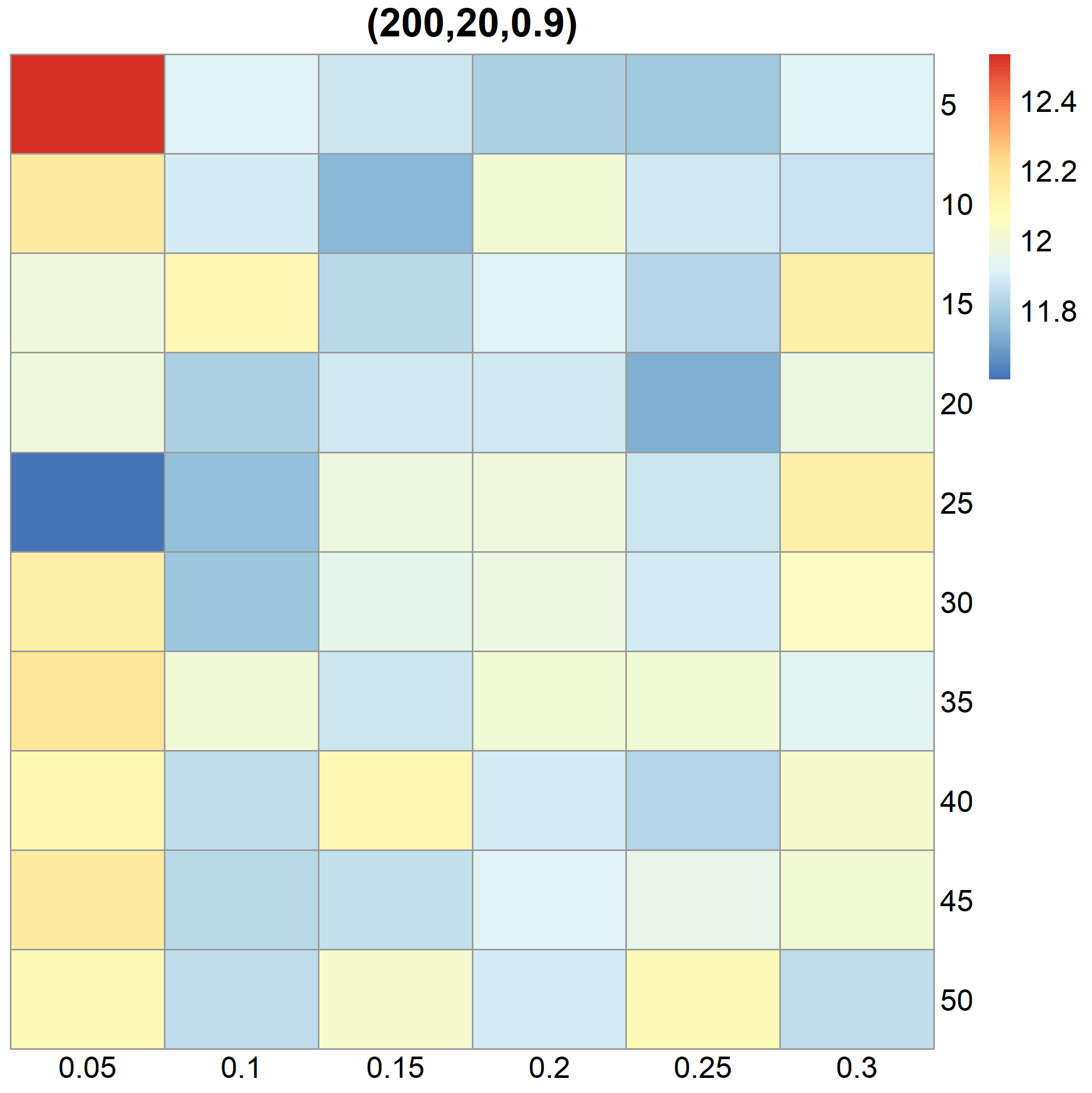}
		\includegraphics[width=0.24\linewidth]{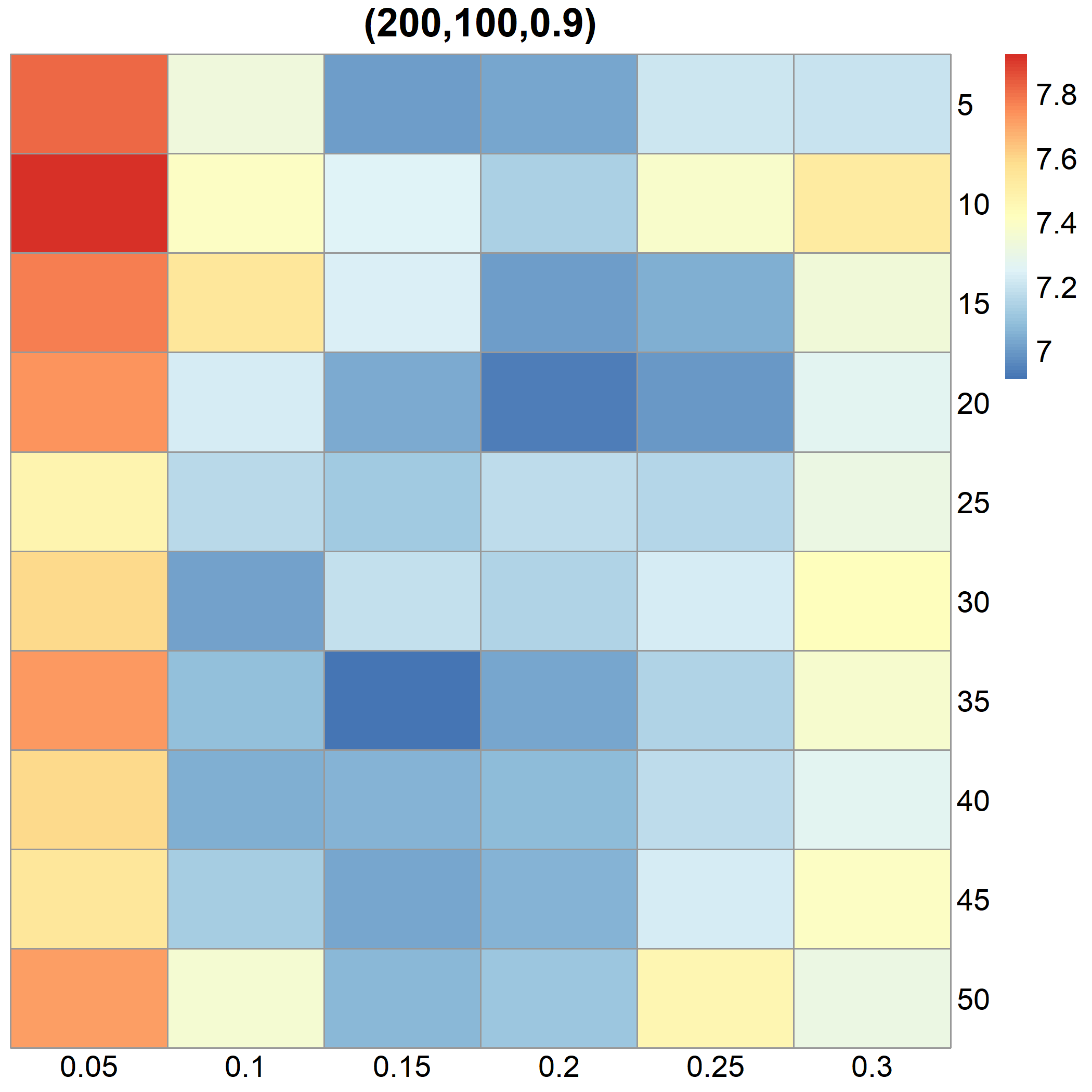}
		\includegraphics[width=0.24\linewidth]{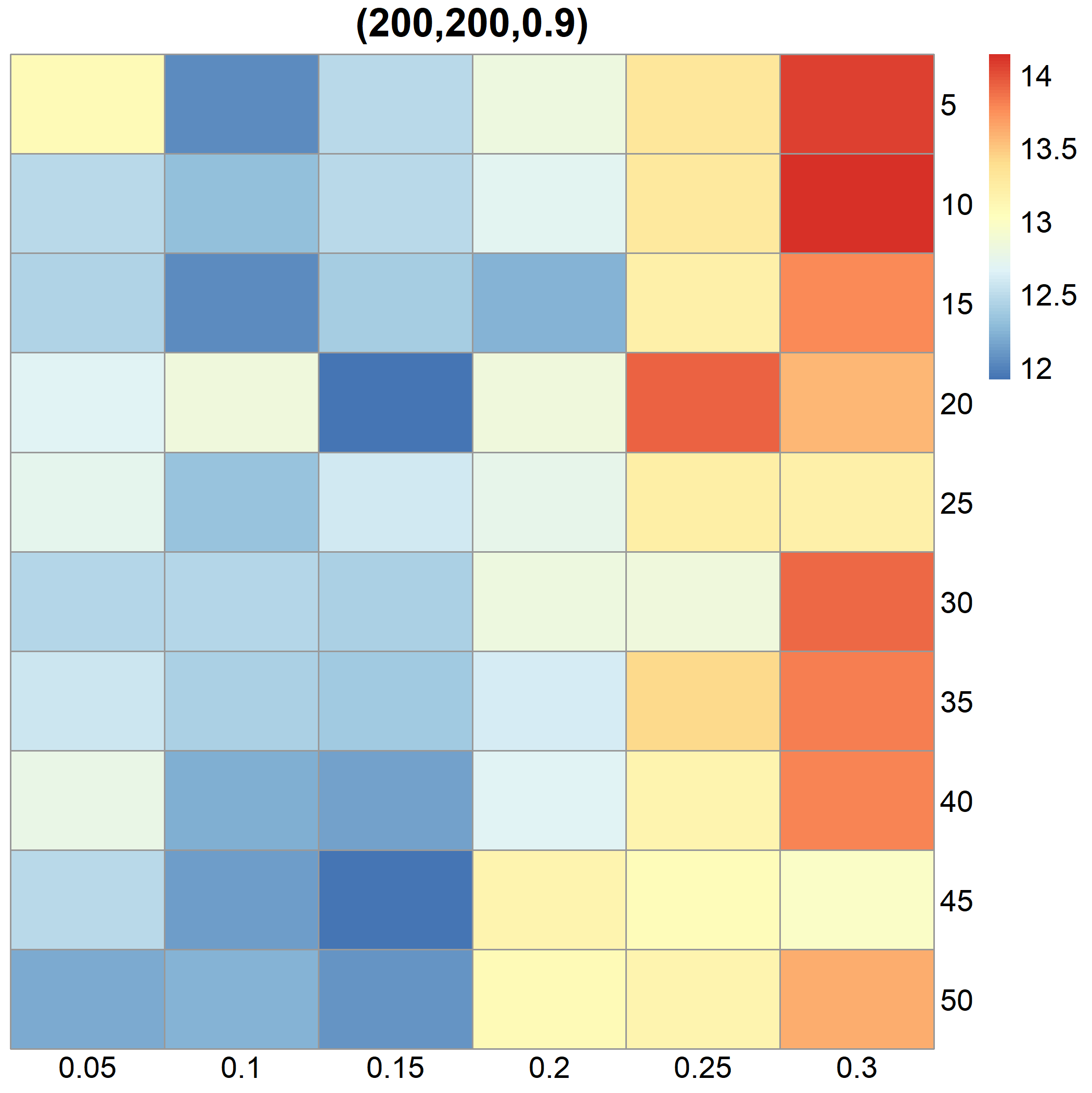}
		\includegraphics[width=0.24\linewidth]{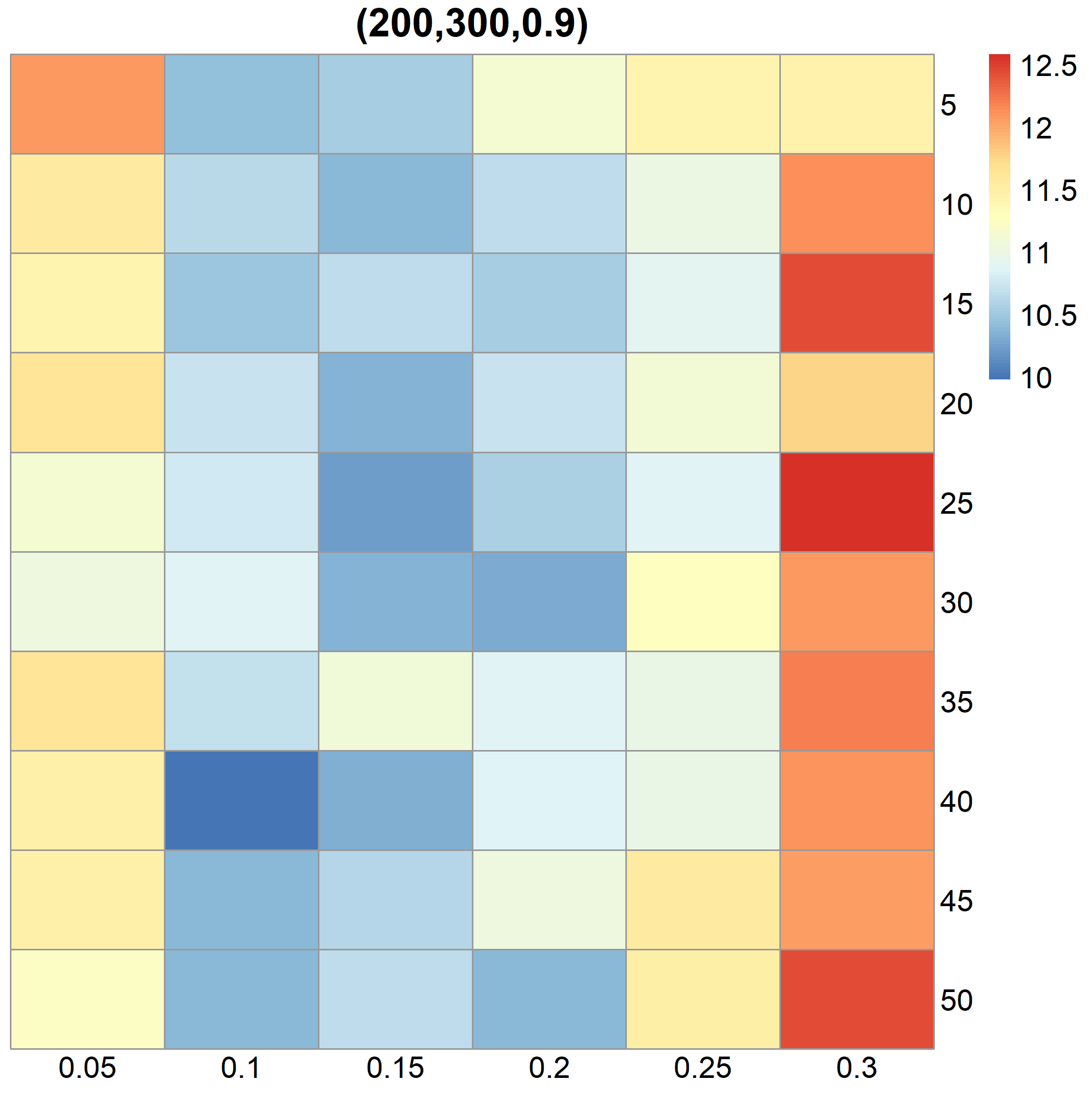}
		
		\vspace{3pt}
		\small (a) $N = 200$
	\end{minipage}
	
	\vspace{6pt}
	
	\begin{minipage}{0.95\textwidth}
		\centering
		\includegraphics[width=0.24\linewidth]{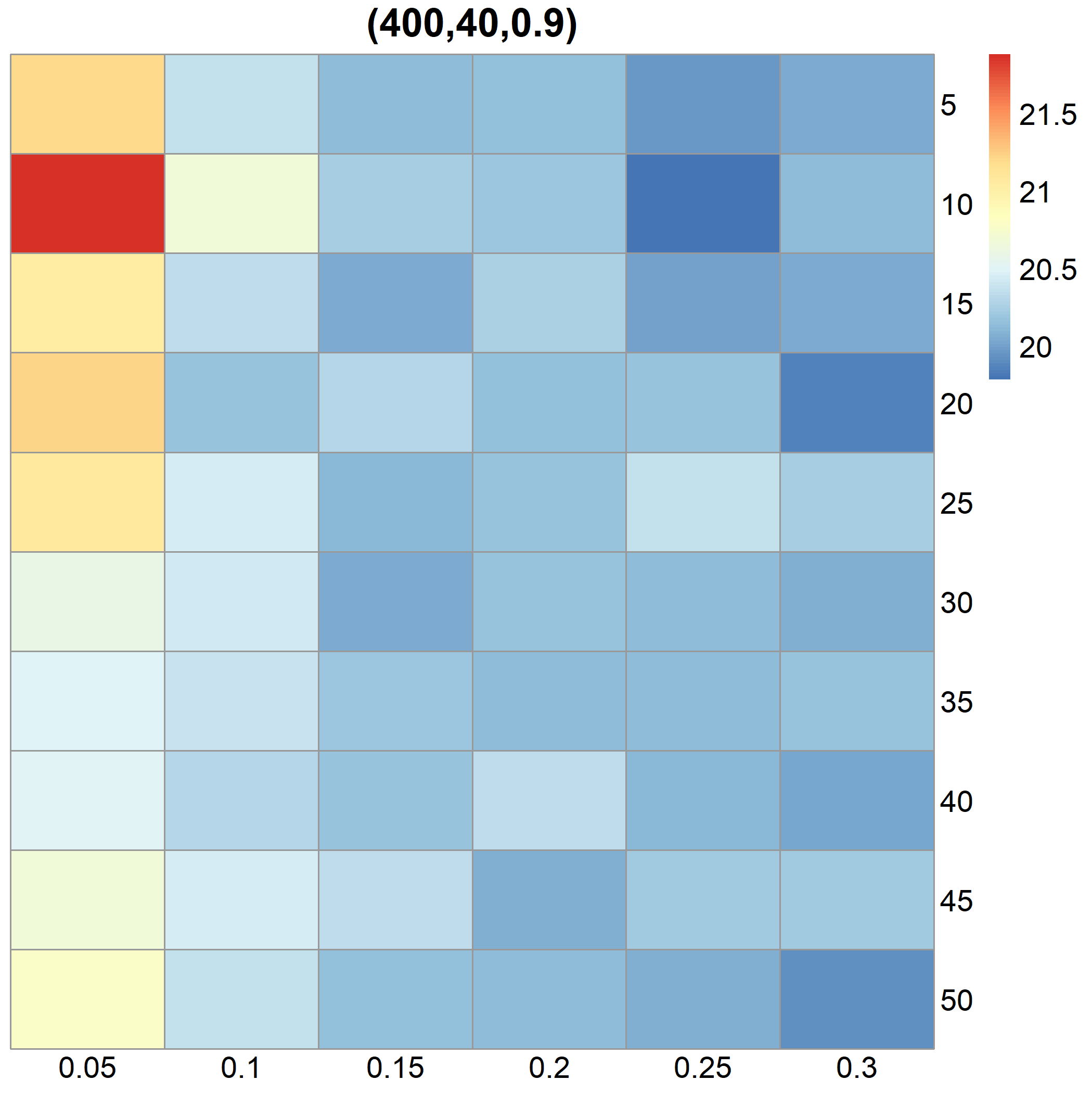}
		\includegraphics[width=0.24\linewidth]{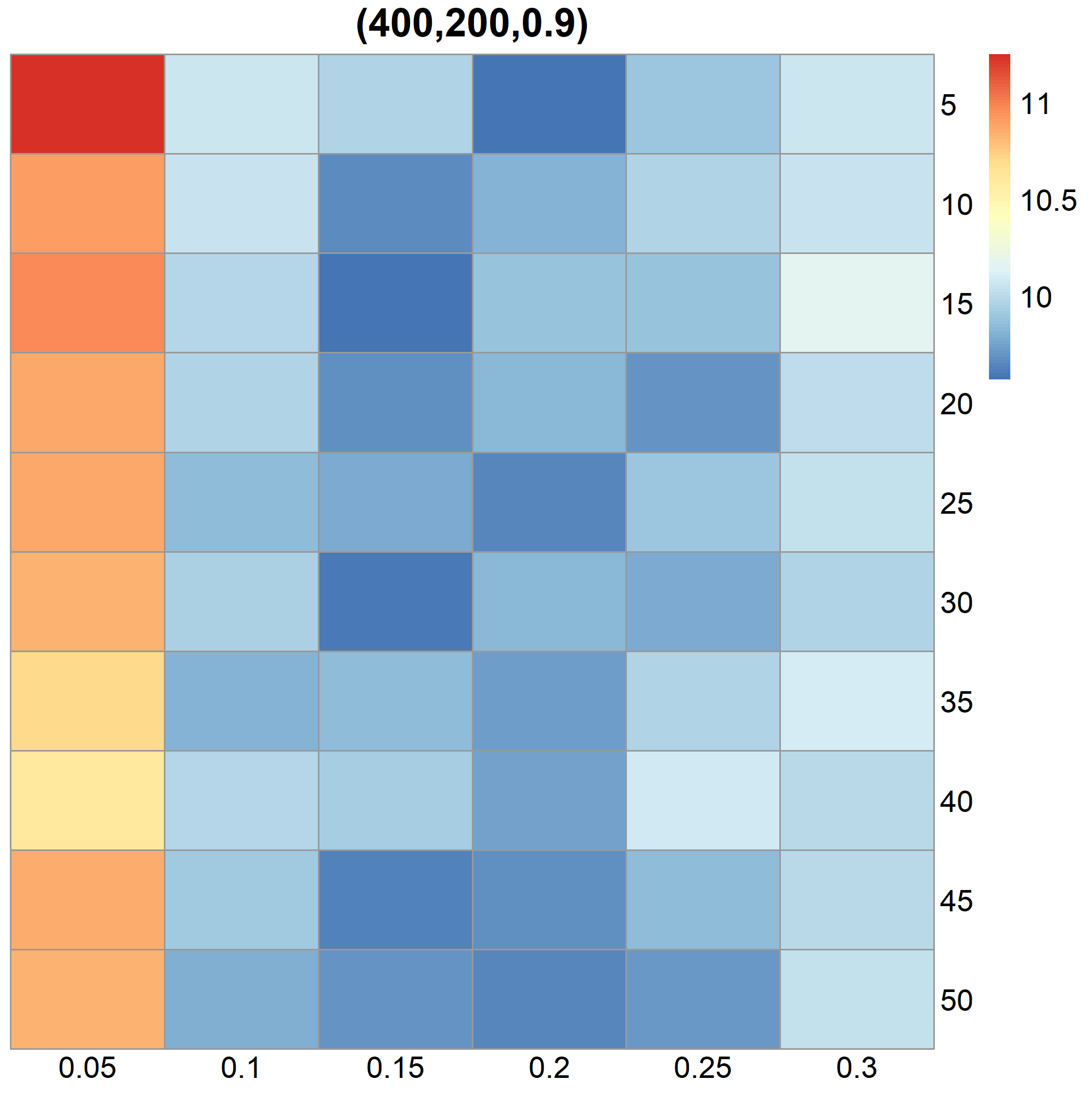}
		\includegraphics[width=0.24\linewidth]{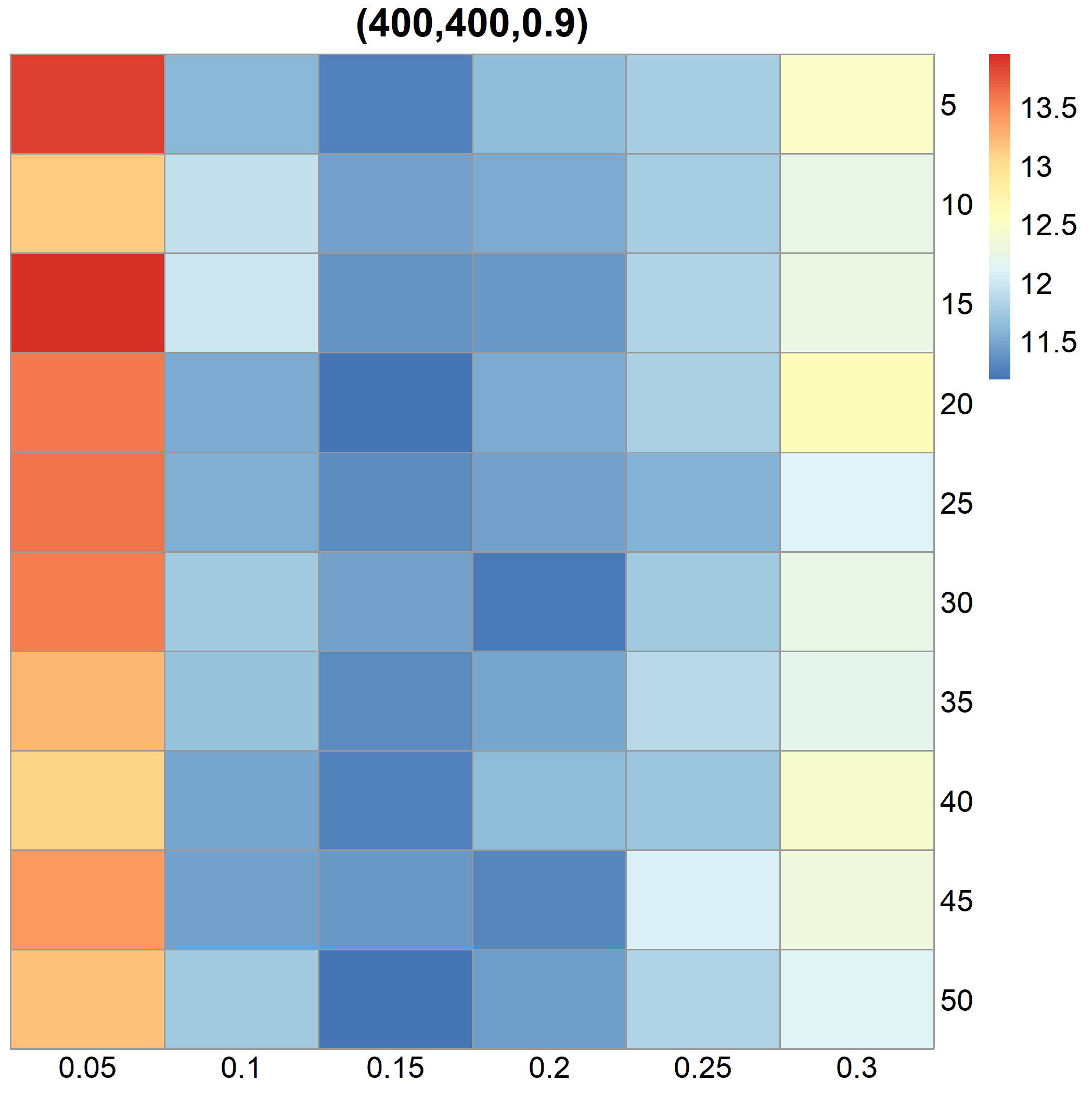}
		\includegraphics[width=0.24\linewidth]{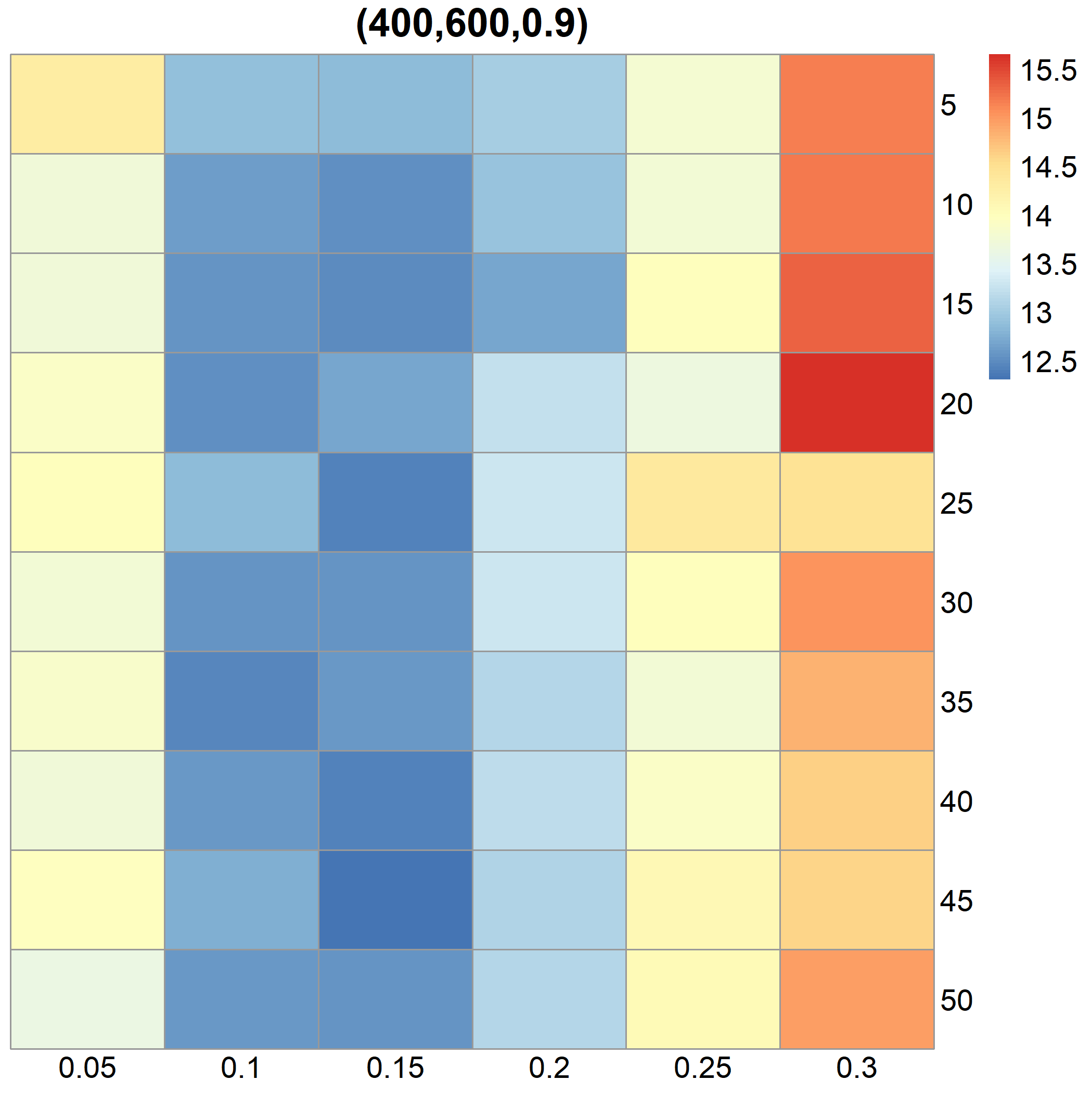}
		
		\vspace{3pt}
		\small (b) $N = 400$
	\end{minipage}
	
	\vspace{6pt}
	
	\begin{minipage}{0.95\textwidth}
		\centering
		\includegraphics[width=0.24\linewidth]{figs/3-FCOV-800-80-0.9.png}
		\includegraphics[width=0.24\linewidth]{figs/3-FCOV-800-400-0.9.png}
		\includegraphics[width=0.24\linewidth]{figs/3-FCOV-800-800-0.9.png}
		\includegraphics[width=0.24\linewidth]{figs/3-FCOV-800-1200-0.9.png}
		
		\vspace{3pt}
		\small (c) $N = 800$
	\end{minipage}
	
	\caption{Cross-validation results for Section \ref{sec3.2} under polynomial decay with $\rho = 0.9$. Values in parentheses denote $(N, K, \rho)$. The horizontal axis represents the selection probability $p$, and the vertical axis indicates the number of candidate models $M$. Darker regions correspond to $(p, M)$ combinations yielding lower cross-validation errors.}
	\label{fig:case0.9}
\end{figure*}

\newpage
\begin{figure*}[!h]
	\centering
	
	\begin{minipage}{0.95\textwidth}
		\centering
		\includegraphics[width=0.24\linewidth]{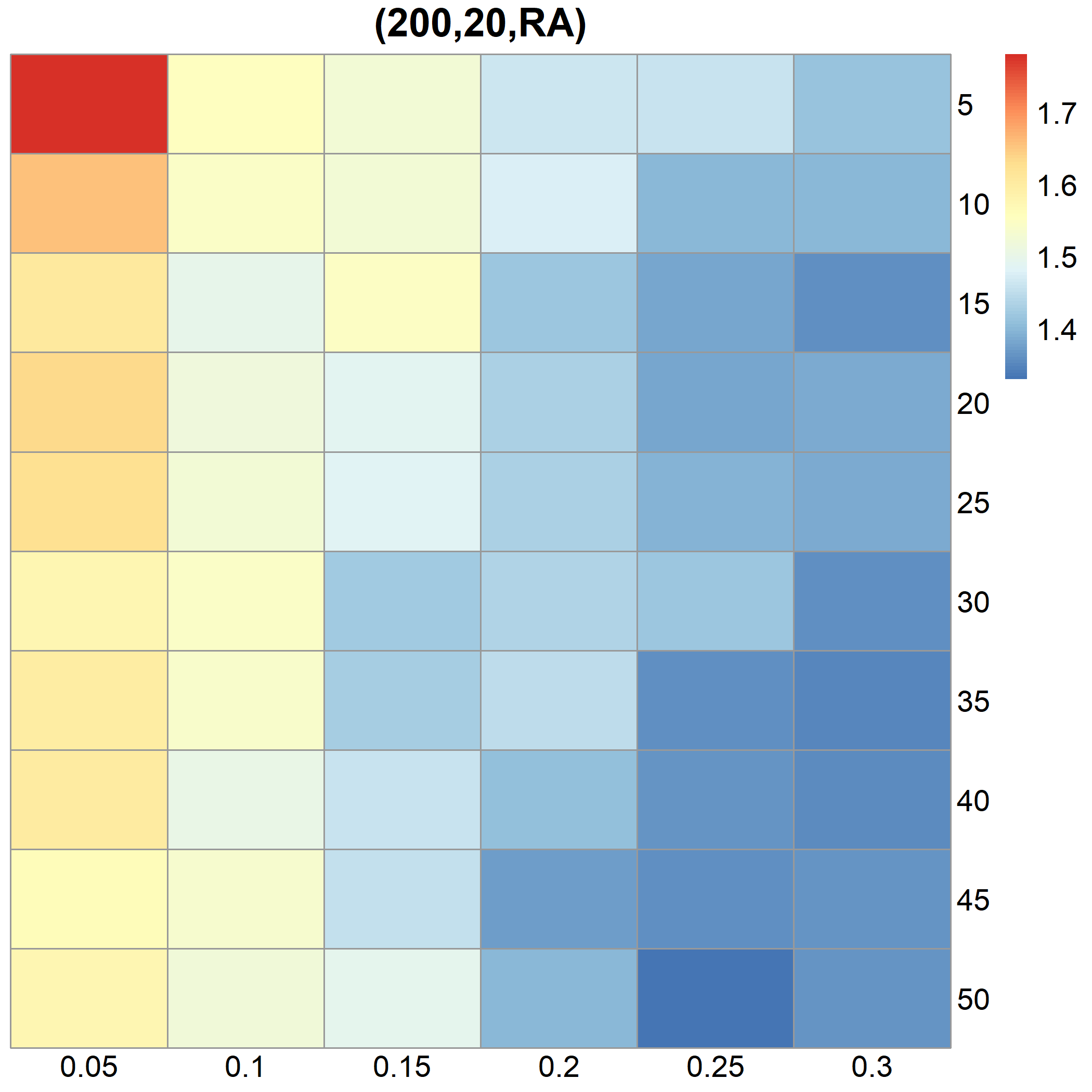}
		\includegraphics[width=0.24\linewidth]{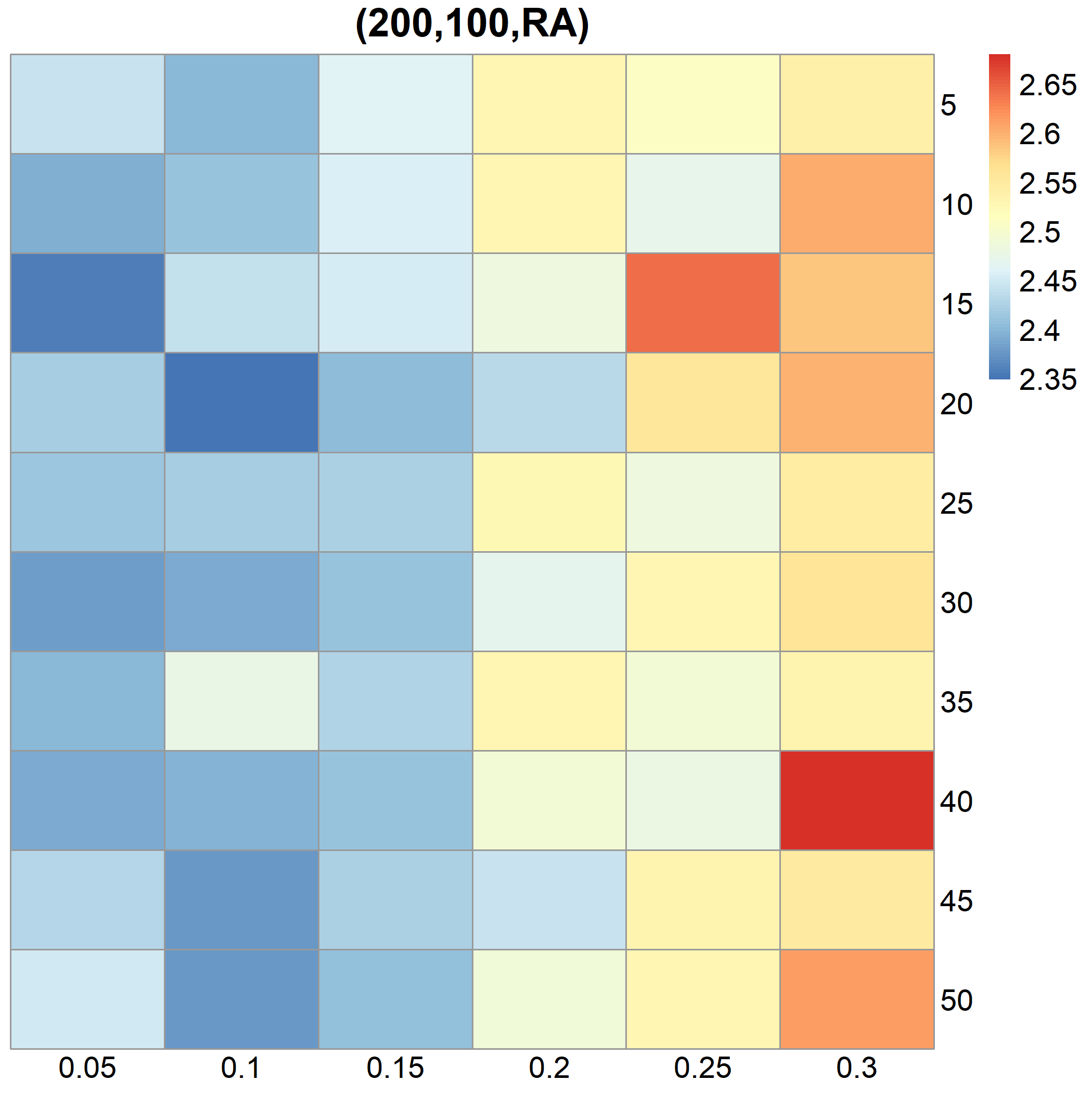}
		\includegraphics[width=0.24\linewidth]{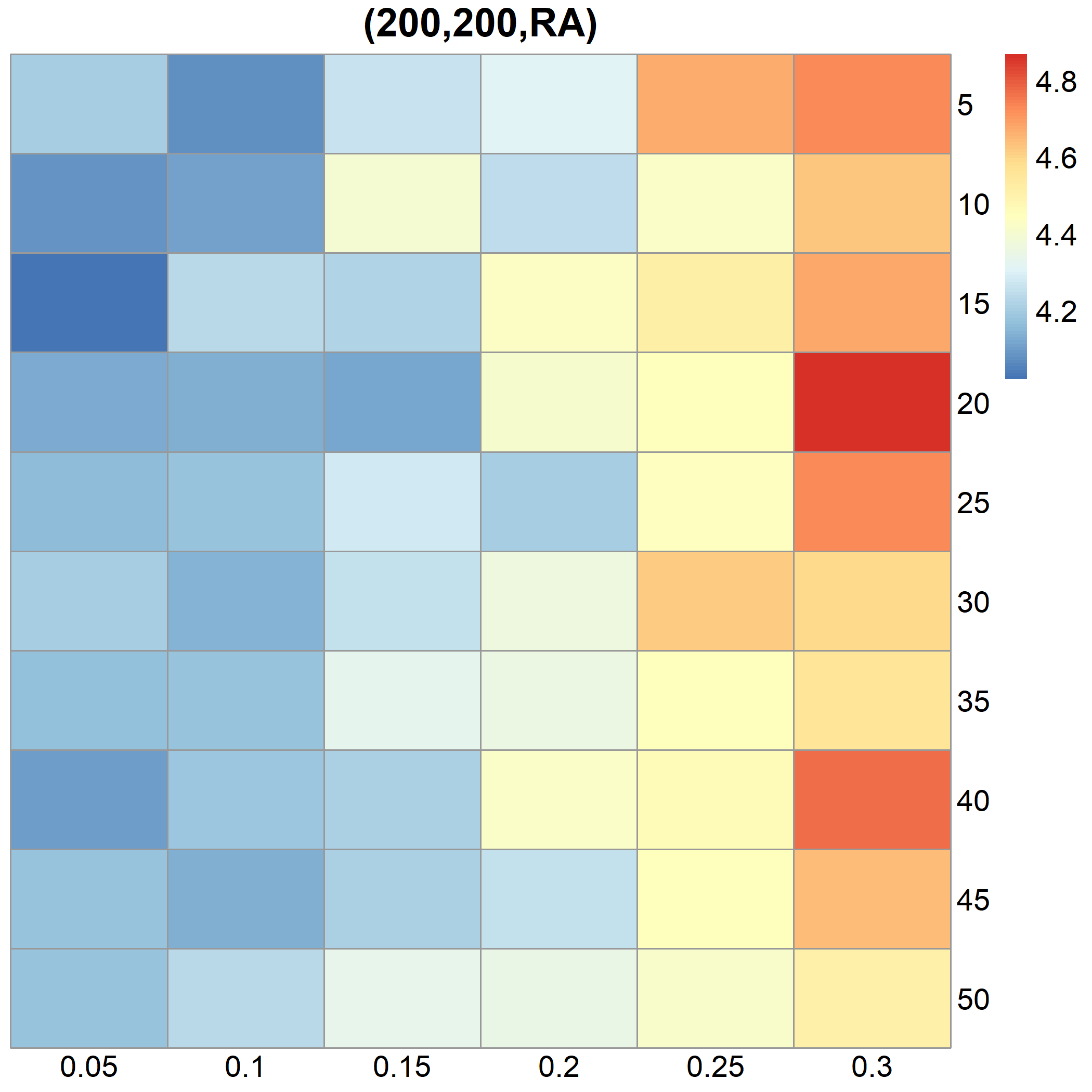}
		\includegraphics[width=0.24\linewidth]{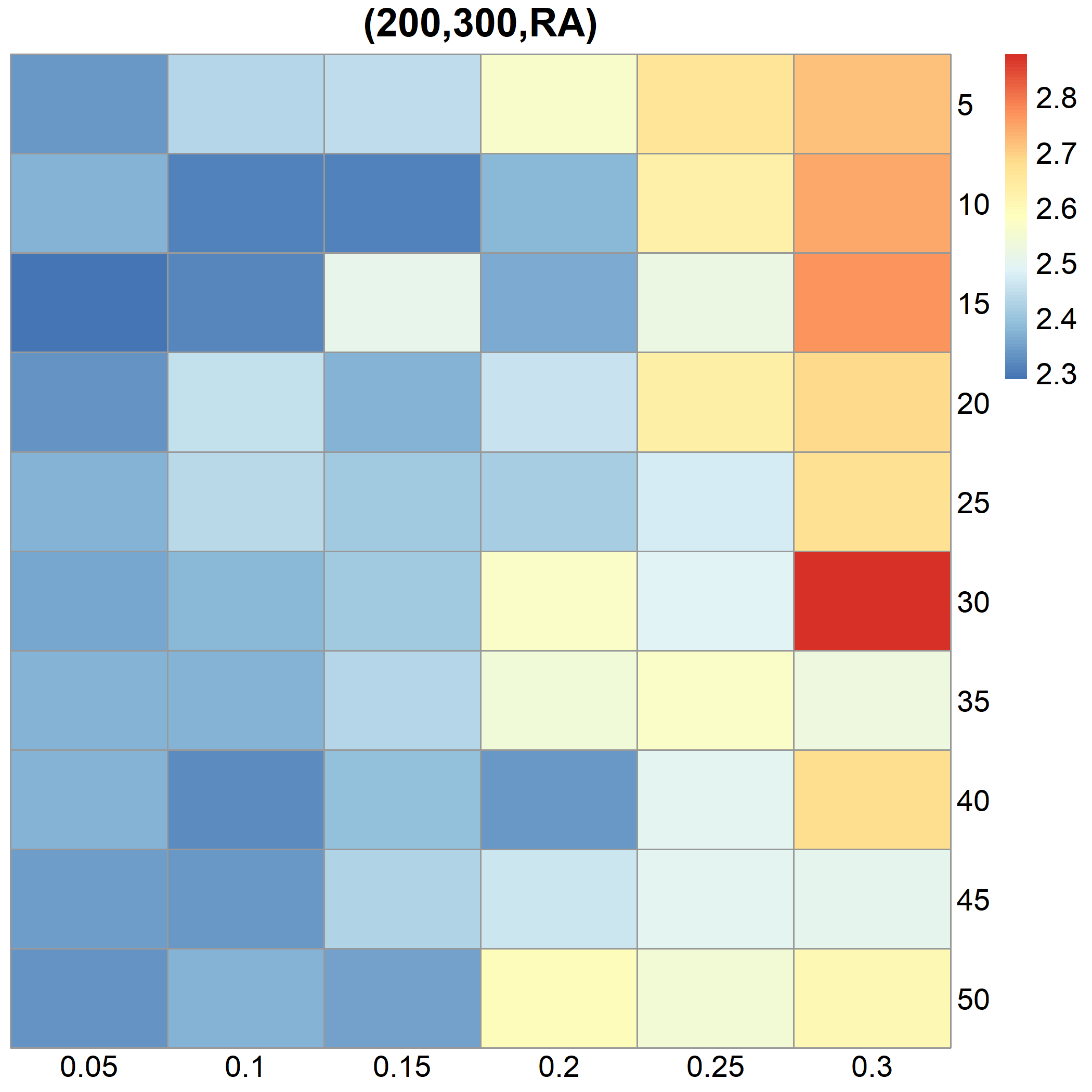}
		
		\vspace{3pt}
		\small (a) $N = 200$
	\end{minipage}
	
	\vspace{6pt}
	
	\begin{minipage}{0.95\textwidth}
		\centering
		\includegraphics[width=0.24\linewidth]{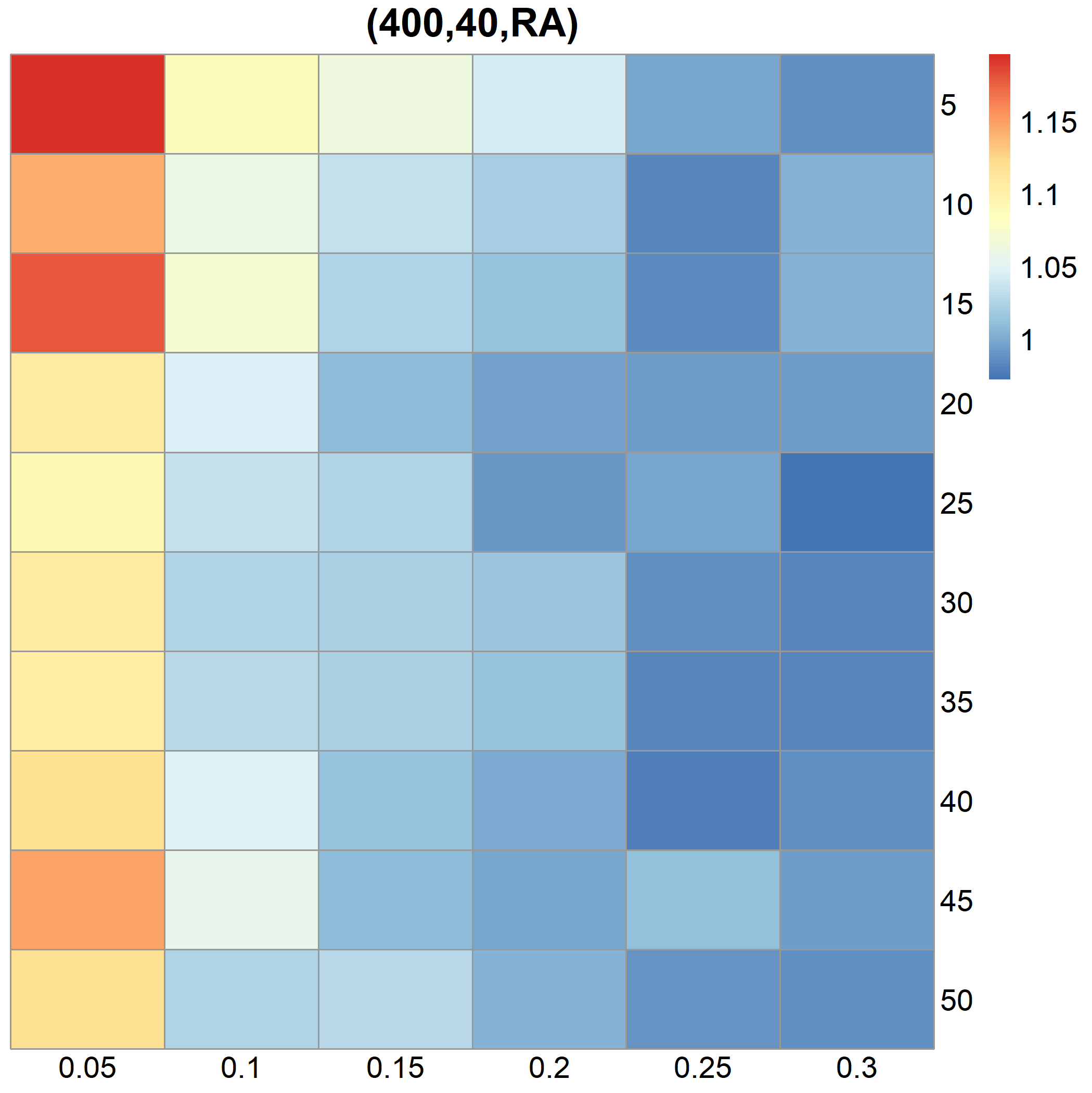}
		\includegraphics[width=0.24\linewidth]{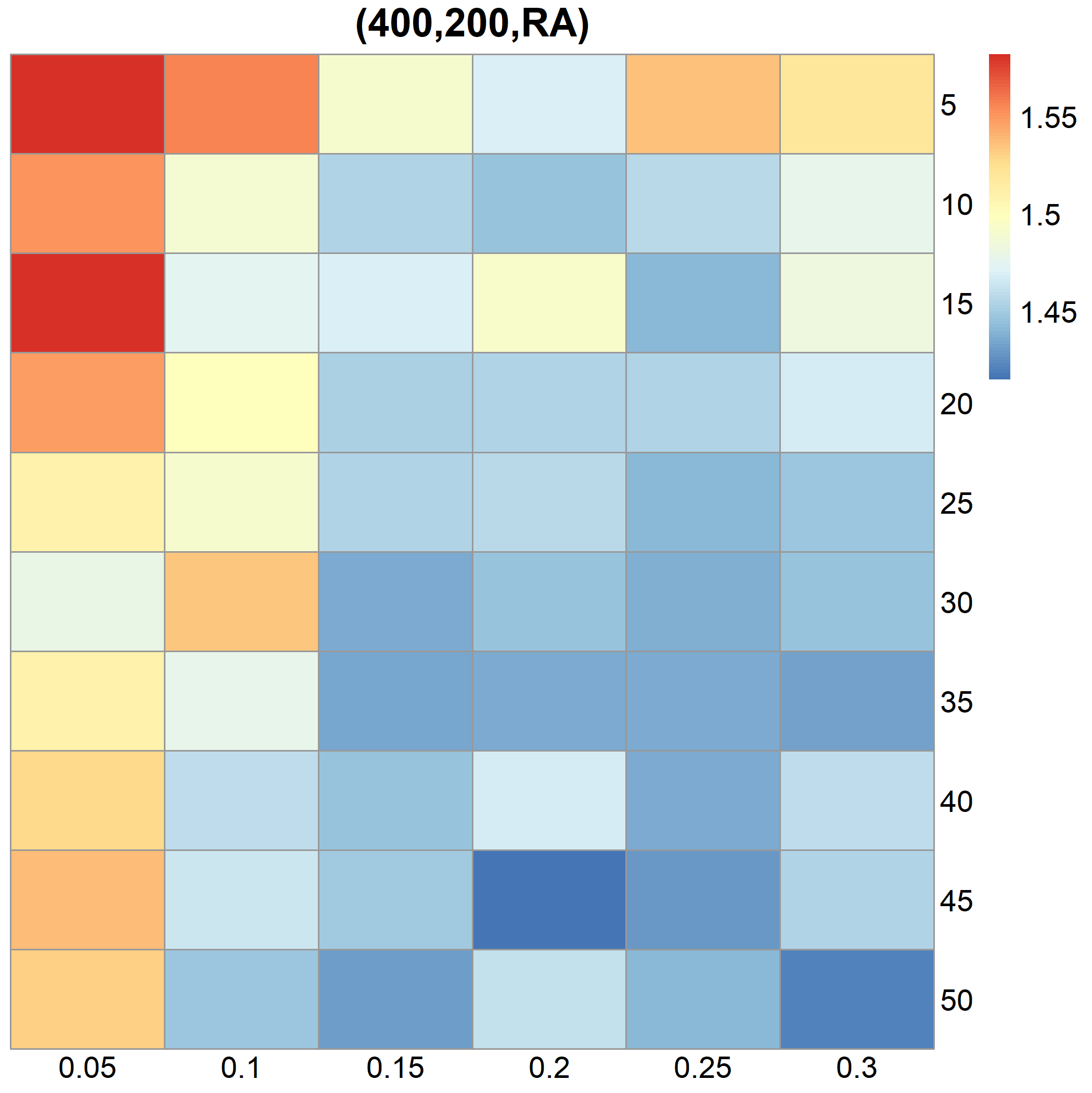}
		\includegraphics[width=0.24\linewidth]{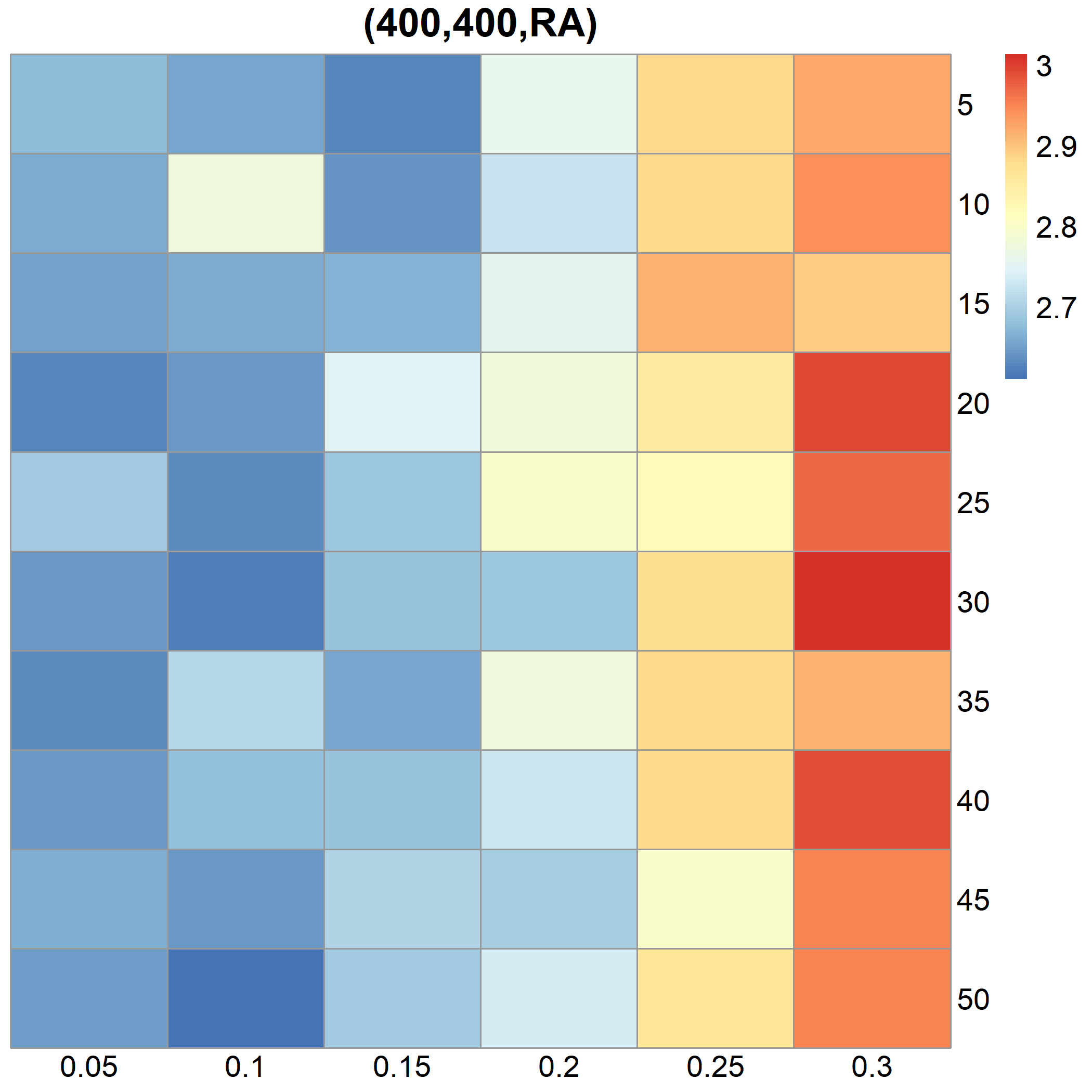}
		\includegraphics[width=0.24\linewidth]{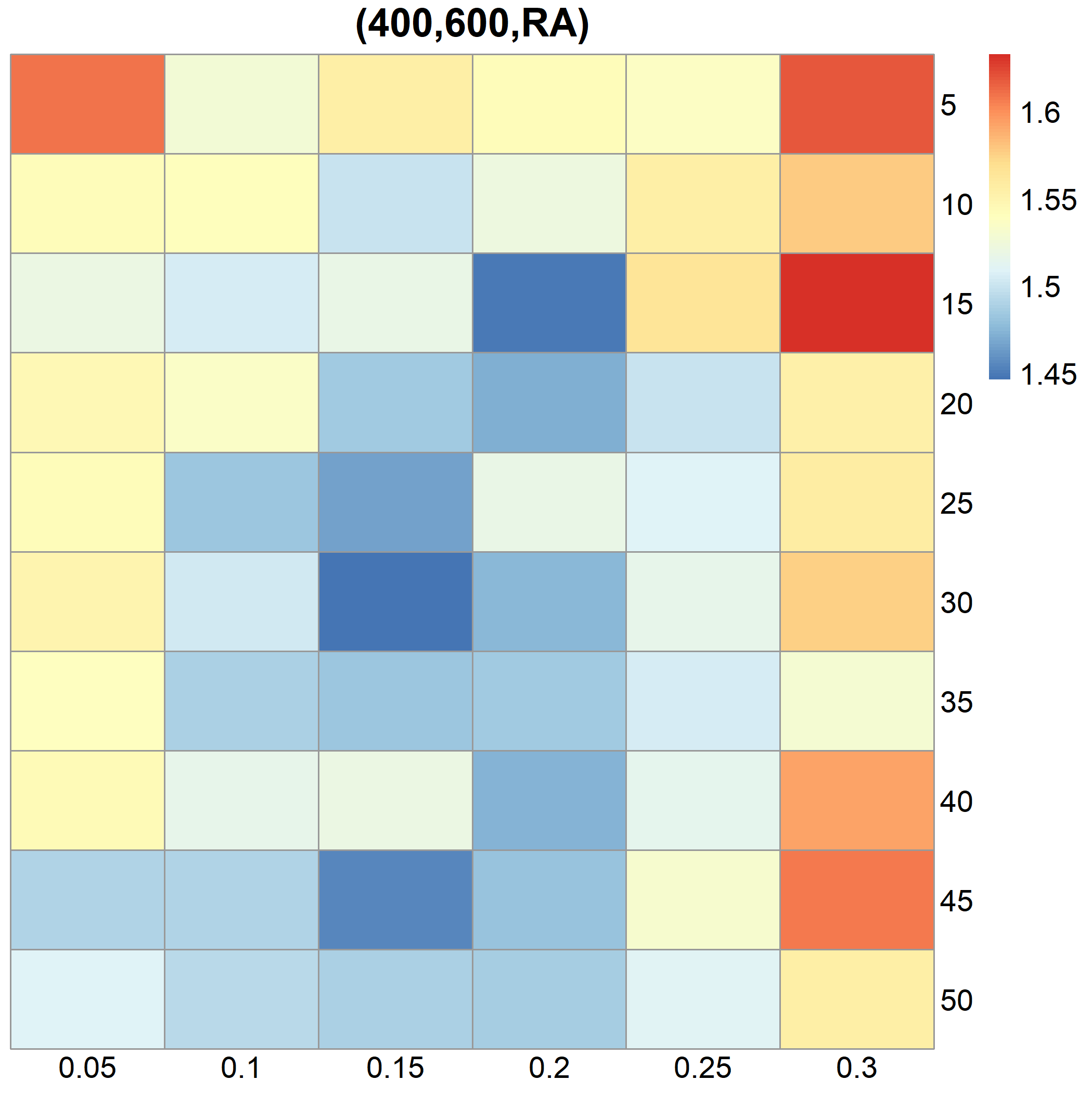}
		
		\vspace{3pt}
		\small (b) $N = 400$
	\end{minipage}
	
	\vspace{6pt}
	
	\begin{minipage}{0.95\textwidth}
		\centering
		\includegraphics[width=0.24\linewidth]{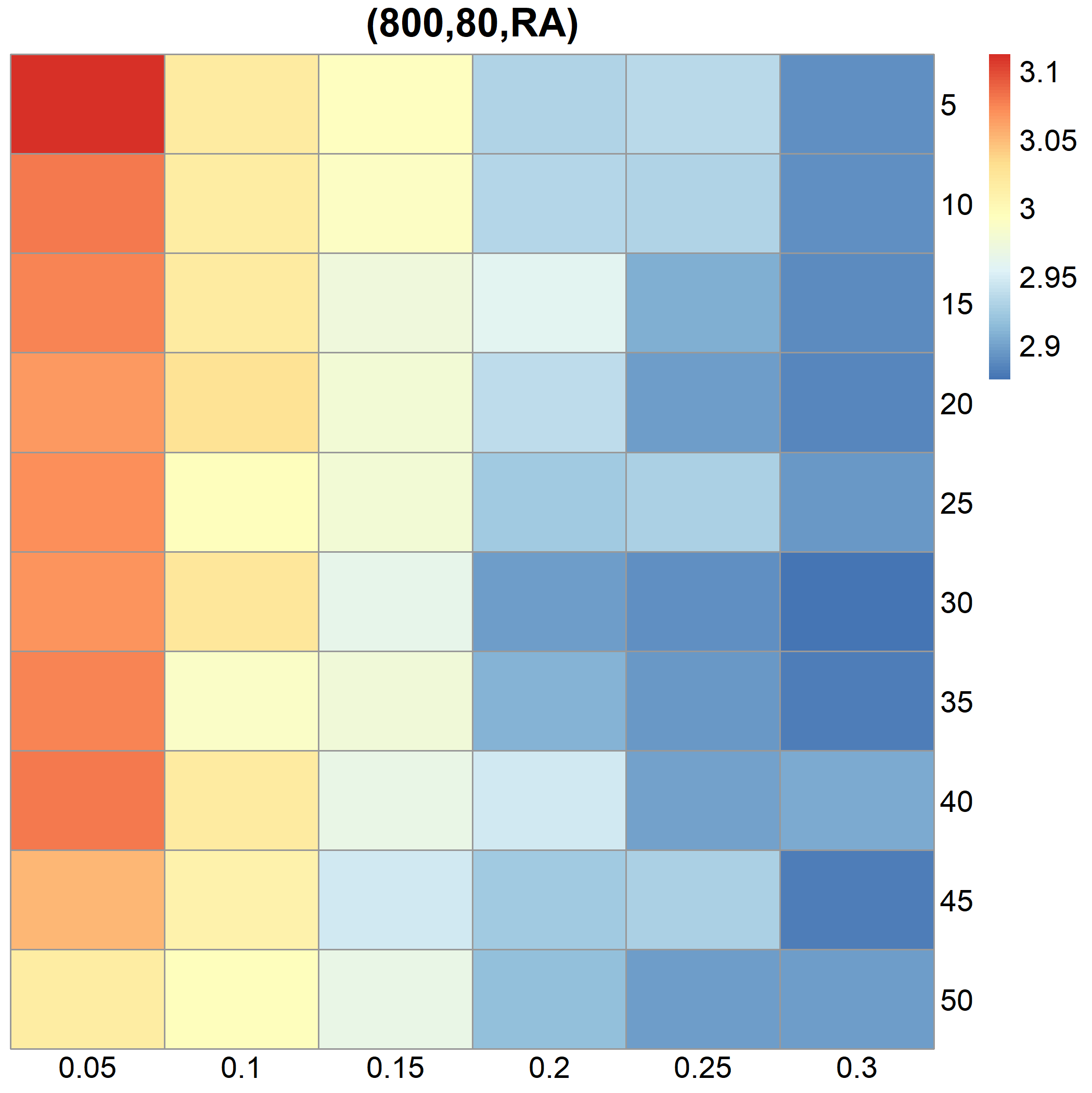}
		\includegraphics[width=0.24\linewidth]{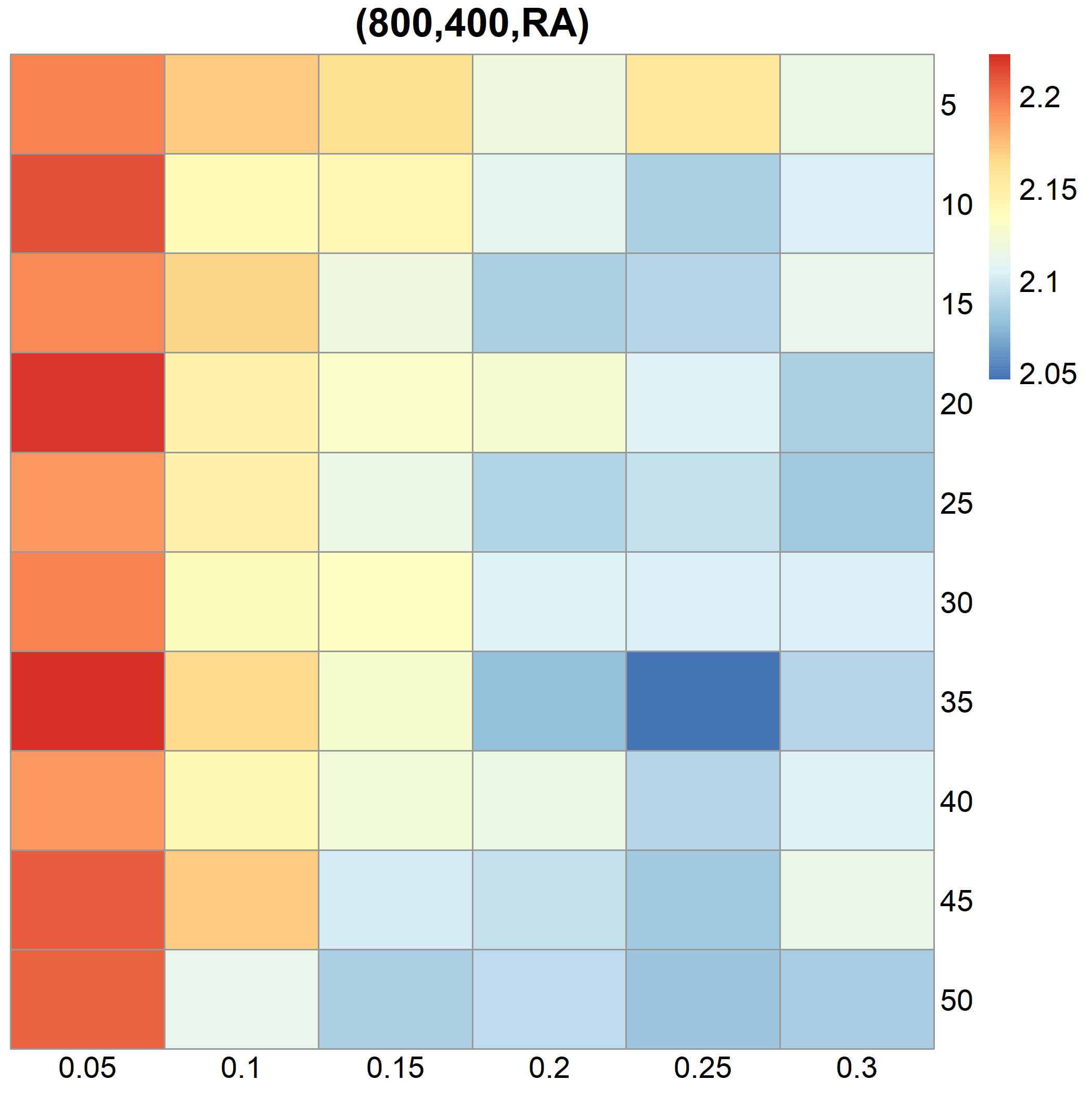}
		\includegraphics[width=0.24\linewidth]{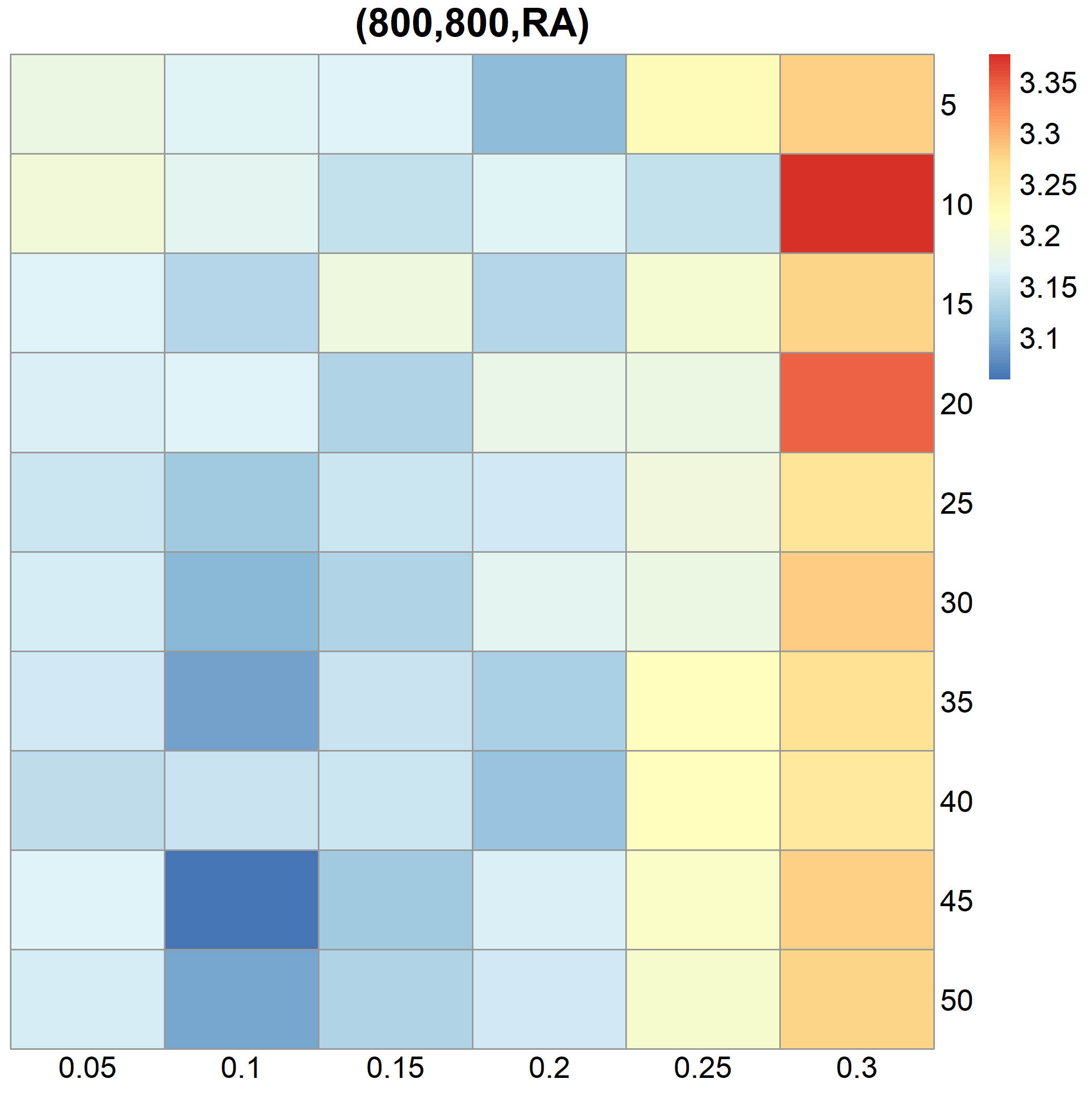}
		\includegraphics[width=0.24\linewidth]{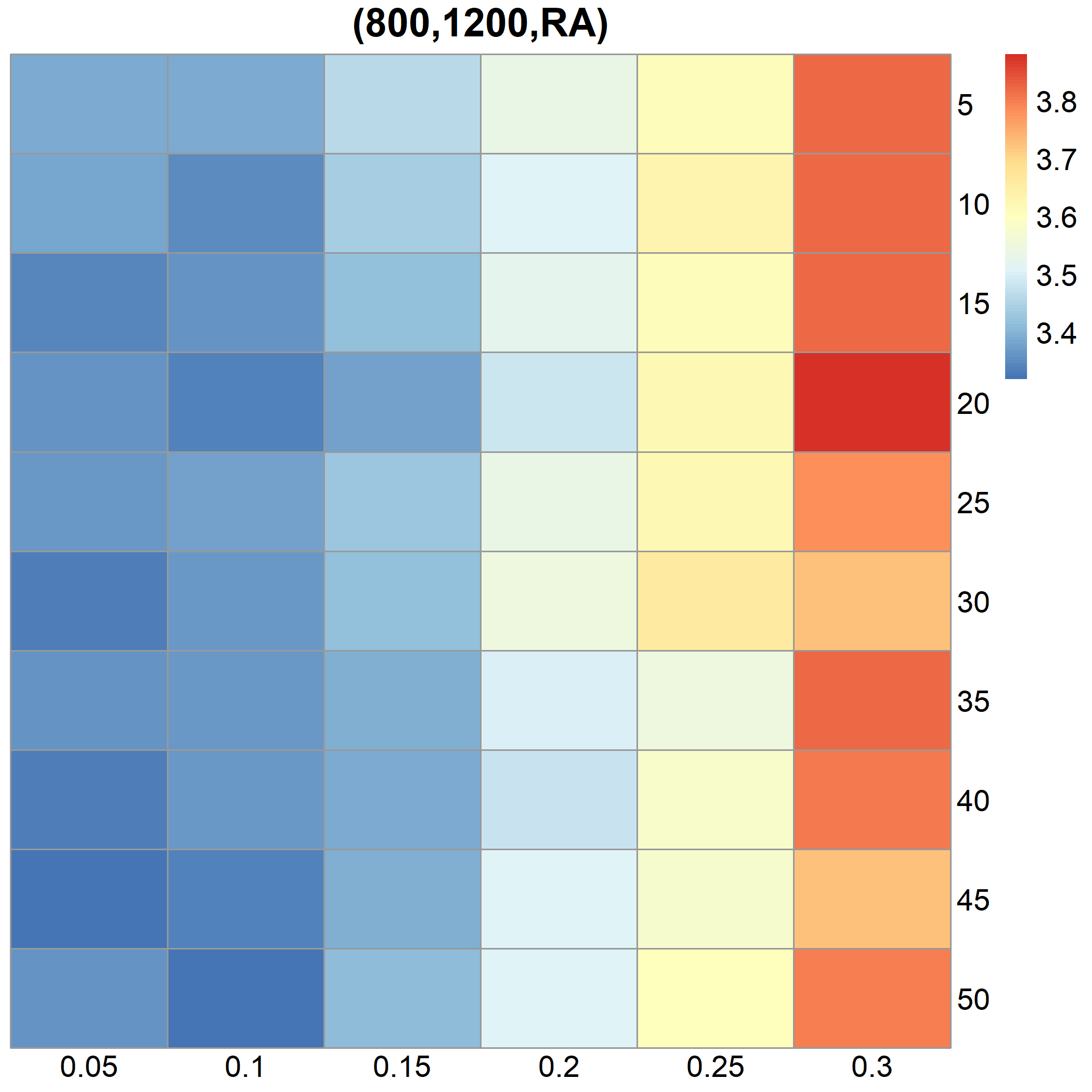}
		
		\vspace{3pt}
		\small (c) $N = 800$
	\end{minipage}
	
	\caption{Cross-validation results for Section \ref{sec3.2} under polynomial decay with a random covariance structure. Values in parentheses denote $(N, K, \rho)$. The horizontal axis represents the selection probability $p$, and the vertical axis indicates the number of candidate models $M$. Darker regions correspond to $(p, M)$ combinations yielding lower cross-validation errors.}
	\label{fig:caseepolyRA}
\end{figure*}

\newpage

Figures~\ref{fig:caseexp0.1}, \ref{fig:caseexp0.9}, and \ref{fig:caseexpRA} report the cross-validation results under three different scenarios: when the covariates are weakly correlated, highly correlated, and randomly correlated, respectively, with exponentially decaying regression coefficients. These findings underscore the sensitivity of the RSA estimator to the choice of selection probability across different dependence structures, while also corroborating the importance and effectiveness of the cross-validation procedure.
\begin{figure*}[!h]
	\centering
	
	\begin{minipage}{0.95\textwidth}
		\centering
		\includegraphics[width=0.24\linewidth]{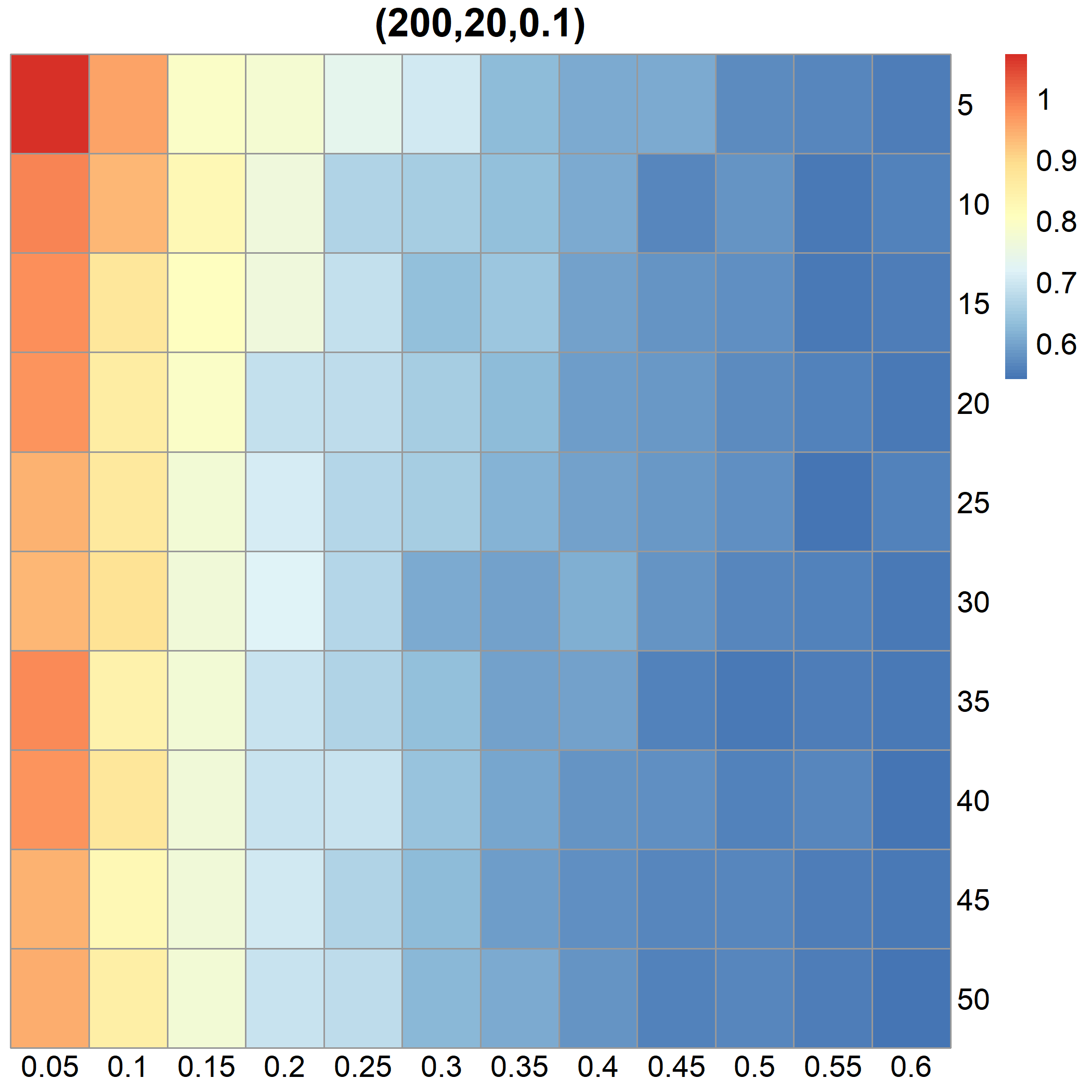}
		\includegraphics[width=0.24\linewidth]{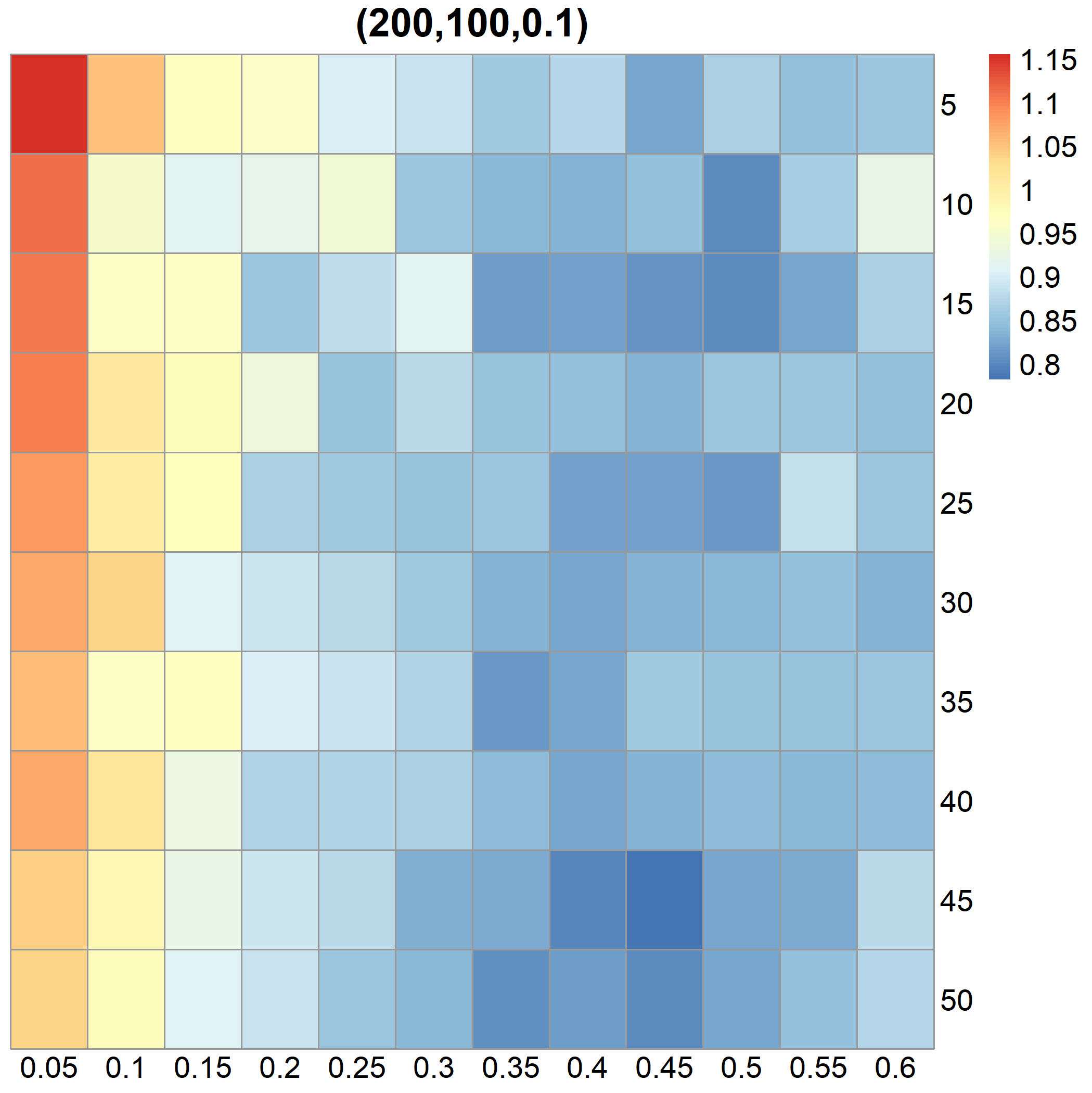}
		\includegraphics[width=0.24\linewidth]{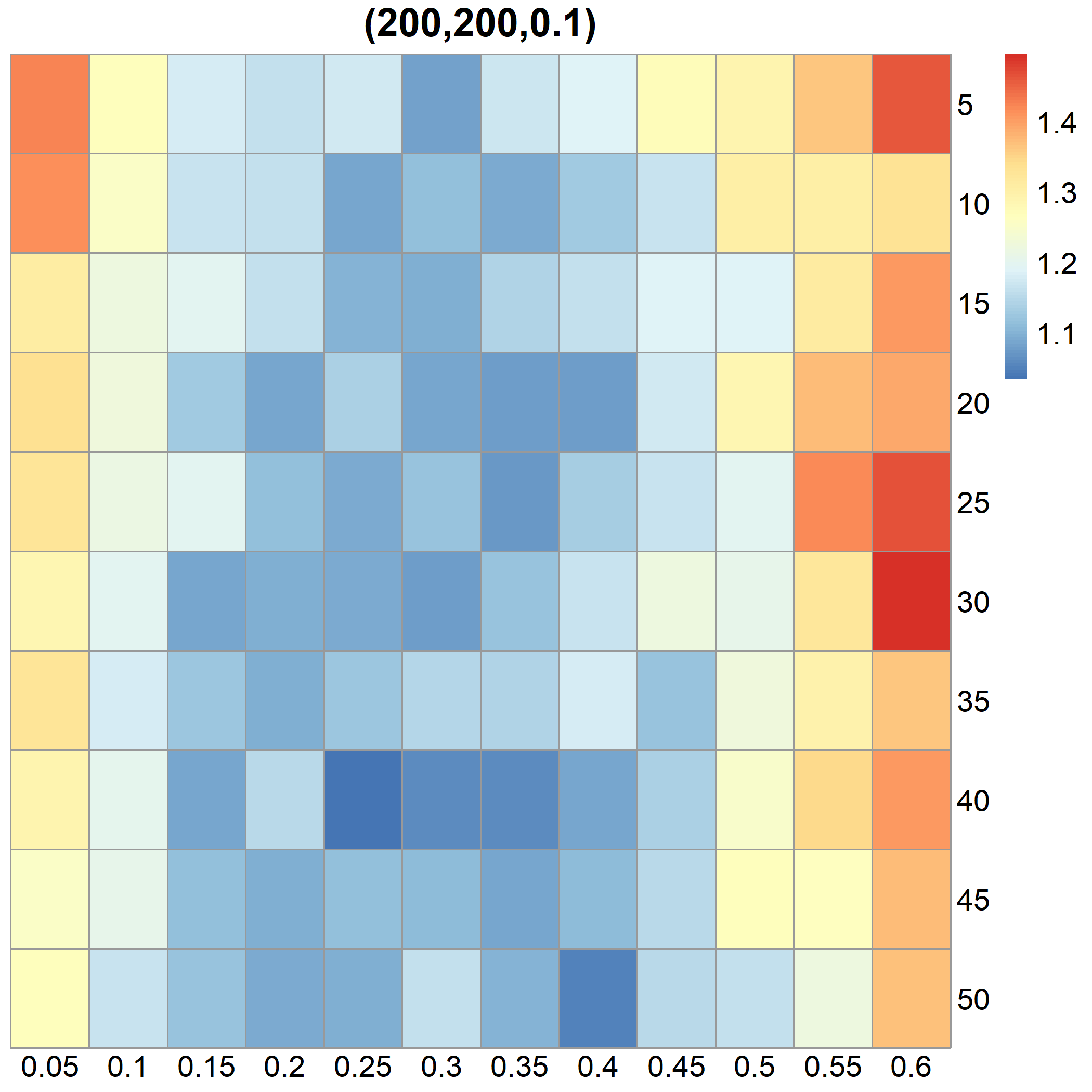}
		\includegraphics[width=0.24\linewidth]{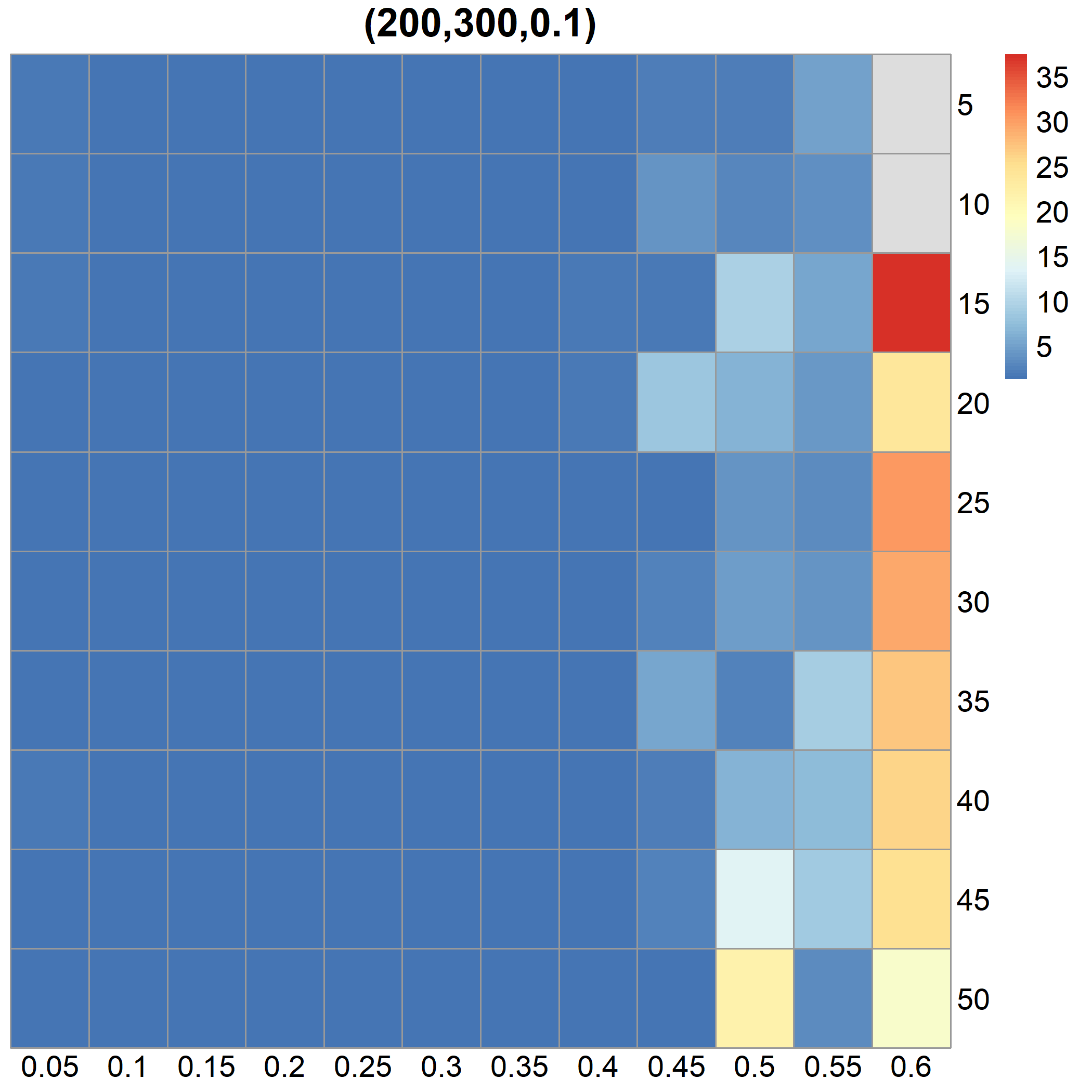}
		
		\vspace{3pt}
		\small (a) $N = 200$
	\end{minipage}
	
	\vspace{6pt}
	
	\begin{minipage}{0.95\textwidth}
		\centering
		\includegraphics[width=0.24\linewidth]{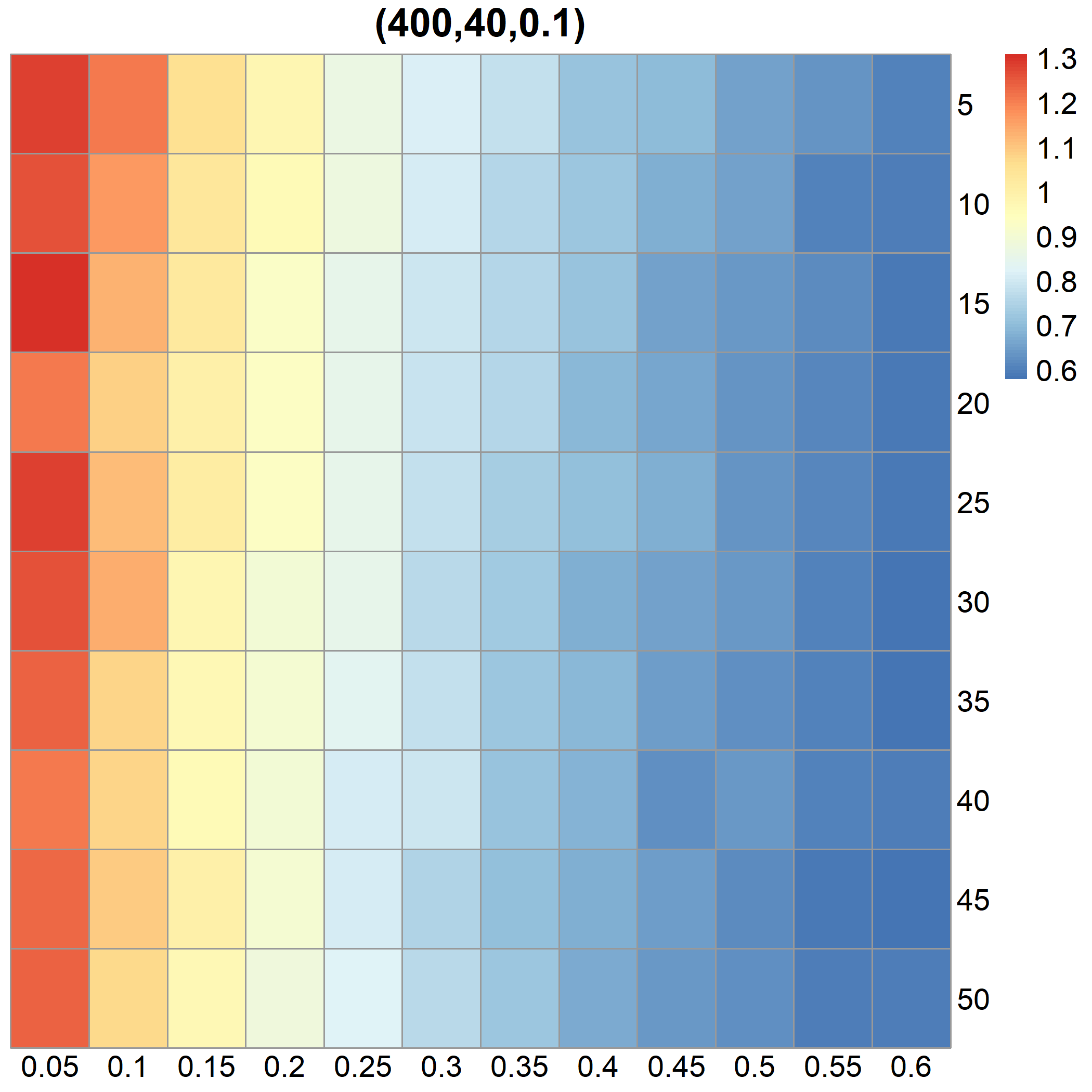}
		\includegraphics[width=0.24\linewidth]{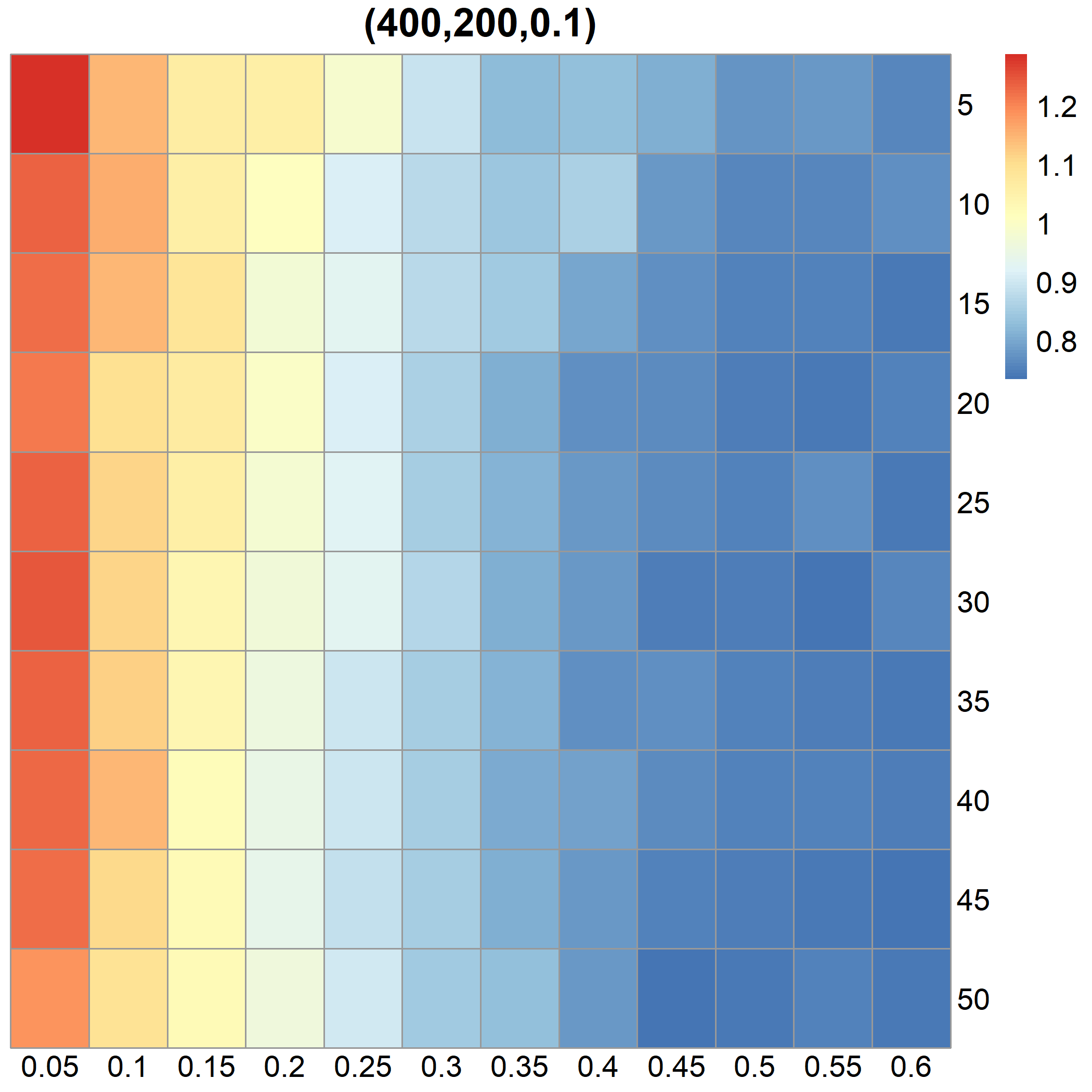}
		\includegraphics[width=0.24\linewidth]{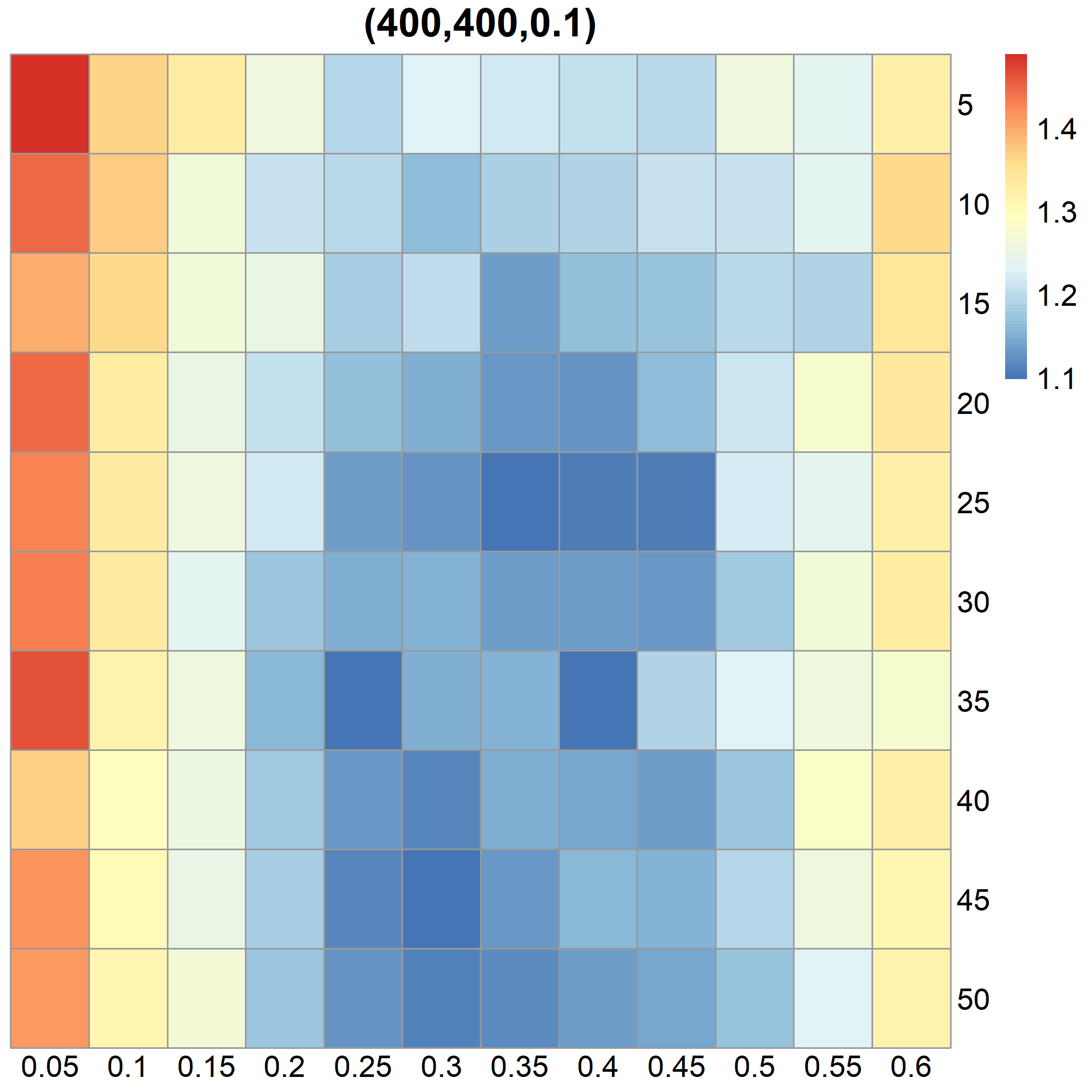}
		\includegraphics[width=0.24\linewidth]{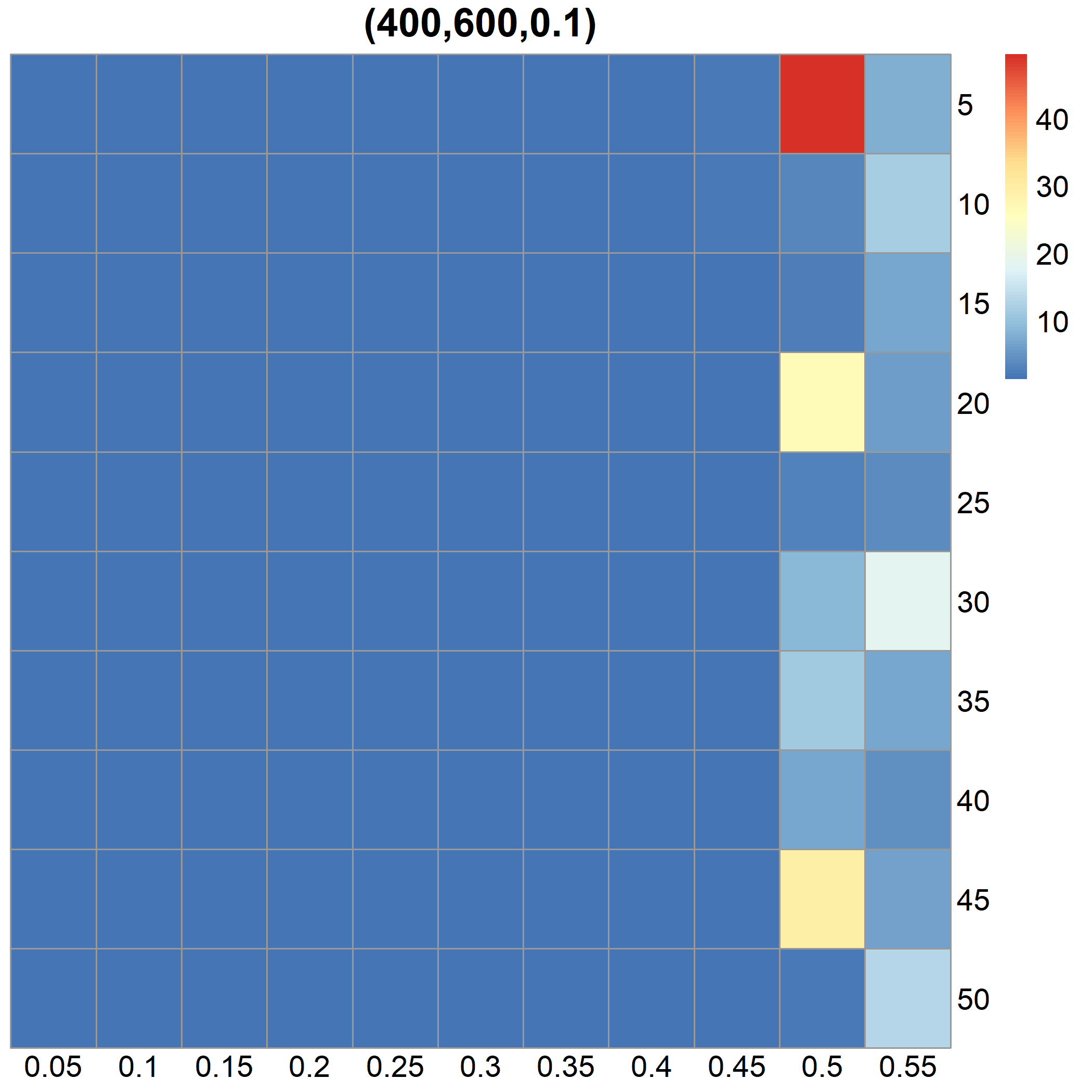}
		
		\vspace{3pt}
		\small (b) $N = 400$
	\end{minipage}
	
	\vspace{6pt}
	
	\begin{minipage}{0.95\textwidth}
		\centering
		\includegraphics[width=0.24\linewidth]{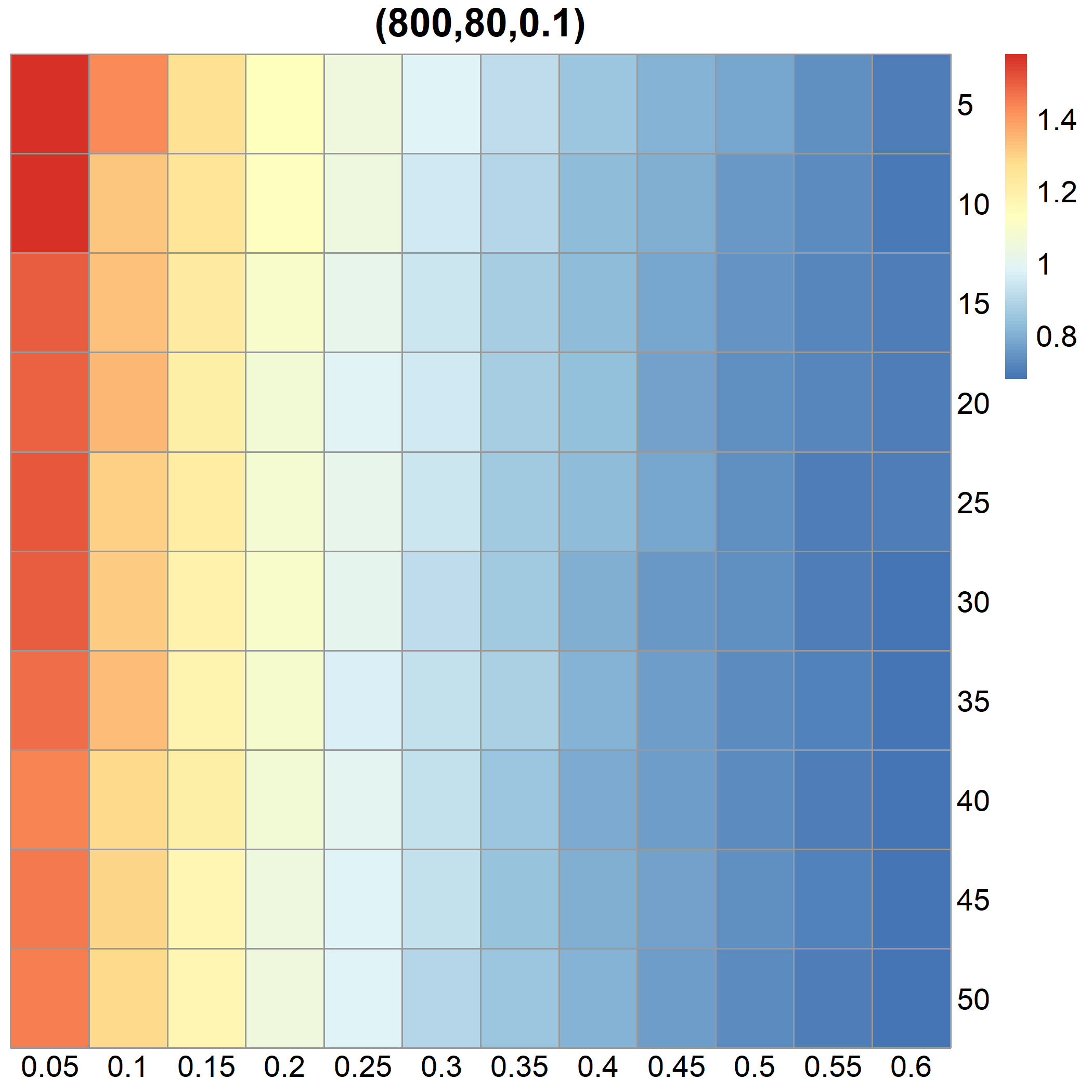}
		\includegraphics[width=0.24\linewidth]{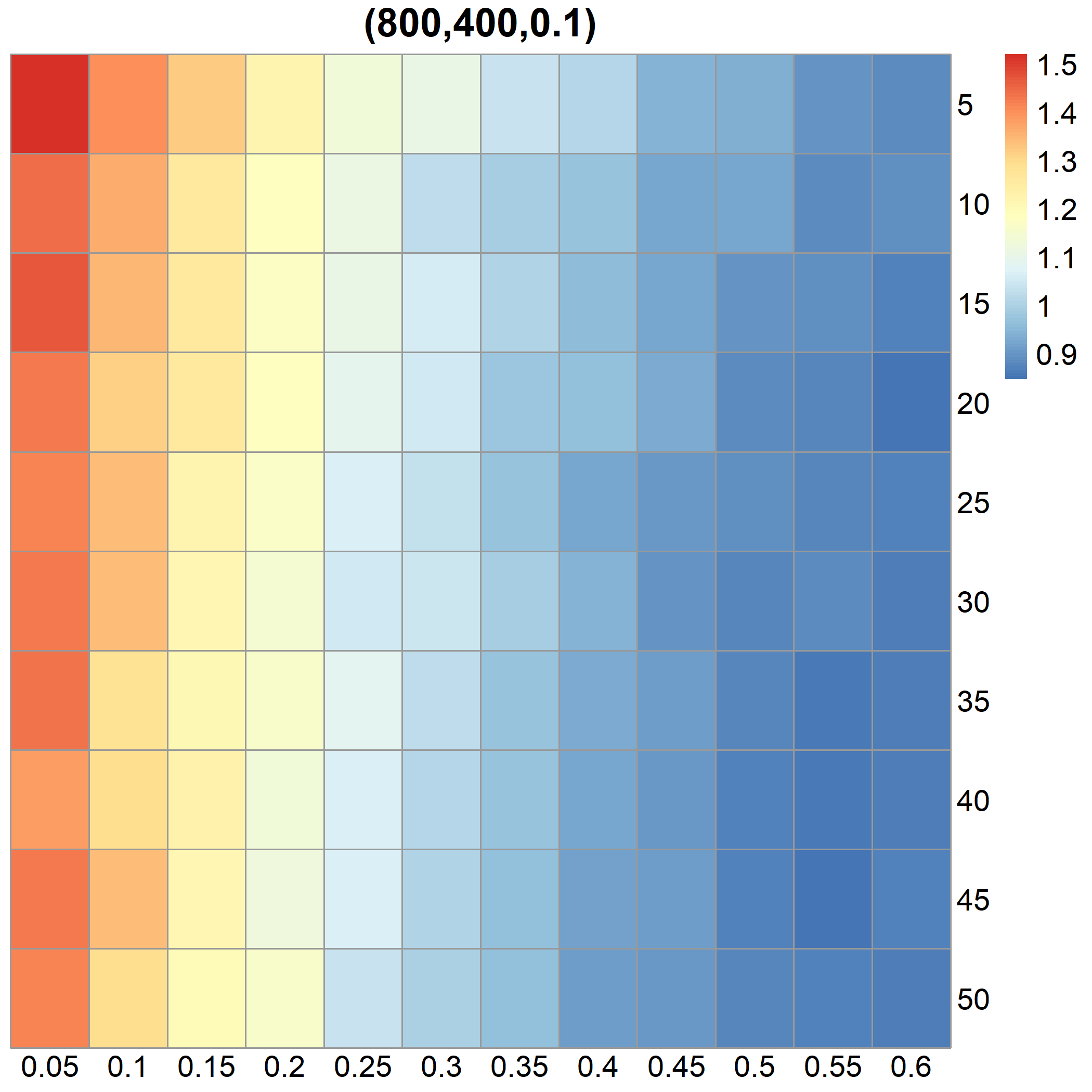}
		\includegraphics[width=0.24\linewidth]{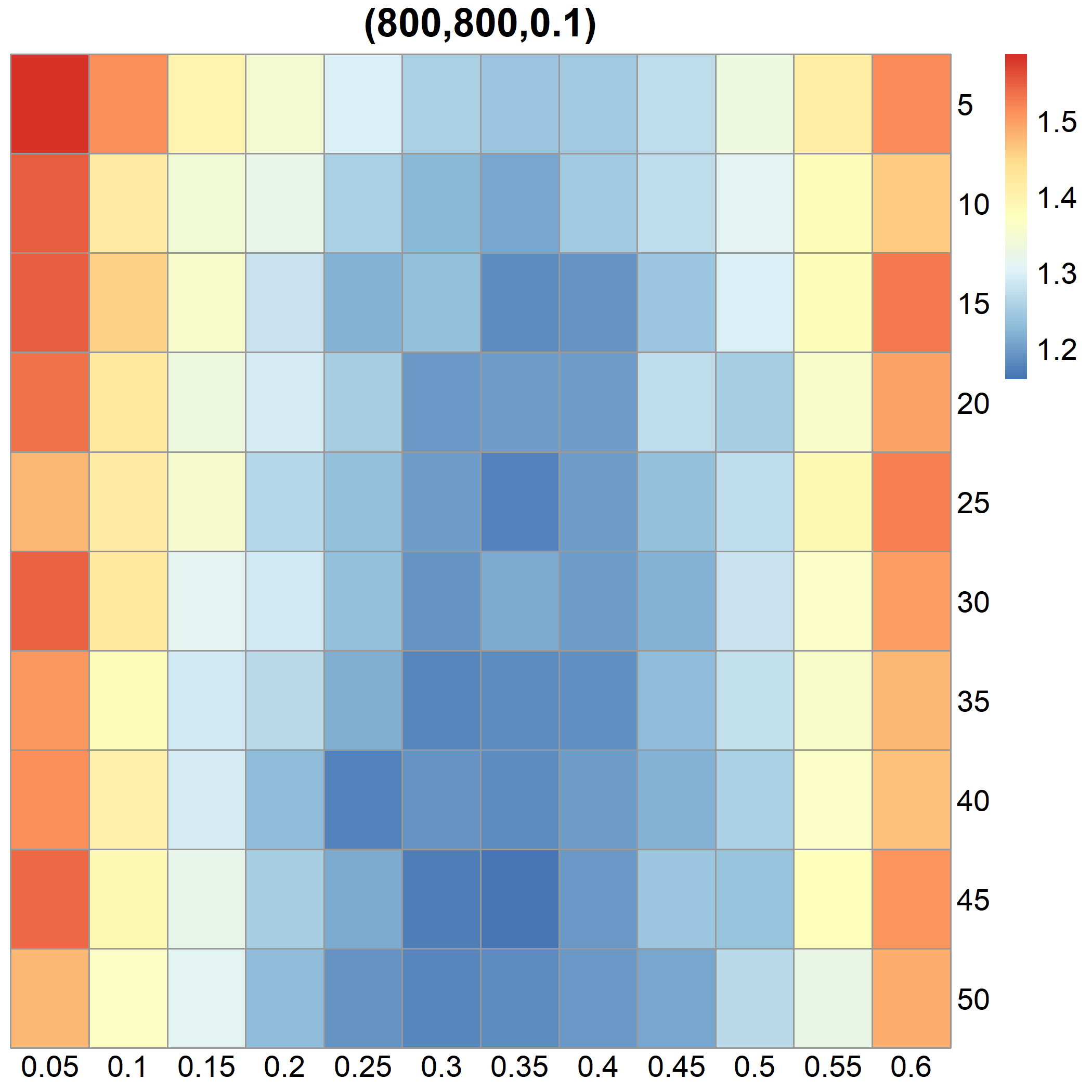}
		\includegraphics[width=0.24\linewidth]{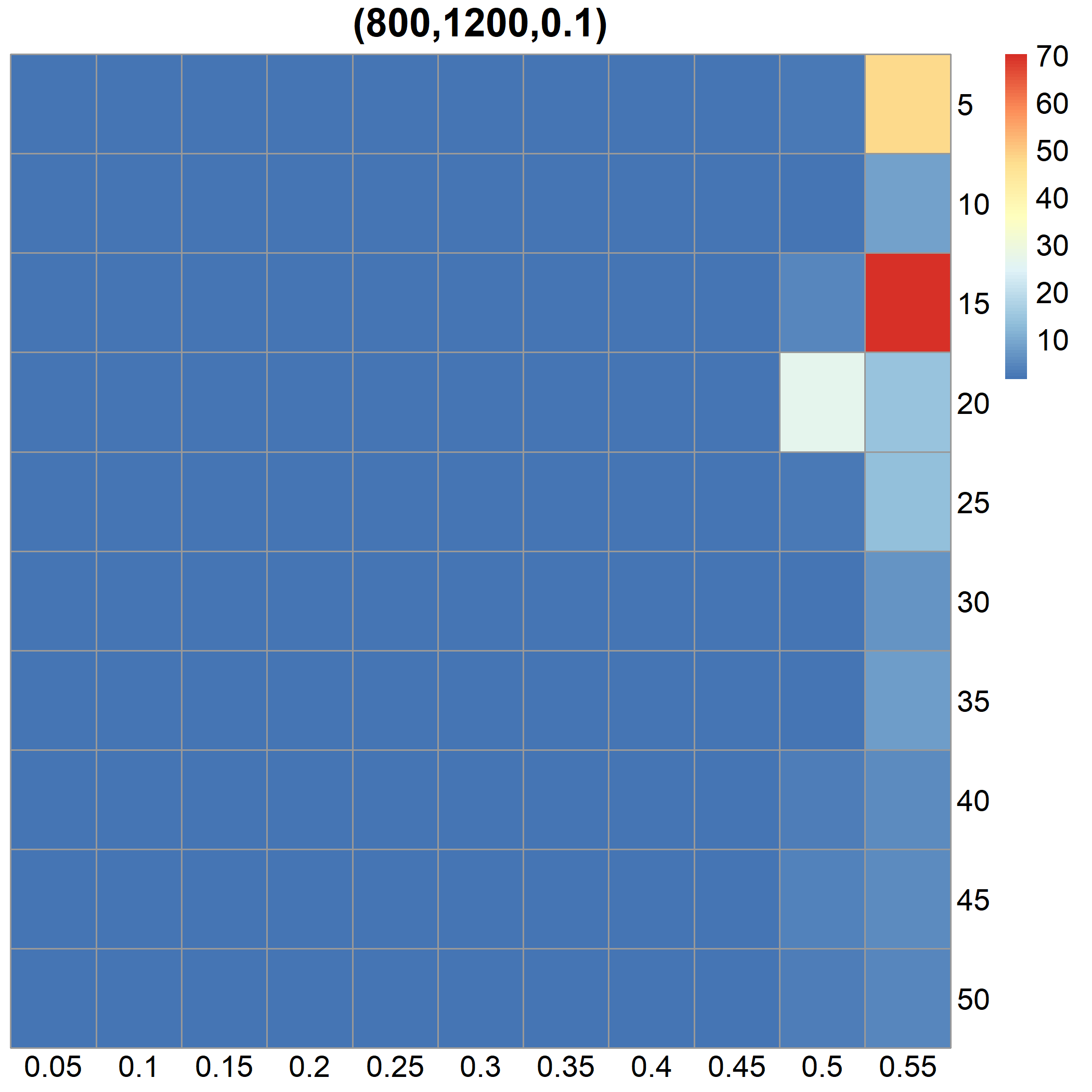}
		
		\vspace{3pt}
		\small (c) $N = 800$
	\end{minipage}
	
	\caption{Cross-validation results for Section \ref{sec3.2} under exponential decay with $\rho = 0.1$. Values in parentheses denote $(N, K, \rho)$. The horizontal axis represents the selection probability $p$, and the vertical axis indicates the number of candidate models $M$. Darker regions correspond to $(p, M)$ combinations yielding lower cross-validation errors.}
	\label{fig:caseexp0.1}
\end{figure*}

\newpage
\begin{figure*}[!h]
	\centering
	
	\begin{minipage}{0.95\textwidth}
		\centering
		\includegraphics[width=0.24\linewidth]{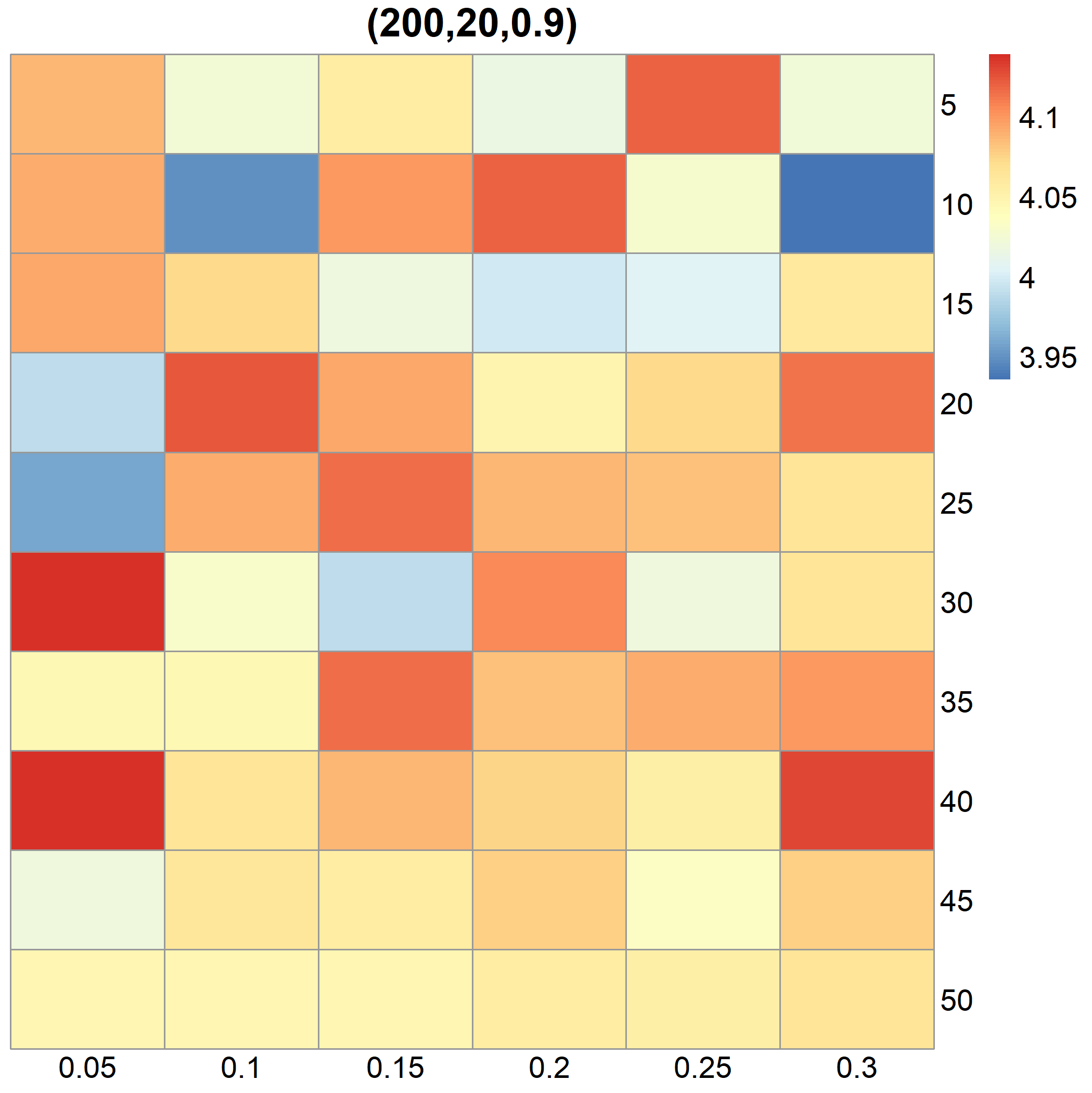}
		\includegraphics[width=0.24\linewidth]{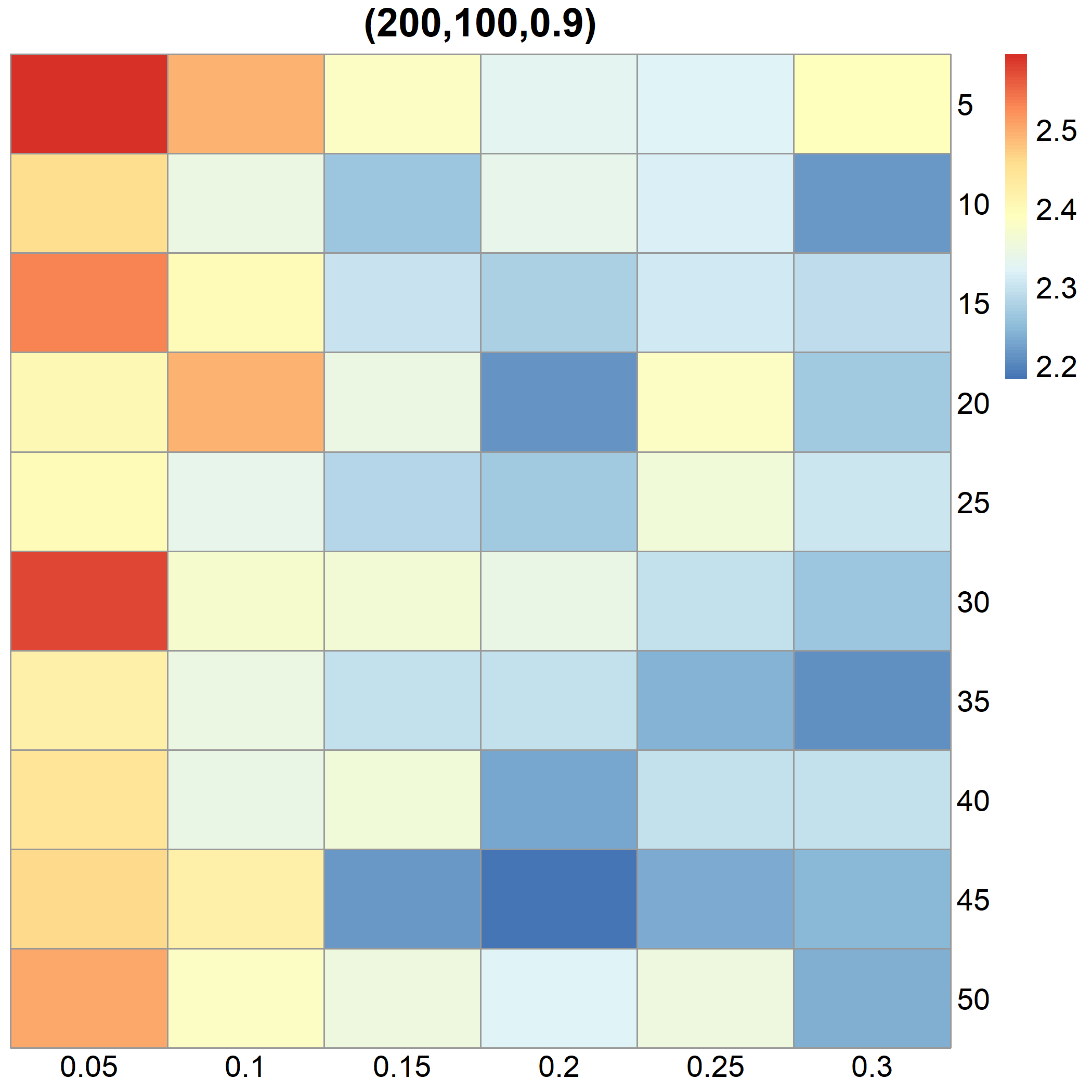}
		\includegraphics[width=0.24\linewidth]{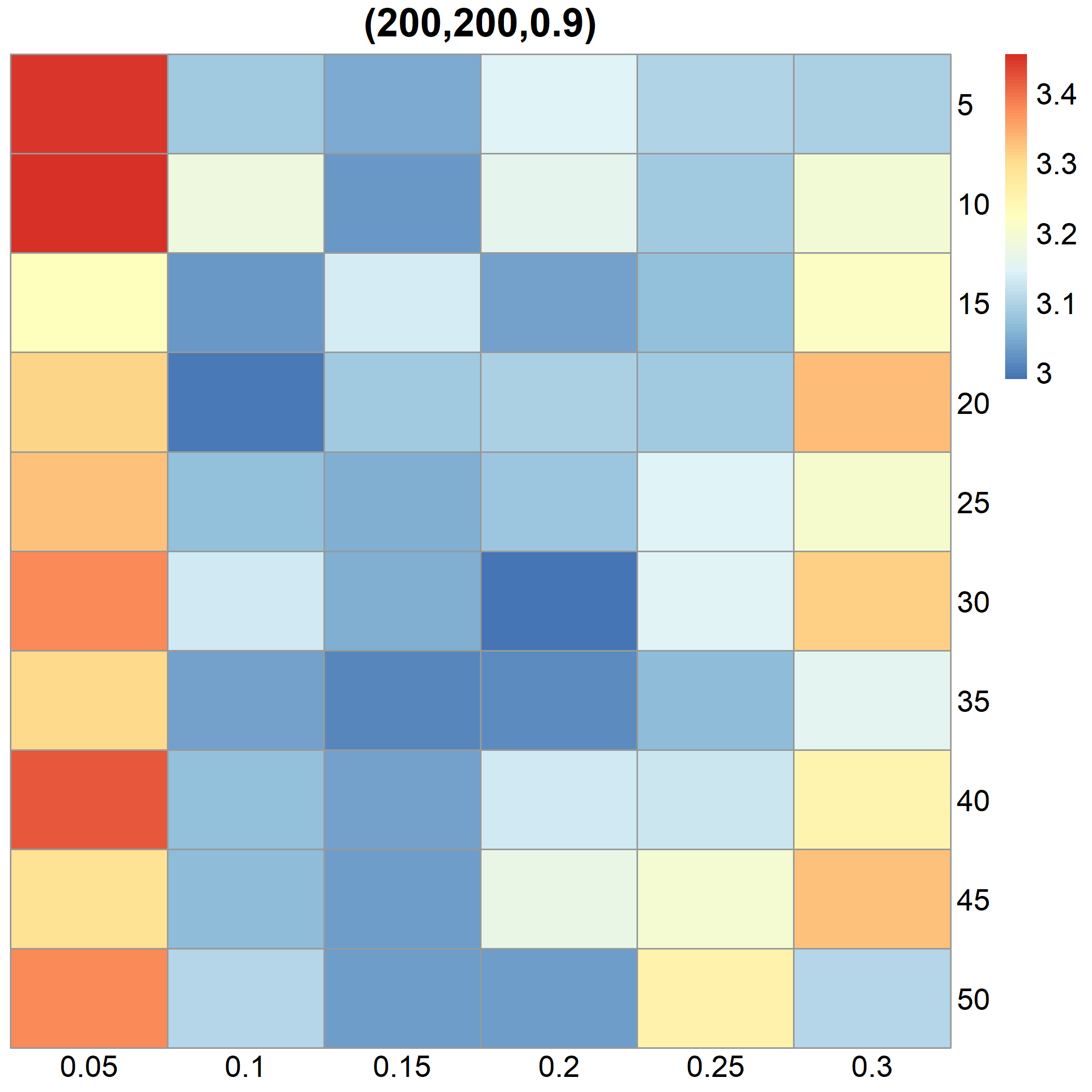}
		\includegraphics[width=0.24\linewidth]{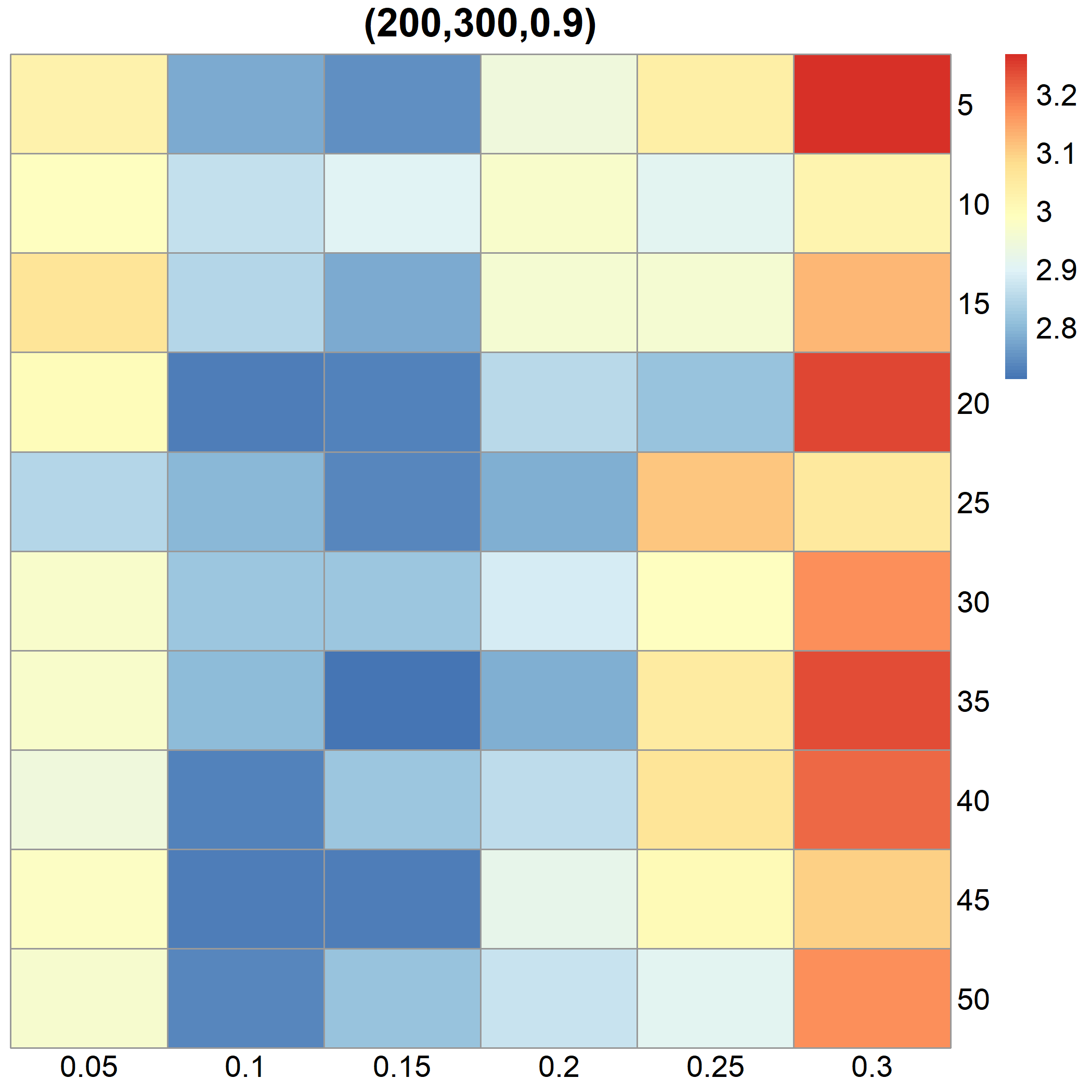}
		
		\vspace{3pt}
		\small (a) $N = 200$
	\end{minipage}
	
	\vspace{6pt}
	
	\begin{minipage}{0.95\textwidth}
		\centering
		\includegraphics[width=0.24\linewidth]{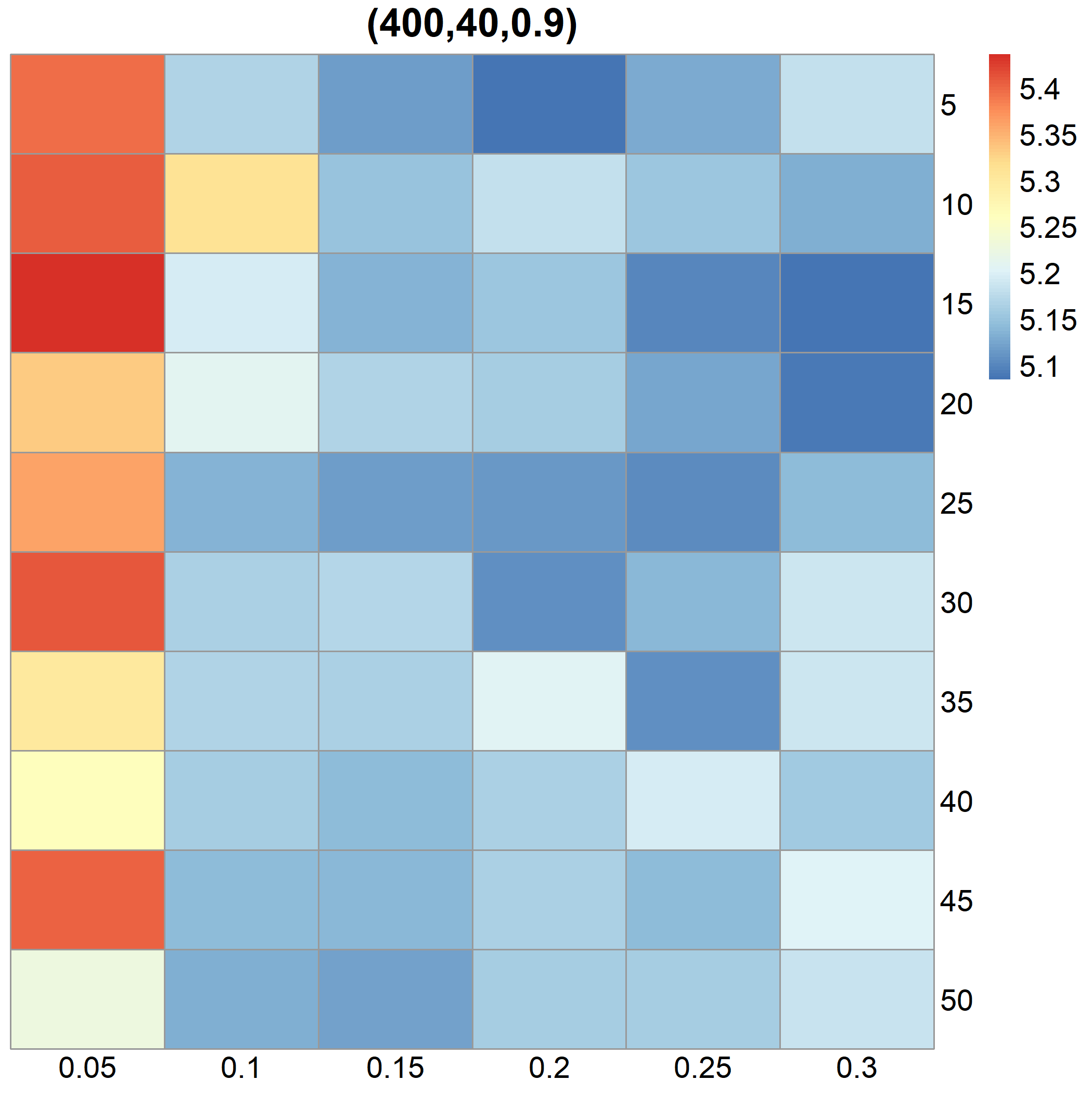}
		\includegraphics[width=0.24\linewidth]{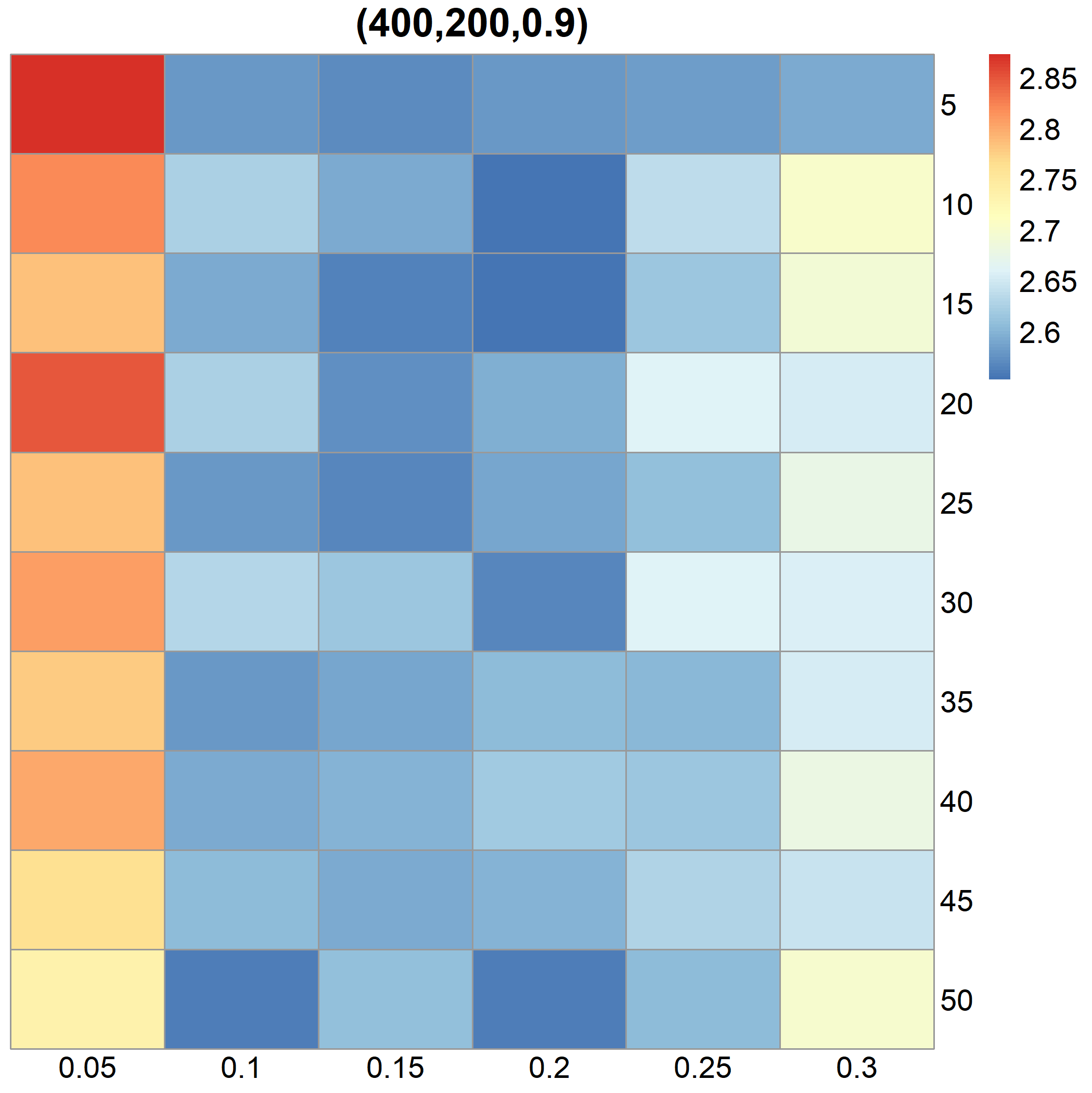}
		\includegraphics[width=0.24\linewidth]{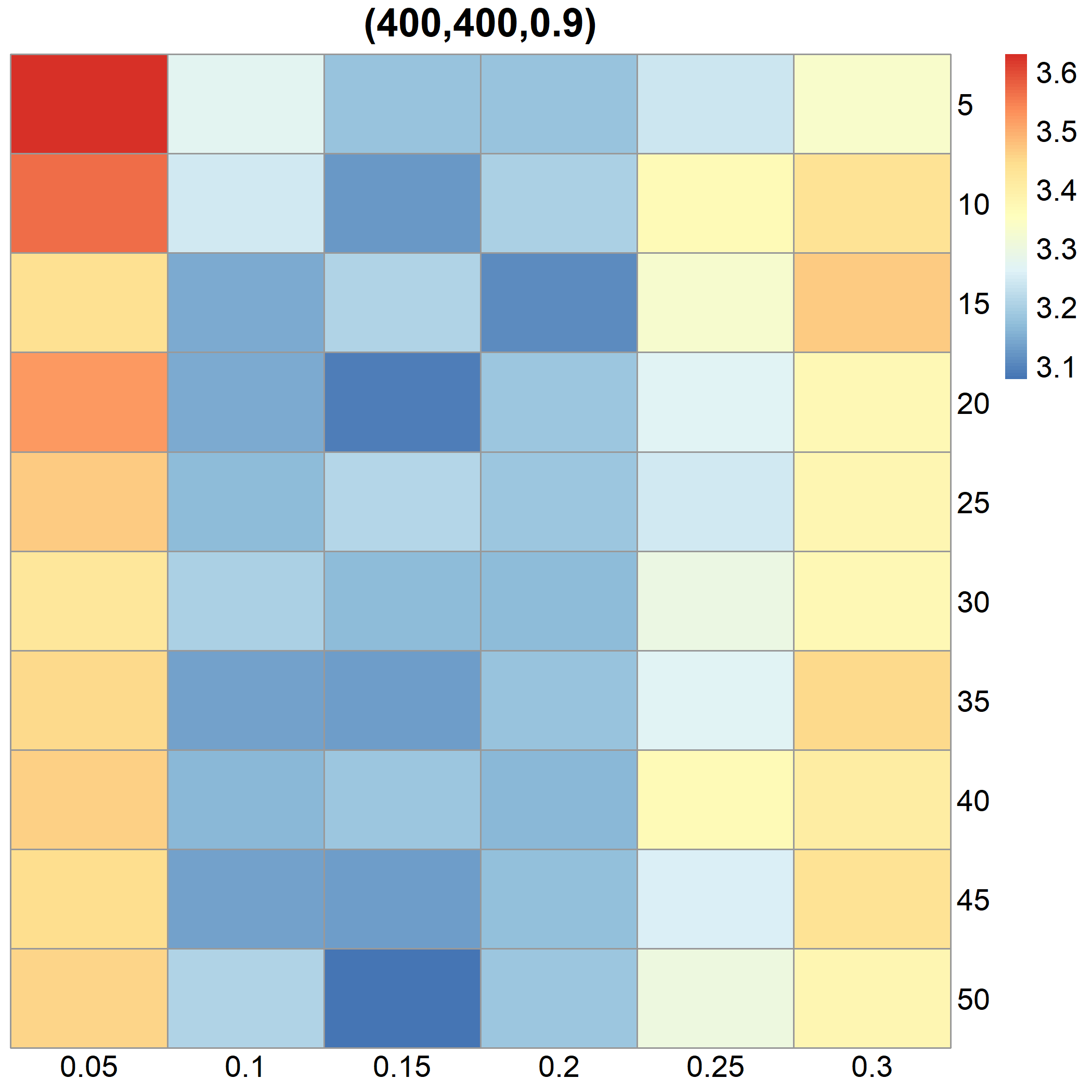}
		\includegraphics[width=0.24\linewidth]{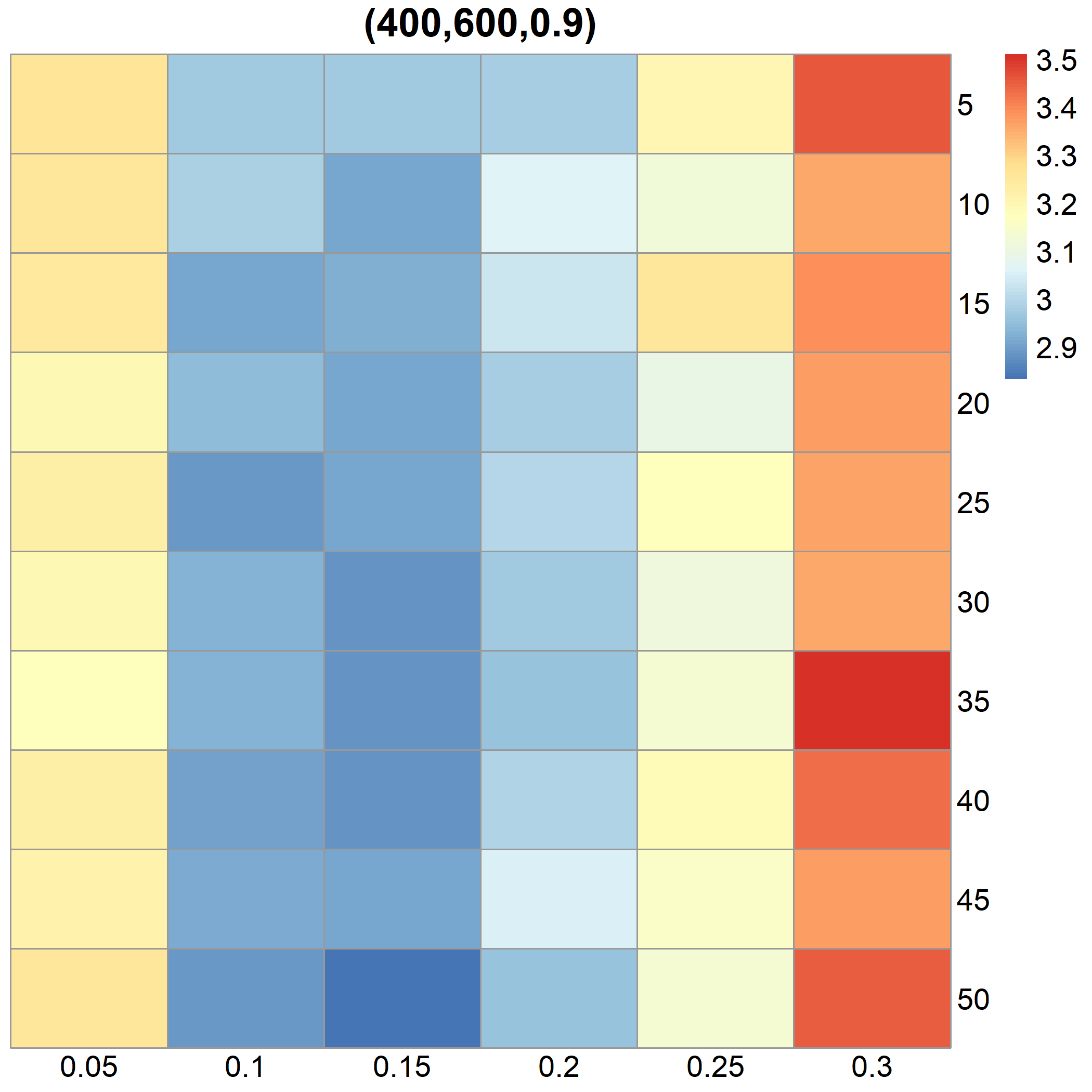}
		
		\vspace{3pt}
		\small (b) $N = 400$
	\end{minipage}
	
	\vspace{6pt}
	
	\begin{minipage}{0.95\textwidth}
		\centering
		\includegraphics[width=0.24\linewidth]{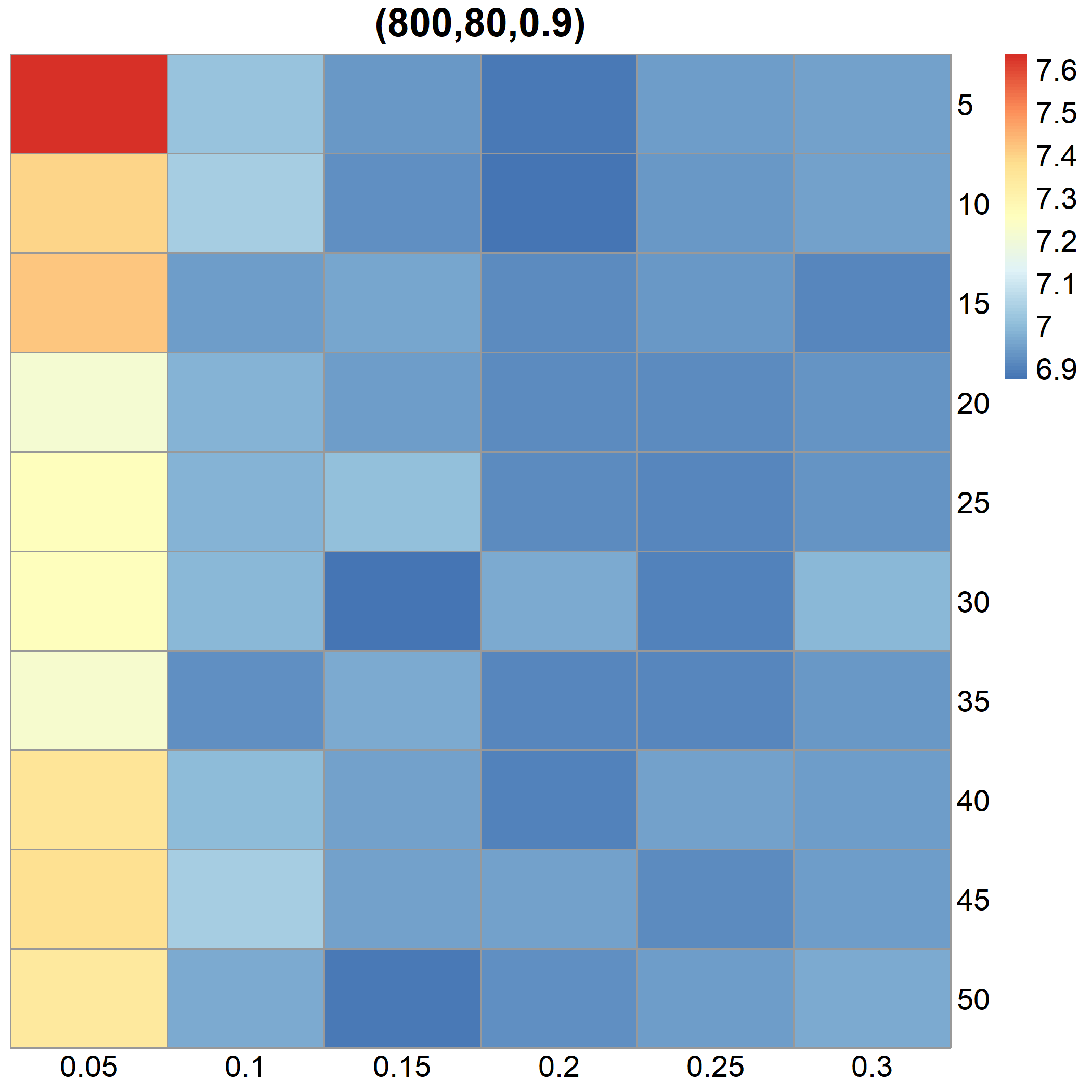}
		\includegraphics[width=0.24\linewidth]{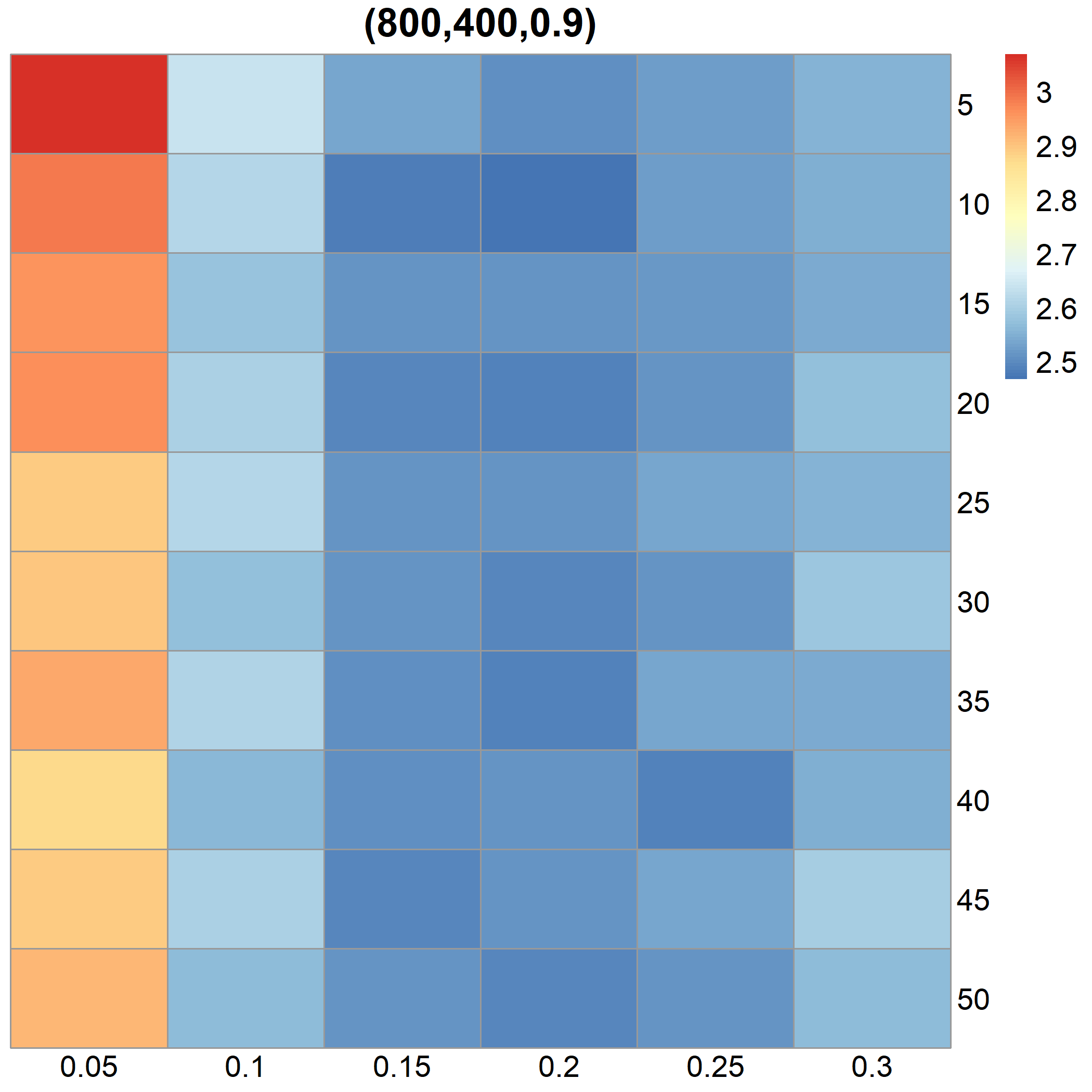}
		\includegraphics[width=0.24\linewidth]{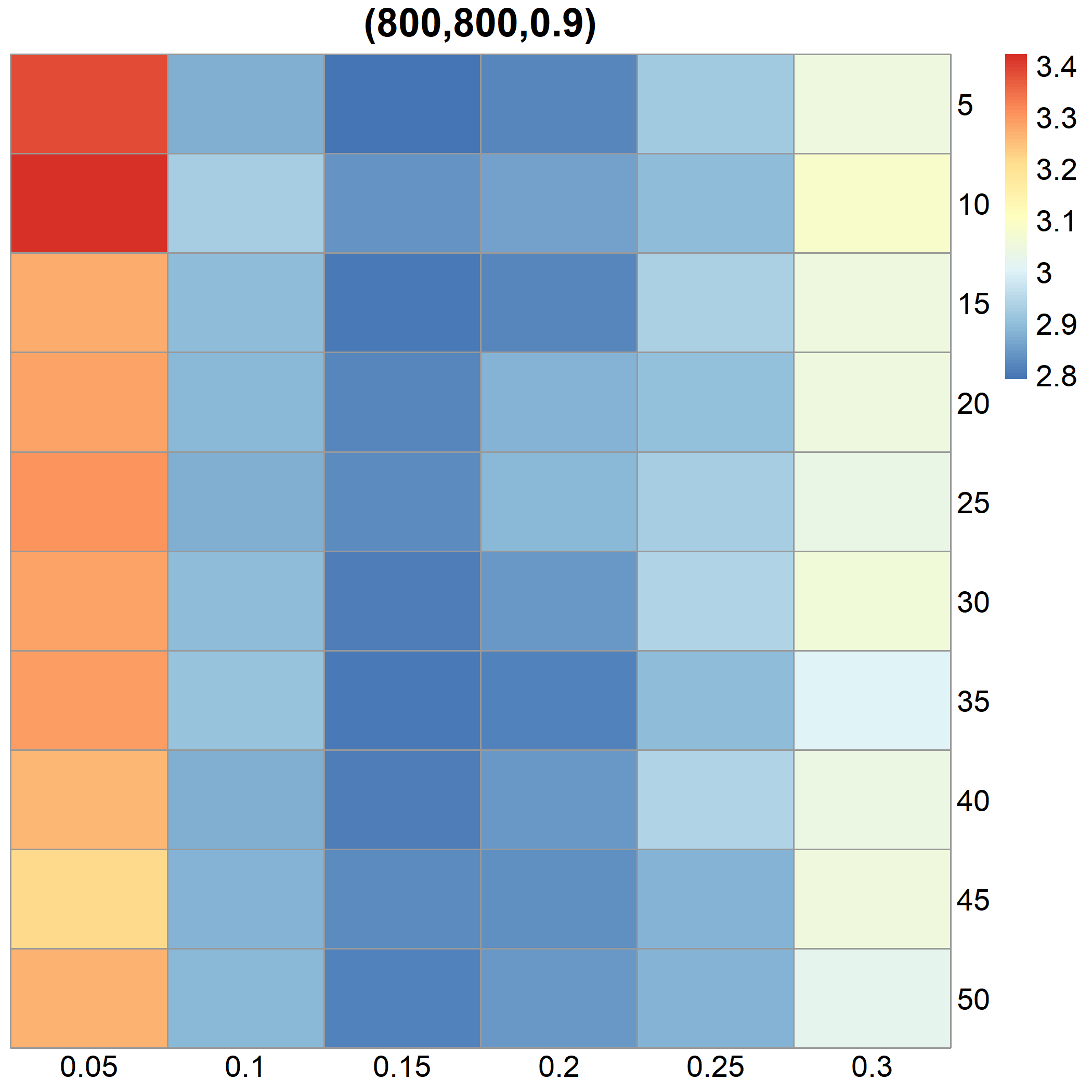}
		\includegraphics[width=0.24\linewidth]{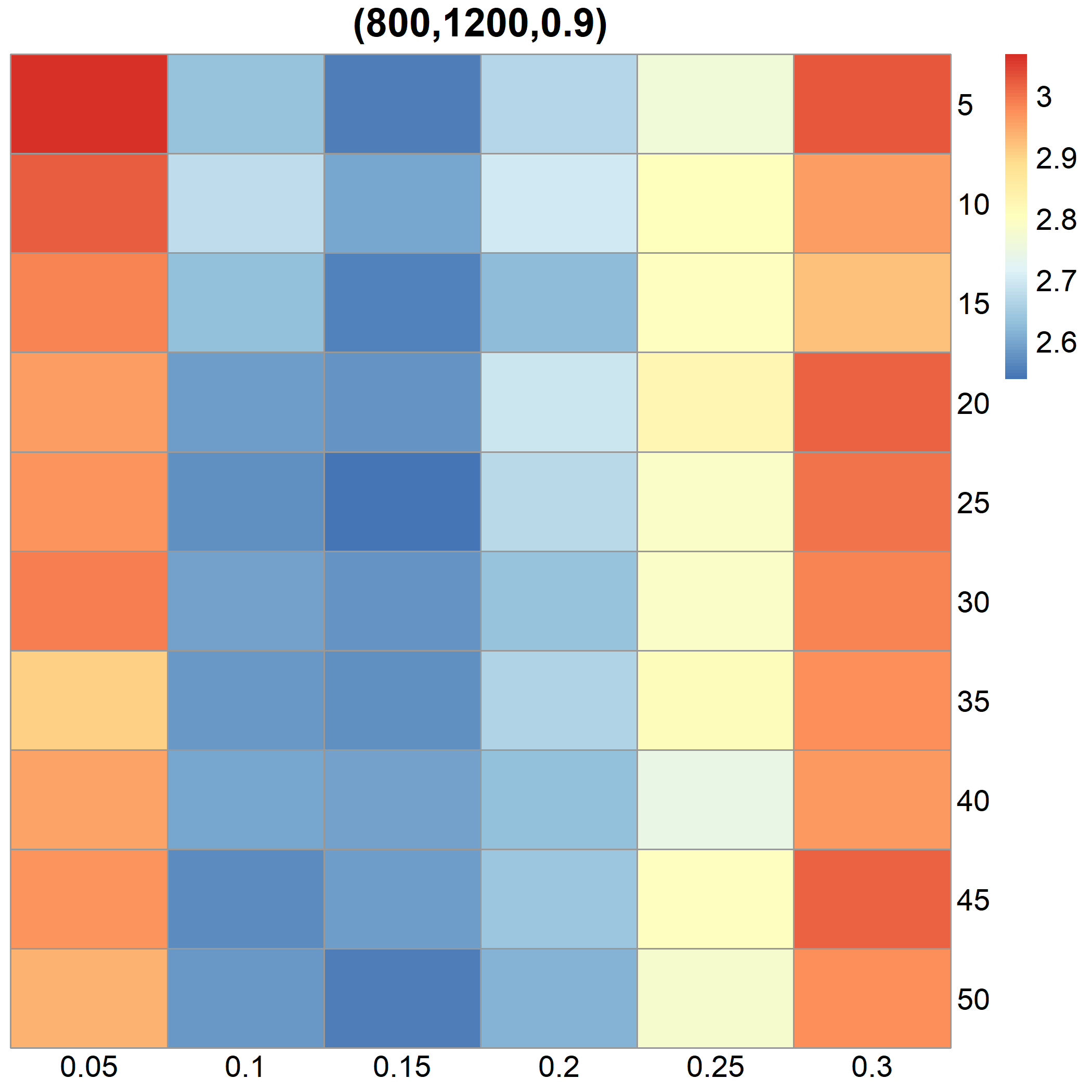}
		
		\vspace{3pt}
		\small (c) $N = 800$
	\end{minipage}
	
	\caption{Cross-validation results for Section \ref{sec3.2} under exponential decay with $\rho = 0.9$. Values in parentheses denote $(N, K, \rho)$. The horizontal axis represents the selection probability $p$, and the vertical axis indicates the number of candidate models $M$. Darker regions correspond to $(p, M)$ combinations yielding lower cross-validation errors.}
	\label{fig:caseexp0.9}
\end{figure*}

\newpage
\begin{figure*}[!h]
	\centering
	
	\begin{minipage}{0.95\textwidth}
		\centering
		\includegraphics[width=0.24\linewidth]{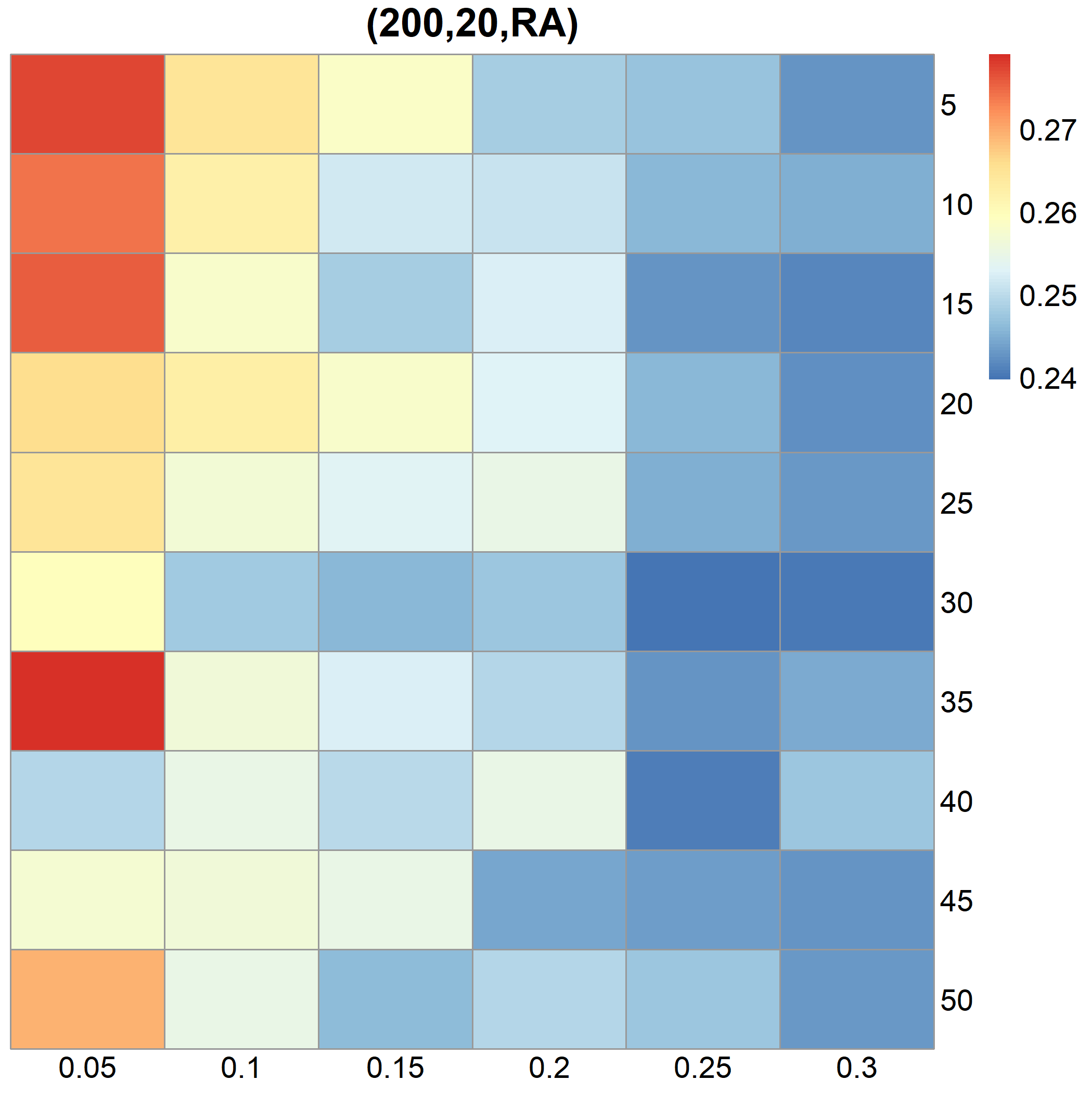}
		\includegraphics[width=0.24\linewidth]{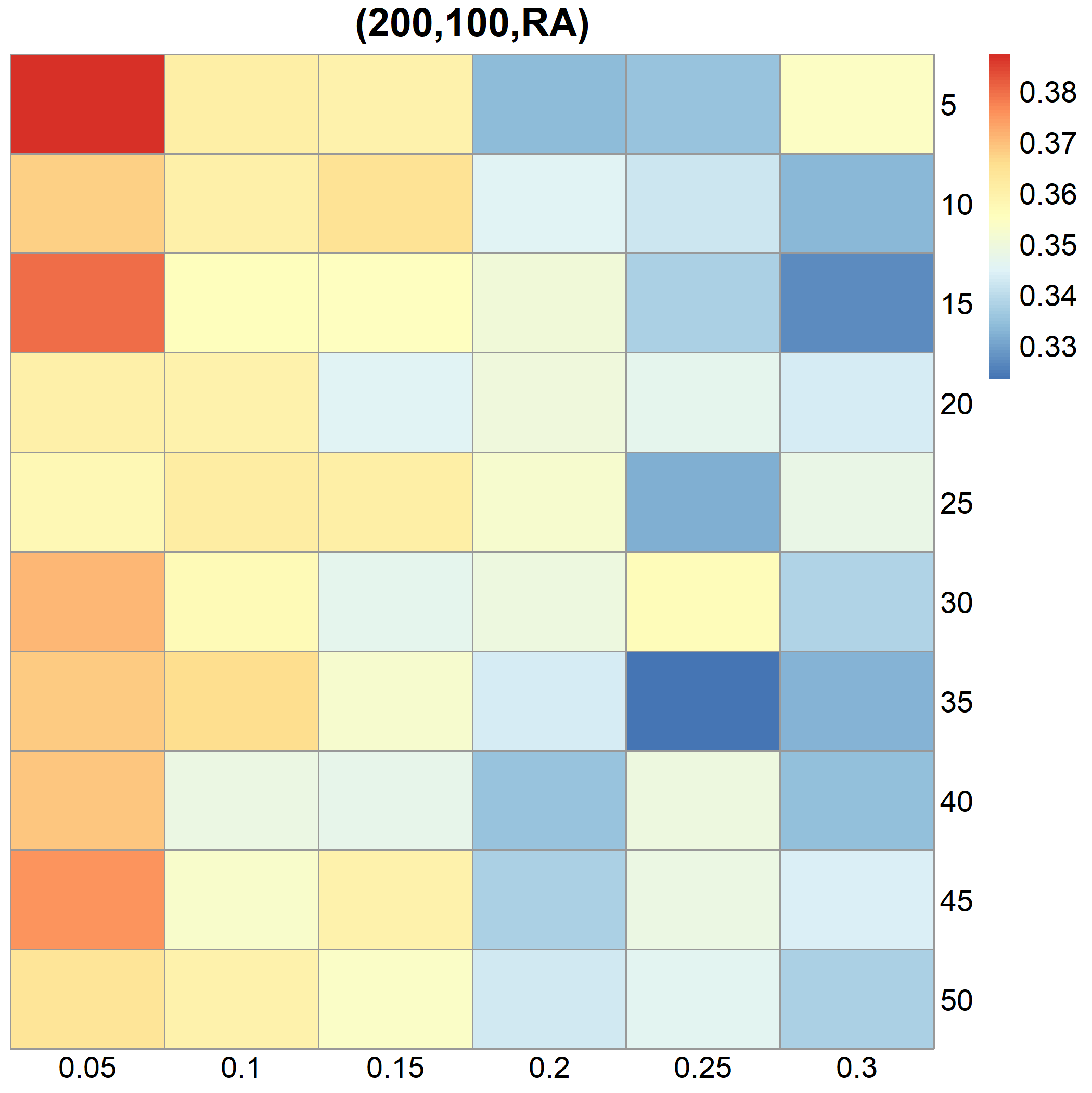}
		\includegraphics[width=0.24\linewidth]{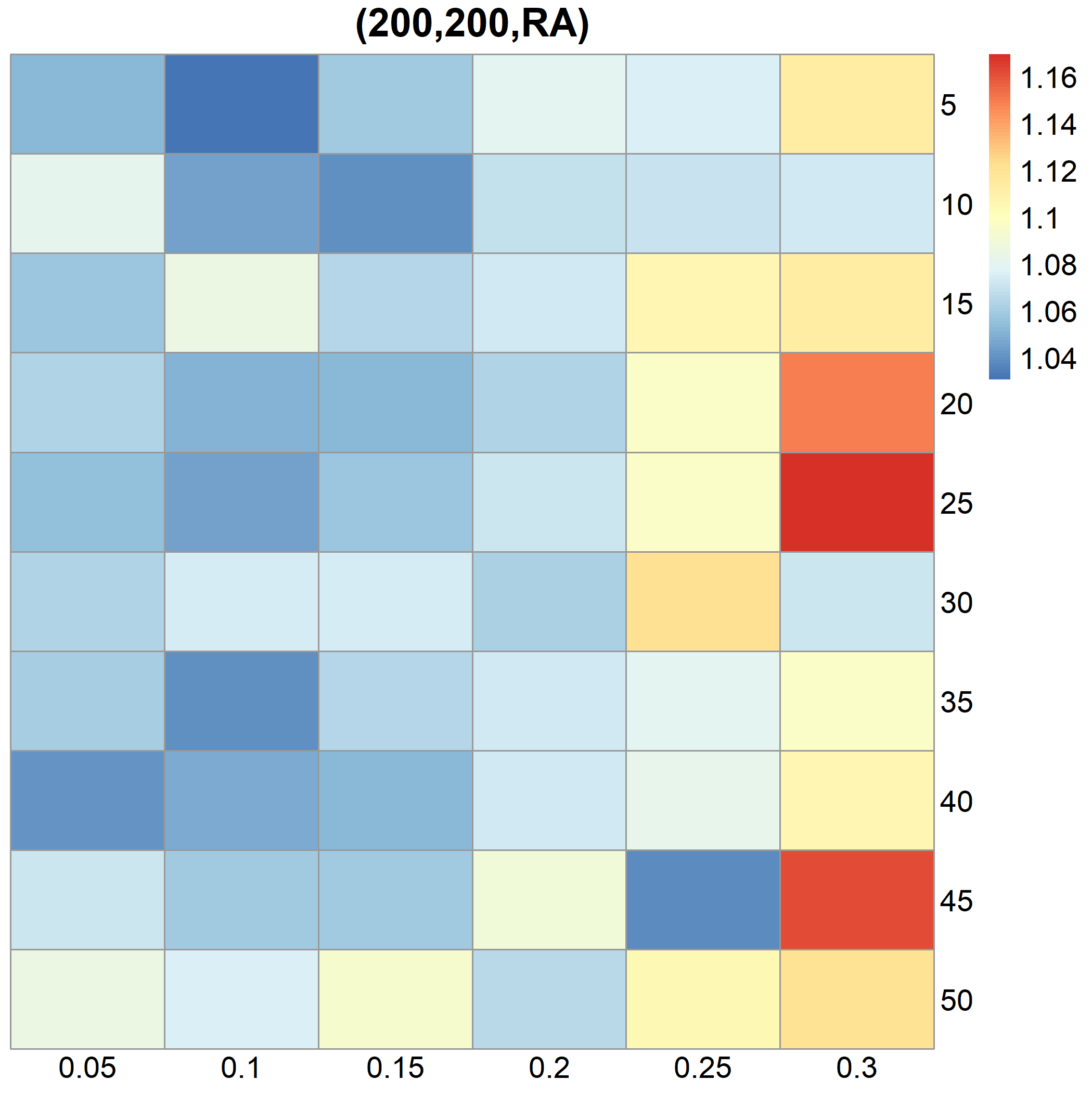}
		\includegraphics[width=0.24\linewidth]{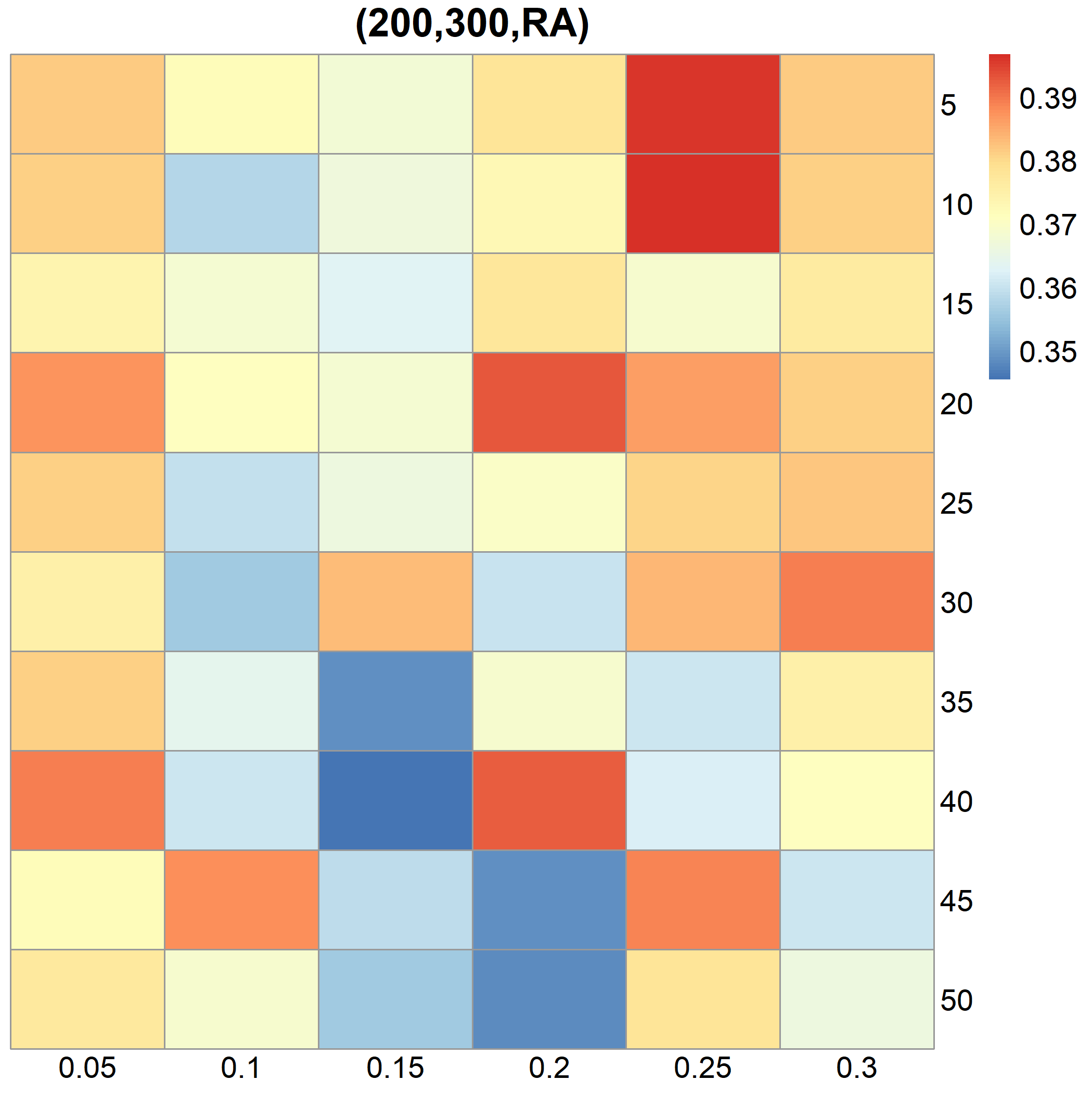}
		
		\vspace{3pt}
		\small (a) $N = 200$
	\end{minipage}
	
	\vspace{6pt}
	
	\begin{minipage}{0.95\textwidth}
		\centering
		\includegraphics[width=0.24\linewidth]{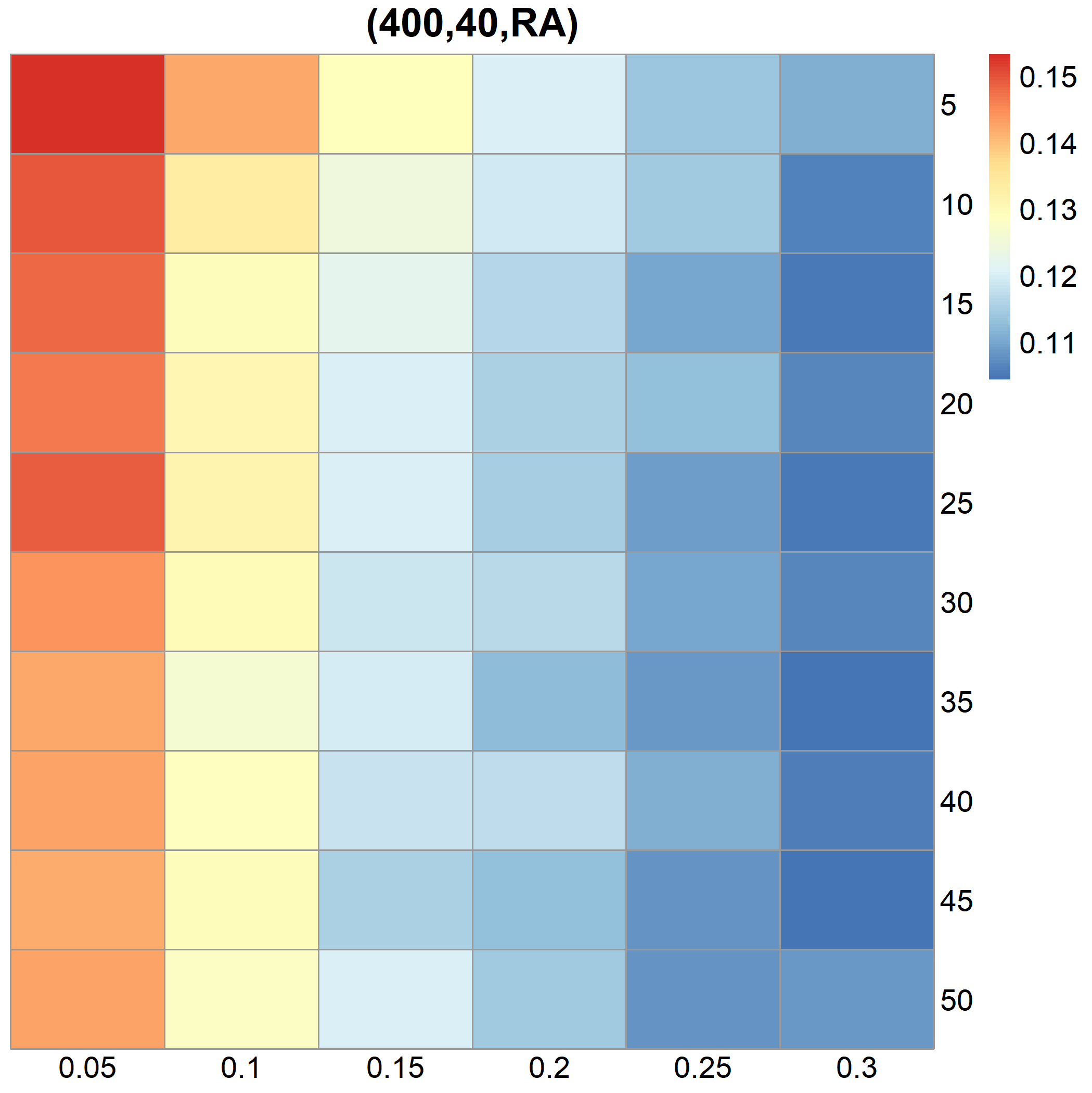}
		\includegraphics[width=0.24\linewidth]{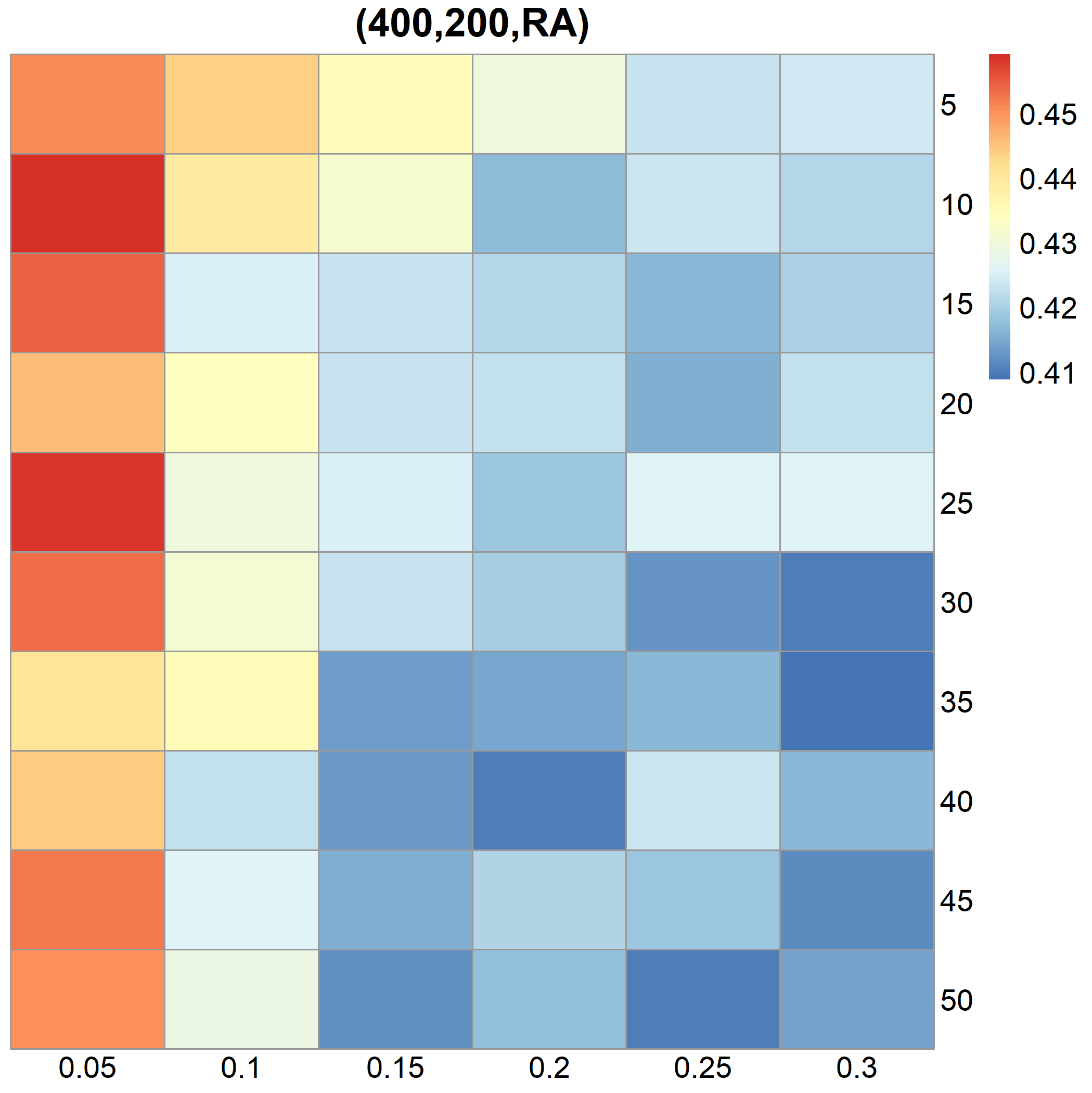}
		\includegraphics[width=0.24\linewidth]{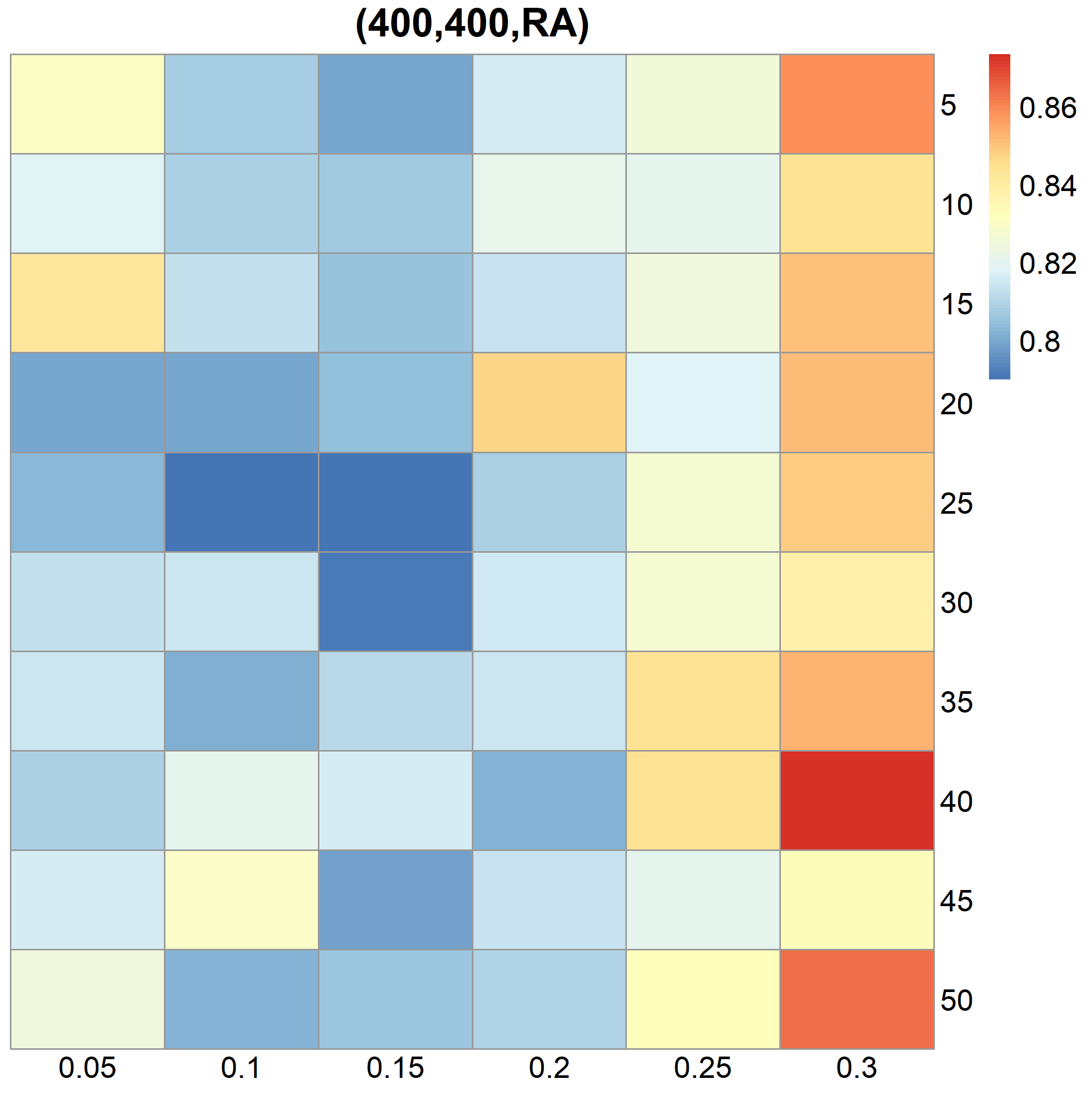}
		\includegraphics[width=0.24\linewidth]{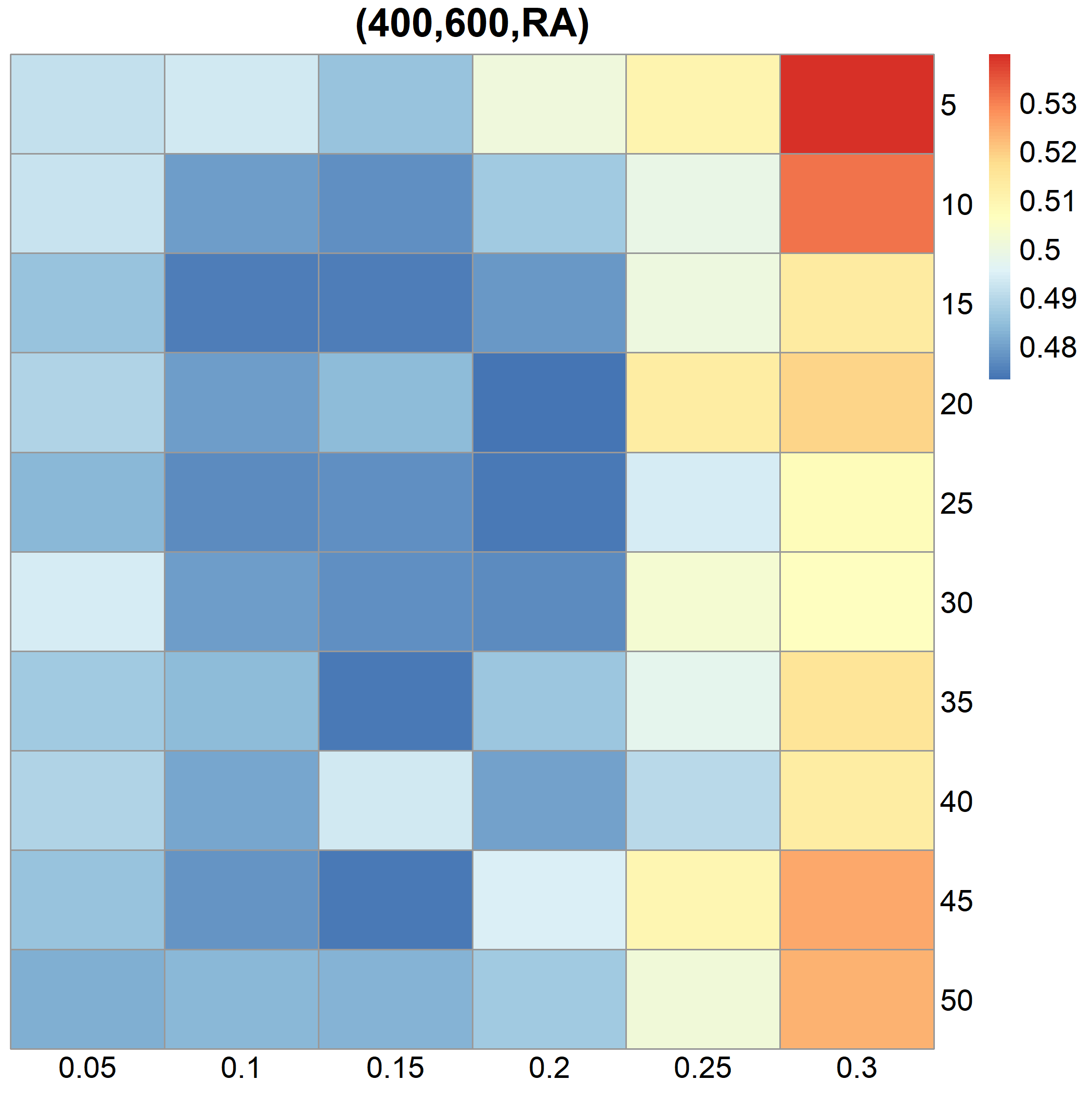}
		
		\vspace{3pt}
		\small (b) $N = 400$
	\end{minipage}
	
	\vspace{6pt}
	
	\begin{minipage}{0.95\textwidth}
		\centering
		\includegraphics[width=0.24\linewidth]{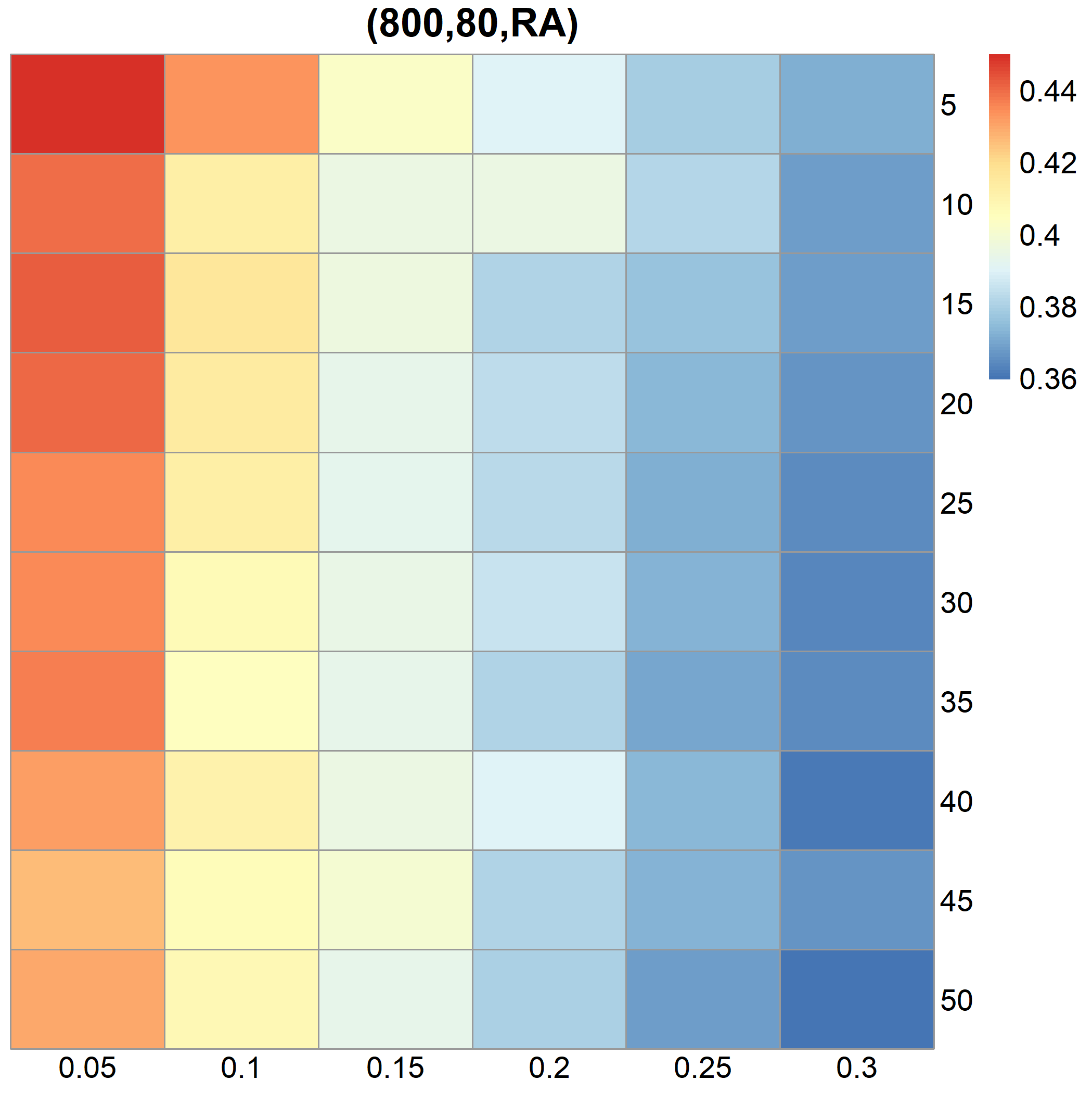}
		\includegraphics[width=0.24\linewidth]{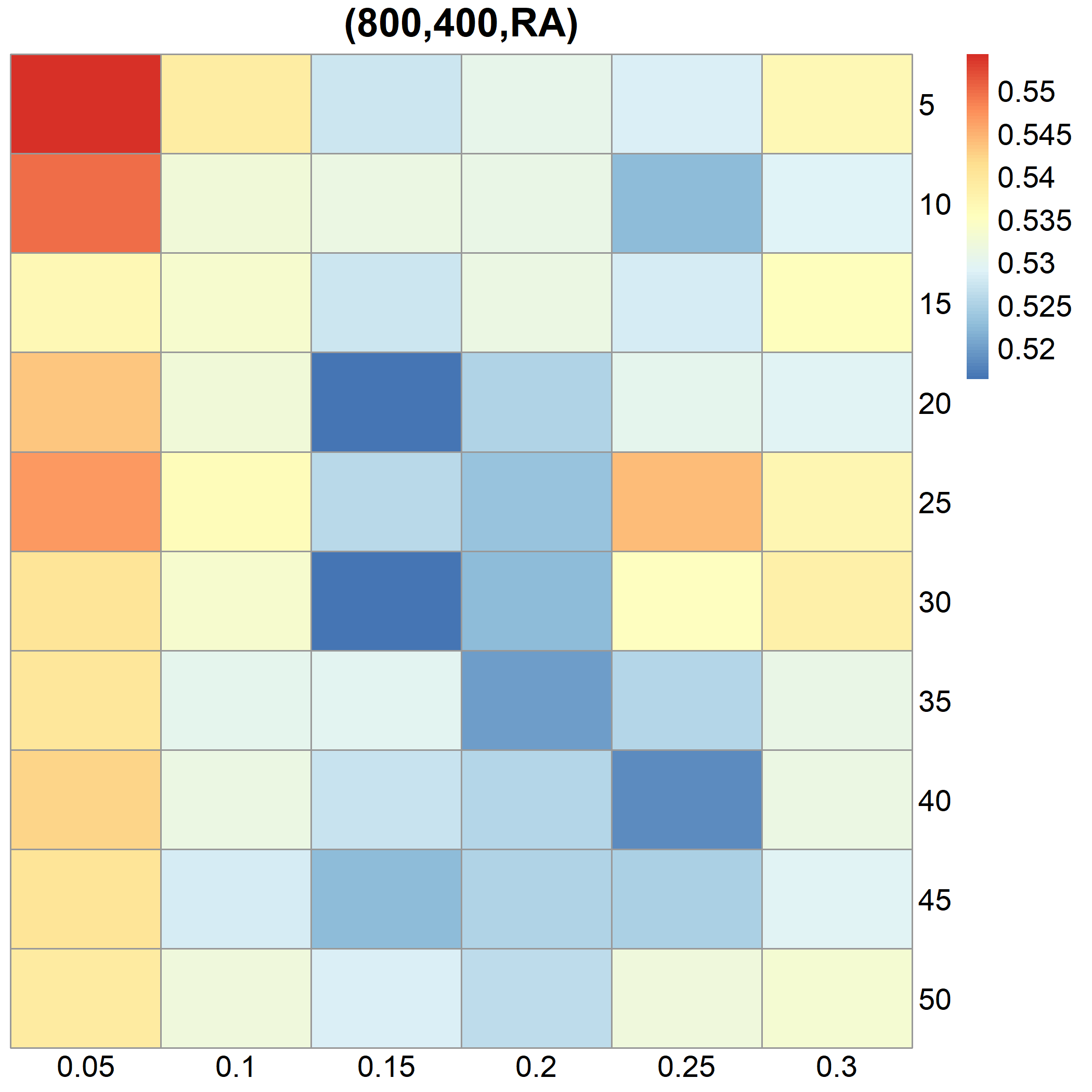}
		\includegraphics[width=0.24\linewidth]{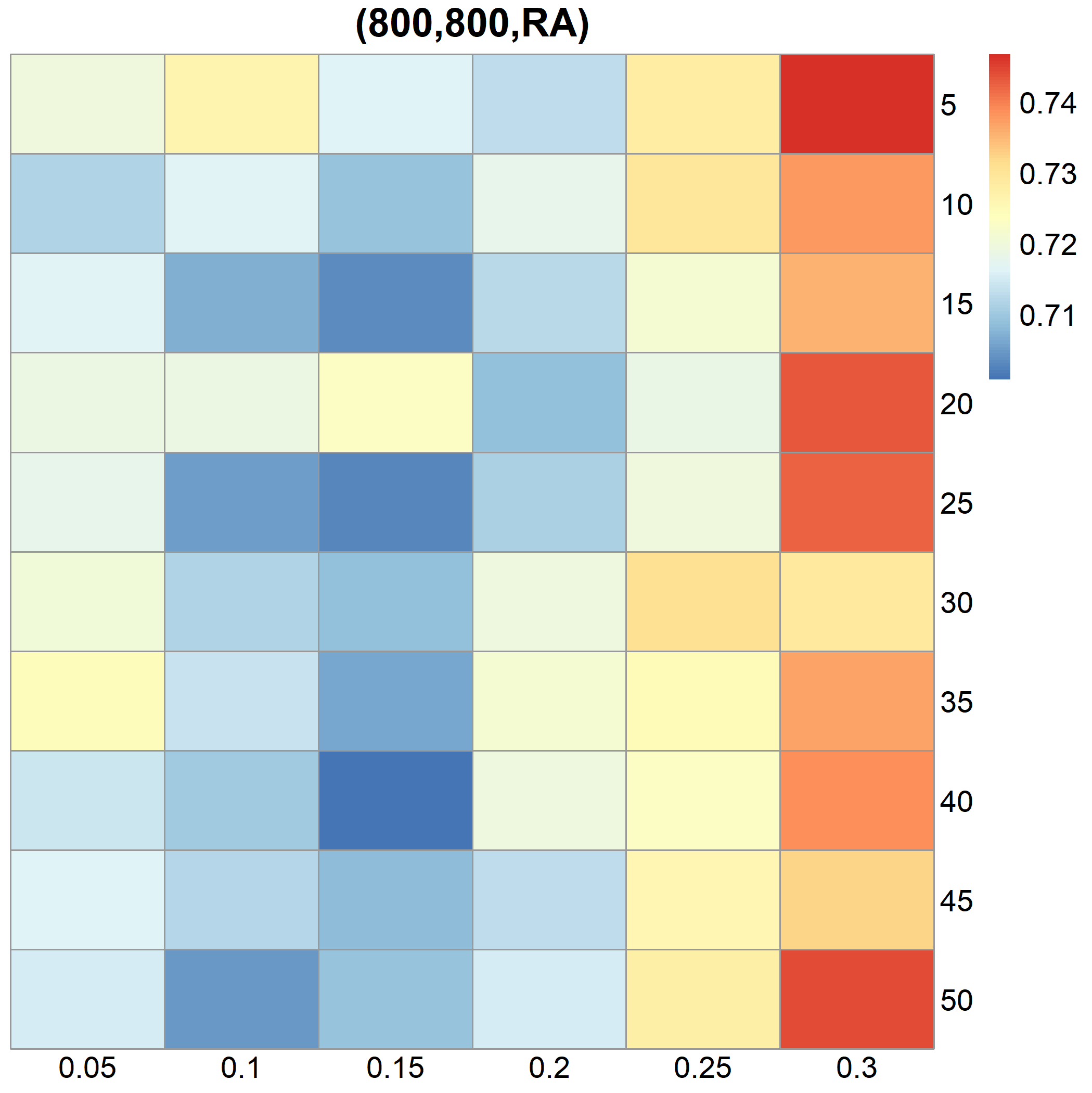}
		\includegraphics[width=0.24\linewidth]{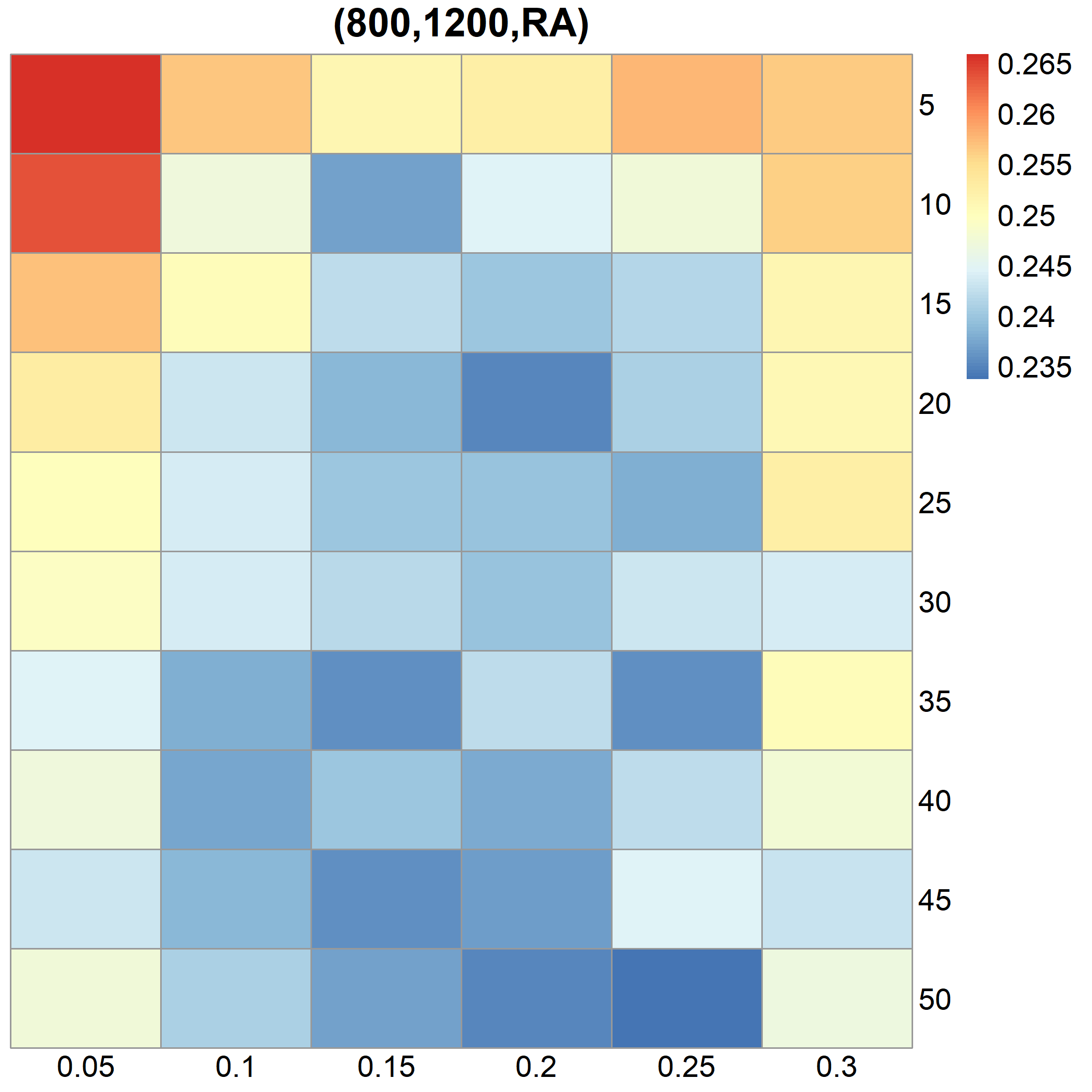}
		
		\vspace{3pt}
		\small (c) $N = 800$
	\end{minipage}
	
	\caption{Cross-validation results for Section \ref{sec3.2} under exponential decay with a random covariance structure. Values in parentheses denote $(N, K, \rho)$. The horizontal axis represents the selection probability $p$, and the vertical axis indicates the number of candidate models $M$. Darker regions correspond to $(p, M)$ combinations yielding lower cross-validation errors.}
	\label{fig:caseexpRA}
\end{figure*}

\newpage

\subsubsection{CV results for Section \ref{app:comparewithpeng}}

Figures~\ref{fig:cv1} and~\ref{fig:cv2} illustrate the cross-validated performance of the RSA estimator under different choices of selection probability. As shown in both figures, the estimator's performance is highly sensitive to this tuning parameter. Moreover, the optimal selection probability depends on the correlation structure among covariates. When the covariates are highly dependent, a smaller selection probability is generally preferred. This finding aligns with the intuition that strong dependence among covariates can be exploited to enhance predictive accuracy.
\begin{figure*}[htbp]
	\centering
	
	\begin{minipage}{0.95\textwidth}
		\centering
		\includegraphics[width=0.32\linewidth]{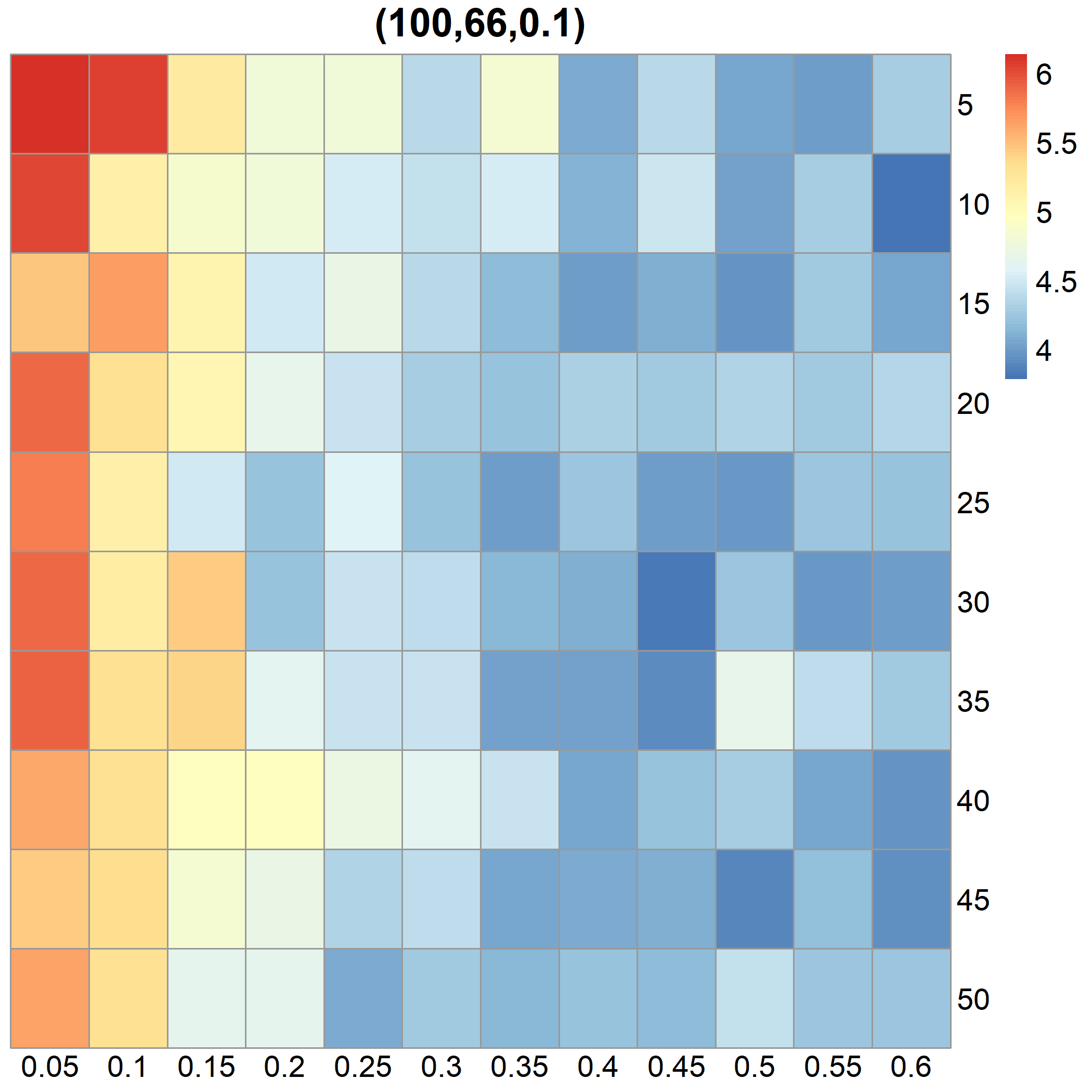}
		\includegraphics[width=0.32\linewidth]{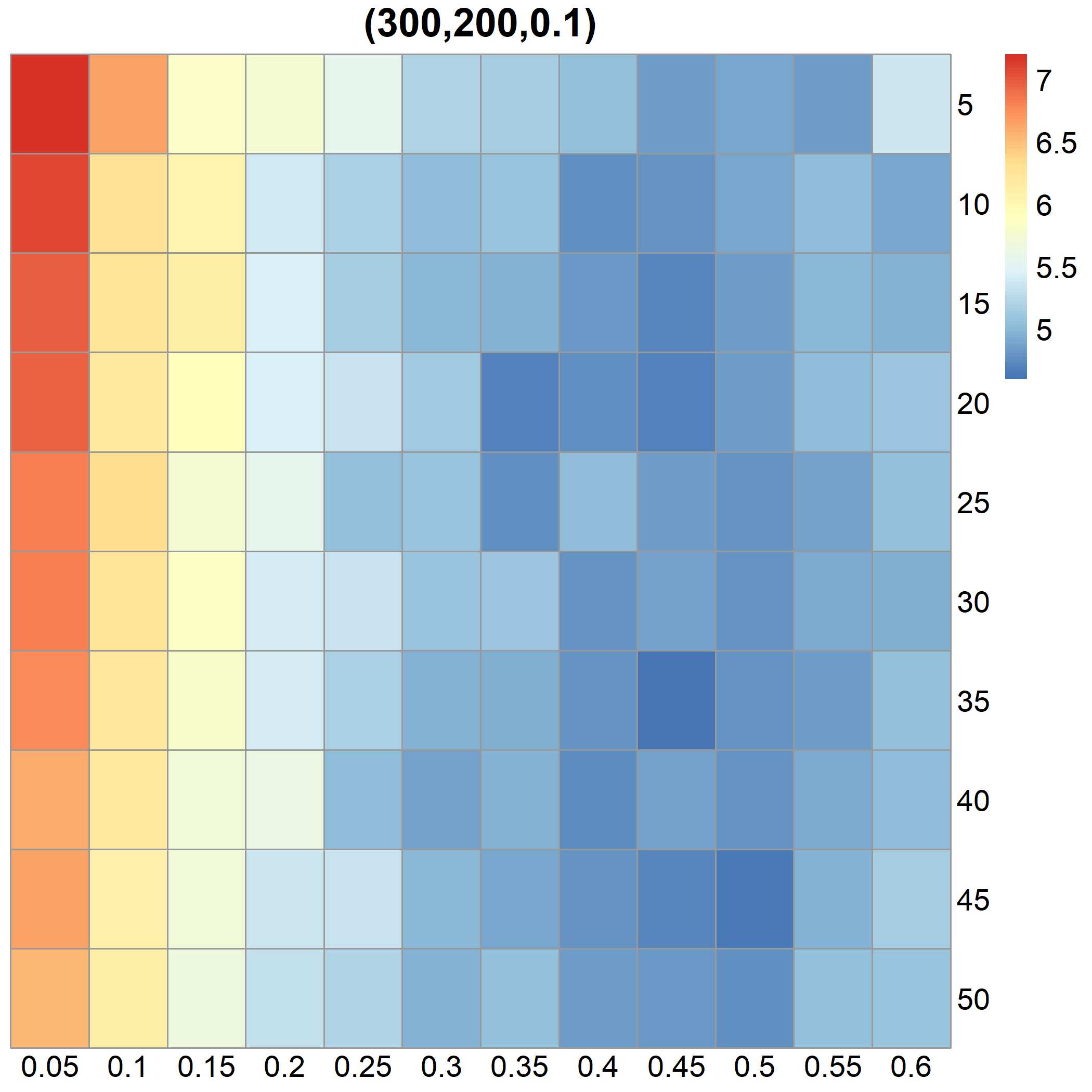}
		\includegraphics[width=0.32\linewidth]{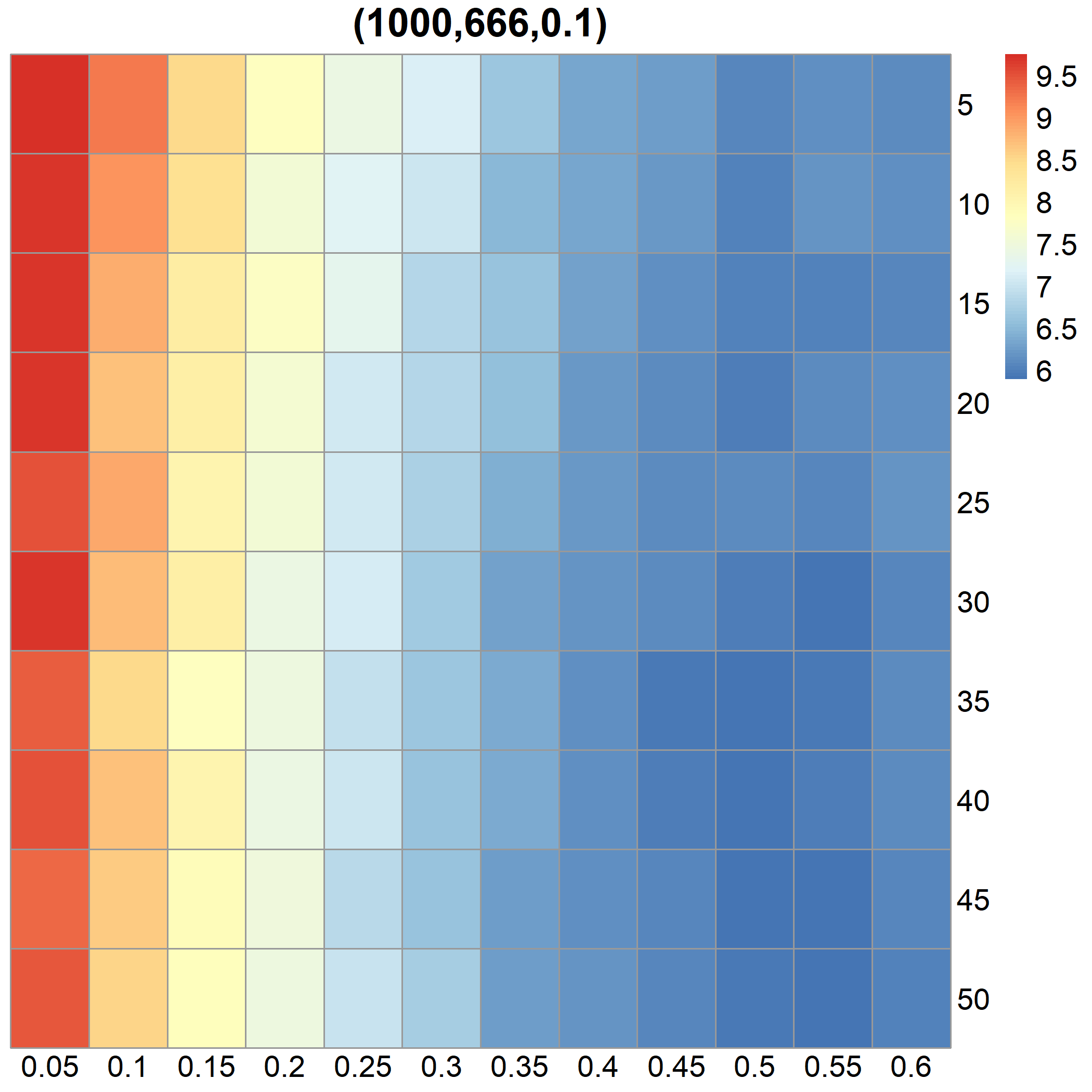}
		
		\vspace{3pt}
		\small (a) $\rho = 0.1$
	\end{minipage}
	
	\vspace{6pt}
	
	\begin{minipage}{0.95\textwidth}
		\centering
		\includegraphics[width=0.32\linewidth]{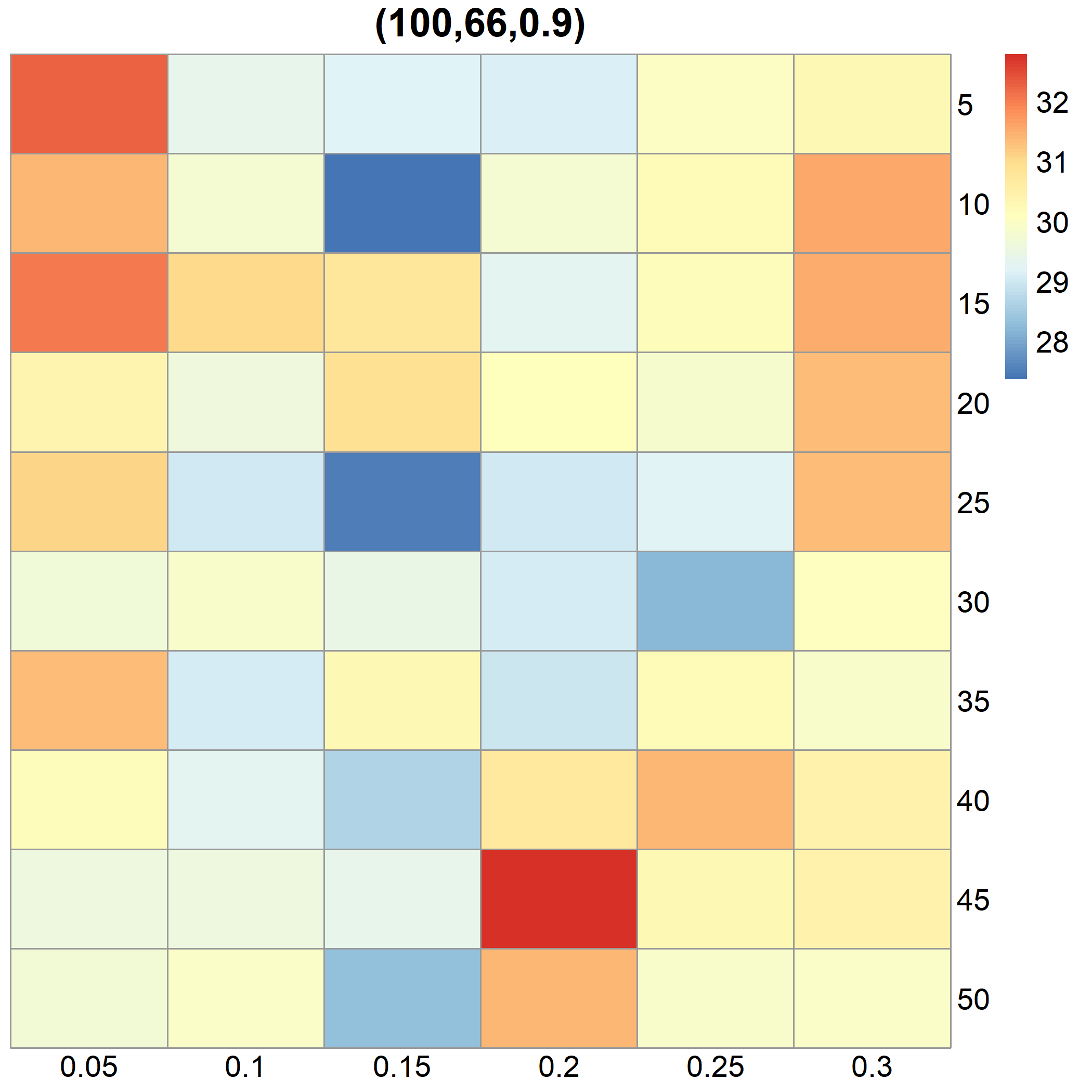}
		\includegraphics[width=0.32\linewidth]{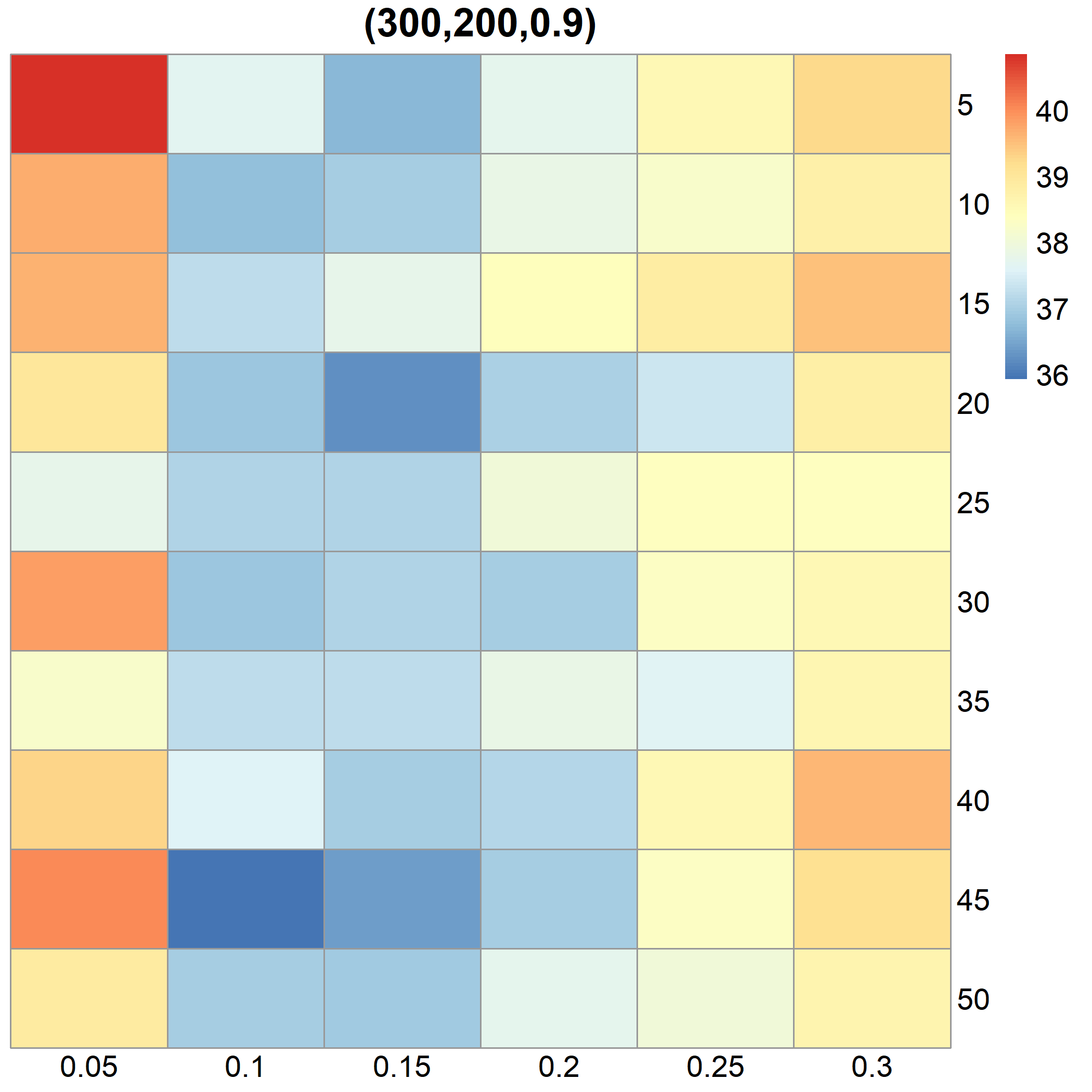}
		\includegraphics[width=0.32\linewidth]{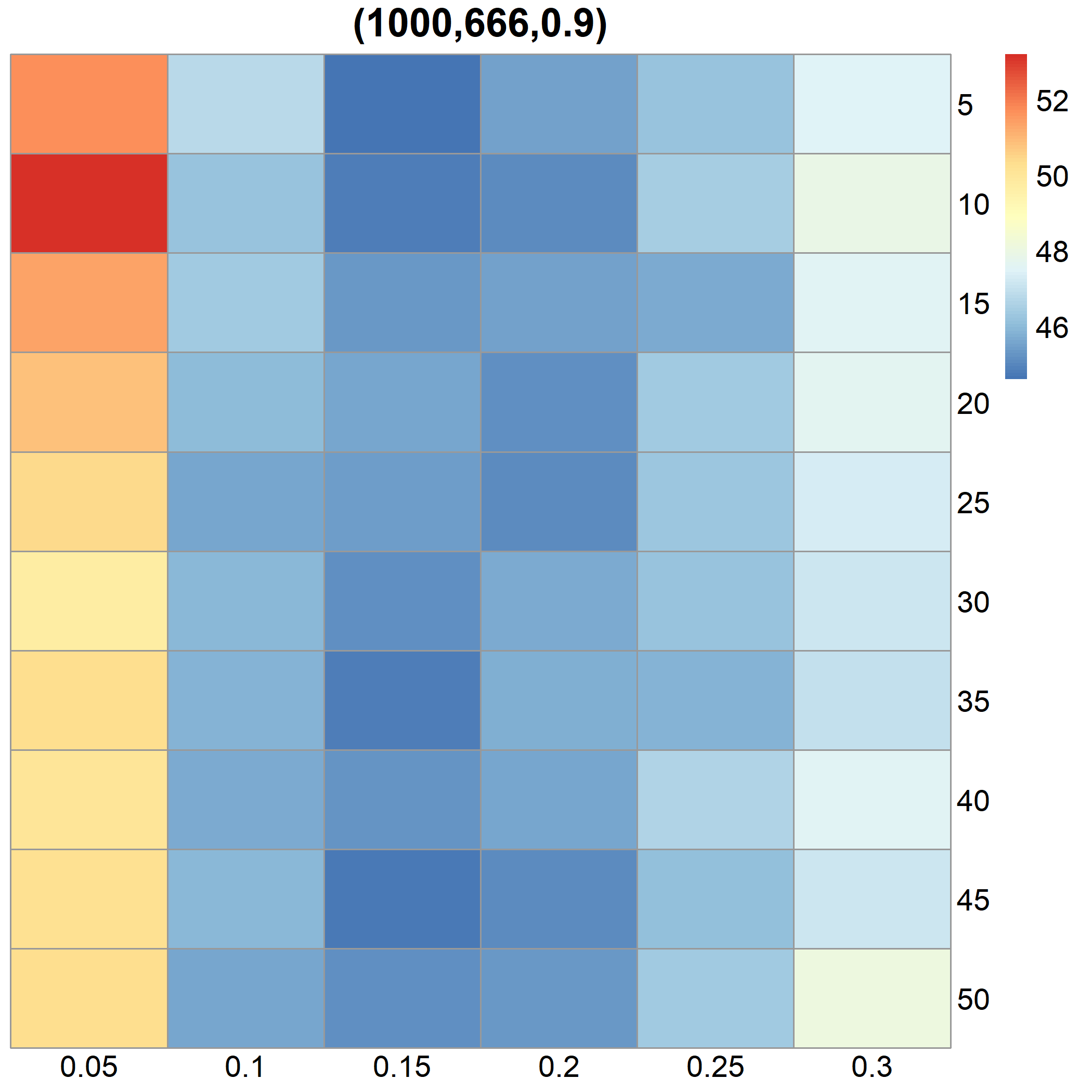}
		
		\vspace{3pt}
		\small (b) $\rho = 0.9$
	\end{minipage}
	
	\caption{Cross-validation results for Section \ref{app:comparewithpeng} under polynomial decay. Values in parentheses denote $(N, K, \rho)$. The horizontal axis represents the selection probability $p$, and the vertical axis indicates the number of candidate models $M$. Darker regions correspond to $(p, M)$ combinations yielding lower cross-validation errors.}
	\label{fig:cv1}
\end{figure*}

\begin{figure*}[htbp]
	\centering
	
	\begin{minipage}{0.95\textwidth}
		\centering
		\includegraphics[width=0.32\linewidth]{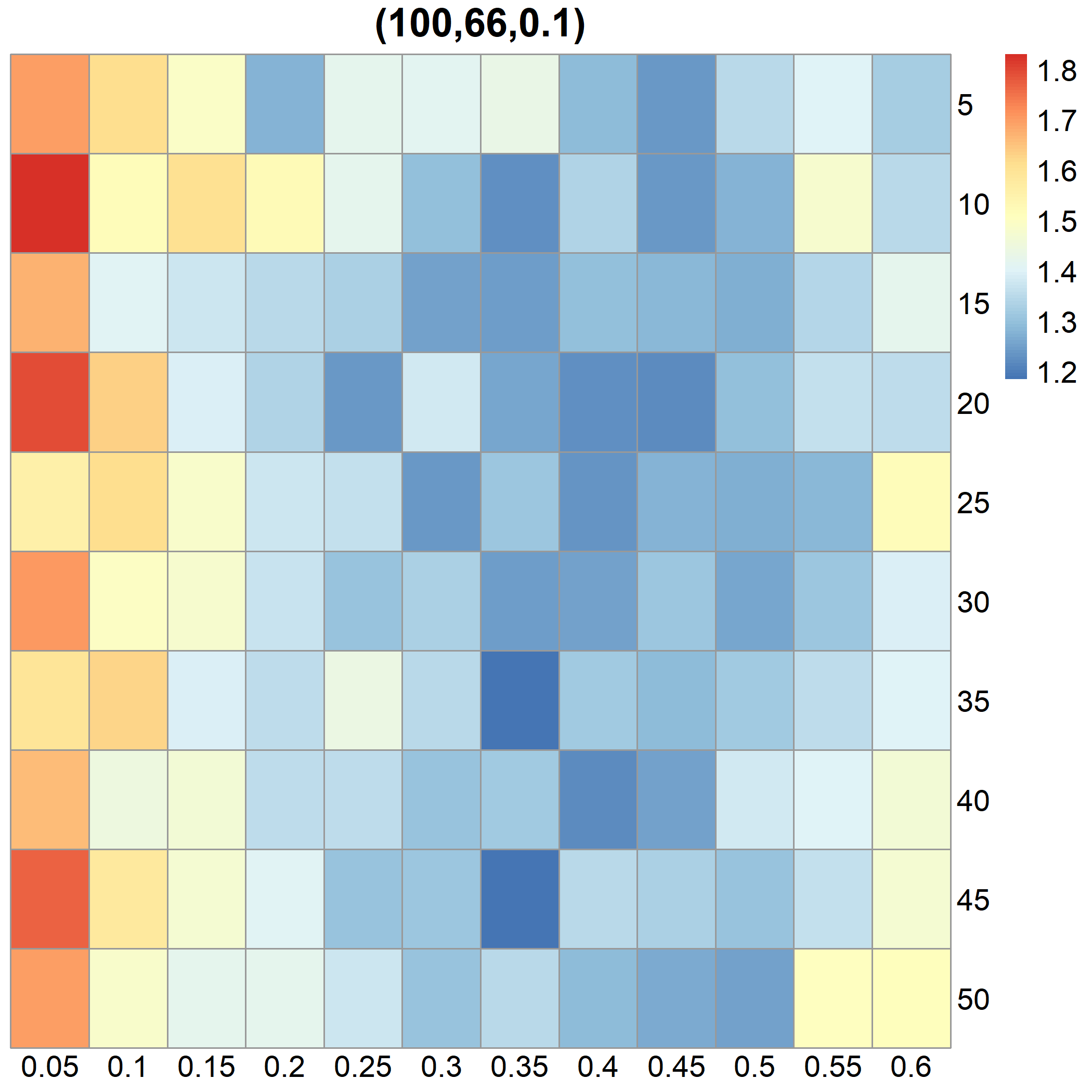}
		\includegraphics[width=0.32\linewidth]{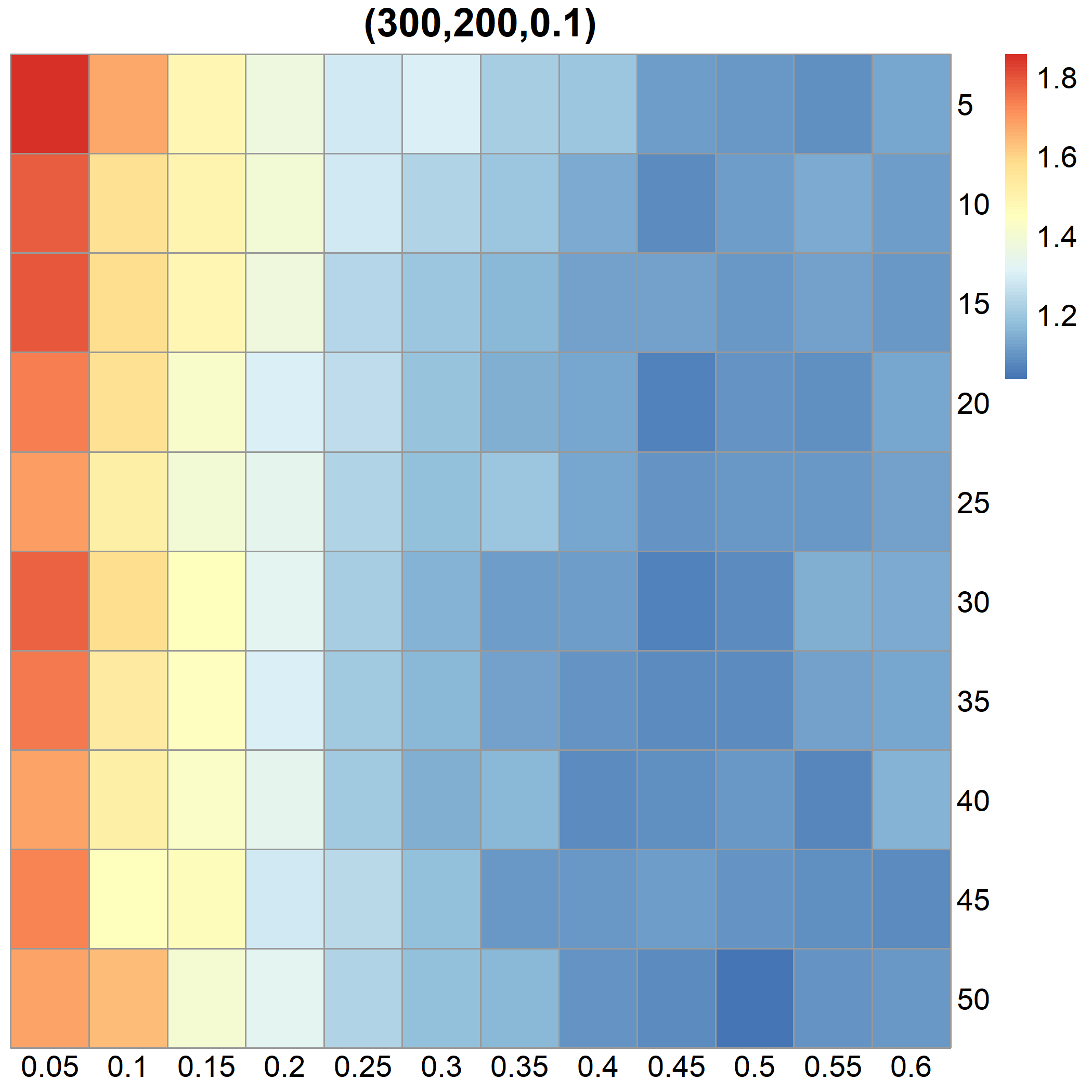}
		\includegraphics[width=0.32\linewidth]{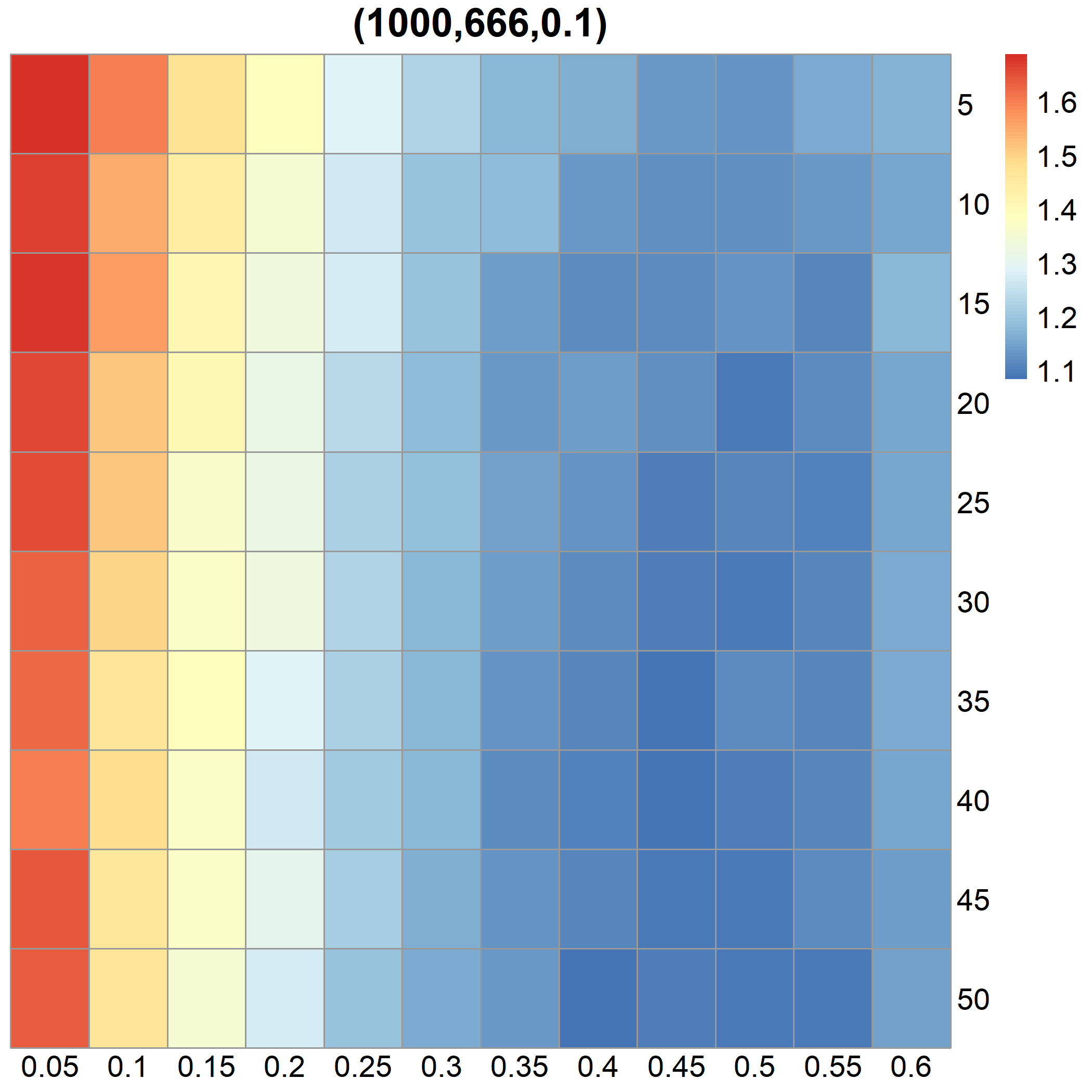}
		
		\vspace{3pt}
		\small (a) $\rho = 0.1$
	\end{minipage}
	
	\vspace{6pt}
	
	\begin{minipage}{0.95\textwidth}
		\centering
		\includegraphics[width=0.32\linewidth]{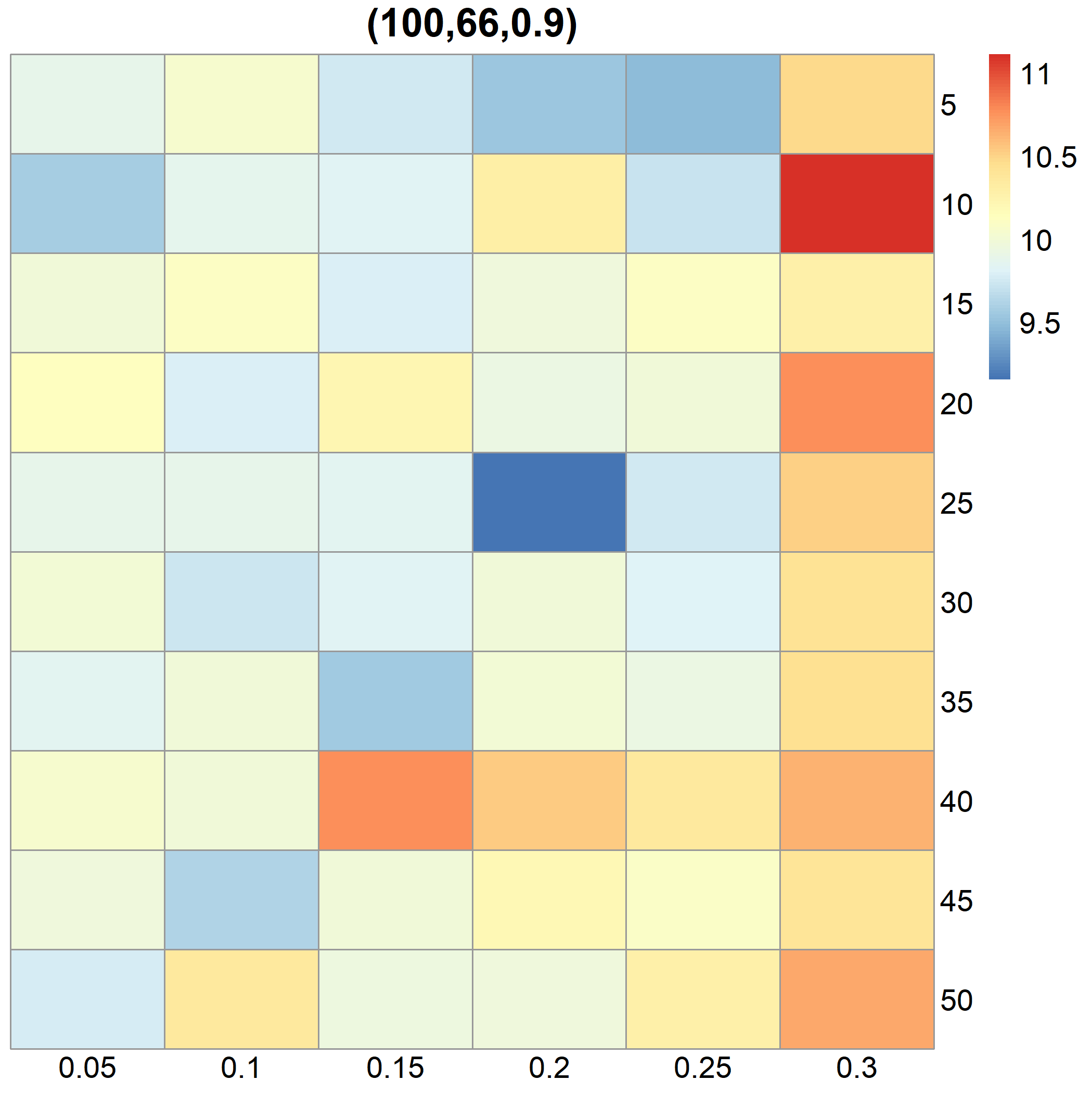}
		\includegraphics[width=0.32\linewidth]{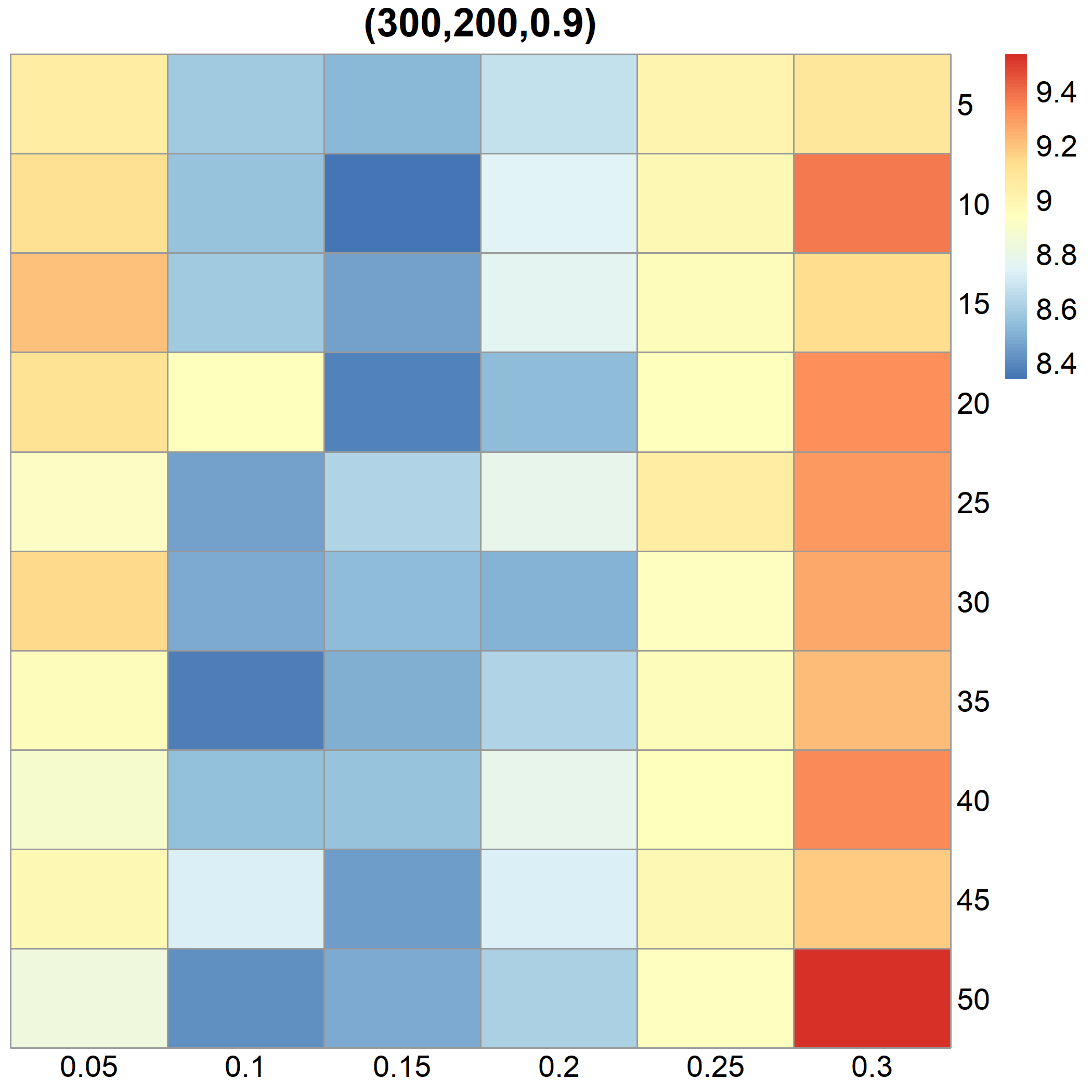}
		\includegraphics[width=0.32\linewidth]{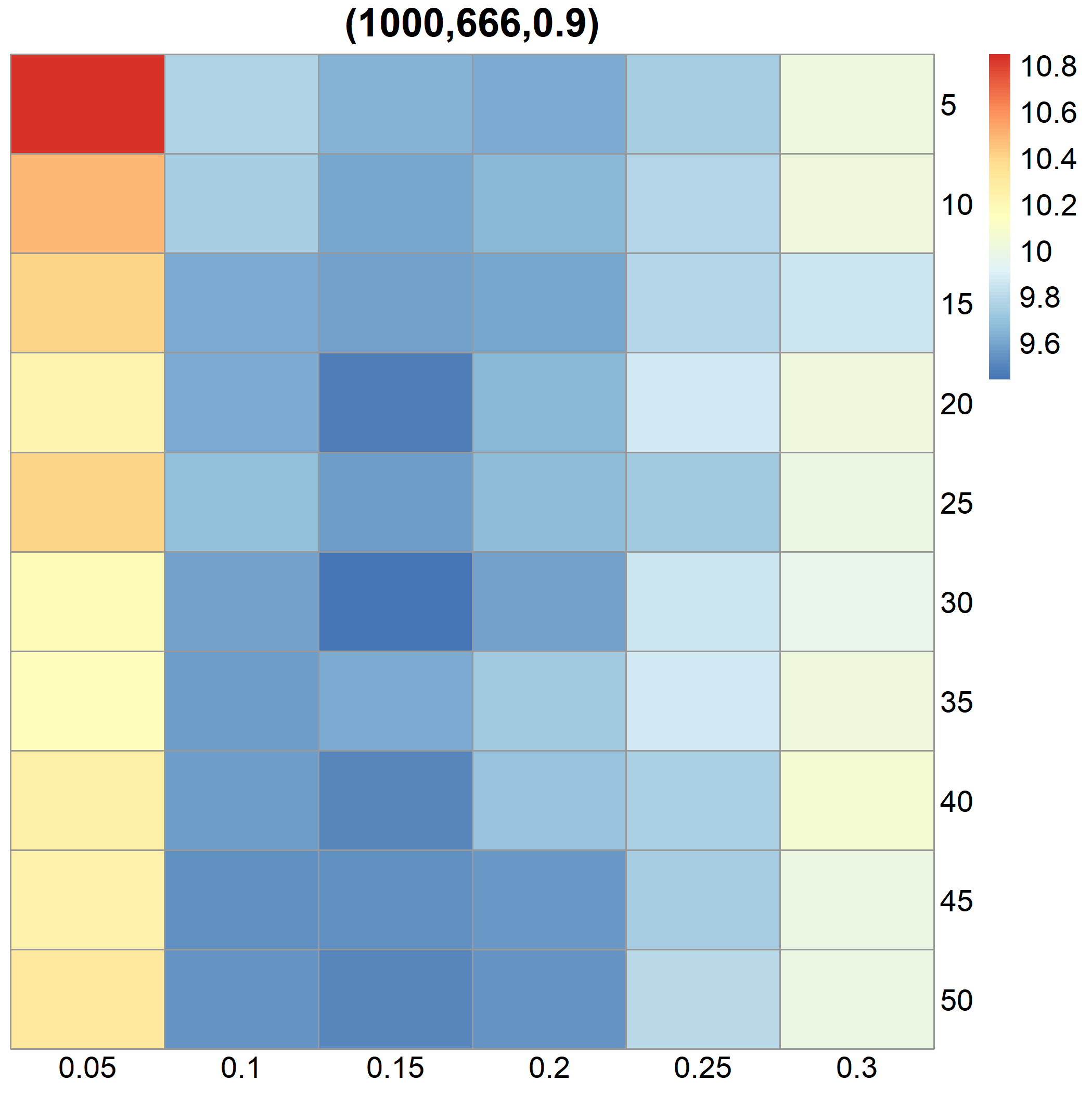}
		
		\vspace{3pt}
		\small (b) $\rho = 0.9$
	\end{minipage}
	
	\caption{Cross-validation results for Section \ref{app:comparewithpeng} under exponential decay. Values in parentheses denote $(N, K, \rho)$. The horizontal axis represents the selection probability $p$, and the vertical axis indicates the number of candidate models $M$. Darker regions correspond to $(p, M)$ combinations yielding lower cross-validation errors.}
	\label{fig:cv2}
\end{figure*}

\newpage
\subsubsection{CV results for Section \ref{app:manyrelevantvariables}}

Figures~\ref{fig:cvMR0.1}--\ref{fig:cvMRexpRA} report the cross-validation results in settings where a large number of covariates are believed to contribute to predictive performance. The findings are consistent with our earlier conclusions and further support the effectiveness of the cross-validation procedure.

\begin{figure*}[htbp]
	\centering
	
	\begin{minipage}{0.95\textwidth}
		\centering
		\includegraphics[width=0.32\linewidth]{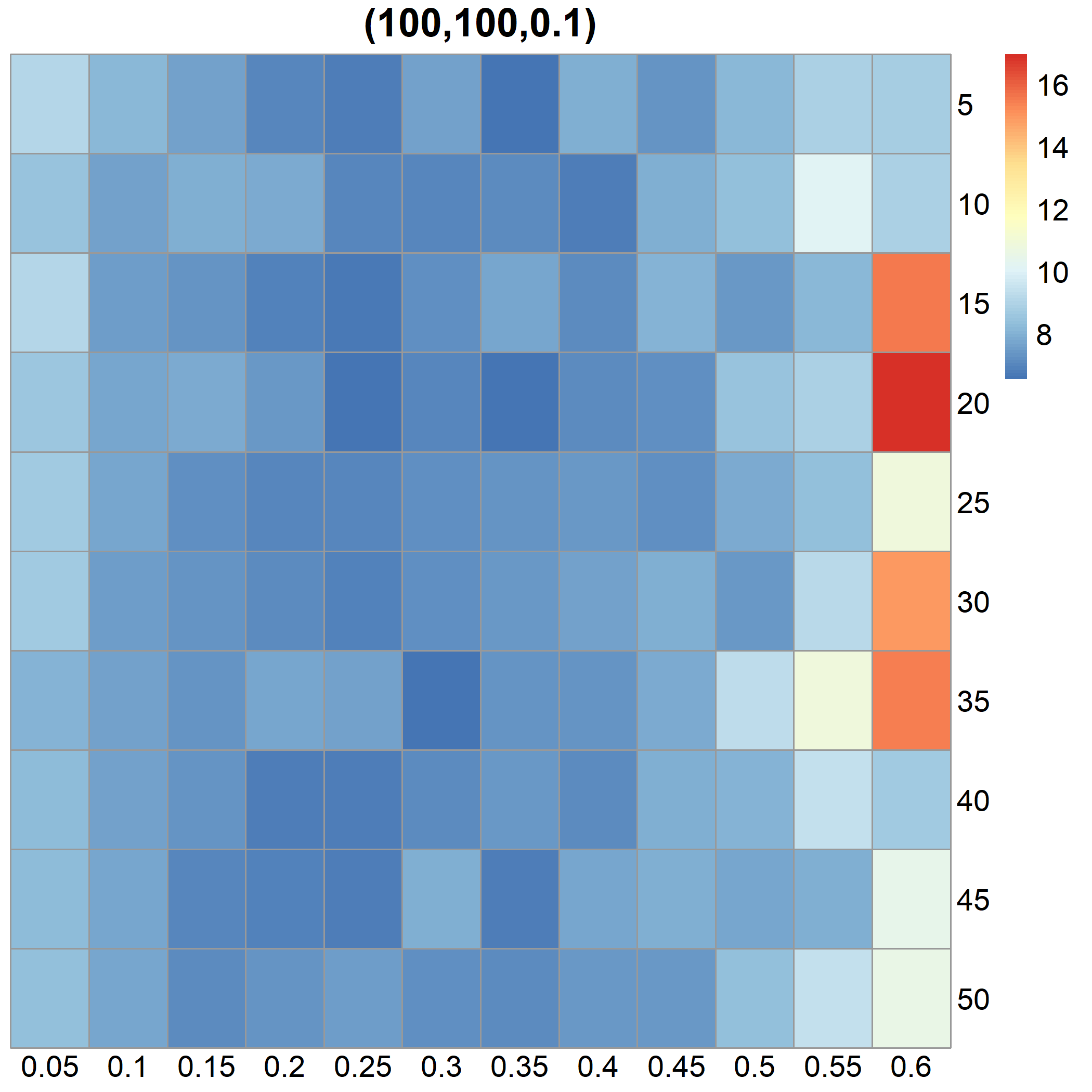}
		\includegraphics[width=0.32\linewidth]{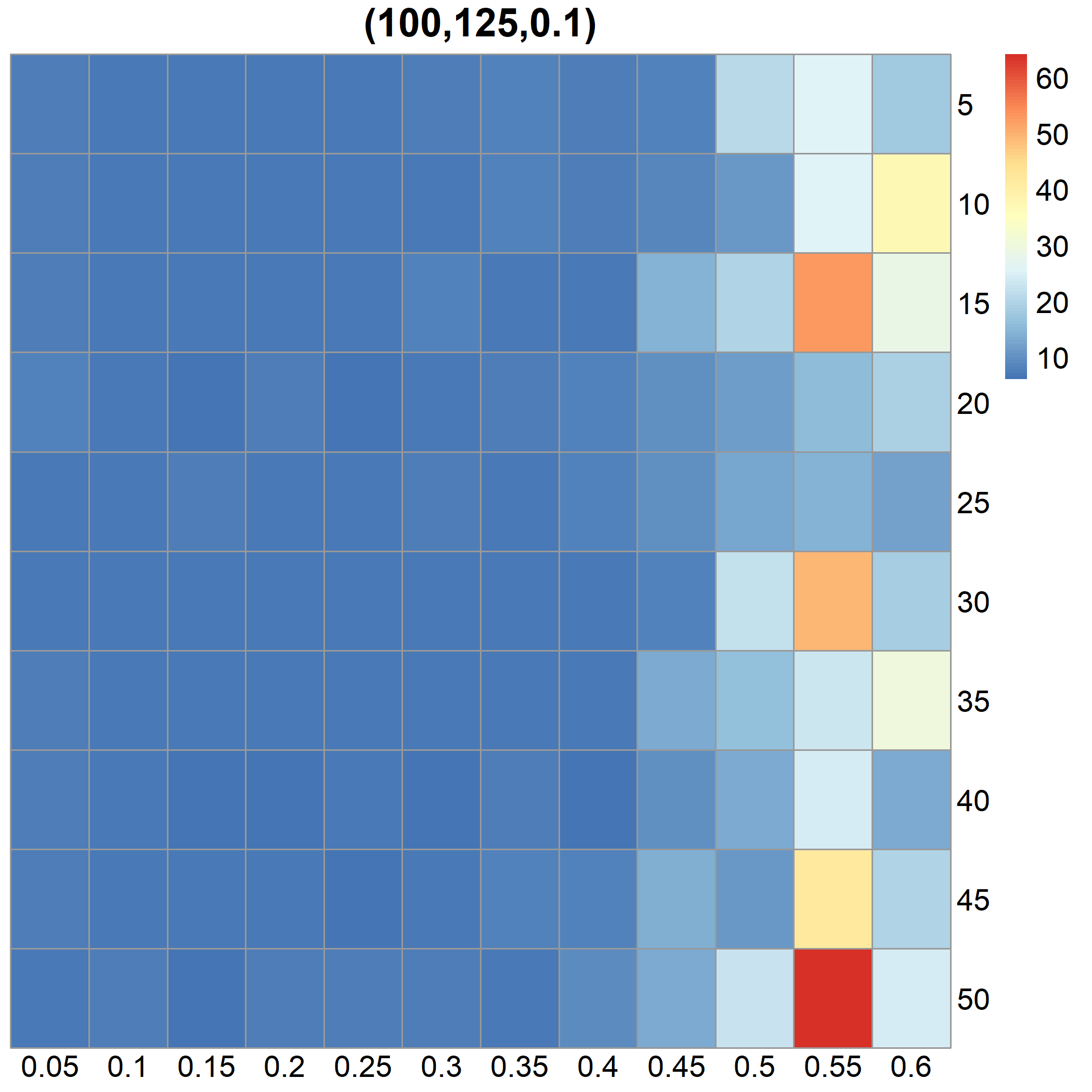}
		\includegraphics[width=0.32\linewidth]{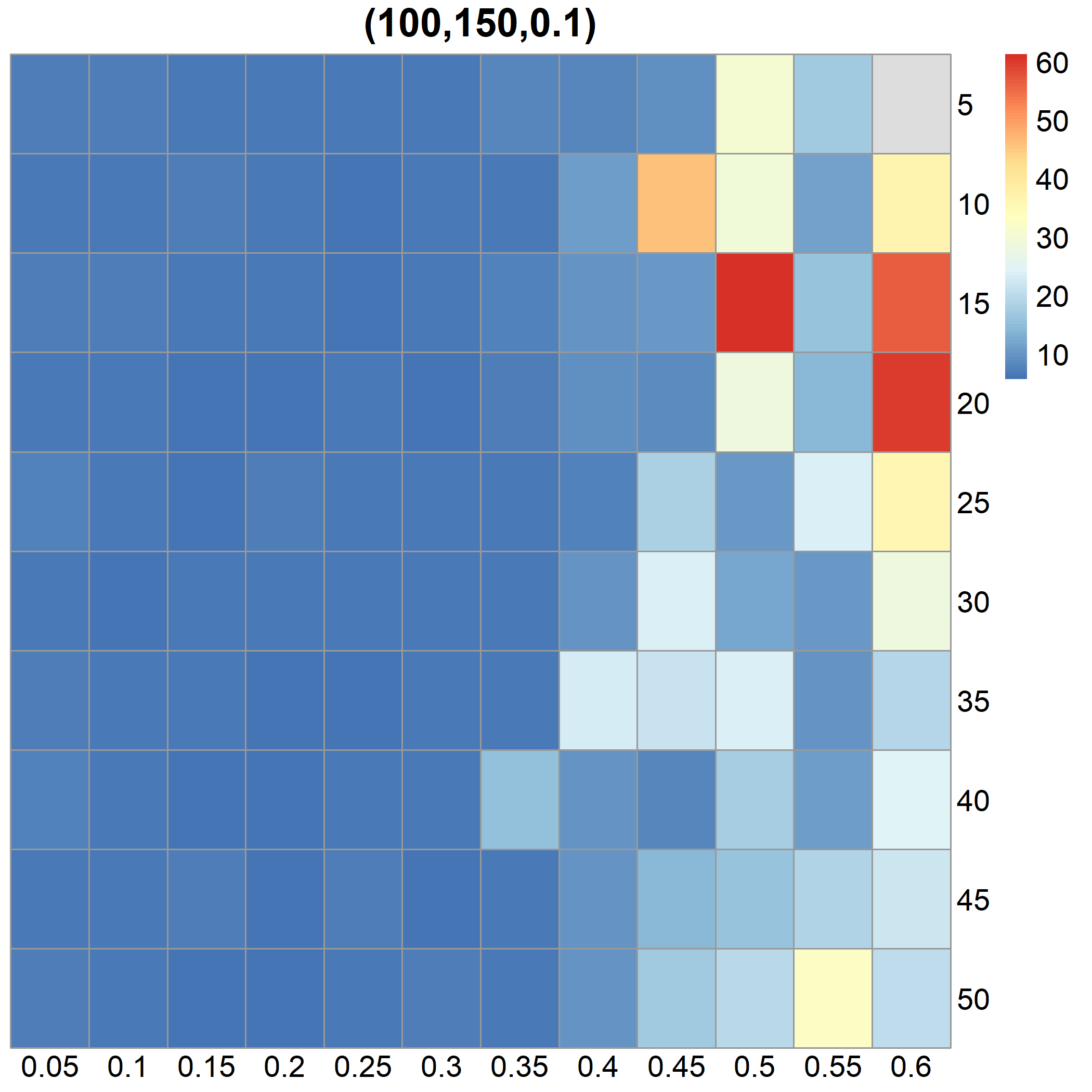}
		
		\vspace{3pt}
		\small (a) $N = 100$
	\end{minipage}
	
	\vspace{6pt}
	
	\begin{minipage}{0.95\textwidth}
		\centering
		\includegraphics[width=0.32\linewidth]{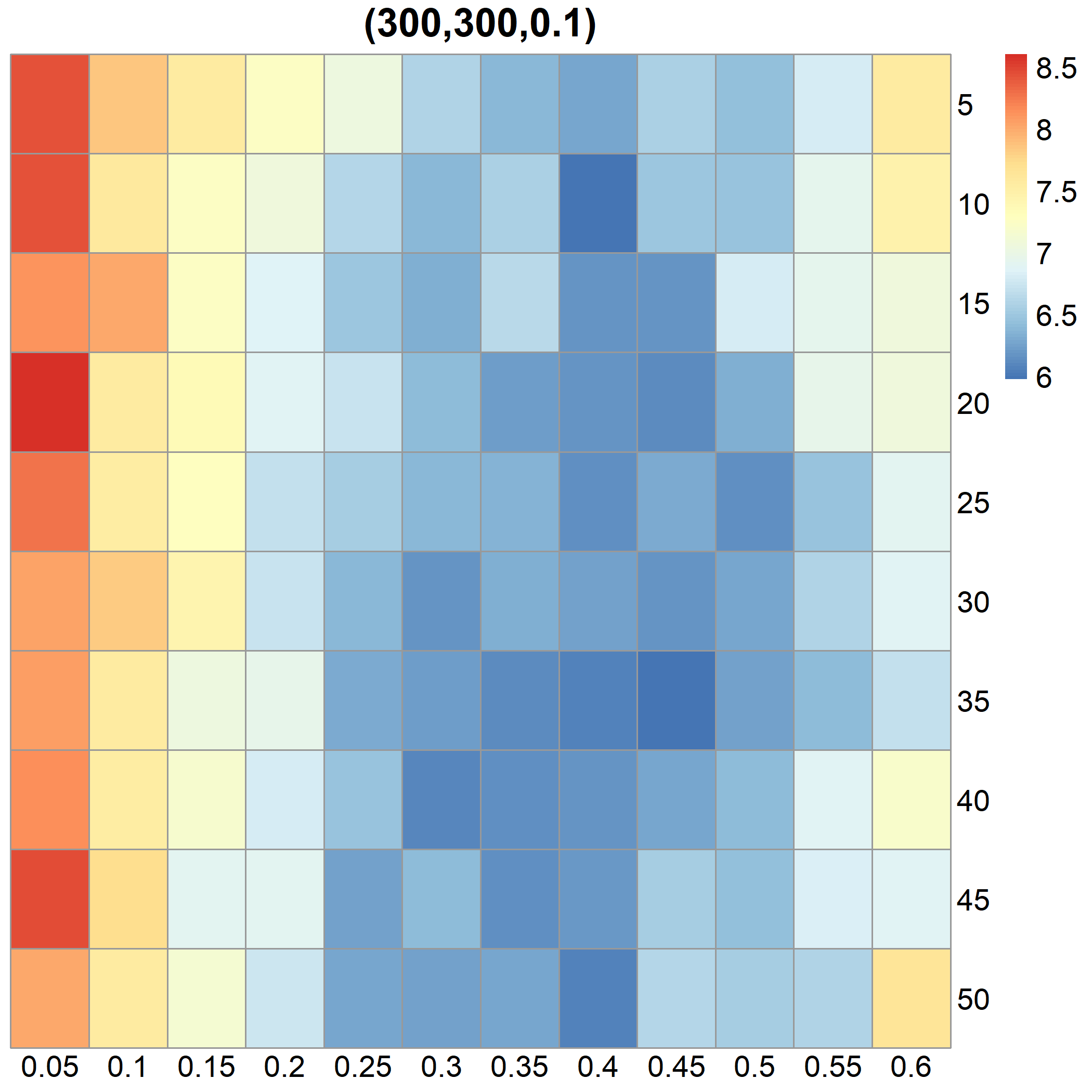}
		\includegraphics[width=0.32\linewidth]{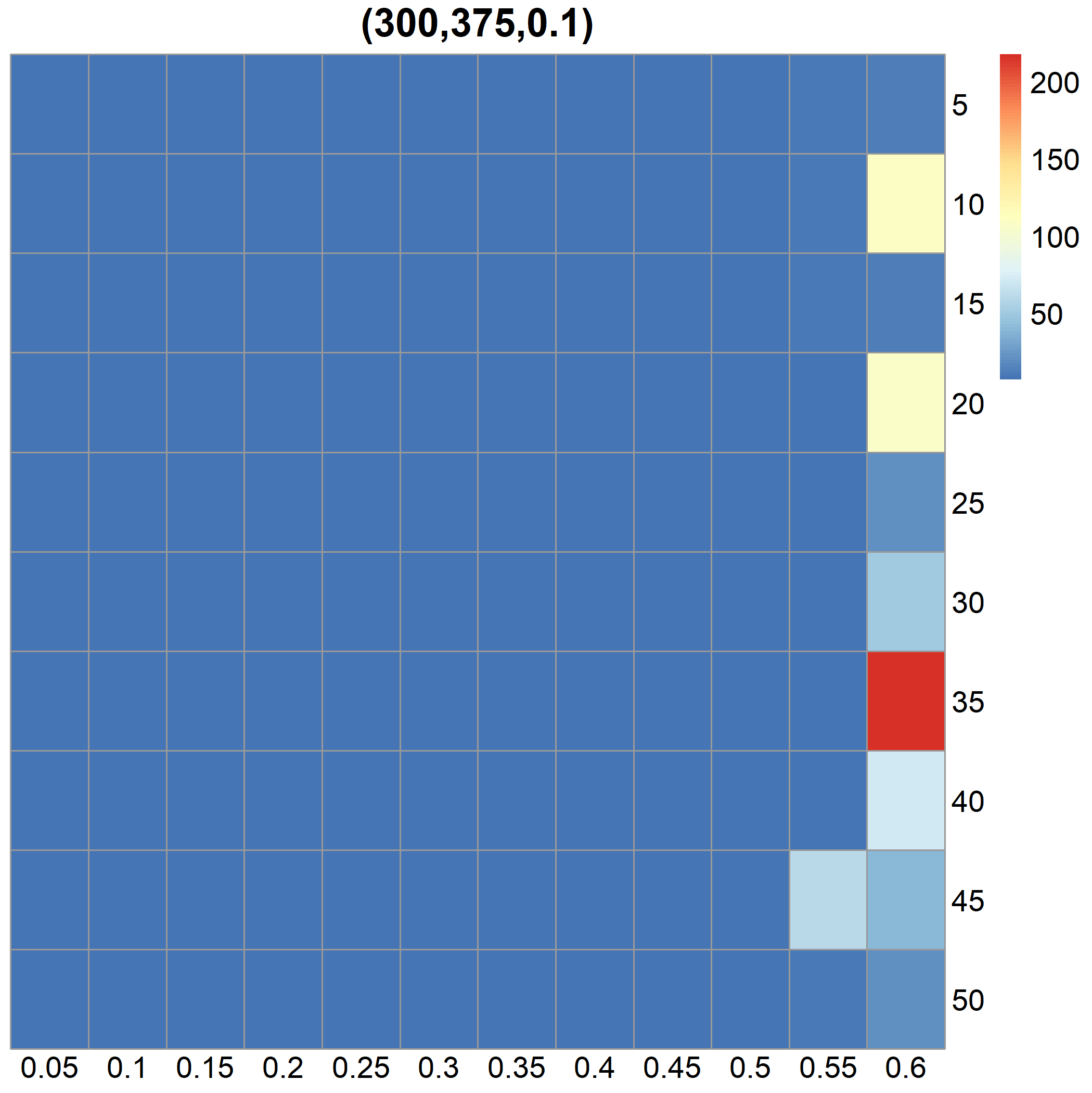}
		\includegraphics[width=0.32\linewidth]{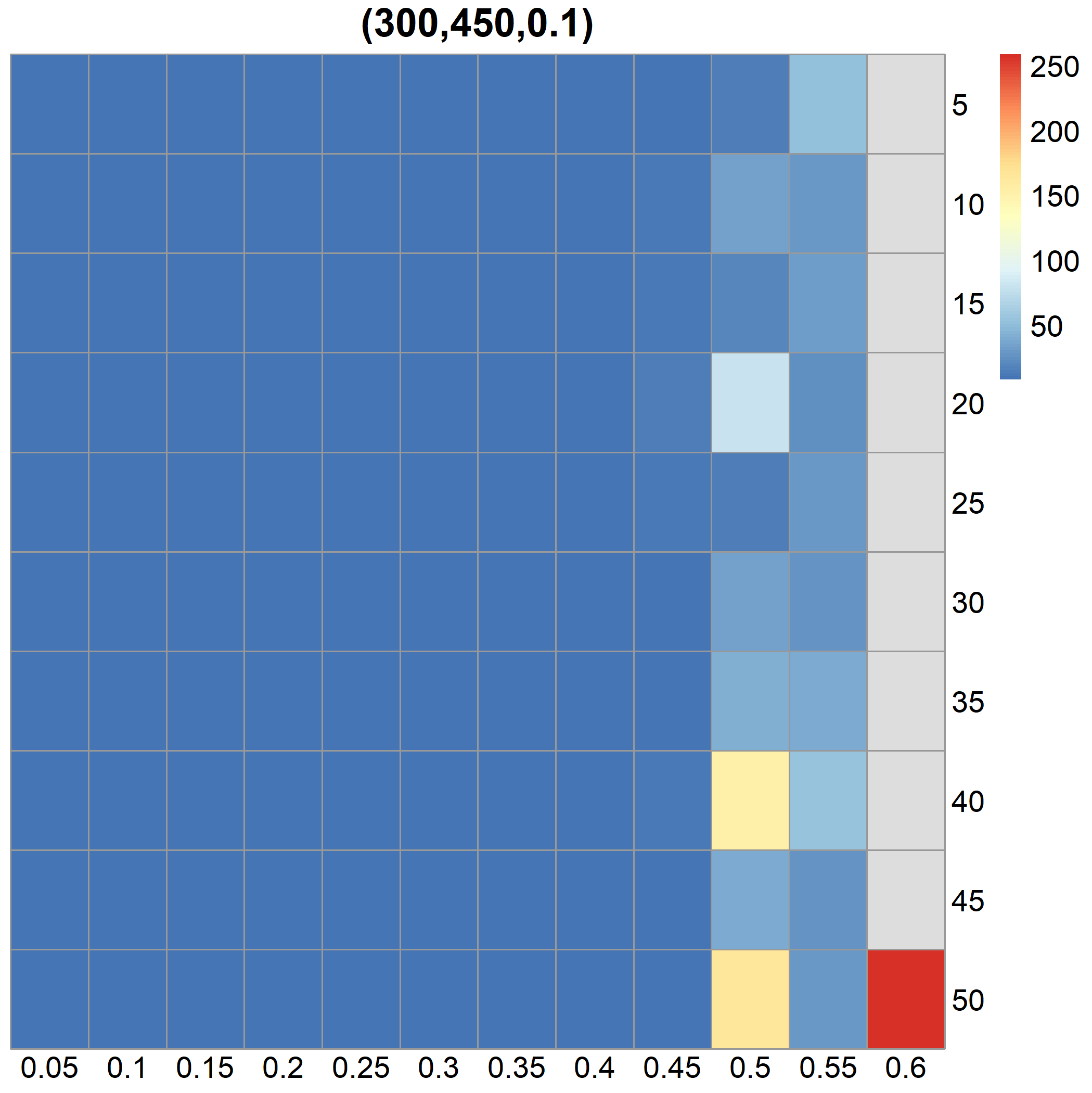}
		
		\vspace{3pt}
		\small (b) $N = 300$
	\end{minipage}
	
	\caption{Cross-validation results for Section \ref{app:manyrelevantvariables} under polynomial decay with $\rho = 0.1$. Values in parentheses denote $(N, K, \rho)$. The horizontal axis represents the selection probability $p$, and the vertical axis indicates the number of candidate models $M$. Darker regions correspond to $(p, M)$ combinations yielding lower cross-validation errors.}
	\label{fig:cvMR0.1}
\end{figure*}

\begin{figure*}[htbp]
	\centering
	
	\begin{minipage}{0.95\textwidth}
		\centering
		\includegraphics[width=0.32\linewidth]{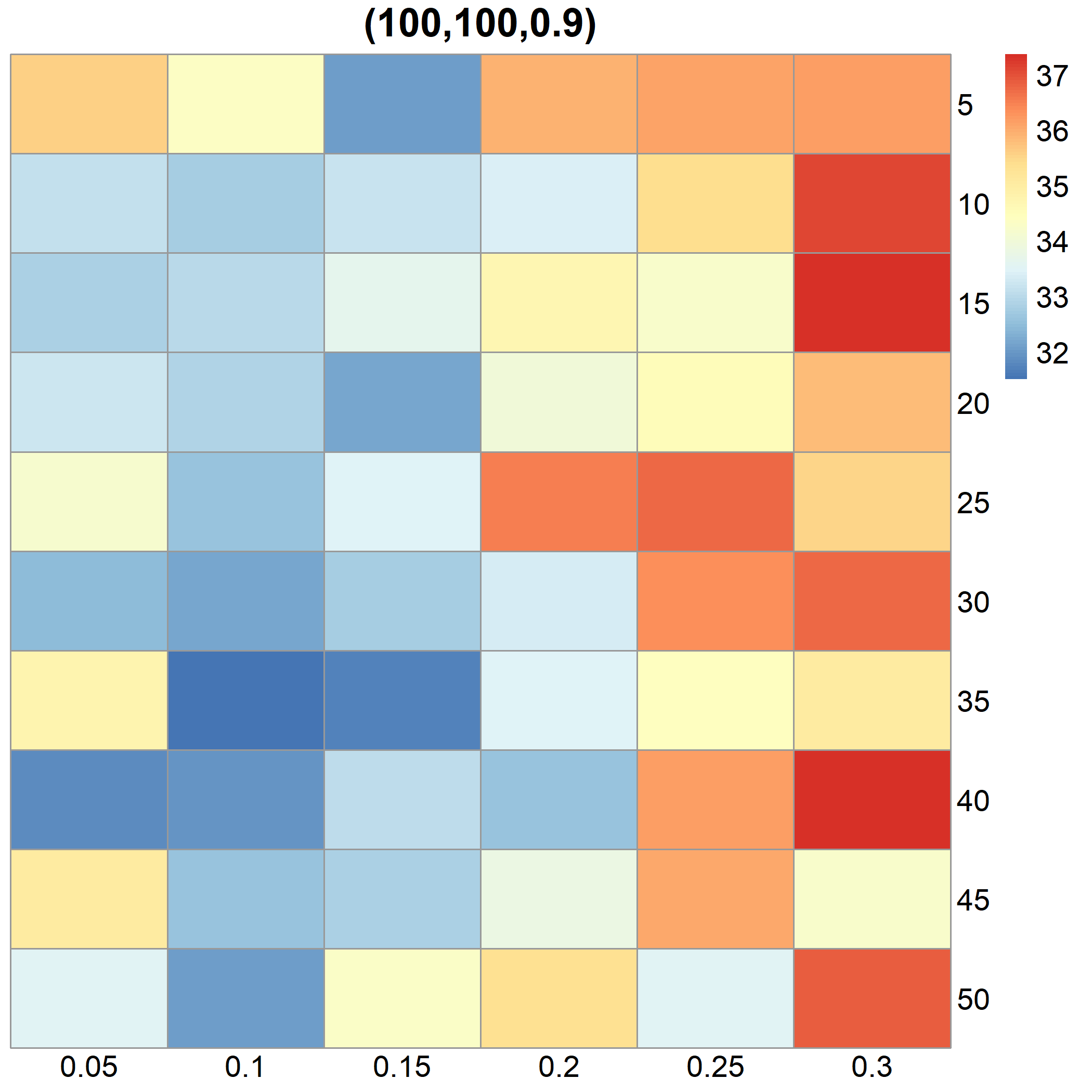}
		\includegraphics[width=0.32\linewidth]{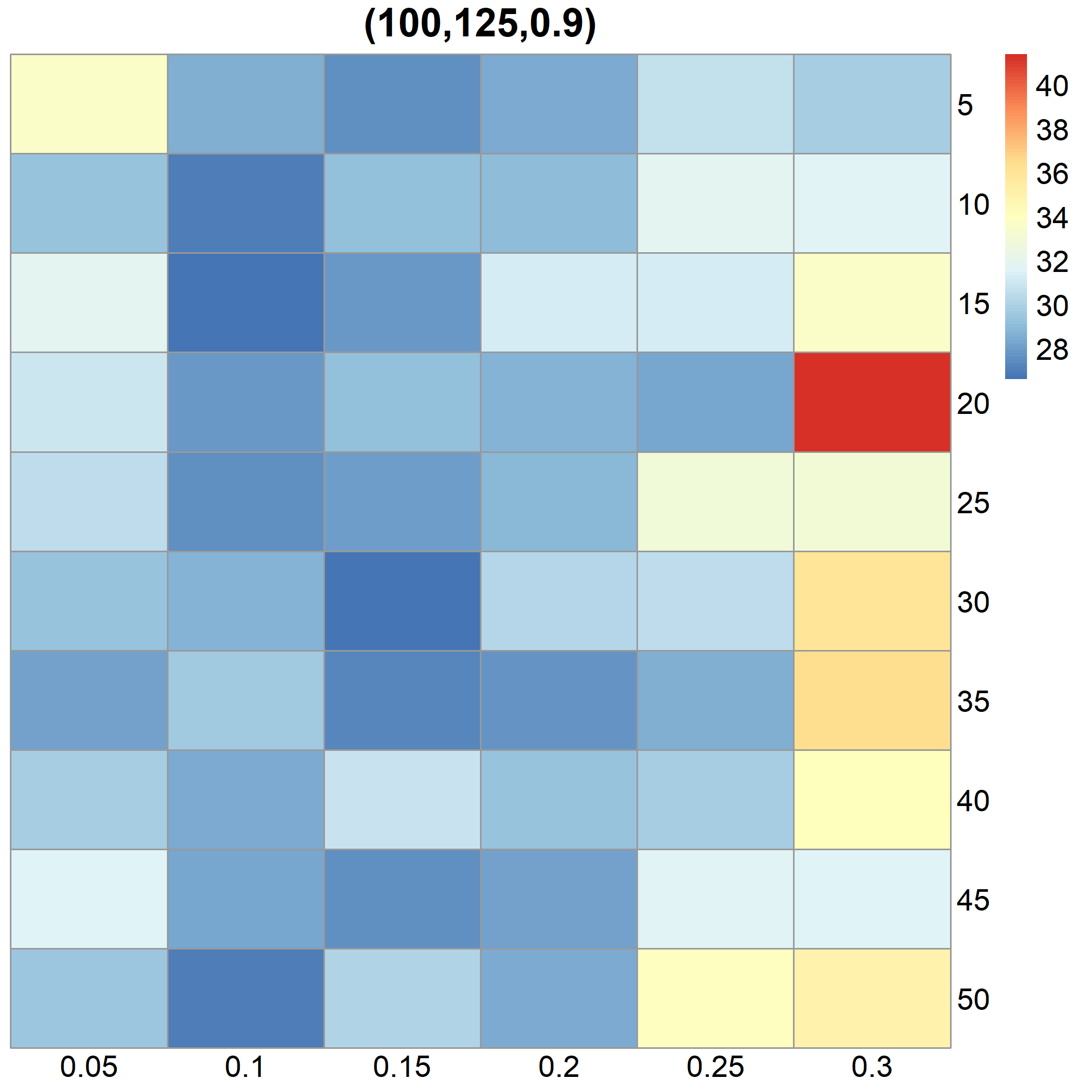}
		\includegraphics[width=0.32\linewidth]{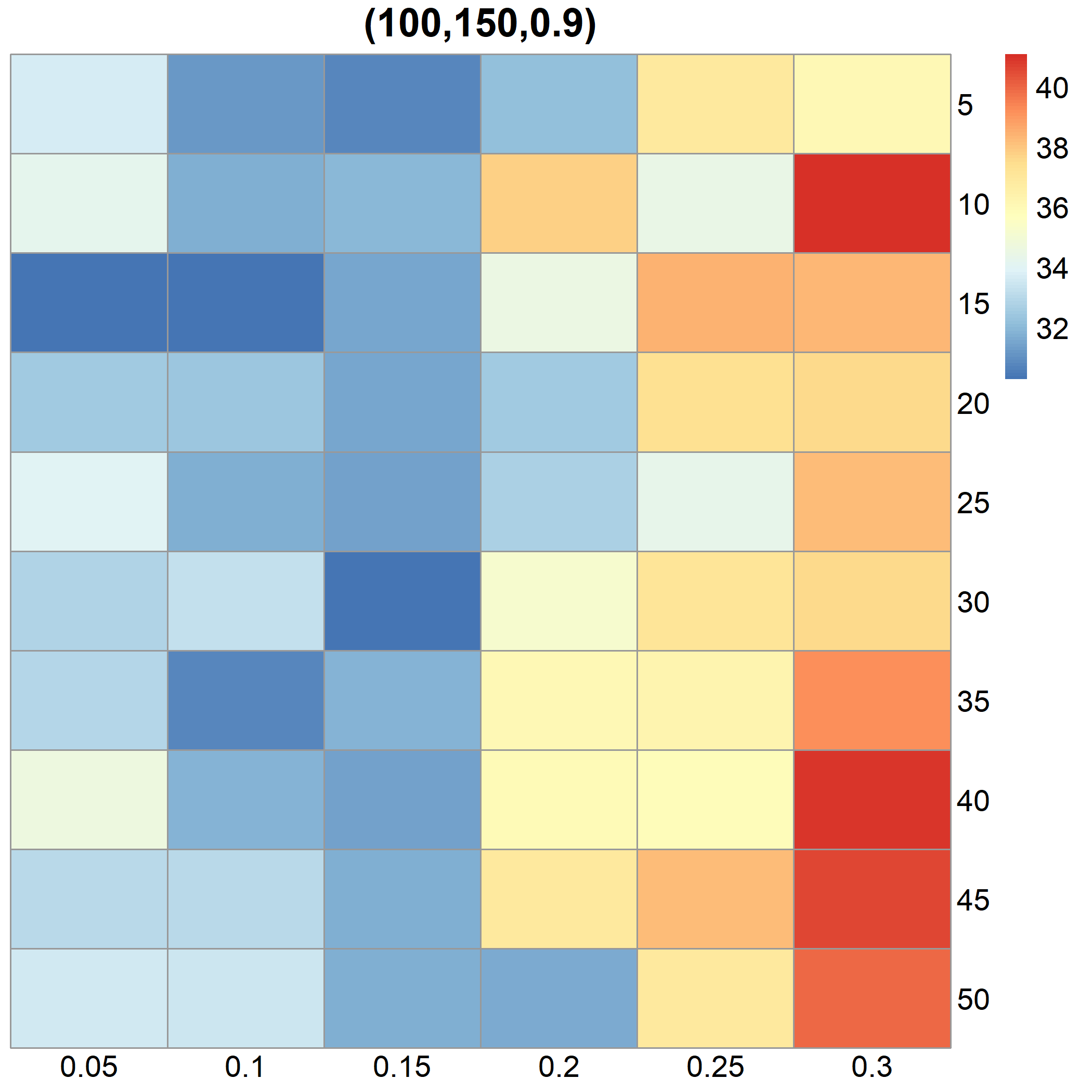}
		
		\vspace{3pt}
		\small (a) $N = 100$
	\end{minipage}
	
	\vspace{6pt}
	
	\begin{minipage}{0.95\textwidth}
		\centering
		\includegraphics[width=0.32\linewidth]{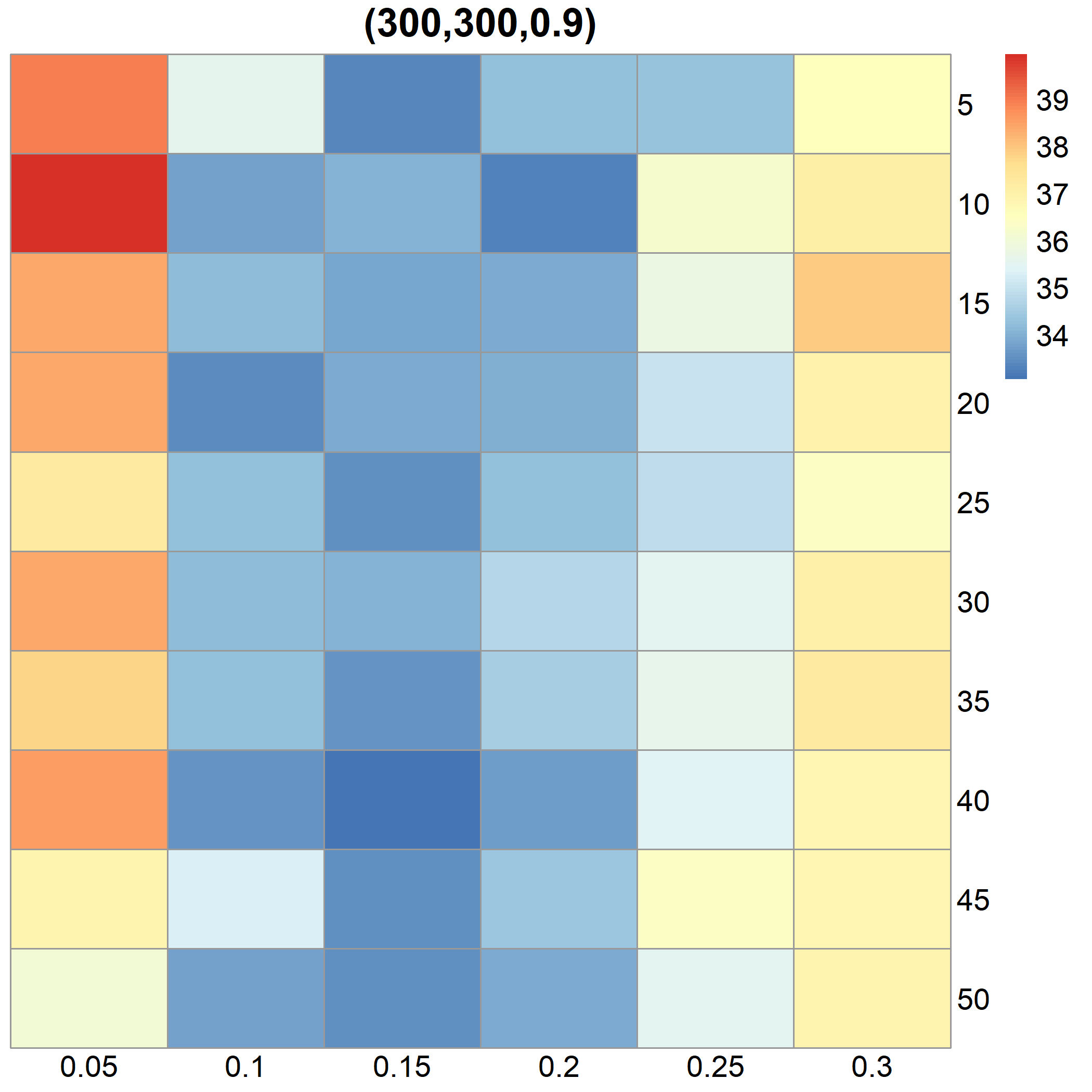}
		\includegraphics[width=0.32\linewidth]{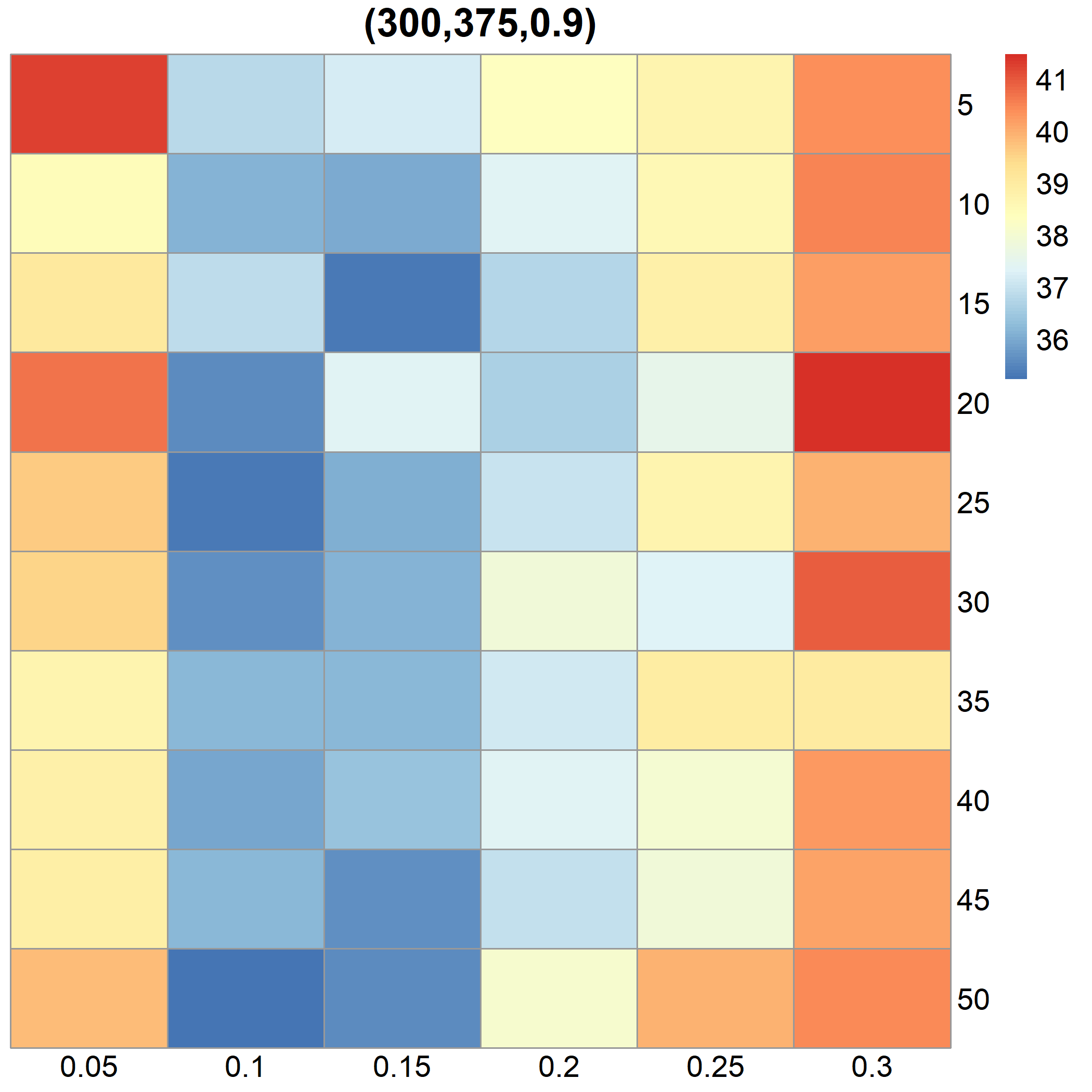}
		\includegraphics[width=0.32\linewidth]{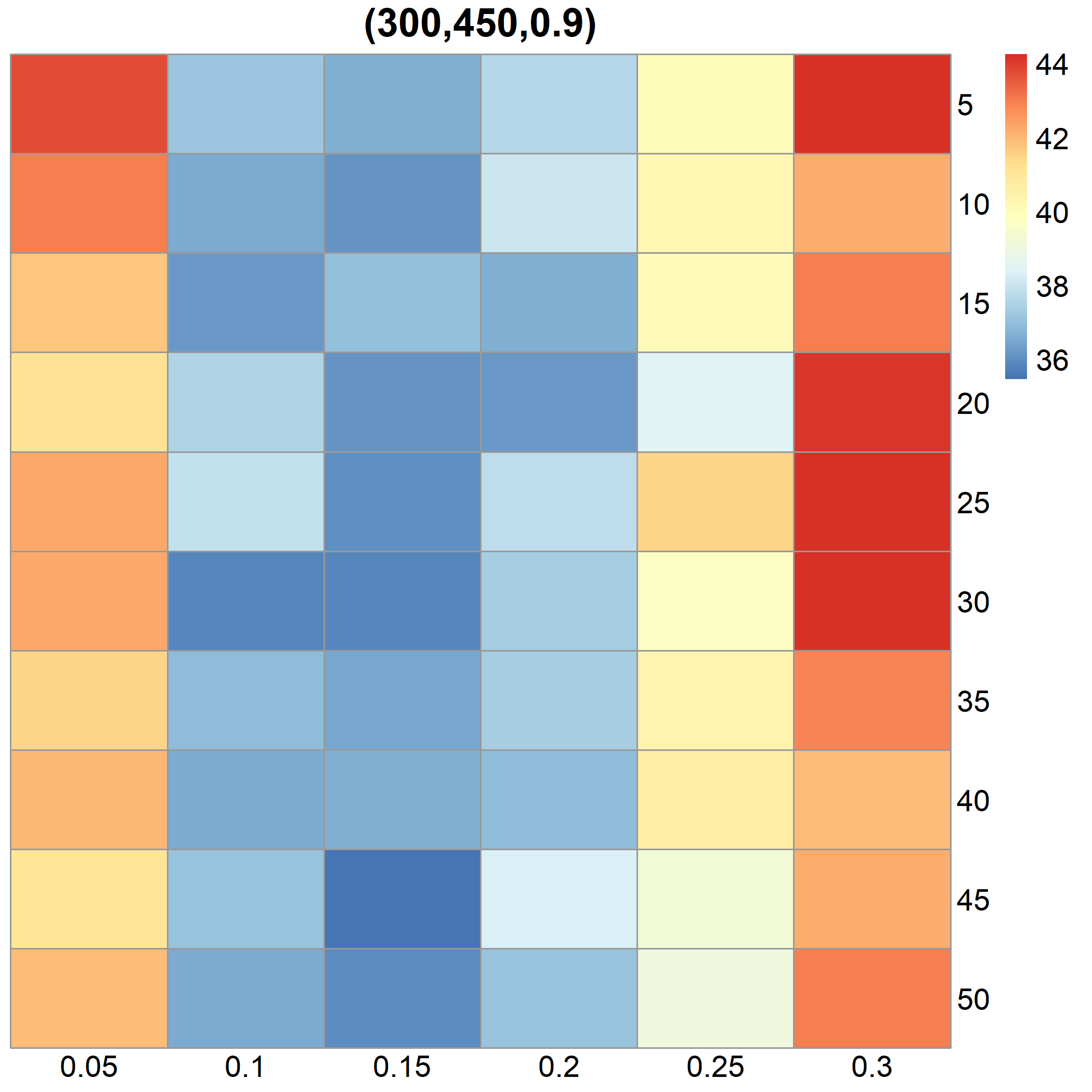}
		
		\vspace{3pt}
		\small (b) $N = 300$
	\end{minipage}
	
	\caption{Cross-validation results for Section \ref{app:manyrelevantvariables} under polynomial decay with $\rho = 0.9$. Values in parentheses denote $(N, K, \rho)$. The horizontal axis represents the selection probability $p$, and the vertical axis indicates the number of candidate models $M$. Darker regions correspond to $(p, M)$ combinations yielding lower cross-validation errors.}
	\label{fig:cvMR0.9}
\end{figure*}

\begin{figure*}[htbp]
	\centering
	
	\begin{minipage}{0.95\textwidth}
		\centering
		\includegraphics[width=0.32\linewidth]{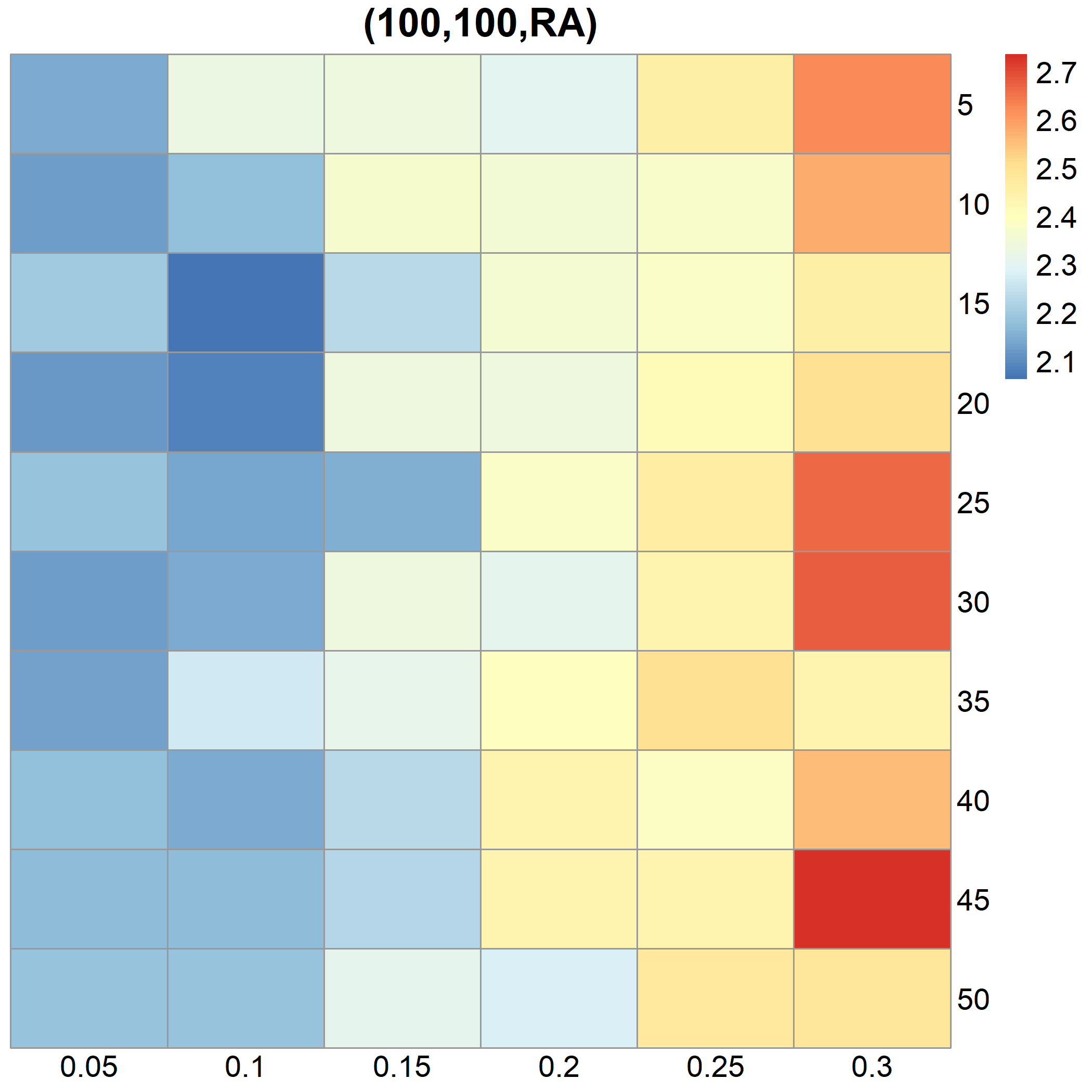}
		\includegraphics[width=0.32\linewidth]{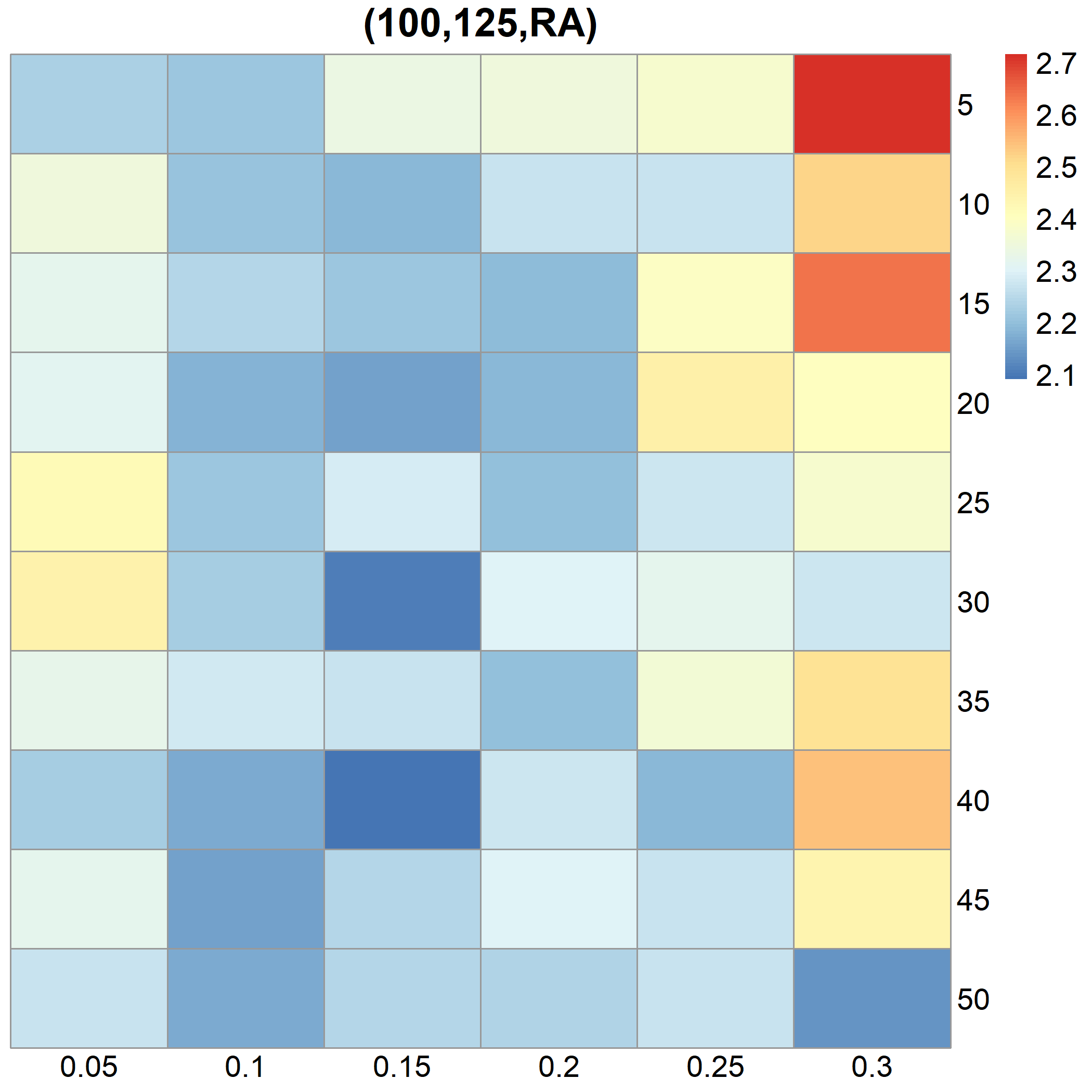}
		\includegraphics[width=0.32\linewidth]{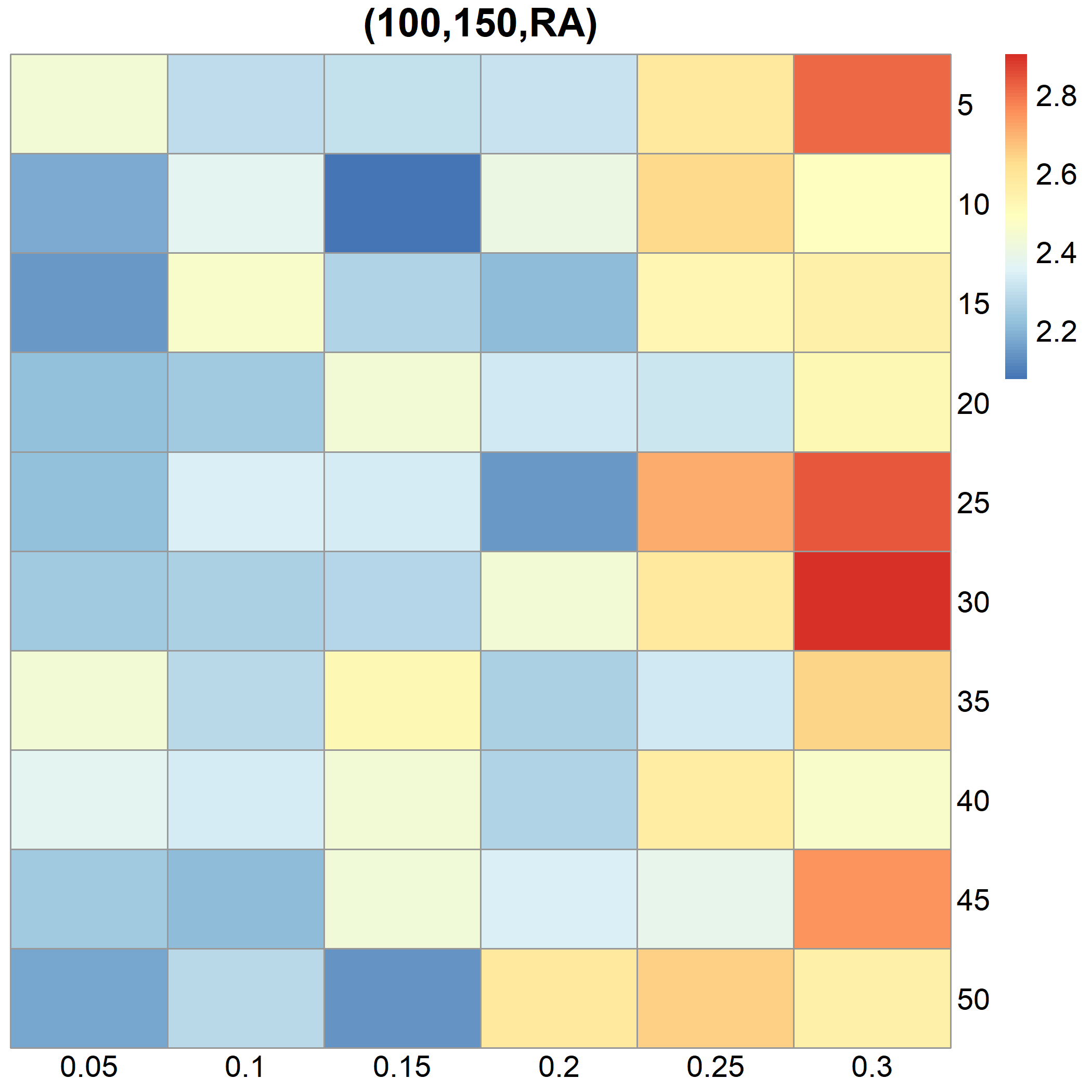}
		
		\vspace{3pt}
		\small (a) $N = 100$
	\end{minipage}
	
	\vspace{6pt}
	
	\begin{minipage}{0.95\textwidth}
		\centering
		\includegraphics[width=0.32\linewidth]{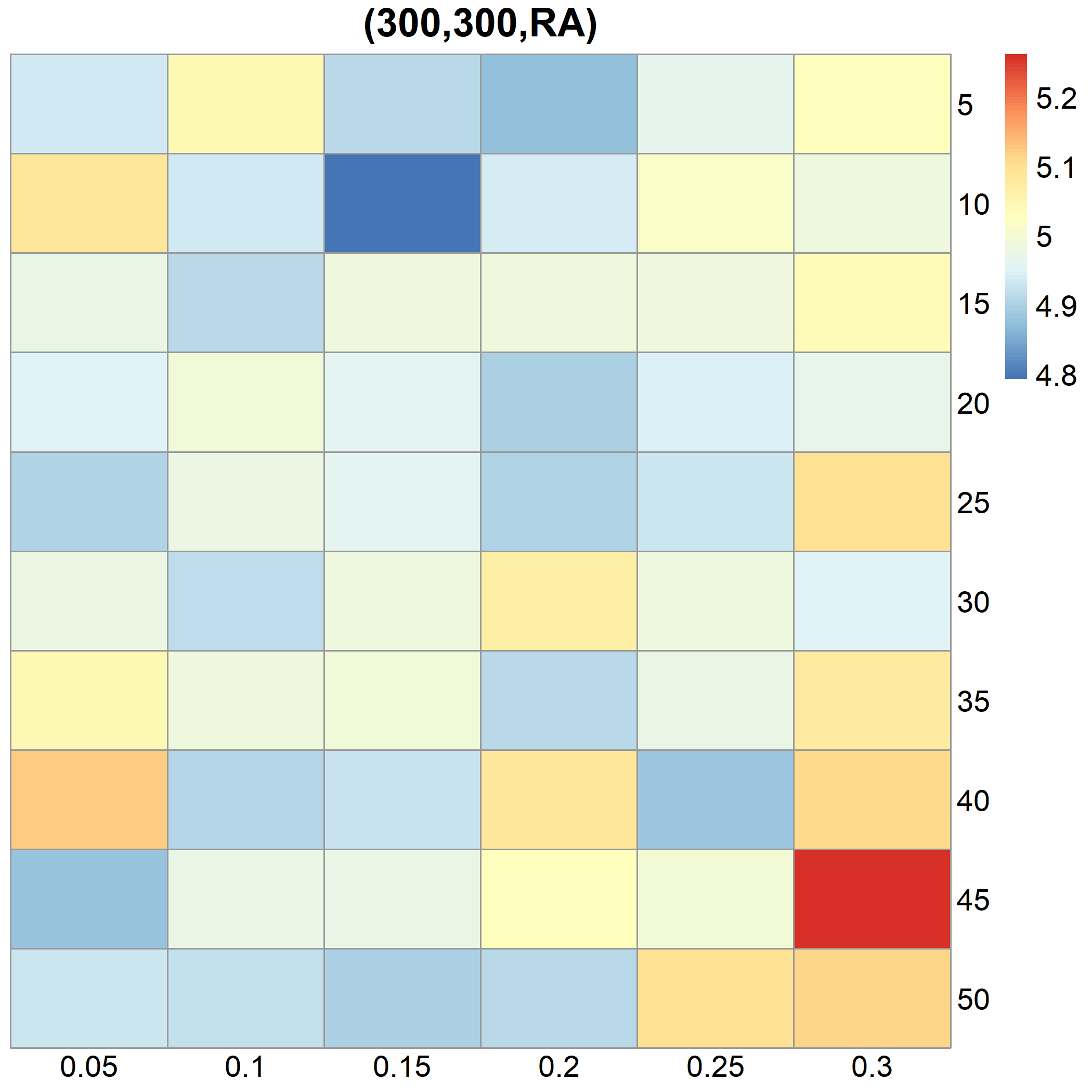}
		\includegraphics[width=0.32\linewidth]{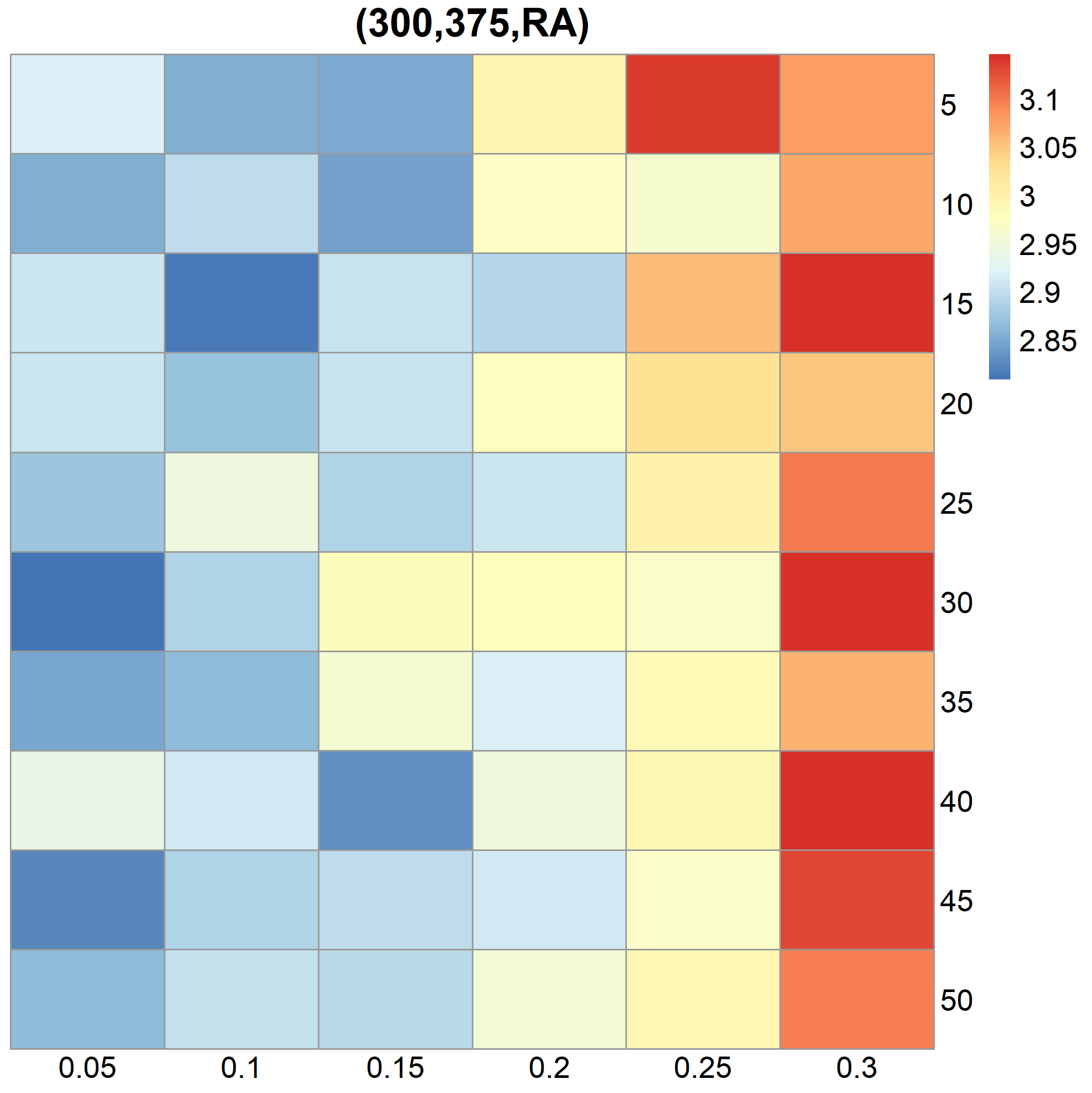}
		\includegraphics[width=0.32\linewidth]{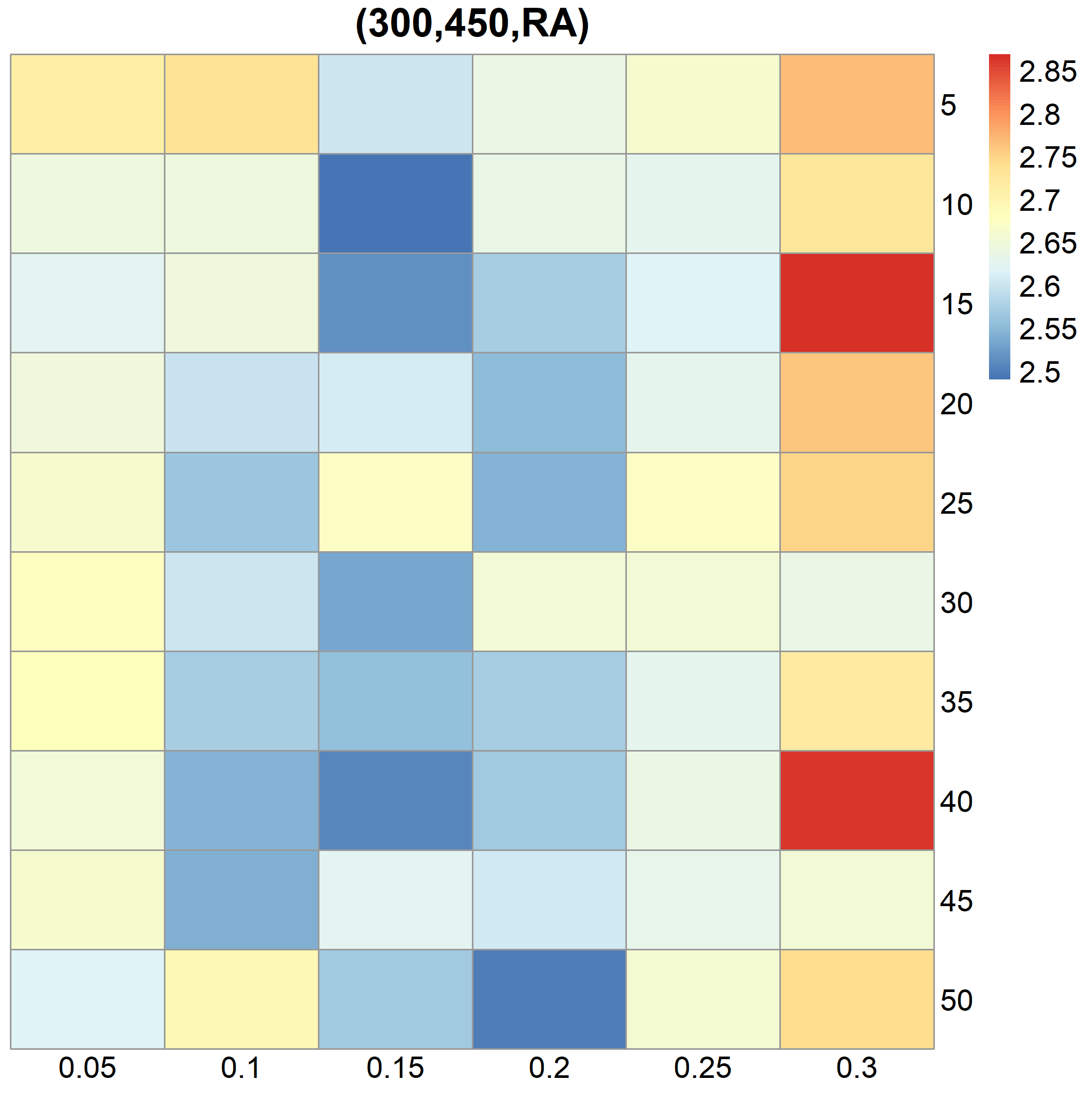}
		
		\vspace{3pt}
		\small (b) $N = 300$
	\end{minipage}
	
	\caption{Cross-validation results for Section \ref{app:manyrelevantvariables} under polynomial decay with a random covariance structure. Values in parentheses denote $(N, K, \rho)$. The horizontal axis represents the selection probability $p$, and the vertical axis indicates the number of candidate models $M$. Darker regions correspond to $(p, M)$ combinations yielding lower cross-validation errors.}
	\label{fig:cvMRRA}
\end{figure*}

\begin{figure*}[htbp]
	\centering
	
	\begin{minipage}{0.95\textwidth}
		\centering
		\includegraphics[width=0.32\linewidth]{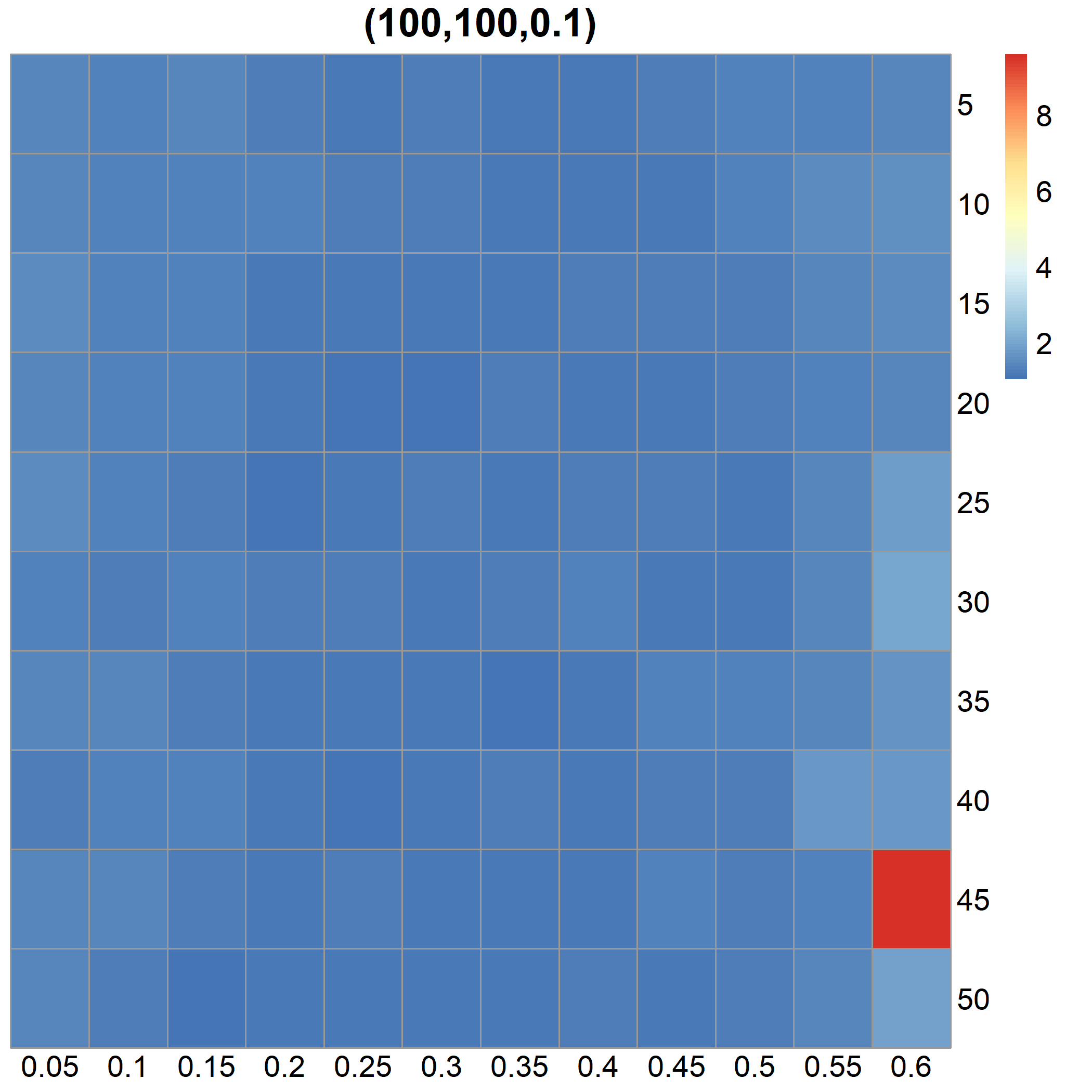}
		\includegraphics[width=0.32\linewidth]{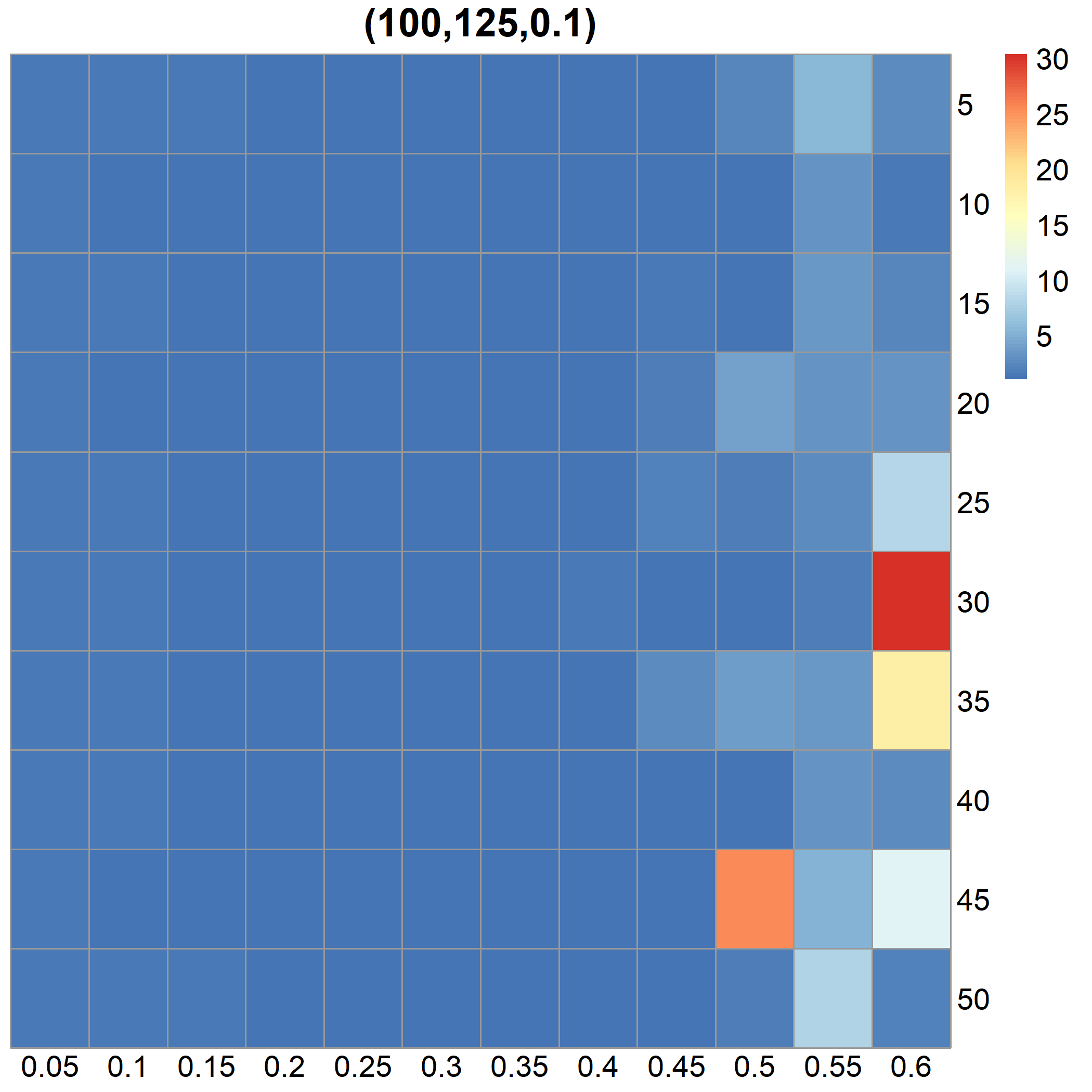}
		\includegraphics[width=0.32\linewidth]{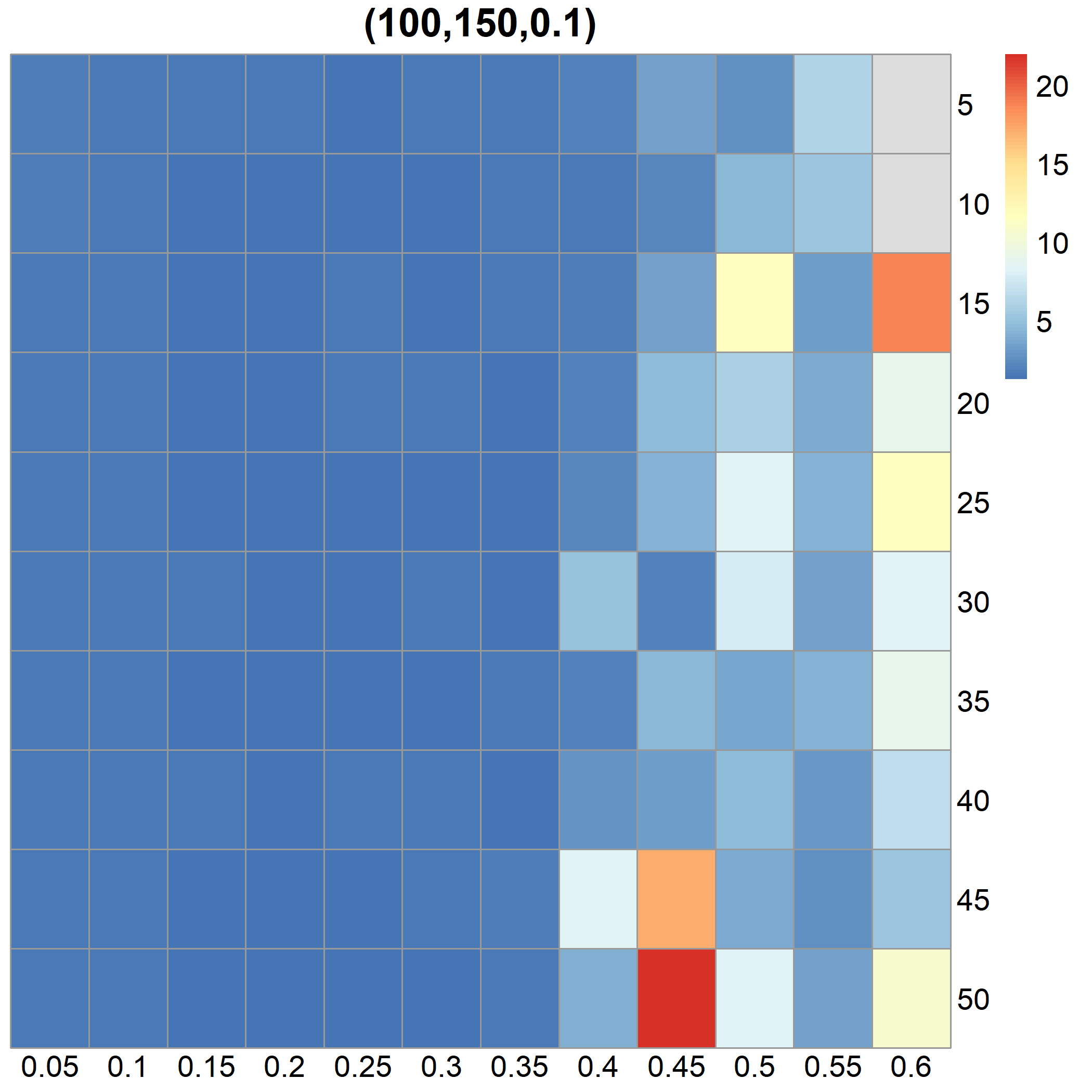}
		
		\vspace{3pt}
		\small (a) $N = 100$
	\end{minipage}
	
	\vspace{6pt}
	
	\begin{minipage}{0.95\textwidth}
		\centering
		\includegraphics[width=0.32\linewidth]{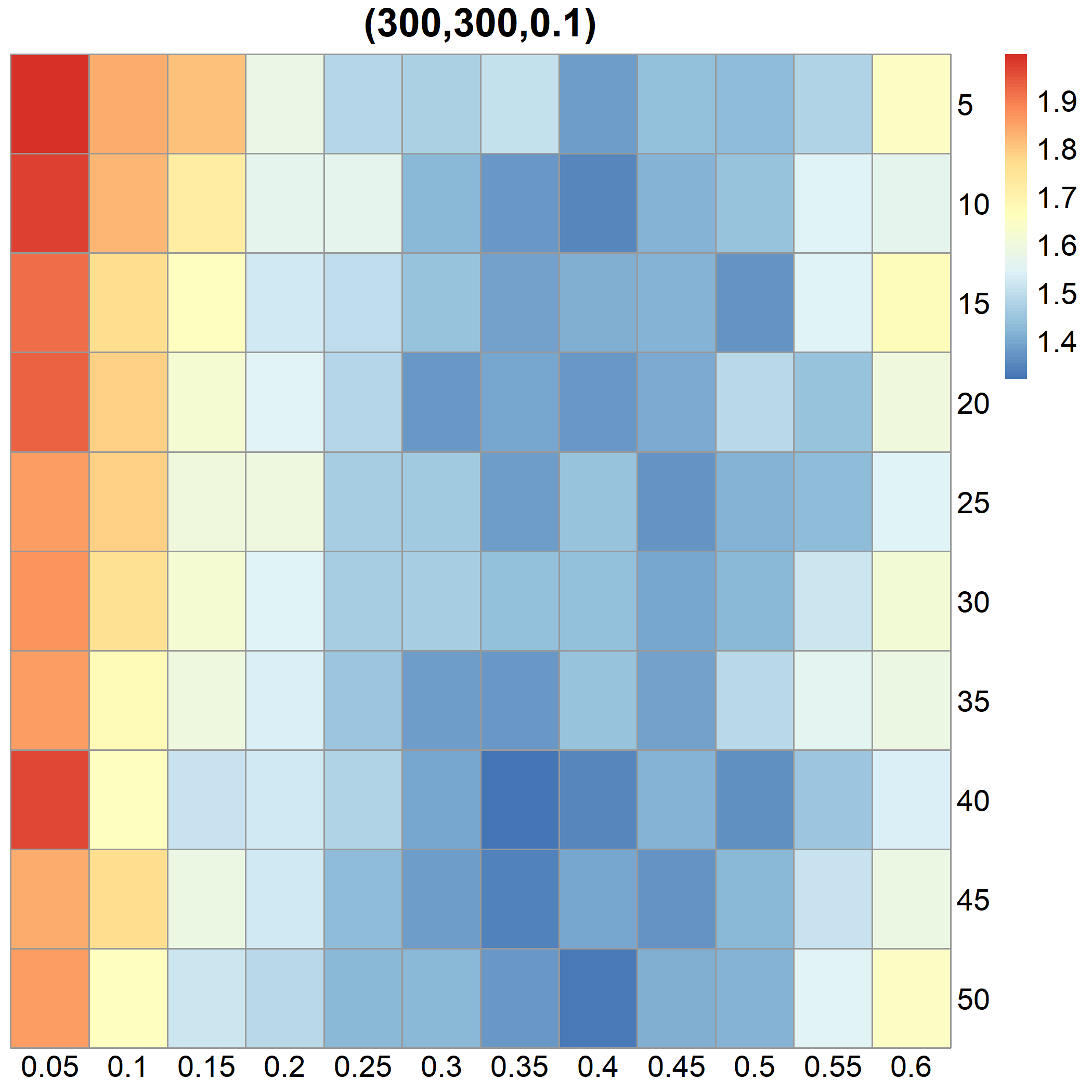}
		\includegraphics[width=0.32\linewidth]{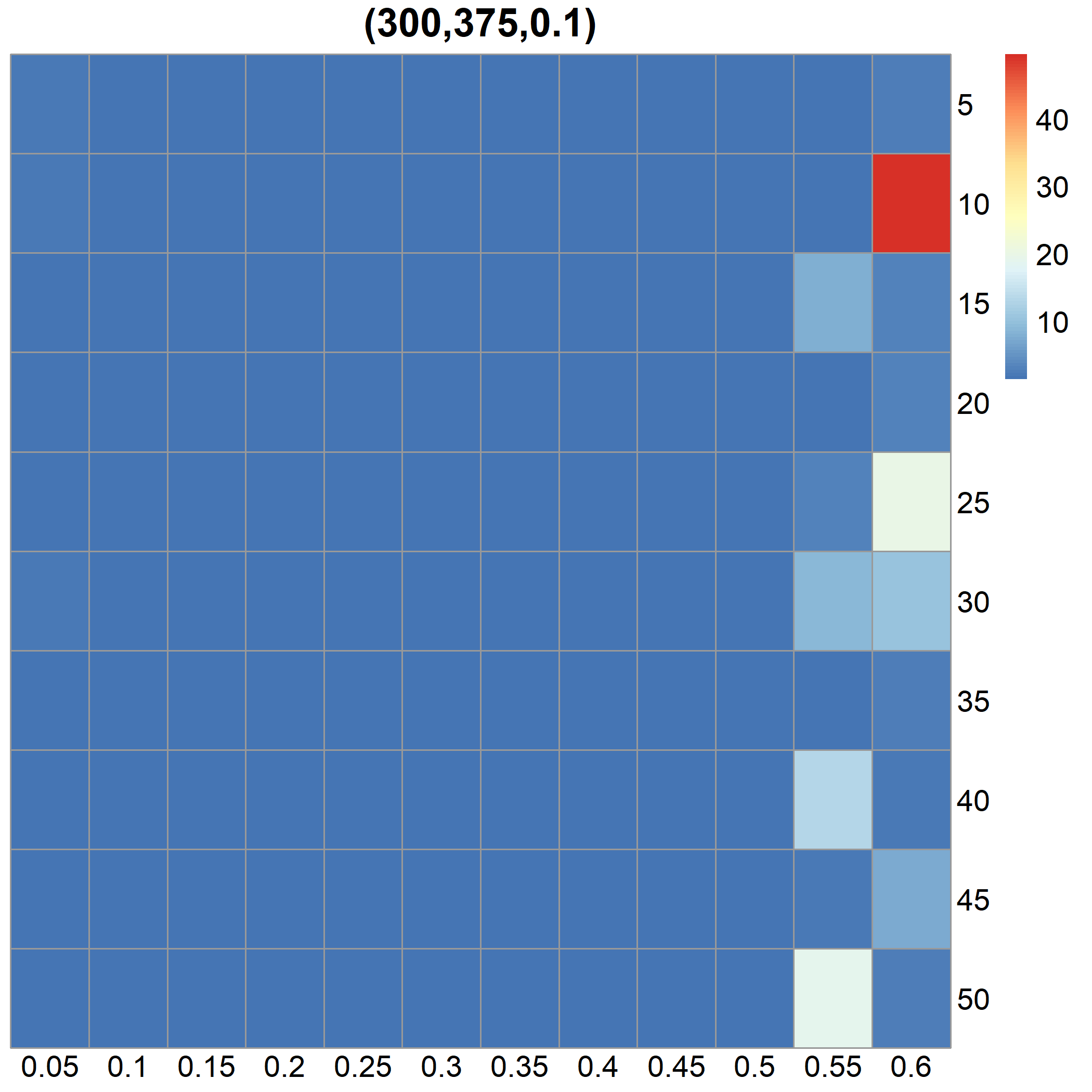}
		\includegraphics[width=0.32\linewidth]{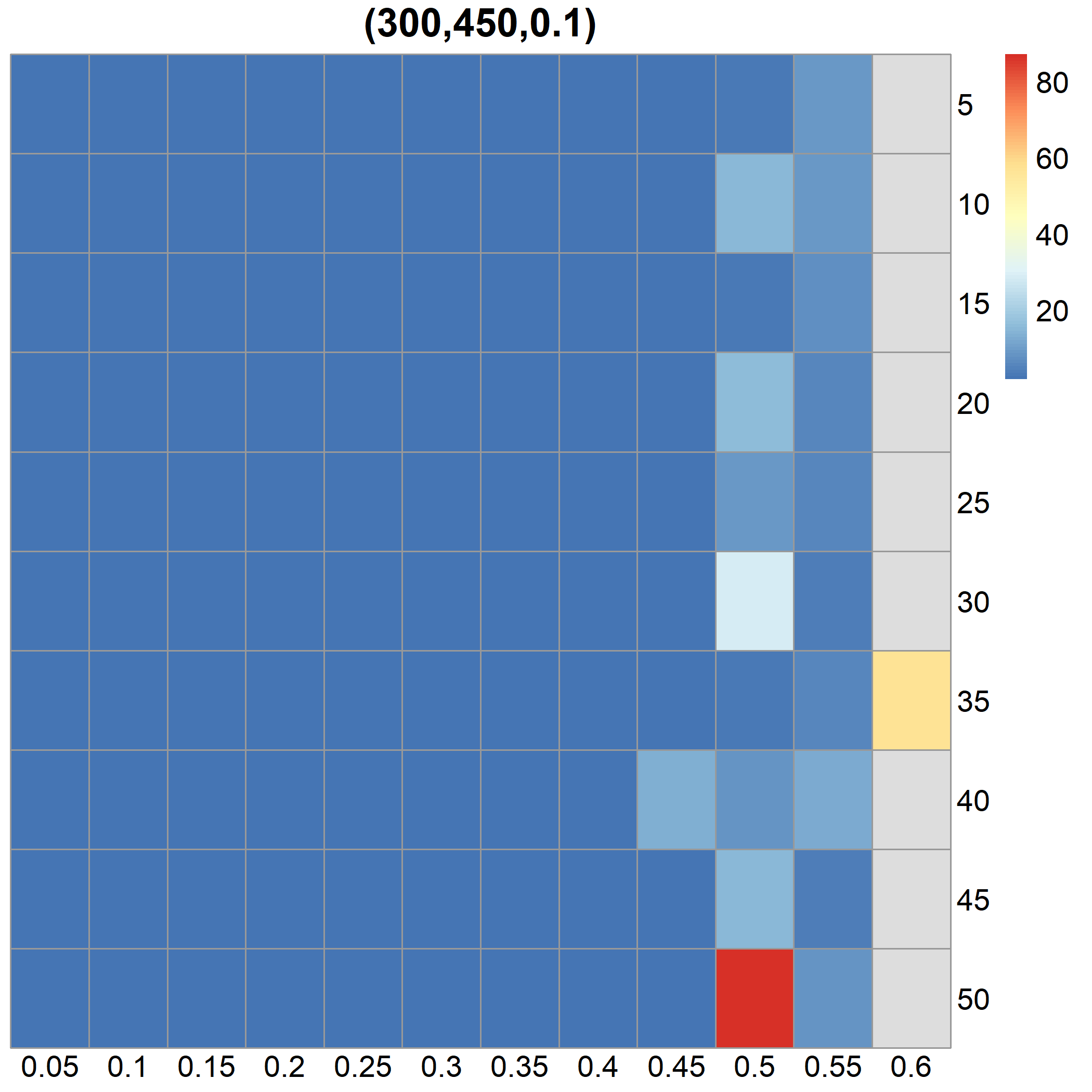}
		
		\vspace{3pt}
		\small (b) $N = 300$
	\end{minipage}
	
	\caption{Cross-validation results for Section \ref{app:manyrelevantvariables} under exponential decay with $\rho = 0.1$. Values in parentheses denote $(N, K, \rho)$. The horizontal axis represents the selection probability $p$, and the vertical axis indicates the number of candidate models $M$. Darker regions correspond to $(p, M)$ combinations yielding lower cross-validation errors.}
	\label{fig:cvMRexp0.1}
\end{figure*}

\begin{figure*}[htbp]
	\centering
	
	\begin{minipage}{0.95\textwidth}
		\centering
		\includegraphics[width=0.32\linewidth]{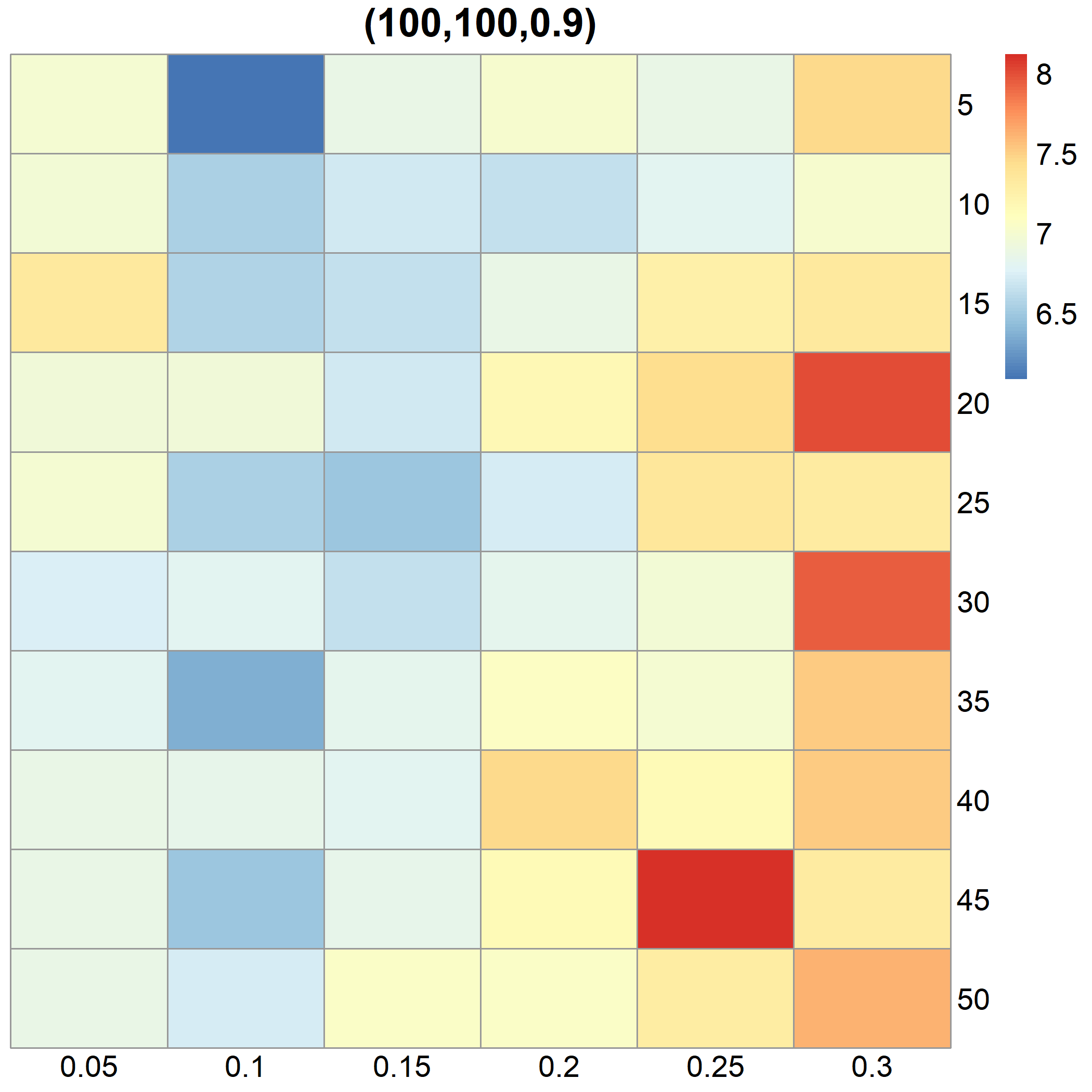}
		\includegraphics[width=0.32\linewidth]{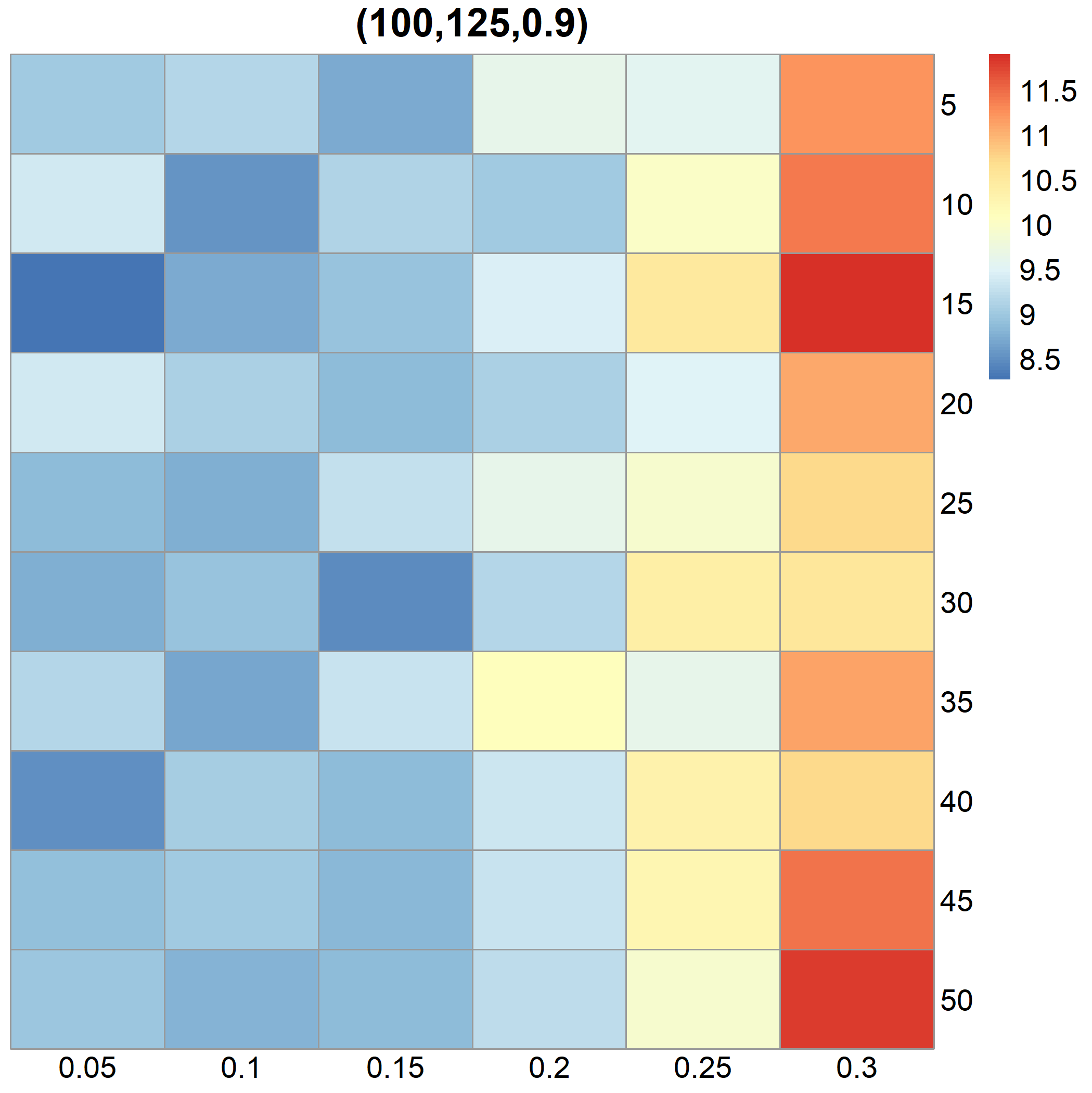}
		\includegraphics[width=0.32\linewidth]{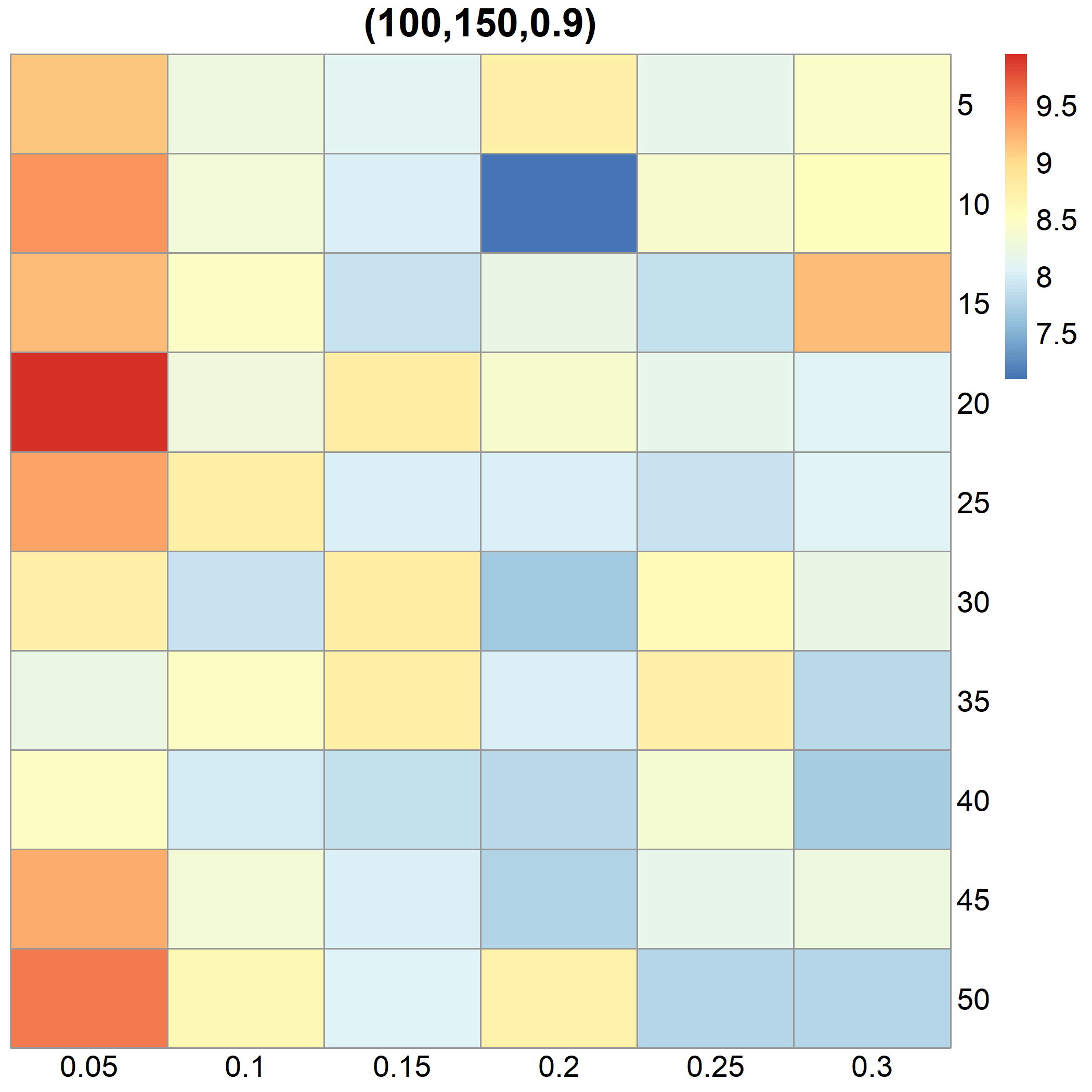}
		
		\vspace{3pt}
		\small (a) $N = 100$
	\end{minipage}
	
	\vspace{6pt}
	
	\begin{minipage}{0.95\textwidth}
		\centering
		\includegraphics[width=0.32\linewidth]{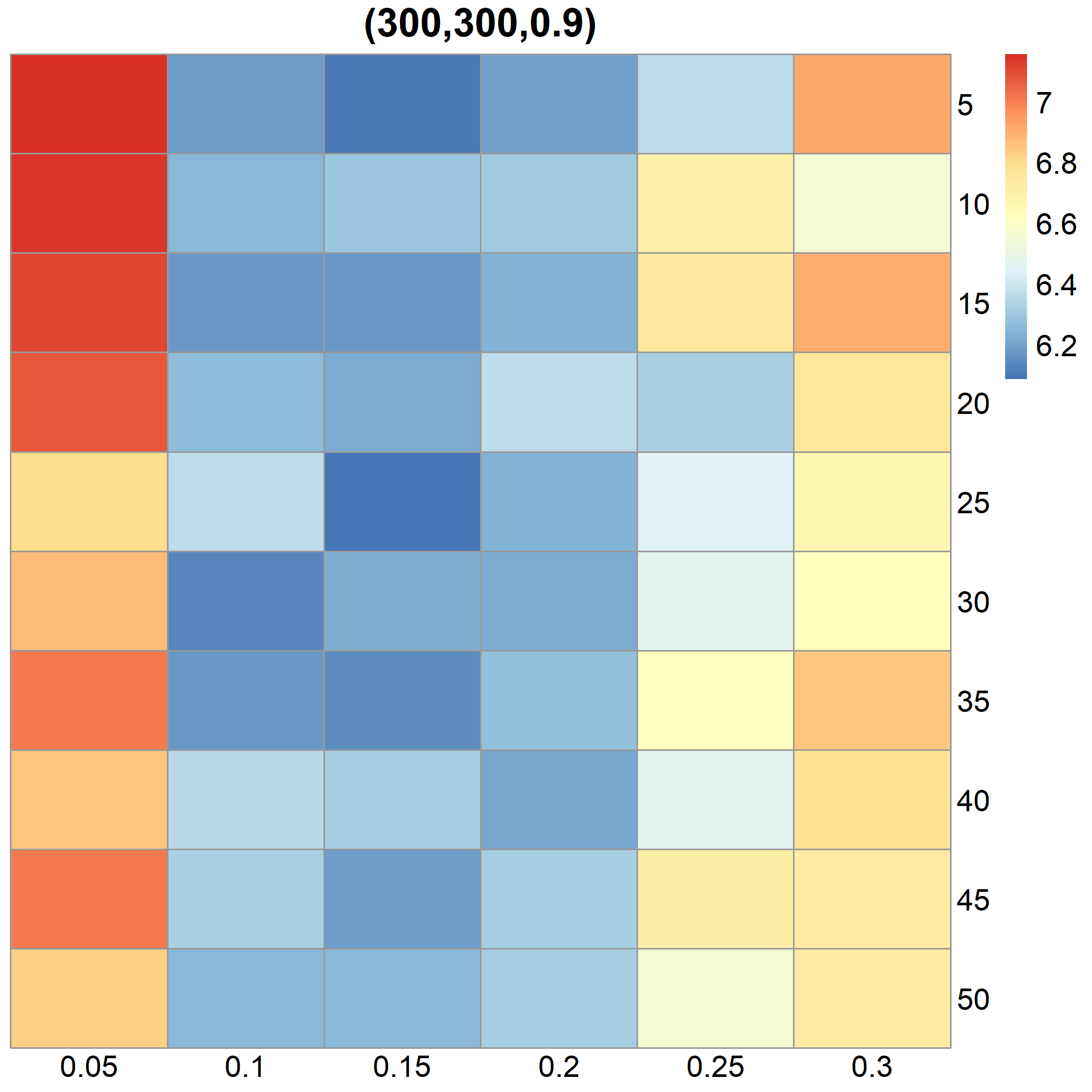}
		\includegraphics[width=0.32\linewidth]{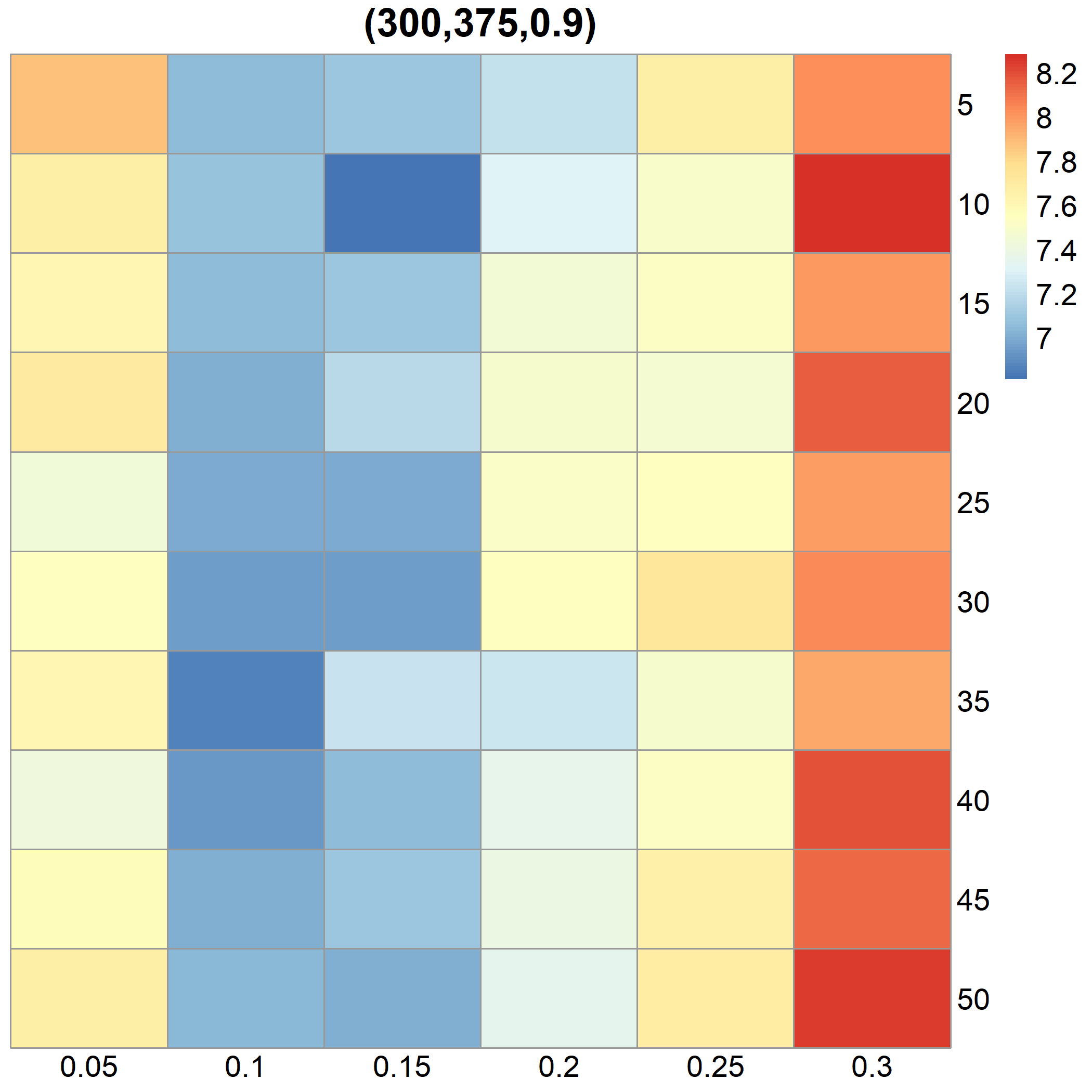}
		\includegraphics[width=0.32\linewidth]{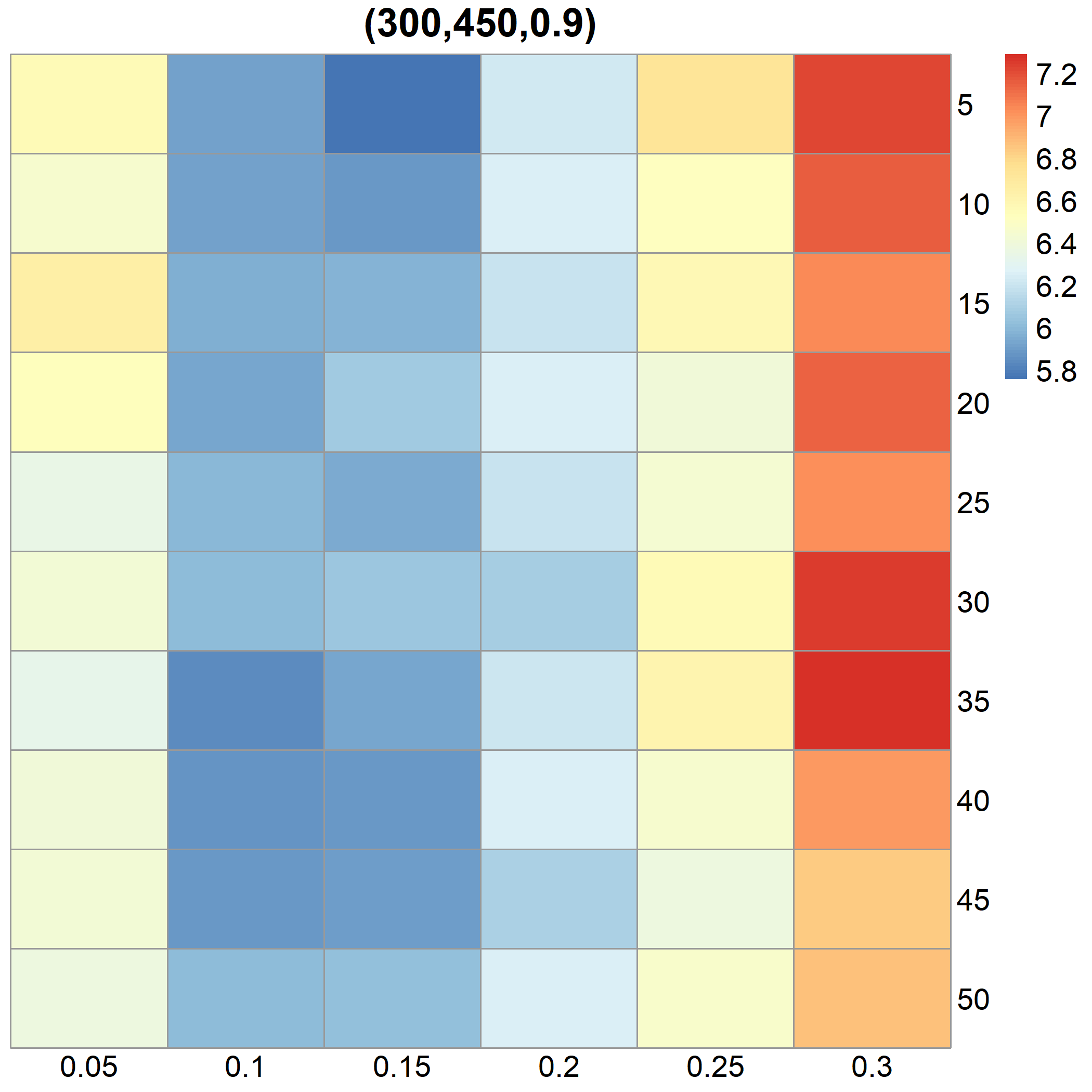}
		
		\vspace{3pt}
		\small (b) $N = 300$
	\end{minipage}
	
	\caption{Cross-validation results for Section \ref{app:manyrelevantvariables} under exponential decay with $\rho = 0.9$. Values in parentheses denote $(N, K, \rho)$. The horizontal axis represents the selection probability $p$, and the vertical axis indicates the number of candidate models $M$. Darker regions correspond to $(p, M)$ combinations yielding lower cross-validation errors.}
	\label{fig:cvMRexp0.9}
\end{figure*}

\begin{figure*}[htbp]
	\centering
	
	\begin{minipage}{0.95\textwidth}
		\centering
		\includegraphics[width=0.32\linewidth]{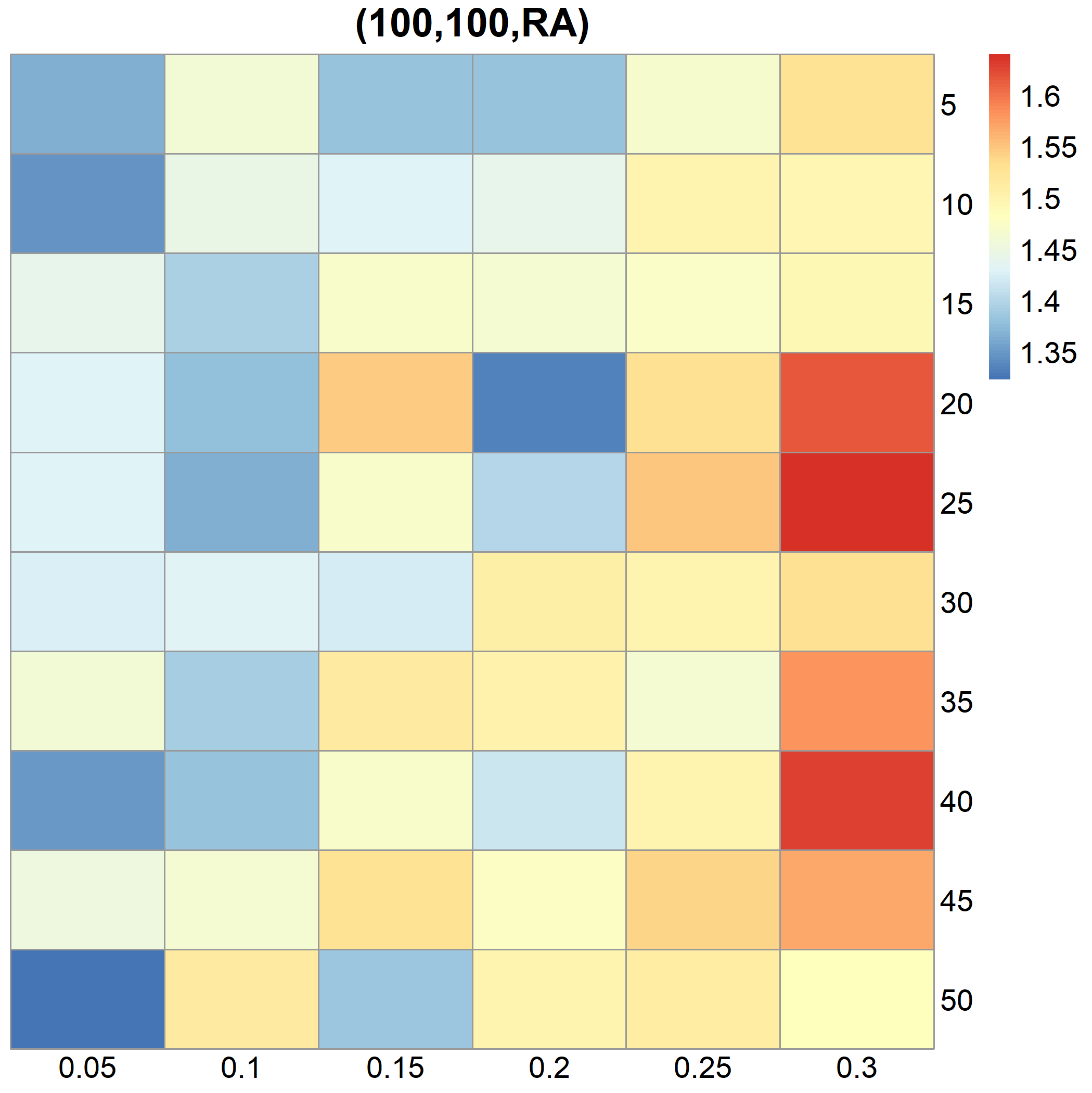}
		\includegraphics[width=0.32\linewidth]{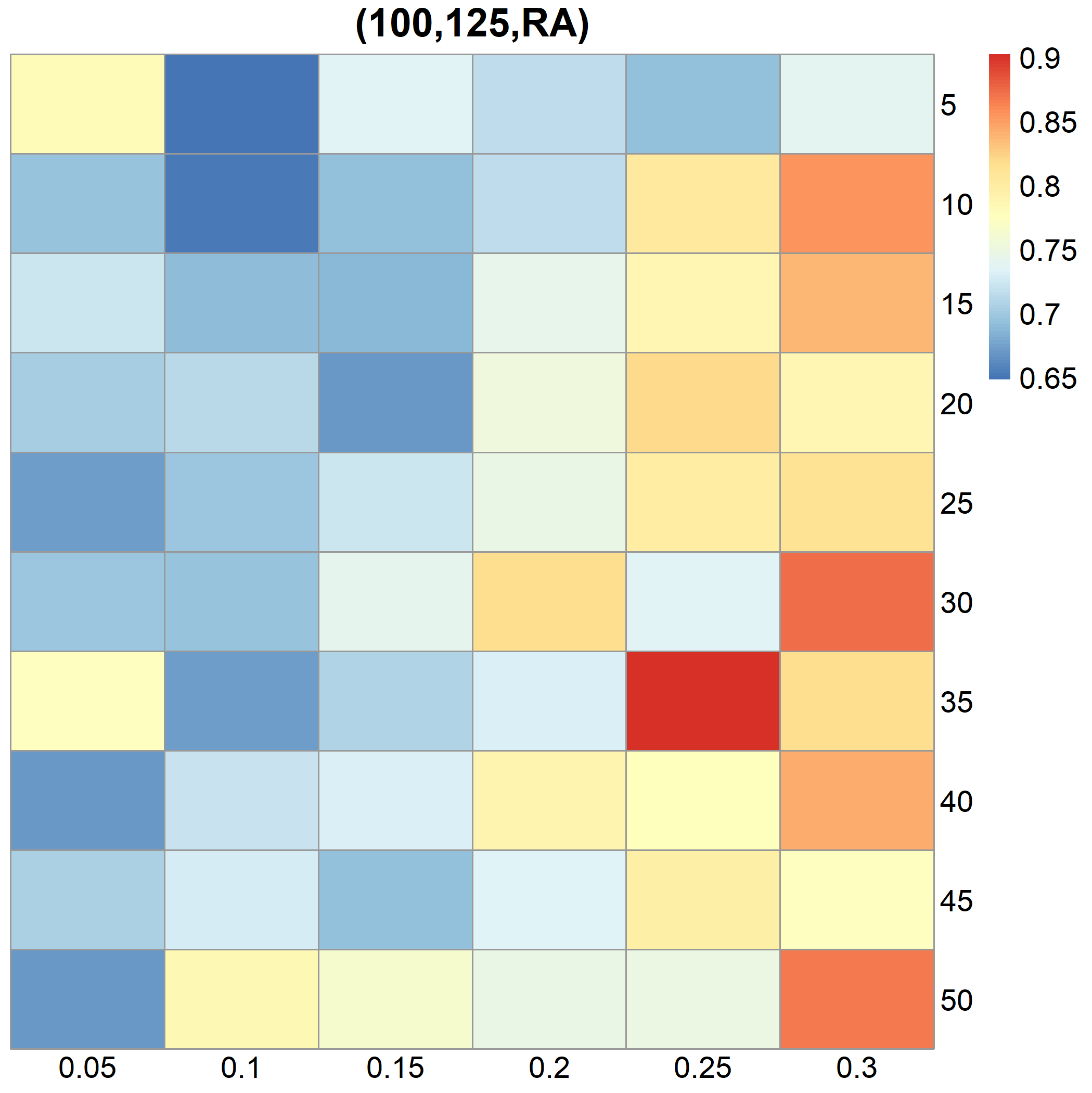}
		\includegraphics[width=0.32\linewidth]{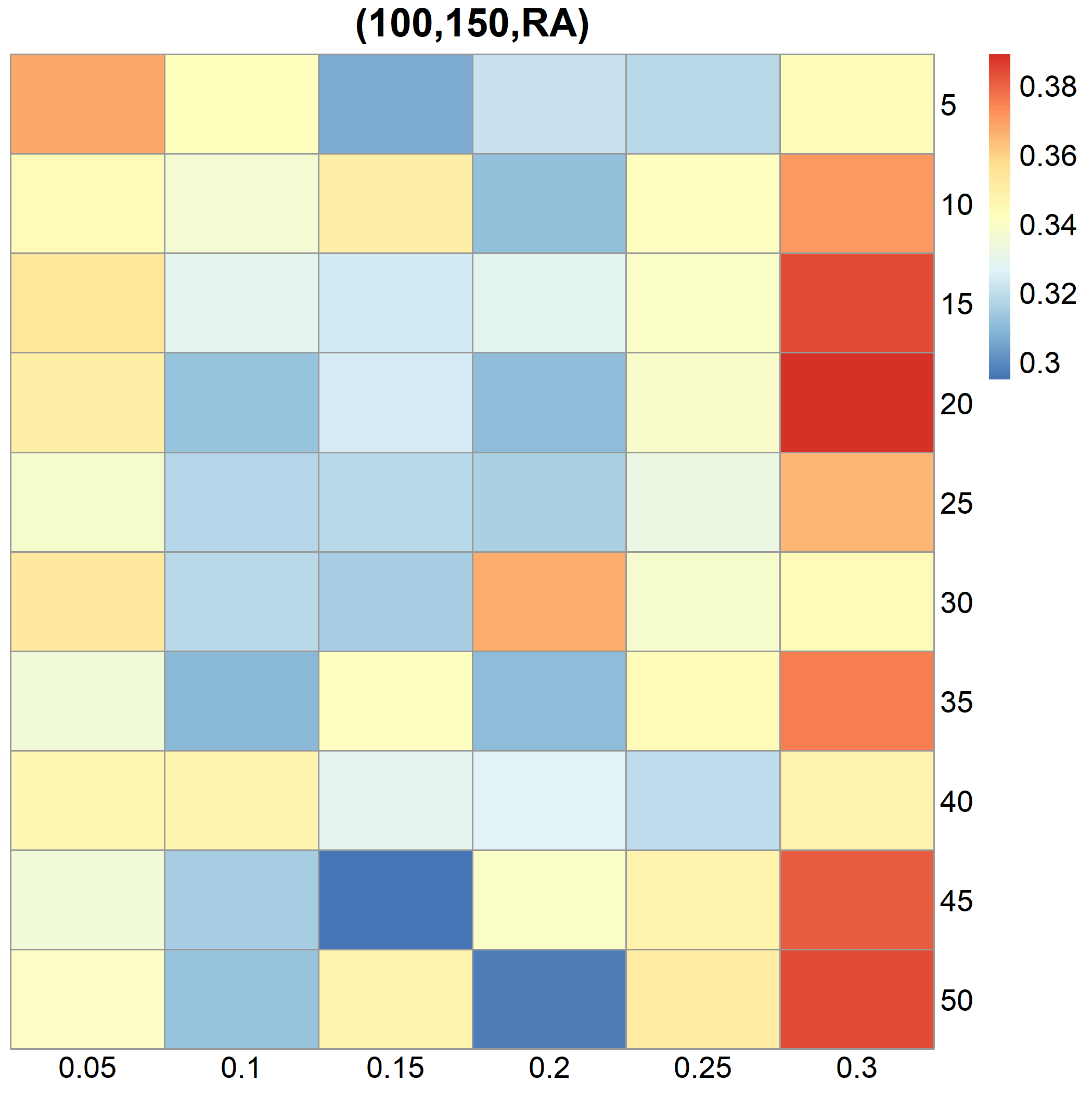}
		
		\vspace{3pt}
		\small (a) $N = 100$
	\end{minipage}
	
	\vspace{6pt}
	
	\begin{minipage}{0.95\textwidth}
		\centering
		\includegraphics[width=0.32\linewidth]{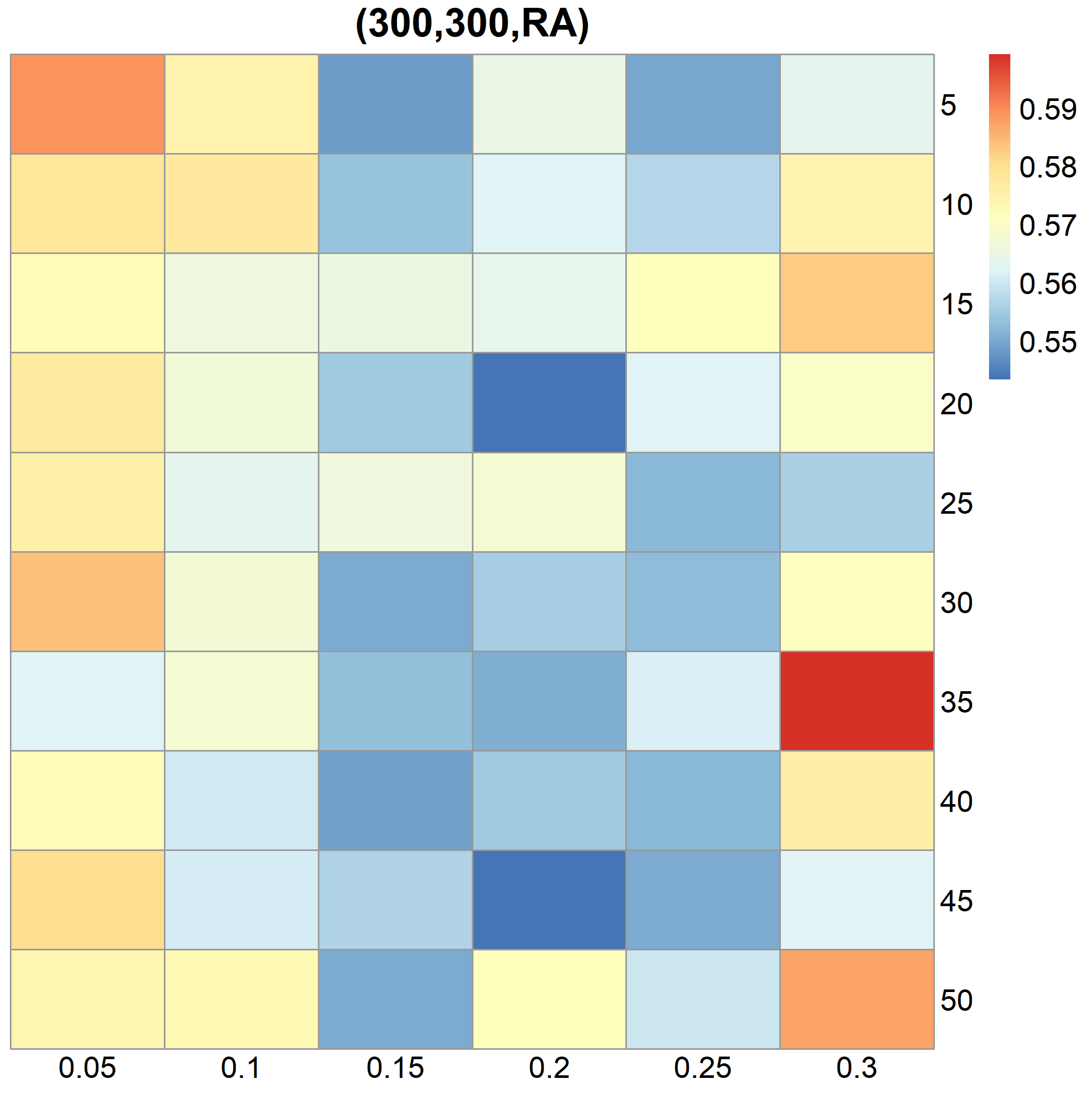}
		\includegraphics[width=0.32\linewidth]{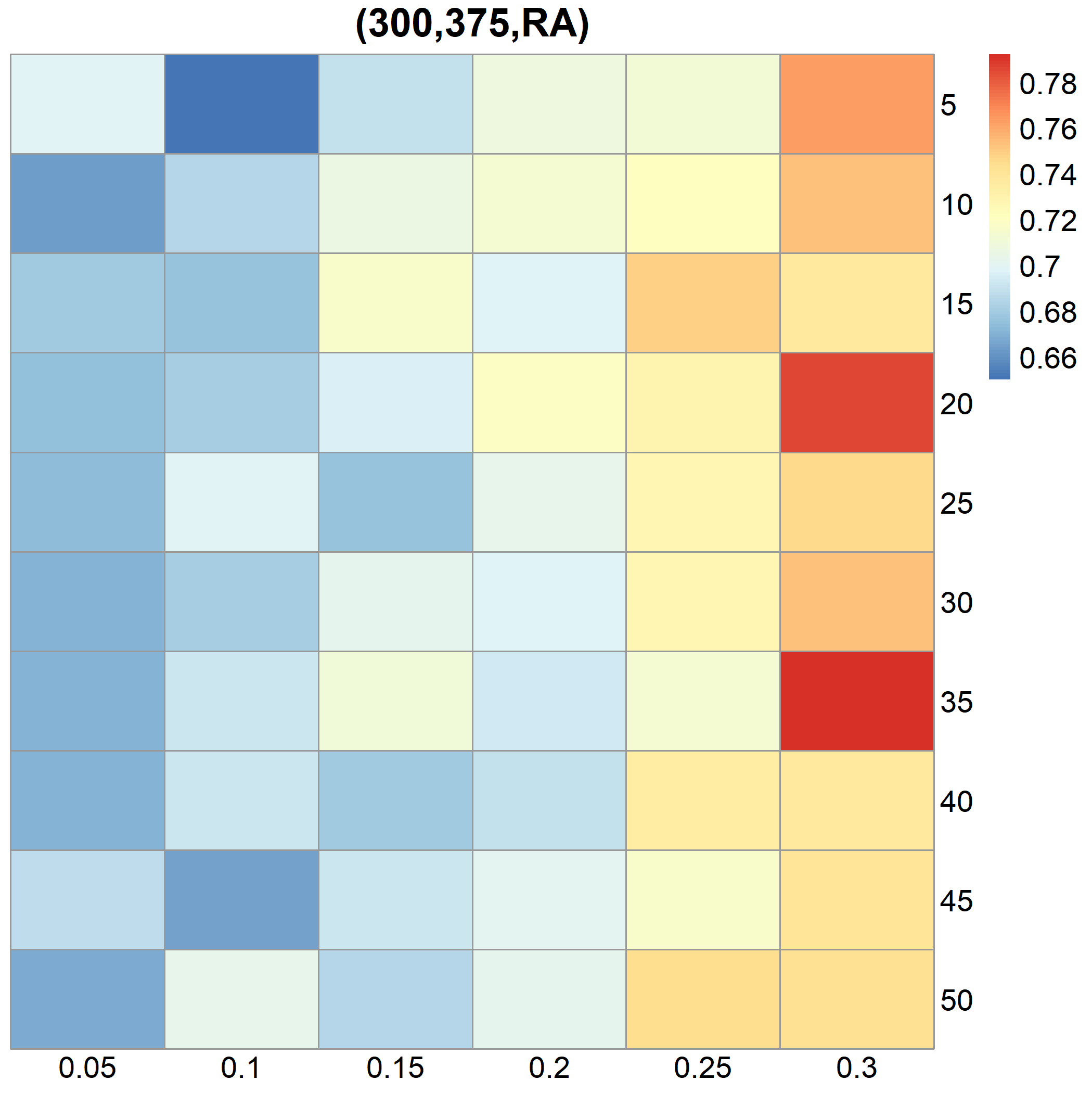}
		\includegraphics[width=0.32\linewidth]{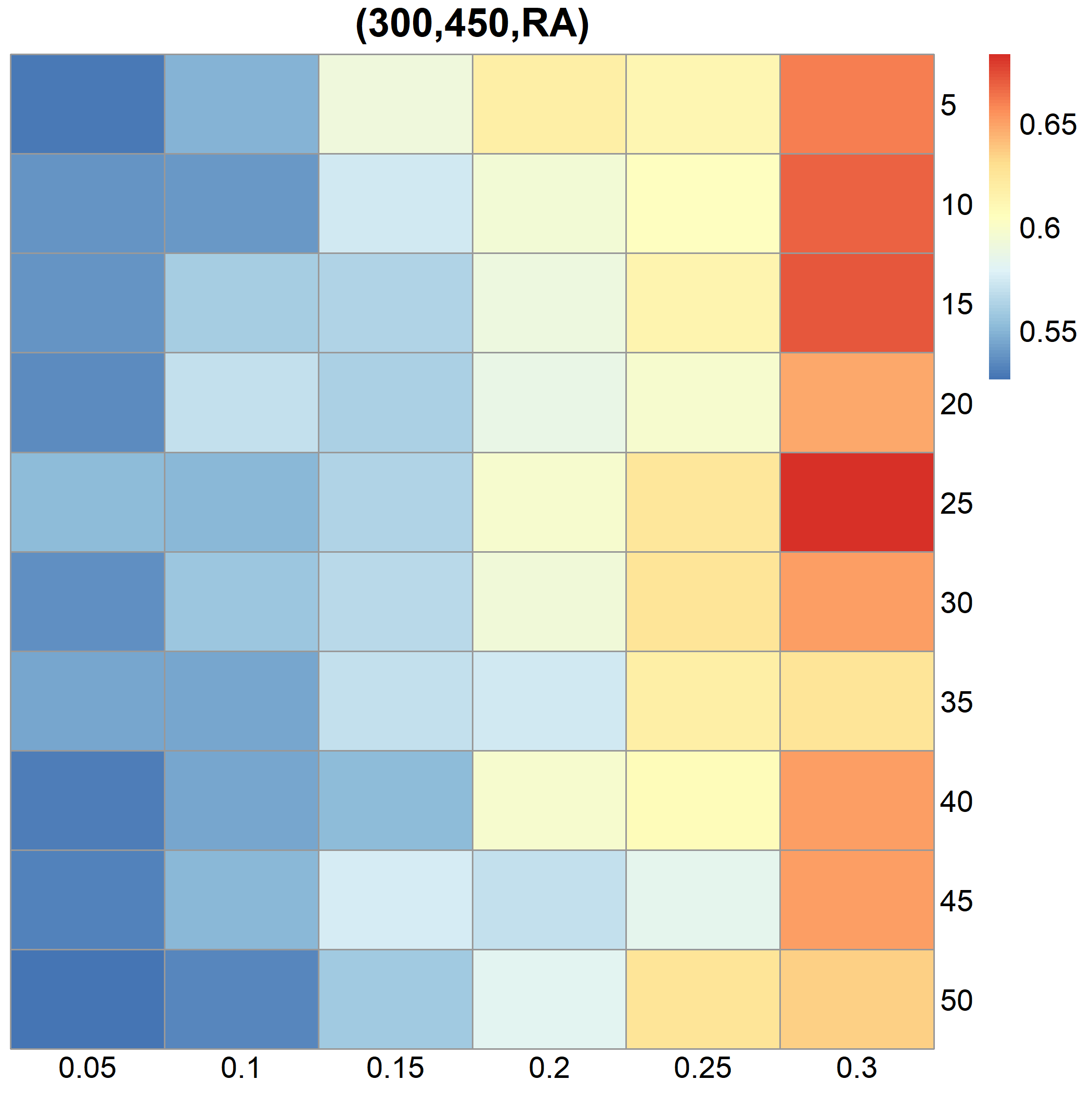}
		
		\vspace{3pt}
		\small (b) $N = 300$
	\end{minipage}
	
	\caption{Cross-validation results for Section \ref{app:manyrelevantvariables} under exponential decay with a random covariance structure. Values in parentheses denote $(N, K, \rho)$. The horizontal axis represents the selection probability $p$, and the vertical axis indicates the number of candidate models $M$. Darker regions correspond to $(p, M)$ combinations yielding lower cross-validation errors.}
	\label{fig:cvMRexpRA}
\end{figure*}

\newpage
\subsection{MSE results in simulations} \label{app:mse}
Tables \ref{tab:mseFXRApoly} to \ref{tab:mse910} reports the MSE comparisons for Sections \ref{sec3.2} and \ref{app:manyrelevantvariables} , with the smallest value in each row highlighted in bold. Notably RSA.o achieves the smallest training error when it has the smallest out-of-sample prediction error. 
\begin{table}[!h]
	\centering
	\caption{MSE comparison in Section \ref{sec3.2}: polynomial decay.}
	\scalebox{0.80}{
		\begin{threeparttable}
			\begin{tabular}{cccccccccccccc}
				\toprule
				$\rho$   & $N$     & $K$     & RSA.o & RSA.f & RSA.1 & RSR   & RF    & Lasso & SCAD  & MCP   & PMA   & MCV   & MMA \\
				\midrule
				\multirow[t]{12}[2]{*}{0.1} & \multirow[t]{4}[1]{*}{200} & 20    & 0.24  & 1.37  & 0.29  & 3.15  & 0.89  & 0.19  & 0.19  & 0.19  & 0.59  & 0.63  & \textbf{0.18} \\
				&       & 100   & \textbf{0.57} & 1.52  & 0.63  & 2.99  & 1.10  & 0.73  & 0.69  & 0.71  & 0.83  & 1.07  & 0.71 \\
				&       & 200   & \textbf{0.98} & 1.90  & 1.14  & 3.13  & 1.37  & 1.38  & 1.28  & 1.38  & 1.35  & 1.59  & 1.71 \\
				&       & 300   & \textbf{1.21} & 2.03  & 1.39  & 3.06  & 1.54  & 1.71  & 1.68  & 1.86  & 1.96  & 1.91  & 2.40 \\
				& \multirow[t]{4}[0]{*}{400} & 40    & 0.31  & 1.94  & 0.43  & 3.66  & 1.10  & \textbf{0.20} & \textbf{0.20} & \textbf{0.20} & 0.74  & 0.71  & \textbf{0.20} \\
				&       & 200   & \textbf{0.65} & 2.01  & 0.72  & 3.44  & 1.28  & 0.85  & 0.79  & 0.80  & 0.95  & 1.19  & 0.80 \\
				&       & 400   & \textbf{1.38} & 2.34  & 1.70  & 3.57  & 1.54  & 1.52  & 1.41  & 1.49  & 1.50  & 1.79  & 2.07 \\
				&       & 600   & \textbf{1.37} & 2.45  & 1.56  & 3.47  & 1.70  & 1.89  & 1.80  & 1.93  & 2.14  & 2.09  & 2.85 \\
				& \multirow[t]{4}[1]{*}{800} & 80    & 0.38  & 2.54  & 0.46  & 4.23  & 1.34  & 0.24  & 0.24  & 0.24  & 0.91  & 0.83  & \textbf{0.23} \\
				&       & 400   & \textbf{0.74} & 2.54  & 0.83  & 3.94  & 1.47  & 0.95  & 0.89  & 0.88  & 1.11  & 1.32  & 0.91 \\
				&       & 800   & \textbf{1.38} & 2.80  & 1.59  & 3.98  & 1.71  & 1.68  & 1.55  & 1.60  & 1.63  & 1.86  & 2.26 \\
				&       & 1200  & \textbf{1.45} & 2.84  & 1.59  & 3.81  & 1.84  & 2.05  & 1.92  & 2.04  & 2.31  & 2.22  & 3.21 \\
				\midrule
				\multirow[t]{12}[2]{*}{0.9} & \multirow[t]{4}[1]{*}{200} & 20    & 1.04  & 0.86  & 1.97  & 2.38  & 3.16  & \textbf{0.84} & 1.26  & 1.22  & 1.98  & 1.27  & 1.44 \\
				&       & 100   & \textbf{0.97} & 1.04  & 1.06  & 1.67  & 2.84  & 1.38  & 1.75  & 1.77  & 2.12  & 1.60  & 3.01 \\
				&       & 200   & \textbf{1.51} & 1.56  & 1.61  & 1.93  & 4.00  & 2.30  & 2.66  & 2.75  & 3.41  & 2.34  & 7.83 \\
				&       & 300   & \textbf{2.05} & \textbf{2.05} & 2.35  & 2.22  & 4.74  & 3.29  & 3.49  & 3.74  & 4.24  & 3.20  & 10.19 \\
				& \multirow[t]{4}[0]{*}{400} & 40    & \textbf{0.78} & 0.98  & 0.91  & 2.78  & 4.96  & 1.06  & 1.62  & 1.59  & 2.67  & 2.77  & 1.95 \\
				&       & 200   & \textbf{0.98} & 1.14  & 1.36  & 1.89  & 3.68  & 1.36  & 1.79  & 1.84  & 2.36  & 2.10  & 3.46 \\
				&       & 400   & \textbf{1.49} & 1.67  & 1.73  & 2.12  & 4.78  & 2.34  & 2.78  & 2.88  & 3.69  & 2.68  & 8.28 \\
				&       & 600   & \textbf{2.01} & 2.08  & 2.17  & 2.31  & 5.40  & 3.41  & 3.83  & 3.91  & 4.52  & 3.43  & 10.79 \\
				& \multirow[t]{4}[1]{*}{800} & 80    & \textbf{0.77} & 1.24  & 1.12  & 3.34  & 7.11  & 1.12  & 1.78  & 1.77  & 3.21  & 4.25  & 2.25 \\
				&       & 400   & \textbf{1.01} & 1.37  & 1.20  & 2.11  & 4.56  & 1.48  & 1.86  & 1.95  & 2.67  & 2.45  & 3.81 \\
				&       & 800   & \textbf{1.58} & 1.88  & 1.80  & 2.36  & 5.63  & 2.52  & 2.99  & 3.09  & 4.01  & 3.06  & 9.94 \\
				&       & 1200  & \textbf{2.05} & 2.24  & 2.23  & 2.50  & 6.15  & 3.40  & 3.67  & 3.93  & 4.95  & 3.66  & 10.93 \\
				\midrule
				\multirow[t]{12}[2]{*}{RA} & \multirow[t]{4}[1]{*}{200} & 20    & 0.16  & 0.25  & 0.20  & 1.36  & 0.41  & 0.14  & 0.15  & 0.15  & 0.26  & 0.24  & \textbf{0.13} \\
				&       & 100   & 0.29  & 0.28  & 0.38  & 0.44  & 0.41  & 0.33  & 0.30  & 0.32  & 0.38  & 0.34  & \textbf{0.26} \\
				&       & 200   & 0.54  & 0.55  & 0.67  & \textbf{0.53} & 1.26  & 0.95  & 0.89  & 0.94  & 0.89  & 0.63  & 2.16 \\
				&       & 300   & 0.49  & \textbf{0.42} & 0.57  & 0.47  & 0.55  & 0.65  & 0.64  & 0.68  & 0.67  & 0.53  & 0.85 \\
				& \multirow[t]{4}[0]{*}{400} & 40    & 0.10  & 0.21  & 0.15  & 0.52  & 0.22  & \textbf{0.04} & 0.05  & 0.05  & 0.12  & 0.26  & \textbf{0.04} \\
				&       & 200   & \textbf{0.25} & 0.32  & 0.29  & 0.50  & 0.48  & 0.35  & 0.31  & 0.32  & 0.33  & 0.37  & 0.30 \\
				&       & 400   & \textbf{0.48} & 0.49  & 0.53  & 0.52  & 1.15  & 0.83  & 0.77  & 0.79  & 0.72  & 0.60  & 1.84 \\
				&       & 600   & 0.98  & 0.77  & 1.03  & \textbf{0.64} & 2.17  & 1.26  & 1.14  & 1.30  & 0.99  & 0.76  & 3.12 \\
				& \multirow[t]{4}[1]{*}{800} & 80    & 0.18  & 0.35  & 0.24  & 0.73  & 0.51  & 0.14  & 0.14  & 0.15  & 0.25  & 0.26  & \textbf{0.12} \\
				&       & 400   & \textbf{0.28} & 0.37  & 0.33  & 0.53  & 0.61  & 0.44  & 0.37  & 0.38  & 0.39  & 0.37  & 0.40 \\
				&       & 800   & \textbf{0.38} & 0.39  & 0.44  & 0.50  & 0.59  & 0.54  & 0.51  & 0.51  & 0.47  & 0.47  & 0.95 \\
				&       & 1200  & \textbf{0.50} & \textbf{0.50} & 0.54  & 0.51  & 1.01  & 0.86  & 0.78  & 0.83  & 0.78  & 0.59  & 1.42 \\
				\bottomrule
			\end{tabular}%
			\vspace{1ex}
			{\raggedright Note: Values in bold indicate the smallest MSE.\par}
		\end{threeparttable}
	}
	\label{tab:mseFXRApoly}%
\end{table}%

\newpage

\begin{table}[!h]
	\centering
	\caption{MSE comparison in Section \ref{sec3.2}: exponential decay.}
	\scalebox{0.80}{
		\begin{threeparttable}
			\begin{tabular}{cccccccccccccc}
				\toprule
				$\rho$   & $N$     & $K$     & RSA.o & RSA.f & RSA.1 & RSR   & RF    & Lasso & SCAD  & MCP   & PMA   & MCV   & MMA \\
				\midrule
				\multirow[t]{12}[2]{*}{0.1} & \multirow[t]{4}[1]{*}{200} & 20    & 0.07  & 0.36  & 0.10  & 0.71  & 0.22  & \textbf{0.04} & \textbf{0.04} & \textbf{0.04} & 0.14  & 0.15  & \textbf{0.04} \\
				&       & 100   & \textbf{0.15} & 0.40  & 0.17  & 0.69  & 0.27  & 0.18  & 0.16  & 0.16  & 0.19  & 0.26  & 0.16 \\
				&       & 200   & \textbf{0.30} & 0.50  & 0.35  & 0.74  & 0.34  & 0.34  & 0.33  & 0.35  & 0.36  & 0.41  & 0.48 \\
				&       & 300   & \textbf{0.30} & 0.52  & 0.34  & 0.73  & 0.37  & 0.42  & 0.41  & 0.47  & 0.48  & 0.48  & 0.57 \\
				& \multirow[t]{4}[0]{*}{400} & 40    & 0.08  & 0.53  & 0.10  & 0.87  & 0.28  & \textbf{0.05} & \textbf{0.05} & \textbf{0.05} & 0.19  & 0.17  & \textbf{0.05} \\
				&       & 200   & \textbf{0.16} & 0.54  & 0.18  & 0.81  & 0.32  & 0.21  & 0.19  & 0.19  & 0.24  & 0.30  & 0.20 \\
				&       & 400   & \textbf{0.27} & 0.60  & 0.30  & 0.83  & 0.37  & 0.35  & 0.33  & 0.34  & 0.36  & 0.42  & 0.46 \\
				&       & 600   & \textbf{0.34} & 0.60  & 0.39  & 0.78  & 0.40  & 0.42  & 0.40  & 0.43  & 0.45  & 0.48  & 0.62 \\
				& \multirow[t]{4}[1]{*}{800} & 80    & 0.10  & 0.68  & 0.12  & 1.01  & 0.34  & 0.06  & 0.06  & 0.06  & 0.21  & 0.19  & \textbf{0.05} \\
				&       & 400   & \textbf{0.18} & 0.66  & 0.21  & 0.92  & 0.36  & 0.23  & 0.21  & 0.20  & 0.24  & 0.30  & 0.22 \\
				&       & 800   & \textbf{0.29} & 0.67  & 0.33  & 0.88  & 0.39  & 0.36  & 0.32  & 0.32  & 0.33  & 0.38  & 0.52 \\
				&       & 1200  & \textbf{0.36} & 0.64  & 0.41  & 0.82  & 0.40  & 0.42  & 0.38  & 0.37  & 0.40  & 0.42  & 0.68 \\
				\midrule
				\multirow[t]{12}[2]{*}{0.9} & \multirow[t]{4}[1]{*}{200} & 20    & \textbf{0.17} & 0.19  & 0.18  & 0.55  & 0.77  & 0.20  & 0.30  & 0.30  & 0.49  & 0.30  & 0.34 \\
				&       & 100   & \textbf{0.24} & 0.25  & 0.25  & 0.41  & 0.71  & 0.35  & 0.45  & 0.46  & 0.55  & 0.41  & 0.78 \\
				&       & 200   & 0.41  & \textbf{0.40} & 0.42  & 0.48  & 1.02  & 0.63  & 0.74  & 0.73  & 0.89  & 0.63  & 2.12 \\
				&       & 300   & \textbf{0.51} & \textbf{0.51} & 0.54  & 0.53  & 1.19  & 0.81  & 0.88  & 0.93  & 1.09  & 0.78  & 2.60 \\
				& \multirow[t]{4}[0]{*}{400} & 40    & \textbf{0.21} & 0.26  & 0.22  & 0.73  & 1.34  & 0.27  & 0.41  & 0.42  & 0.73  & 0.71  & 0.51 \\
				&       & 200   & \textbf{0.25} & 0.30  & 0.30  & 0.47  & 0.95  & 0.36  & 0.47  & 0.49  & 0.64  & 0.52  & 0.91 \\
				&       & 400   & \textbf{0.37} & 0.42  & 0.41  & 0.51  & 1.18  & 0.60  & 0.71  & 0.73  & 0.93  & 0.65  & 2.11 \\
				&       & 600   & \textbf{0.46} & 0.48  & 0.50  & 0.53  & 1.24  & 0.77  & 0.86  & 0.89  & 1.05  & 0.78  & 2.30 \\
				& \multirow[t]{4}[1]{*}{800} & 80    & \textbf{0.21} & 0.32  & 0.31  & 0.84  & 1.84  & 0.31  & 0.48  & 0.47  & 0.84  & 1.09  & 0.61 \\
				&       & 400   & \textbf{0.25} & 0.34  & 0.31  & 0.50  & 1.12  & 0.35  & 0.46  & 0.47  & 0.67  & 0.59  & 0.93 \\
				&       & 800   & \textbf{0.35} & 0.41  & 0.49  & 0.51  & 1.20  & 0.53  & 0.63  & 0.65  & 0.84  & 0.66  & 1.96 \\
				&       & 1200  & \textbf{0.40} & 0.44  & 0.45  & 0.50  & 1.18  & 0.64  & 0.69  & 0.74  & 0.90  & 0.71  & 2.28 \\
				\midrule
				\multirow[t]{12}[2]{*}{RA} & \multirow[t]{4}[1]{*}{200} & 20    & \textbf{0.17} & 0.23  & 0.21  & 1.25  & 0.54  & 0.18  & 0.21  & 0.20  & 0.31  & 0.28  & \textbf{0.17} \\
				&       & 100   & \textbf{0.22} & 0.27  & 0.25  & 0.46  & 0.40  & 0.34  & 0.34  & 0.33  & 0.33  & 0.34  & 0.32 \\
				&       & 200   & 0.27  & \textbf{0.26} & 0.35  & 0.30  & 0.41  & 0.42  & 0.43  & 0.45  & 0.45  & 0.33  & 0.69 \\
				&       & 300   & \textbf{0.28} & 0.31  & 0.33  & 0.46  & 0.39  & 0.41  & 0.38  & 0.40  & 0.41  & 0.38  & 0.60 \\
				& \multirow[t]{4}[0]{*}{400} & 40    & 0.19  & 0.35  & 0.23  & 1.07  & 0.59  & 0.17  & 0.19  & 0.19  & 0.32  & 0.34  & \textbf{0.15} \\
				&       & 200   & 0.39  & \textbf{0.37} & 0.42  & 0.58  & 0.94  & 0.57  & 0.49  & 0.51  & 0.51  & 0.47  & 0.58 \\
				&       & 400   & 0.63  & \textbf{0.62} & 0.70  & 0.67  & 1.65  & 1.06  & 0.98  & 1.05  & 1.02  & 0.78  & 2.49 \\
				&       & 600   & \textbf{0.33} & \textbf{0.33} & 0.37  & 0.41  & 0.47  & 0.50  & 0.46  & 0.50  & 0.47  & 0.46  & 0.78 \\
				& \multirow[t]{4}[1]{*}{800} & 80    & 0.15  & 0.27  & 0.18  & 0.45  & 0.45  & 0.13  & 0.14  & 0.14  & 0.26  & 0.30  & \textbf{0.12} \\
				&       & 400   & 0.31  & 0.35  & 0.36  & 0.50  & 0.44  & 0.32  & \textbf{0.29} & 0.30  & 0.31  & 0.34  & 0.30 \\
				&       & 800   & \textbf{0.49} & 0.50  & 0.55  & 0.59  & 1.12  & 0.91  & 0.75  & 0.75  & 0.72  & 0.63  & 1.76 \\
				&       & 1200  & 1.23  & \textbf{0.68} & 1.28  & 0.70  & 1.78  & 1.03  & 0.86  & 0.88  & 0.83  & 0.75  & 2.59 \\
				\bottomrule
			\end{tabular}%
			\vspace{1ex}
			{\raggedright Note: Values in bold indicate the smallest MSE.\par}
		\end{threeparttable}
	}
	\label{tab:mseFXRAexp}%
\end{table}%

\newpage

\begin{table}[!h]
	\centering
	\caption{MSE comparison in Section \ref{app:manyrelevantvariables}: polynomial decay.}
	\scalebox{0.80}{
		\begin{threeparttable}
			\begin{tabular}{cccccccccccccc}
				\toprule
				$\rho$   & $N$     & $K$     & RSA.o & RSA.f & RSA.1 & RSR   & RF    & Lasso & SCAD  & MCP   & PMA   & MCV   & MMA \\
				\midrule
				\multirow[t]{6}[2]{*}{0.1} & \multirow[t]{3}[1]{*}{100} & 100   & \textbf{1.36} & 2.29  & 1.84  & 3.62  & 1.75  & 1.95  & 2.05  & 2.64  & 2.34  & 2.50  & 2.29 \\
				&       & 125   & \textbf{1.59} & 2.38  & 1.91  & 3.63  & 1.90  & 2.22  & 2.34  & 2.89  & 2.61  & 2.64  & 2.74 \\
				&       & 150   & \textbf{1.75} & 2.37  & 2.14  & 3.58  & 1.94  & 2.33  & 2.49  & 3.01  & 2.78  & 2.73  & 3.08 \\
				& \multirow[t]{3}[1]{*}{300} & 300   & \textbf{1.46} & 3.13  & 1.75  & 4.43  & 2.03  & 2.10  & 2.23  & 2.61  & 2.62  & 2.72  & 2.62 \\
				&       & 375   & \textbf{1.65} & 3.08  & 1.87  & 4.24  & 2.06  & 2.35  & 2.45  & 2.87  & 2.92  & 2.90  & 2.86 \\
				&       & 450   & 2.55  & 3.10  & 3.03  & 4.17  & \textbf{2.15} & 2.56  & 2.66  & 3.20  & 3.30  & 3.03  & 3.28 \\
				\midrule
				\multirow[t]{6}[2]{*}{0.9} & \multirow[t]{3}[1]{*}{100} & 100   & 4.77  & 4.71  & 5.03  & \textbf{4.22} & 10.15 & 7.03  & 8.64  & 9.12  & 10.53 & 6.22  & 21.39 \\
				&       & 125   & 5.06  & 5.18  & 5.82  & \textbf{4.22} & 11.26 & 8.04  & 9.17  & 9.90  & 11.54 & 6.87  & 25.66 \\
				&       & 150   & 5.69  & 5.78  & 6.30  & \textbf{4.83} & 12.58 & 9.67  & 11.08 & 10.89 & 12.50 & 7.67  & 28.24 \\
				& \multirow[t]{3}[1]{*}{300} & 300   & \textbf{4.07} & 4.36  & 4.38  & 4.56  & 13.78 & 7.19  & 8.96  & 8.93  & 12.36 & 5.89  & 26.62 \\
				&       & 375   & 4.90  & 4.85  & 5.42  & \textbf{4.76} & 14.60 & 8.27  & 9.98  & 10.16 & 13.27 & 6.78  & 28.86 \\
				&       & 450   & 5.37  & 5.41  & 5.64  & \textbf{5.09} & 15.50 & 9.91  & 11.27 & 11.47 & 14.41 & 7.97  & 29.71 \\
				\midrule
				\multirow[t]{6}[2]{*}{RA} & \multirow[t]{3}[1]{*}{100} & 100   & \textbf{0.40} & \textbf{0.40} & 0.46  & 0.42  & 0.58  & 0.65  & 0.61  & 0.63  & 0.66  & 0.48  & 0.93 \\
				&       & 125   & \textbf{0.52} & 0.53  & 0.56  & 0.65  & 0.72  & 0.90  & 0.84  & 0.88  & 0.87  & 0.79  & 1.09 \\
				&       & 150   & 0.51  & \textbf{0.50} & 0.57  & 0.52  & 0.68  & 0.82  & 0.76  & 0.81  & 0.80  & 0.63  & 1.19 \\
				& \multirow[t]{3}[1]{*}{300} & 300   & 0.63  & \textbf{0.60} & 0.72  & \textbf{0.60} & 1.33  & 1.10  & 1.00  & 1.07  & 0.99  & 0.70  & 2.00 \\
				&       & 375   & 0.56  & \textbf{0.51} & 0.64  & 0.56  & 0.82  & 0.79  & 0.78  & 0.81  & 0.83  & 0.63  & 1.30 \\
				&       & 450   & \textbf{0.43} & 0.45  & 0.51  & 0.53  & 0.60  & 0.73  & 0.69  & 0.73  & 0.71  & 0.60  & 0.90 \\
				\bottomrule
			\end{tabular}%
			\vspace{1ex}
			{\raggedright Note: Values in bold indicate the smallest MSE.\par}
		\end{threeparttable}
	}
	\label{tab:mse56}%
\end{table}%

\begin{table}[!h]
	\centering
	\caption{MSE comparison in Section \ref{app:manyrelevantvariables}: exponential decay.}
	\scalebox{0.80}{
		\begin{threeparttable}
			\begin{tabular}{cccccccccccccc}
				\toprule
				$\rho$   & $N$     & $K$     & RSA.o & RSA.f & RSA.1 & RSR   & RF    & Lasso & SCAD  & MCP   & PMA   & MCV   & MMA \\
				\midrule
				\multirow[t]{6}[2]{*}{0.1} & \multirow[t]{3}[1]{*}{100} & 100   & \textbf{0.32} & 0.58  & 0.36  & 0.88  & 0.43  & 0.46  & 0.55  & 0.66  & 0.55  & 0.62  & 0.54 \\
				&       & 125   & \textbf{0.35} & 0.59  & 0.37  & 0.84  & 0.45  & 0.51  & 0.59  & 0.70  & 0.60  & 0.66  & 0.62 \\
				&       & 150   & \textbf{0.42} & 0.58  & 0.49  & 0.84  & 0.46  & 0.57  & 0.66  & 0.80  & 0.64  & 0.67  & 0.72 \\
				& \multirow[t]{3}[1]{*}{300} & 300   & \textbf{0.33} & 0.71  & 0.38  & 0.96  & 0.44  & 0.44  & 0.46  & 0.53  & 0.52  & 0.58  & 0.55 \\
				&       & 375   & \textbf{0.33} & 0.68  & 0.37  & 0.90  & 0.45  & 0.48  & 0.49  & 0.55  & 0.58  & 0.59  & 0.64 \\
				&       & 450   & 0.52  & 0.66  & 0.60  & 0.86  & \textbf{0.45} & 0.51  & 0.51  & 0.59  & 0.61  & 0.60  & 0.71 \\
				\midrule
				\multirow[t]{6}[2]{*}{0.9} & \multirow[t]{3}[1]{*}{100} & 100   & 1.21  & 1.22  & 1.75  & \textbf{1.09} & 2.66  & 1.84  & 2.15  & 2.26  & 2.73  & 1.63  & 5.46 \\
				&       & 125   & 1.60  & 1.25  & 2.19  & \textbf{1.01} & 2.79  & 1.96  & 2.34  & 2.40  & 2.86  & 1.68  & 5.74 \\
				&       & 150   & 1.52  & 1.32  & 1.59  & \textbf{1.07} & 2.87  & 2.18  & 2.48  & 2.58  & 2.93  & 1.92  & 6.41 \\
				& \multirow[t]{3}[1]{*}{300} & 300   & \textbf{0.77} & 0.83  & 0.85  & 0.89  & 2.64  & 1.35  & 1.70  & 1.73  & 2.32  & 1.20  & 4.59 \\
				&       & 375   & \textbf{0.84} & 0.89  & 0.99  & 0.89  & 2.64  & 1.51  & 1.81  & 1.86  & 2.35  & 1.27  & 5.25 \\
				&       & 450   & 0.90  & 0.92  & 1.13  & \textbf{0.89} & 2.59  & 1.64  & 1.86  & 1.95  & 2.40  & 1.37  & 5.42 \\
				\midrule
				\multirow[t]{6}[2]{*}{RA} & \multirow[t]{3}[1]{*}{100} & 100   & \textbf{0.14} & \textbf{0.14} & 0.16  & \textbf{0.14} & 0.22  & 0.23  & 0.21  & 0.21  & 0.19  & 0.16  & 0.36 \\
				&       & 125   & \textbf{0.11} & \textbf{0.11} & 0.14  & 0.12  & 0.14  & 0.16  & 0.17  & 0.17  & 0.17  & 0.14  & 0.21 \\
				&       & 150   & 0.19  & 0.18  & 0.19  & \textbf{0.17} & 0.28  & 0.31  & 0.27  & 0.27  & 0.27  & 0.22  & 0.49 \\
				& \multirow[t]{3}[1]{*}{300} & 300   & \textbf{0.07} & 0.09  & 0.08  & 0.12  & 0.09  & 0.10  & 0.10  & 0.10  & 0.10  & 0.10  & 0.12 \\
				&       & 375   & 0.08  & 0.08  & 0.11  & 0.10  & \textbf{0.07} & 0.09  & 0.10  & 0.10  & 0.10  & 0.08  & 0.11 \\
				&       & 450   & 0.14  & \textbf{0.13} & 0.15  & 0.14  & 0.22  & 0.22  & 0.20  & 0.21  & 0.20  & 0.16  & 0.36 \\
				\bottomrule
			\end{tabular}%
			\vspace{1ex}
			{\raggedright Note: Values in bold indicate the smallest MSE.\par}
		\end{threeparttable}
	}
	\label{tab:mse910}%
\end{table}%

\newpage

\subsection{Additional results for Section \ref{sec4}} \label{app:empirical}

\subsubsection{Plot of Log returns for two periods}
Figure \ref{fig:empiricalsp500} displays the log returns during the pre- and post-crisis periods. Due to macroeconomic shocks, such as the European sovereign debt crisis around 2010-2012, S\&P 500 log returns exhibits significantly higher volatility in the post-crisis period compared to the pre-crisis period. 
\begin{figure}[!h]
	\centering
	\begin{minipage}[b]{0.48\textwidth}
		\centering
		\includegraphics[width=\linewidth, trim={0 0cm 0 0cm}]{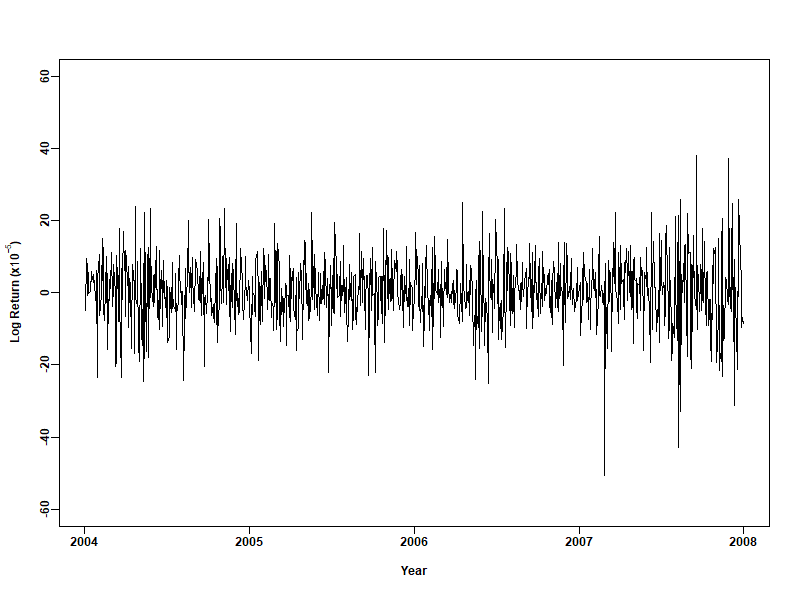}
		
		\small (a) Pre-crisis
	\end{minipage}
	\hfill
	\begin{minipage}[b]{0.48\textwidth}
		\centering
		\includegraphics[width=\linewidth, trim={0 0cm 0 0cm}]{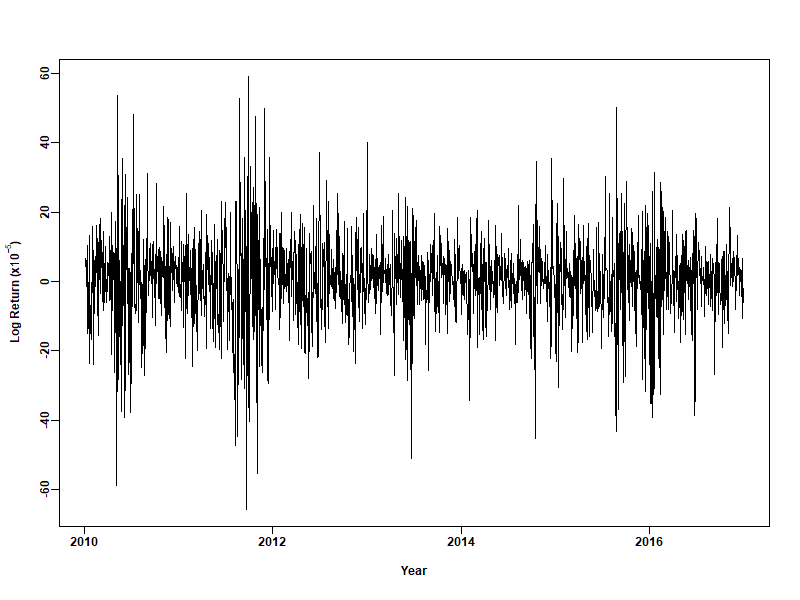}
		
		\small (b) Post-crisis
	\end{minipage}
	\caption{Log returns of the S\&P 500 index across two periods.}
	\label{fig:empiricalsp500}
\end{figure}

\subsubsection{Tuning parameters selected by CV}
Based on the CV-selected tuning parameters from the simulation and our preliminary exploration, we set the tuning grid as $p \in [0.01, 0.3]$ with an increment of 0.02 and $M \in [1,29]$ with an increment of 2 for the pre-crisis period. For the post-crisis period, we use $p \in [0.1, 0.3]$ with an increment of 0.02 and $M \in [1, 29]$ with an increment of 2. Figure~\ref{EP_cv} displays the cross-validation results for both periods. Because the factors are orthogonalized, RSA tends to select each factor with a relatively high selection probability, for example, in the post-crisis period, while its performance is not highly sensitive to the number of candidate models. These results closely resemble the CV findings under the low-correlation setting in Section \ref{app:manyrelevantvariables}.
\begin{figure}[!h]
	\centering
	\begin{minipage}[b]{0.48\textwidth}
		\centering
		\includegraphics[width=0.85\linewidth, trim={0 0cm 0 0cm}]{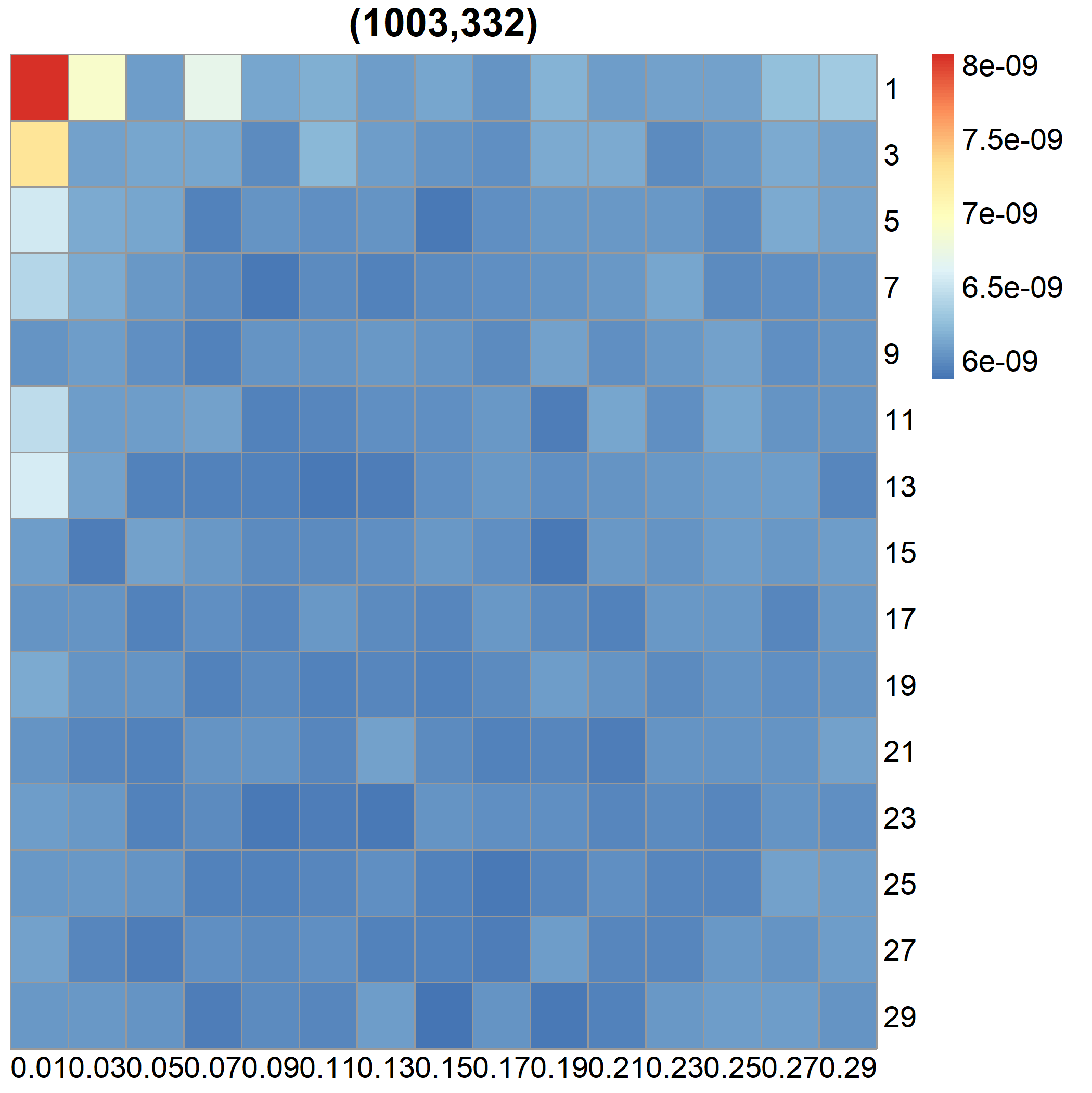}
		
		\small (a) Pre-crisis
	\end{minipage}
	\hfill
	\begin{minipage}[b]{0.48\textwidth}
		\centering
		\includegraphics[width=0.85\linewidth, trim={0 0cm 0 0cm}]{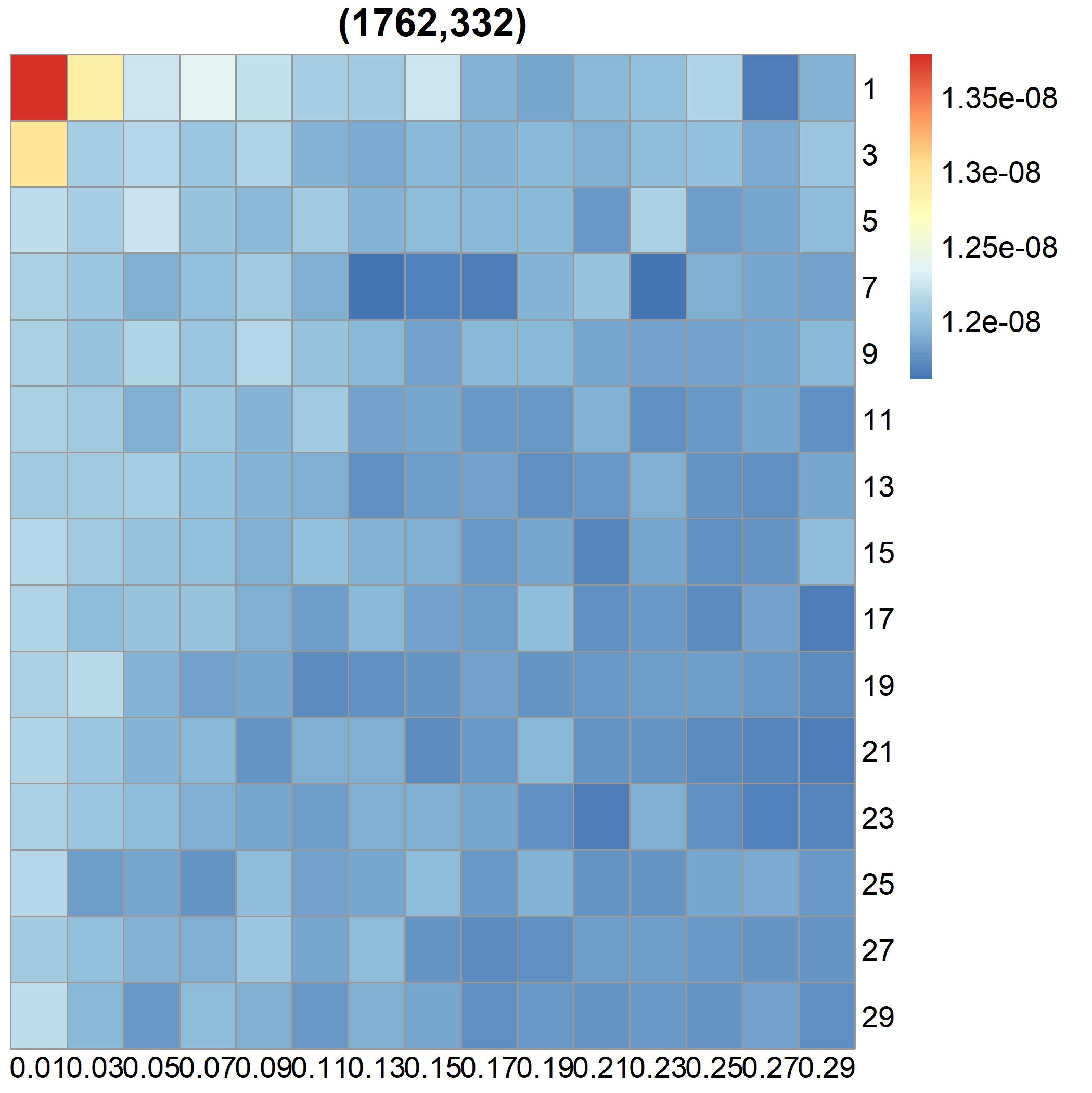}
		
		\small (b) Post-crisis
	\end{minipage}
	\caption{Cross-validation results for Section \ref{sec4}. Values in parentheses denote $(N, K)$. The horizontal axis represents the selection probability $p$, and the vertical axis indicates the number of candidate models $M$. Darker regions correspond to $(p, M)$ combinations yielding lower cross-validation errors.}
	\label{EP_cv}
\end{figure}

\subsubsection{Standard deviation of MSFE for different methods}

Figure \ref{fig:empiricalmsfesd} reports the standard deviation of MSFE for each forecast horizon. In the pre-crisis period, RSA exhibits lower volatility for most horizons, while in the post-crisis period, it consistently achieves the lowest prediction volatility in each horizon. 
\begin{figure}[htbp]
	\centering
	\begin{minipage}[b]{0.48\textwidth}
		\centering
		\includegraphics[width=\linewidth, trim={0 0cm 0 0cm}]{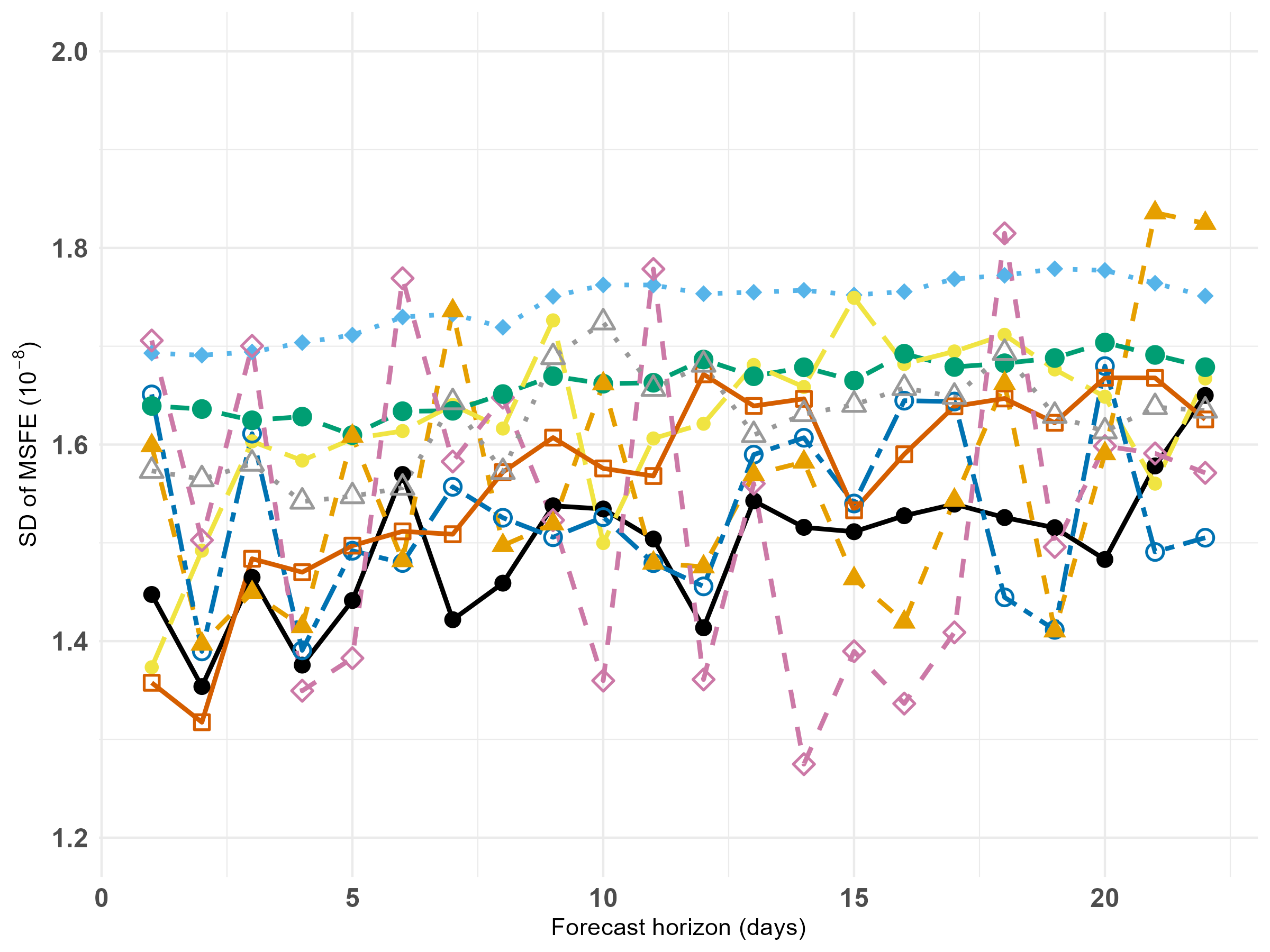}
		
		\small (a) Pre-crisis
	\end{minipage}
	\hfill
	\begin{minipage}[b]{0.48\textwidth}
		\centering
		\includegraphics[width=\linewidth, trim={0 0cm 0 0cm}]{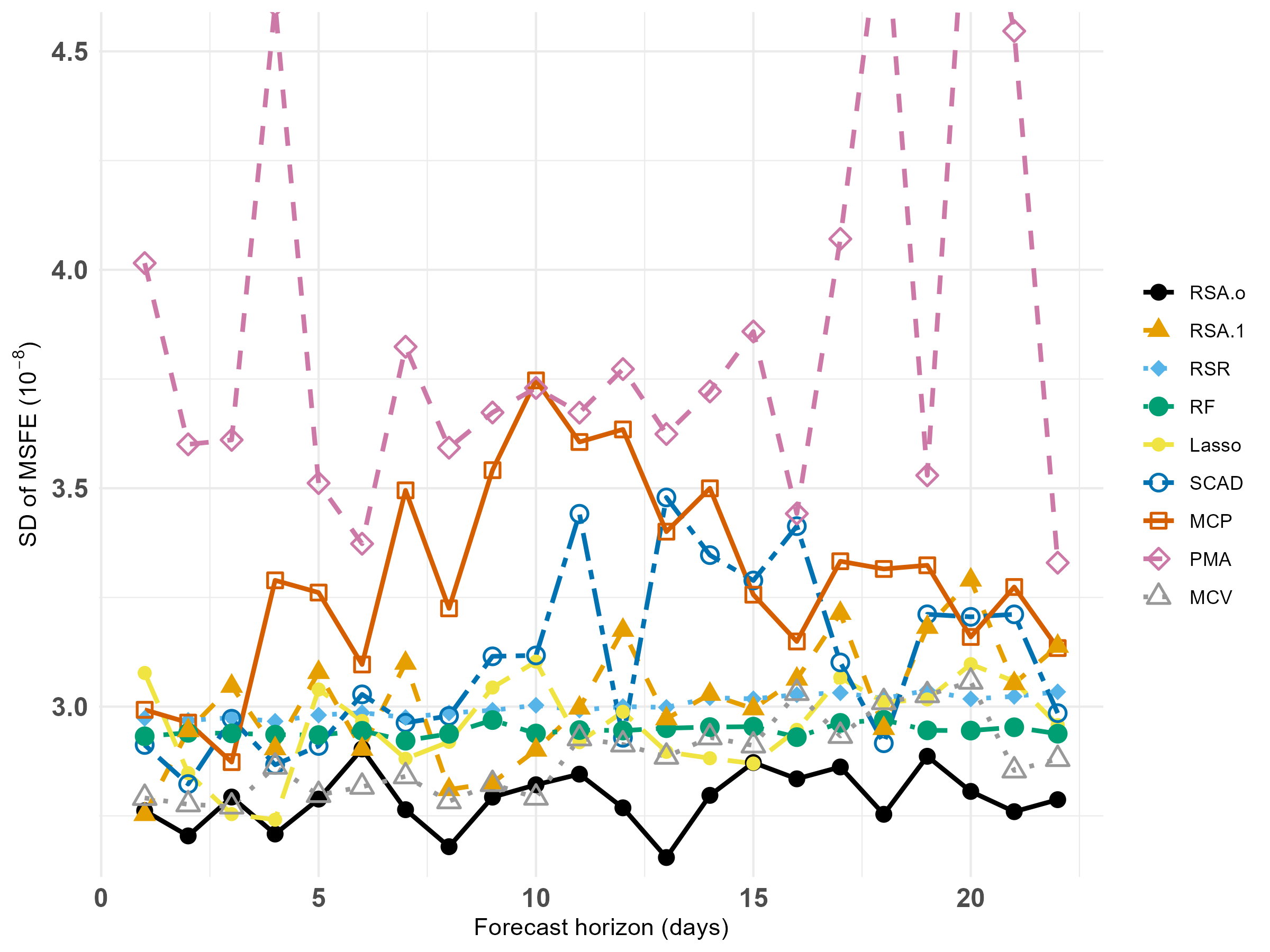}
		
		\small (b) Post-crisis
	\end{minipage}
	\caption{Standard deviation of MSFE across forecast horizons.}
	\label{fig:empiricalmsfesd}
\end{figure}

\subsubsection{MSE results for different methods}
Table \ref{EP_MSE} presents the training error for each subperiod, showing that RSA's superior out-of-sample performance is not achieved at the expense of higher training error.

\begin{table}[!h]
	\centering
	\caption{MSE ($\times 10^{-9}$) comparison for empirical analysis.}
	{\begin{tabular}{lcccccccccc} \toprule  
			Period & 	{RSA} & {RSR} & {Lasso} & {SCAD} & {MCP} & {PMA} & {RF} & RSA.1 & MCV  \\ 
			\midrule
			pre-crisis & 2.56  & 4.43  & 2.20   & 2.53  & 3.20   & 4.30   & 1.31  & 2.94  & 3.56 \\
			post-crisis & 6.77  & 9.40   & 5.25  & 6.07  & 8.03  & 8.73  & 2.93  & 8.12  & 7.66 \\
			\bottomrule
	\end{tabular}}
	\label{EP_MSE}
\end{table}

\end{document}

%% file: plot-of-two-layer-RSA-v2.tex
	
	\begin{tikzpicture}[
		>=stealth,->,
		input neuron/.style={circle, draw=blue!60, fill=blue!20, blur shadow, minimum size=18pt, inner sep = 0pt},
		mu neuron/.style={circle, draw=black!70, fill=gray!15, blur shadow, minimum size=18pt, inner sep = 0pt},
		mu2 neuron/.style={circle, draw=purple!70, fill=purple!15, blur shadow, minimum size=18pt, inner sep = 0pt},
		output neuron/.style={circle, draw=red!70, fill=red!20, blur shadow, minimum size=18pt, inner sep = 0pt},
		font=\small
		]
		\def\layersep{15mm}   
		\def\dy{0.45}         
		\def\boxgap{10pt}     
		\def \sidelen{5.8}
		
		\node[input neuron] (X) at (0,0) {$X$};
		
		\node (inputLabel) at ($(X)+(-\sidelen,0)$) {\textbf{Input}};
		\coordinate (L1) at ($(X)+(0,-\layersep)$);
		\node (layer1Label) at ($(L1)+(-\sidelen,0)$) {\textbf{Layer 1}};
		\coordinate (L2) at ($(X)+(0,-2*\layersep)$);
		\node (layer2Label) at ($(L2)+(-\sidelen,0)$) {\textbf{Layer 2}};
		\coordinate (L3) at ($(X)+(0,-3*\layersep)$);
		\node (outputLabel) at ($(L3)+(-\sidelen,0)$) {\textbf{Output}};
		
		\coordinate (arrowX) at ($(inputLabel.east)+(-0.6,0)$);
		
		\draw[->]  ($(arrowX |- inputLabel) + (0, -0.2)$)  -- 
		node[pos=0.5,right] {$R_m^{(\ell)}$}
		($(arrowX |- layer1Label)+(0,0.2)$);
		
		\draw[->]  ($(arrowX |- layer1Label)+(0,-0.2)$) -- 
		node[pos=0.5,right] {$\hat{w}_m^{(\ell)}$}
		($(arrowX |- layer2Label)+(0,0.2)$);
		
		\draw[->]  ($(arrowX |- layer2Label)+ (0,-0.2)$) -- 
		node[pos=0.5,right] {$\mathbbm{\hat w}_{\ell}$}
		($(arrowX |- outputLabel)+(0,0.2)$);

		\node[mu neuron] (mu11) at ($(L1)+(-4.4,0)$) {$\hat{\mu}_1^{(1)}$};
		\node[] (mu12) at ($(mu11)+(0.8,0)$) {$\cdots$};
		\node[mu neuron] (mu13) at ($(mu12)+(0.8,0)$) {$\hat{\mu}_M^{(1)}$};
		\node[draw=blue!60, thick, rounded corners, fit=(mu11)(mu13), inner sep=4pt]{};
		\foreach \i/\lab in {1/,3/}{\draw (X) -- node[midway, pos=0.5]{\lab} (mu1\i);} 
		
		\node[] (mu3dots) at ($(L1)+(-1.8,0)$) {$\dots$};
		
		\node[mu neuron] (mu21) at ($(L1)+(-0.8,0)$) {$\hat{\mu}_1^{(\ell)}$};
		\node[] (mu22) at ($(mu21)+(0.8,0)$) {$\cdots$};
		\node[mu neuron] (mu23) at ($(mu22)+(0.8,0)$) {$\hat{\mu}_M^{(\ell)}$};
		\node[draw=blue!60, thick, rounded corners, fit=(mu21)(mu23), inner sep=4pt]{};
		\foreach \i/\lab in {1/,3/}{\draw (X) -- node[midway, pos=0.5]{\lab} (mu2\i);}
		
		\node[] (mu3dots) at ($(L1)+(1.8,0)$) {$\dots$};
		
		\node[mu neuron] (mu41) at ($(L1)+(2.8,0)$) {$\hat{\mu}_1^{(L)}$};
		\node[] (mu42) at ($(mu41)+(0.8,0)$) {$\cdots$};
		\node[mu neuron] (mu43) at ($(mu42)+(0.8,0)$) {$\hat{\mu}_M^{(L)}$};
		\node[draw=blue!60, thick, rounded corners, fit=(mu41)(mu43), inner sep=4pt]{};
		\foreach \i/\lab in {1/,3/}{\draw (X) -- node[midway, pos=0.5]{\lab} (mu4\i);} 
		
		\node[mu2 neuron] (MU1) at ($(L2)+(-3.6,0)$) {$\hat{\mu}^{(1)}$};
		\node[] (MUdots1) at ($(L2)+(-1.8,0)$) {$\cdots$};
		\node[mu2 neuron] (MU2) at ($(L2)+(0,0)$) {$\hat{\mu}^{(\ell)}$};
		\node[mu2 neuron] (MUdots) at (L2) {$\hat{\mu}^{(\ell)}$};
		\node[] (MUdots2) at ($(L2)+(1.8,0)$) {$\cdots$};
		\node[mu2 neuron] (MUL) at ($(L2)+(3.6,0)$) {$\hat{\mu}^{(L)}$};
		
		\foreach \k/\lbl in {1/,3/}{\draw (mu1\k) -- node[midway, pos=0.5]{\lbl} (MU1);} 
		\foreach \k/\lbl in {1/,3/}{\draw (mu2\k) -- node[midway, pos=0.5]{\lbl} (MU2);}
		\foreach \k/\lbl in {1/,3/}{\draw (mu4\k) -- node[midway, pos=0.5]{\lbl} (MUL);} 
		
		\node[output neuron] (Yhat) at (L3) {$\hat{\mu}_{RSA}$};
		\draw (MU1) -- node[midway, pos=0.5] {} (Yhat);
		\draw (MU2) -- node[midway, pos=0.5] {} (Yhat);
		\draw[dashed] (MUdots) -- node[midway,pos=0.5] {} (Yhat);
		\draw (MUL) -- node[midway, pos=0.5] {} (Yhat);
		
	\end{tikzpicture}
	

%% file: ref-1.bib
@article{tibshirani1996regression,
	title        = {{Regression shrinkage and selection via the lasso}},
	author       = {Tibshirani, Robert},
	year         = 1996,
	journal      = {Journal of the Royal Statistical Society Series B: Statistical Methodology},
	publisher    = {Oxford University Press},
	volume       = 58,
	number       = 1,
	pages        = {267--288}
}

@article{fan2001variable,
	title        = {{Variable selection via nonconcave penalized likelihood and its oracle properties}},
	author       = {Fan, Jianqing and Li, Runze},
	year         = 2001,
	journal      = {Journal of the American Statistical Association},
	publisher    = {Taylor \& Francis},
	volume       = 96,
	number       = 456,
	pages        = {1348--1360}
}

@article{zhang2010nearly,
	title        = {{Nearly unbiased variable selection under minimax concave penalty}},
	author       = {Zhang, Cun Hui},
	year         = 2010,
	journal      = {Annals of Statistics},
	publisher    = {Institute of Mathematical Statistics},
	volume       = 38,
	number       = 2,
	pages        = {894--942}
}

@article{akaike1974new,
	title        = {{A new look at the statistical model identification}},
	author       = {Akaike, Hirotugu},
	year         = 1974,
	journal      = {IEEE Transactions on Automatic Control},
	publisher    = {Ieee},
	volume       = 19,
	number       = 6,
	pages        = {716--723}
}

@article{schwarz1978estimating,
	title        = {{Estimating the Dimension of a Model}},
	author       = {Schwarz, Gideon},
	year         = 1978,
	journal      = {Annals of Statistics},
	volume       = 6,
	number       = 2,
	pages        = {461--464}
}

@article{yuan2005combining,
	title        = {{Combining linear regression models: When and how?}},
	author       = {Yuan, Zheng and Yang, Yuhong},
	year         = 2005,
	journal      = {Journal of the American Statistical Association},
	publisher    = {Taylor \& Francis},
	volume       = 100,
	number       = 472,
	pages        = {1202--1214}
}

@article{peng2022improvability,
	title        = {{On improvability of model selection by model averaging}},
	author       = {Peng, Jingfu and Yang, Yuhong},
	year         = 2022,
	journal      = {Journal of Econometrics},
	publisher    = {Elsevier},
	volume       = 229,
	number       = 2,
	pages        = {246--262}
}

@article{boot2019forecasting,
	title        = {{Forecasting using random subspace methods}},
	author       = {Boot, Tom and Nibbering, Didier},
	year         = 2019,
	journal      = {Journal of Econometrics},
	publisher    = {Elsevier},
	volume       = 209,
	number       = 2,
	pages        = {391--406}
}

@article{elliott2013complete,
	title        = {{Complete subset regressions}},
	author       = {Elliott, Graham and Gargano, Antonio and Timmermann, Allan},
	year         = 2013,
	journal      = {Journal of Econometrics},
	publisher    = {Elsevier},
	volume       = 177,
	number       = 2,
	pages        = {357--373}
}

@article{hansen2007least,
	title        = {{Least squares model averaging}},
	author       = {Hansen, Bruce E},
	year         = 2007,
	journal      = {Econometrica},
	publisher    = {Wiley Online Library},
	volume       = 75,
	number       = 4,
	pages        = {1175--1189}
}

@article{liang2011optimal,
	title        = {{Optimal weight choice for frequentist model average estimators}},
	author       = {Liang, Hua and Zou, Guohua and Wan, Alan TK and Zhang, Xinyu},
	year         = 2011,
	journal      = {Journal of the American Statistical Association},
	publisher    = {Taylor \& Francis},
	volume       = 106,
	number       = 495,
	pages        = {1053--1066}
}

@article{zhang2019parsimonious,
	title        = {{Parsimonious model averaging with a diverging number of parameters}},
	author       = {Zhang, Xinyu and Zou, Guohua and Liang, Hua and Carroll, Raymond J},
	year         = 2019,
	journal      = {Journal of the American Statistical Association},
	publisher    = {Taylor \& Francis}
}

@article{breiman2001random,
	title        = {{Random forests}},
	author       = {Breiman, Leo},
	year         = 2001,
	journal      = {Machine Learning},
	publisher    = {Springer},
	volume       = 45,
	pages        = {5--32}
}

@article{wan2010least,
	title        = {{Least squares model averaging by Mallows criterion}},
	author       = {Wan, Alan TK and Zhang, Xinyu and Zou, Guohua},
	year         = 2010,
	journal      = {Journal of Econometrics},
	publisher    = {Elsevier},
	volume       = 156,
	number       = 2,
	pages        = {277--283}
}

@article{hansen2011model,
	title        = {{The model confidence set}},
	author       = {Hansen, Peter R and Lunde, Asger and Nason, James M},
	year         = 2011,
	journal      = {Econometrica},
	publisher    = {Wiley Online Library},
	volume       = 79,
	number       = 2,
	pages        = {453--497}
}

@manual{ishwaran2023fast,
	title        = {{Fast Unified Random Forests for Survival, Regression, and Classification (RF-SRC)}},
	author       = {H. Ishwaran and U.B. Kogalur},
	year         = 2023,
	publisher    = {manual},
	url          = {https://cran.r-project.org/package=randomForestSRC},
	note         = {R package version 3.2.3}
}

@article{feng2020taming,
	title        = {{Taming the factor zoo: A test of new factors}},
	author       = {Feng, Guanhao and Giglio, Stefano and Xiu, Dacheng},
	year         = 2020,
	journal      = {The Journal of Finance},
	publisher    = {Wiley Online Library},
	volume       = 75,
	number       = 3,
	pages        = {1327--1370}
}

@article{wan2024mining,
	title        = {{Mining the factor zoo: Estimation of latent factor models with sufficient proxies}},
	author       = {Wan, Runzhe and Li, Yingying and Lu, Wenbin and Song, Rui},
	year         = 2024,
	journal      = {Journal of Econometrics},
	publisher    = {Elsevier},
	volume       = 239,
	number       = 2,
	pages        = 105386
}

@article{hwang2022bayesian,
	title        = {{Bayesian selection of asset pricing factors using individual stocks}},
	author       = {Hwang, Soosung and Rubesam, Alexandre},
	year         = 2022,
	journal      = {Journal of Financial Econometrics},
	publisher    = {Oxford University Press},
	volume       = 20,
	number       = 4,
	pages        = {716--761}
}

@article{zhang2021new,
	title        = {A new study on asymptotic optimality of least squares model averaging},
	author       = {Zhang, Xinyu},
	year         = 2021,
	journal      = {Econometric Theory},
	publisher    = {Cambridge University Press},
	volume       = 37,
	number       = 2,
	pages        = {388--407}
}

@article{peng2024optimality,
	title        = {On optimality of Mallows model averaging},
	author       = {Peng, Jingfu and Li, Yang and Yang, Yuhong},
	year         = 2024,
	journal      = {Journal of the American Statistical Association},
	publisher    = {Taylor \& Francis},
	pages        = {1--12}
}

@article{nan2014variable,
	title        = {Variable selection diagnostics measures for high-dimensional regression},
	author       = {Nan, Ying and Yang, Yuhong},
	year         = 2014,
	journal      = {Journal of Computational and Graphical Statistics},
	publisher    = {Taylor \& Francis},
	volume       = 23,
	number       = 3,
	pages        = {636--656}
}

@article{zhang2016dominance,
	title        = {On the dominance of Mallows model averaging estimator over ordinary least squares estimator},
	author       = {Zhang, Xinyu and Ullah, Aman and Zhao, Shangwei},
	year         = 2016,
	journal      = {Economics Letters},
	publisher    = {Elsevier},
	volume       = 142,
	pages        = {69--73}
}

@article{ho1998random,
	title        = {The random subspace method for constructing decision forests},
	author       = {Ho, Tin Kam},
	year         = 1998,
	journal      = {IEEE Transactions on Pattern Analysis and Machine Intelligence},
	publisher    = {Ieee},
	volume       = 20,
	number       = 8,
	pages        = {832--844}
}

@article{fama1993common,
	title        = {Common risk factors in the returns on stocks and bonds},
	author       = {Fama, Eugene F and French, Kenneth R},
	year         = 1993,
	journal      = {Journal of Financial Economics},
	publisher    = {Elsevier},
	volume       = 33,
	number       = 1,
	pages        = {3--56}
}

@article{mclean2016does,
	title        = {Does academic research destroy stock return predictability?},
	author       = {McLean, R David and Pontiff, Jeffrey},
	year         = 2016,
	journal      = {The Journal of Finance},
	publisher    = {Wiley Online Library},
	volume       = 71,
	number       = 1,
	pages        = {5--32}
}

@article{fama2015five,
	title        = {A five-factor asset pricing model},
	author       = {Fama, Eugene F and French, Kenneth R},
	year         = 2015,
	journal      = {Journal of Financial Economics},
	publisher    = {Elsevier},
	volume       = 116,
	number       = 1,
	pages        = {1--22}
}

@article{gu2020empirical,
	title        = {Empirical asset pricing via machine learning},
	author       = {Gu, Shihao and Kelly, Bryan and Xiu, Dacheng},
	year         = 2020,
	journal      = {The Review of Financial Studies},
	publisher    = {Oxford University Press},
	volume       = 33,
	number       = 5,
	pages        = {2223--2273}
}

@article{cochrane2011presidential,
	title        = {Presidential address: Discount rates},
	author       = {Cochrane, John H},
	year         = 2011,
	journal      = {The Journal of Finance},
	publisher    = {Wiley Online Library},
	volume       = 66,
	number       = 4,
	pages        = {1047--1108},
	doi          = {10.1111/j.1540-6261.2011.01671.x}
}

@article{chen2021open,
	title        = {Open source cross-sectional asset pricing},
	author       = {Chen, Andrew Y and Zimmermann, Tom},
	year         = 2021,
	journal      = {Critical Finance Review},
	volume       = 10,
	number       = 1,
	pages        = {1--40}
}

@article{pelger2019understanding,
	title        = {Understanding Systematic Risk: A High-Frequency Approach},
	author       = {Pelger, Markus},
	year         = 2019,
	journal      = {The Journal of Finance}
}

@article{liew1976inequality,
	title        = {Inequality constrained least-squares estimation},
	author       = {Liew, Chong Kiew},
	year         = 1976,
	journal      = {Journal of the American Statistical Association},
	publisher    = {Taylor \& Francis},
	volume       = 71,
	number       = 355,
	pages        = {746--751}
}
